\documentclass[11pt,english]{article}
\usepackage[longnamesfirst]{natbib}
\usepackage{newpxtext,newpxmath}
\usepackage[T1]{fontenc}

\usepackage[bottom]{footmisc}
\usepackage{tabularx}
\usepackage{booktabs}
\usepackage{amsmath,enumerate}
\usepackage{amsthm}
\usepackage{dsfont}
\usepackage{amssymb}
\usepackage{mathtools}
\usepackage{mathrsfs}
\usepackage{graphicx}
\usepackage[labelfont=bf,justification=justified,font=footnotesize,skip=2pt,labelsep=period]{caption}
\usepackage{subcaption}
\usepackage{setspace}
\usepackage{arydshln}
\usepackage{chngpage}
\usepackage{float,appendix}
\usepackage{afterpage}
\usepackage{verbatim}
\usepackage{indentfirst}
\usepackage[left=1in, right=1in, top=1.5in, bottom=1.5in]{geometry}
\usepackage{hyperref}
\usepackage{adjustbox}
\usepackage{booktabs, makecell, longtable,ltxtable}
\usepackage{lipsum}
\usepackage{pdflscape}
\usepackage{siunitx}
\usepackage{chngcntr}
\usepackage{lipsum}
\usepackage{tikz}
\usepackage{pdflscape}
\usepackage[normalem]{ulem}
\usepackage{chngcntr,apptools}
\AtAppendix{\counterwithin{lemma}{section}}
\usepackage{siunitx}
\usepackage{tabularx}
\usepackage{xr}
\usepackage{etoolbox}
\usetikzlibrary{positioning}

\usepackage{tikz}
\usepackage{pdflscape}
\usepackage[normalem]{ulem}
\usepackage{chngcntr,apptools}
\AtAppendix{\counterwithin{lemma}{section}}
\usepackage{siunitx}
\usepackage{tabularx}
\usepackage{xr}
\usepackage{etoolbox}
\usetikzlibrary{positioning}
\usetikzlibrary{decorations.pathreplacing}
\usetikzlibrary{calligraphy}
\usepackage[shortlabels]{enumitem}
\usepackage{accents} 
\usepackage{comment}
\usepackage{tocloft}

\cftsetpnumwidth{1em} 
\cftsetrmarg{3em}       

\usepackage{etoolbox}
\robustify{\underaccent}
\usepackage{dutchcal} 
\usepackage{minitoc} 
\definecolor{UCLAblue}{RGB}{30, 75, 135} 
\definecolor{UCLAgold}{RGB}{255, 232, 0} 
\definecolor{usccardinal}{rgb}{0.6, 0.0, 0.0}
\definecolor{Matlabblue}{rgb}{0,0.447,0.741}
\definecolor{bleudefrance}{rgb}{0.19,0.55, 0.91}
\definecolor{cobalt}{rgb}{0.0, 0.28,0.67}
\definecolor{darkblue}{rgb}{0.0, 0.0,0.55}
\definecolor{darkcerulean}{rgb}{0.03,0.27, 0.49}
\definecolor{darkpowderblue}{rgb}{0.0,0.2, 0.6}
\definecolor{bleudefrance}{rgb}{0.19,0.55, 0.91}
\definecolor{darkblue}{rgb}{0.0, 0.2, 0.6}
\definecolor{black}{rgb}{0.0, 0.0, 0.0}

\hypersetup{
	colorlinks,
	linkcolor={black},
	citecolor={black},
	urlcolor={black},
}

\bibpunct{(}{)}{;}{a}{,}{,}
\newcommand\one{\mathds{1}}

\newcommand\R{\mathbb{R}}

\newcommand\Var{\mathbb{V}\text{ar}}

\renewcommand\implies{\Rightarrow}
\newcommand\E{\mathbb{E}}

\newcommand\diag{\text{diag}}
\newcommand\Diag{\text{Diag}}

\newcommand{\udot}[1]{\underaccent{\dot}{#1}}

\newcommand{\K}{\mathcal{K}}
\newcommand{\kk}{\mathbcal{k}}
\newcommand\grad{\nabla_{z_{k,t}}}

\newtheorem{prop}{Proposition}

\usepackage{pgfplots} 
\pgfplotsset{compat=newest} 
\pgfplotsset{plot coordinates/math parser=false} 

\usepackage{rotating}

\pgfplotsset{compat=newest} 
\pgfplotsset{plot coordinates/math parser=false} 
\newlength\figureheight 
\newlength\figurewidth 
\usepackage{mathtools}
\usetikzlibrary{calc}

\usepackage[T3,T1]{fontenc}
\DeclareSymbolFont{tipa}{T3}{cmr}{m}{n}
\DeclareMathAccent{\invbreve}{\mathalpha}{tipa}{16}

\usepackage[en-US]{datetime2}
\usepackage{mathtools}
\usepackage{xcolor}

\title{\textbf{The Elasticity of Quantitative Investment}\thanks{I thank Michael Weber, Joachim Freyberger, and Andreas Neuhierl for providing excellent data. I thank Lubos Pastor, Stefan Nagel, Ralph Koijen, Niels Gormsen for countless hours of guidance, advice, and help. I also thank Elena Asparouhova, Philip Bond, Brian Boyer, Jonathan Brogaard, Svetlana Bryzgalova, James Choi, Christopher Clayton, Roberto Gomez Cram, Victor DeMiguel, Doug Diamond, Karl Diether, Stefano Giglio, Lars Hansen, John Heaton, Christopher Hennessy, Christian Heyerdahl-Larsen, Craig Holden, Christopher Hrdlicka, Avner Kalay, Bryan Kelly, Steven Kou, Jiacui Li, Song Ma, Toby Moskowitz, Scott Robertson, Tyler Shumway, Stephan Siegel, Yang Song, Léa Stern, Noah Stoffman, Hao Xing, Anthony Zhang, and seminar participants at University of Utah, University of Washington, Yale, Purdue, London Business School, Boston University, Brigham Young University Economics Department, Brigham Young University Finance Department and Indiana University for their helpful comments and input. This paper was earlier circulated with the title "Machine Learning, Quantitative Portfolio Choice, and Mispricing." All remaining errors are my own. }}
\author{Carter Davis\thanks{Indiana University Kelley School of Business. \href{mailto:ckd1@iu.edu}{ckd1@iu.edu}}}
\date{Original: January 3, 2021 $\;\;\;$ Current: September 25, 2024} 

\begin{document}
{\fontfamily{qpl}\selectfont
\doparttoc 
\faketableofcontents 
\onehalfspacing
\setcounter{page}{1}

\newpage
\clearpage
\setcounter{section}{0}
\setcounter{table}{0}

\sloppy
\maketitle
\thispagestyle{empty}

    \begin{abstract}

    What is the demand elasticity of statistical arbitrageurs that invest according to the advice of modern cross-sectional asset pricing models? Thirteen models from the literature exhibit strikingly inelastic demand, in contrast to classical models that rely on statistical arbitrageurs to create elastic market demand for assets. This inelasticity arises from the difficulty of trading against price changes. A quantitative equilibrium model shows that aggregate demand remains inelastic even with these statistical arbitrageurs in the market. 
    
    \vspace{3mm} 
    \bigskip
    
  
  \end{abstract}

\maketitle
\thispagestyle{empty}

\newpage
\doublespacing
\setcounter{page}{1}
\setcounter{equation}{0}

How investors react to price changes---the price elasticity of demand---is a central quantity of interest in financial economics \citep{ky}. It is the arbitrageurs, or in the case of statistical alpha, it is the statistical arbitrageurs that trade aggressively against price fluctuations and generate a highly elastic demand curve for assets in classical models. In fact, the famous Arbitrage Pricing Theory model of \cite{apt} relies only on these statistical arbitrageurs to generate elastic demand, absorbing the rest of the flows of the market and generating the key pricing predictions.\footnote{See, for example, the \cite{huberman1982simple} version of the APT, and more generally, those reviewed in Section 4 of \cite{huberman2005arbitrage}.} The literature has uncovered many limits to arbitrage---frictions that prevent arbitrageurs from trading against price deviations \citep{limitsarb}. These limits, by keeping statistical arbitrageurs out of their classical flow-absorbing position in the market, can generate inelastic demand for the entire market. Elasticity is the inverse of price impact from flows and supply shocks, meaning that inelastic demand corresponds to large price impacts of trading flows \citep{GK}. Thus, with inelastic demand and concomitant high price impacts, flows induced from behavioral biases, sentiment, and other frictions can have significant pricing impacts. 

In this paper, I quantify the elasticity of statistical arbitrage strategies using portfolio choice methods from the academic literature. Statistical arbitrageurs are defined here as investors who aim to maximize expected return per unit of volatility, possibly with or without hedging out systematic risk. In this study, I examine a broad array of such strategies, and the results indicate that these specific methods consistently result in inelastic demand. I consider thirteen different statistical portfolio choice models from the literature \citep[e.g.,][]{ff3, brandt, kelly, shrinking} and calculate not only the portfolio weights from these models, but also how these models react to changing prices. 

The elasticity of demand matters economically. To illustrate, consider a simple case where a stock is priced efficiently based on its fundamentals. If an inflow equal to 10\% of shares outstanding occurs for non-fundamental reasons, the extent of overpricing depends critically on the demand elasticity. Calibrations of classic models result in high elasticities, e.g., 6,000 in \cite{petajisto} and the cited 5,000 in \cite{GK}. I provide a calibration of classic model elasticity and find a value-weighted average demand elasticity of about 900. With a demand elasticity of 900, the stock would be overpriced by only 0.01\% ($\approx (1 / 900) \times 10\%$) due to this inflow. However, if the elasticity is in the range of 0.3 and 1.6, which is the approximate literature range as surveyed by \cite{GK}, the stock could be overpriced by anywhere from 6\% ($\approx (1 / 1.6) \times 10\%$) to 33\% ($\approx (1 / 0.3) \times 10\%$). This example highlights the economic significance of elasticity: smaller elasticities lead to larger price distortions from trading flows. 

The extensive array of factor models in the academic literature---including those based on modern machine learning---provides a large set of statistical arbitrageur demand functions with measurable elasticities. While standard factor models that price the cross-section\footnote{I use the typical definition of a factor model that prices the cross-section: all test assets have zero alpha relative to this model.} are usually viewed through the lens of describing returns across assets, they can also be seen from a second perspective: as investment advice for Sharpe-maximizing investors (i.e., mean-variance investors).\footnote{A factor model prices the cross-section if and only if it delivers the highest possible Sharpe ratio. See the "agnostic interpretation" of factor models in \cite{famaprize}, based on the logic in \cite{hubermankandel}.} Although the first perspective is more common in academic papers, the second is well-established. A rich literature explores factor models through this portfolio optimization lens. For example, \cite{pastorjmp}, \cite{brandt}, and \cite{demiguel} study methods for computing optimal factor weights from a portfolio choice perspective. More recently, machine learning-based factor models have been developed to create statistical arbitrage strategies designed to maximize the Sharpe ratio out-of-sample using a wide range of input variables \citep[e.g.,][]{gu, forest}. In this paper, I measure the price elasticity of many of these statistical arbitrageur demand functions.

Measuring demand elasticity requires assessing how price changes affect a large set of input variables and how these variables influence demand. The asset pricing literature has identified a wide range of asset characteristics, referred to as predictors here, which serve as input variables for statistical arbitrageur models. The large set of portfolios formed based on these predictors is often known as the factor zoo and has been discussed at length in the literature \citep{cochrane, harvey, taming}. Predictors fall into two categories: those that are a direct function of price (e.g., price ratios like book-to-market and earnings-to-market) and those independent of price \citep[e.g.,][profitability and investment are purely accounting-based]{famafrench15}. To calculate a statistical arbitrageur’s elasticity, one must assess how a 1\% price change affects predictors, and subsequently how these changes influence portfolio weights. I use the 62 predictors from \cite{weber}, though some models, such as \cite{famafrench15} and \cite{zhang}, rely on smaller subsets. While portfolio weights are often formed using discontinuous functions of predictors \citep[e.g.,][]{famafrench15}, I use differentiable approximations of portfolio weight functions as in \cite{kelly} and \cite{shrinking} to calculate price responsiveness with derivatives.

Applying this methodology, the main finding of this paper is that value-weighted statistical arbitrageur demand elasticities are typically below 10. Pure-alpha trading strategies exhibit similarly inelastic demand. Estimates of systematic elasticities---representing demand responses to correlated price changes---are even more inelastic, ranging from approximately $-1$ to 1. Finally, I introduce a simple demand system asset pricing model that incorporates the demand functions of these statistical arbitrageurs and demonstrate that aggregate demand elasticity for individual stocks changes very little when these model-based arbitrageurs manage capital counterfactually.

The reason for the inelastic demand of these methods is straightforward: trading against stock-level mispricing is risky and difficult. While this concept is known, quantifying it using methods from the literature is novel. Statistical arbitrage strategies invest based on historical performance and, unlike many classical models, have fewer exogenous, fixed, and known components; often, nearly all aspects of their investment rules are modeled with uncertainty \citep[e.g.,][]{shrinking}. The inelastic demand is not due to a lack of alpha—these statistical arbitrageur models actually generate relatively high alpha, with much of the elastic response stemming from the pure-alpha component. However, even with perceived alpha, significant uncertainties about hedging strategies and future price movements make trading challenging, resulting in a relatively inelastic demand response. In an elastic-APT world, a price drop and increase in alpha would prompt aggressive trading as if it were a near-arbitrage opportunity. In reality, however, these situations involve substantial risk and uncertainty, leading arbitrageurs to respond more cautiously. Consequently, these models provide an upper bound on the extent to which statistical arbitrageurs can counter price movements, even before accounting for common market frictions \citep[e.g.,][]{limitsarb}.

Counterfactual experiments reveal that statistical arbitrageurs have little effect on the elasticity of aggregate stock-level demand and fail to arbitrage away alpha, consistent with observed alpha, inelastic demand, and the lack of hyperelastic investors \citep{koijen2023investors}. While model-based statistical arbitrageurs may differ from their real-world counterparts, this paper highlights the significant gap between classic models predicting highly elastic demand and portfolio choice models that produce inelastic demand.

The paper is laid out as follows. Section \ref{sec:data} describes the data, data normalization, and notation. Section \ref{sec:elas_models} gives a calibration of demand elasticity in classic models. Section \ref{sec:quant_elas} describes the statistical arbitrage investment strategies, their out-of-sample performance, and their elasticities. Section \ref{sec:variations} presents various elasticity results, including pure-alpha strategy elasticities and systematic elasticities. Section \ref{sec:model} describes a simple demand system asset pricing model that can accommodate these statistical arbitrageur models, and shows that aggregate demand for individual assets remains inelastic when these statistical arbitrageurs manage capital counterfactually. In Section \ref{sec:conclusion}, I conclude. 

\section{Data and Notation} \label{sec:data}

In this section, I describe the data, data normalization, and notation. Below I describe the returns, return predictors, and holdings data. 


I use quarterly institutional long-only holdings data from the SEC 13F filings dataset. See \cite{ky} (KY hereafter) for more details about these data. The sample period is from 1984 to 2020. The 13F holdings often do not account for all shares outstanding, so following KY, I create a "household" investor that holds the residual shares.

I use the monthly equity returns and stock predictors data from \cite{weber}, which contains 62 predictors (i.e., asset characteristics). A thorough description of the variables can be found in \cite{weber}. These data were extended from 2014 to the end of 2020 by \cite{babayara}. The sample starts in July 1967. The 13F filing data are merged with the predictors data. These data are merged with the CRSP monthly returns data, which contain changes in shares outstanding, returns, and dividends. Finally, I merge these data with the standard market portfolio excess return and risk-free rate data from Kenneth French's website. The risk-free rate is the standard one-month Treasury bill rate. 

\subsection{Notation and Data Preparation}

Let $P_{i,t}$ denote the market equity of stock $i$ at time $t$, while $p_{i,t}$ denotes the share price of the asset. Therefore, $P_{i,t} = S_{i,t} p_{i,t}$, where $S_{i,t}$ denotes the number of shares outstanding. Let $N_t$ be the number of stocks in a given period, though the subscript is sometimes dropped for simplicity. 

I take the raw predictor $k$ for stock $i$ at time $t$, denoted $x_{i,k,t}$, and normalize it as shown in Figure \ref{fig:demand}. For example, $x_{i,k,t}$ could be the book-to-market or profitability ratio in \cite{weber}. As \cite{shrinking} and \cite{kelly} discuss, it is important to normalize the predictor into a well-behaved predictor, $z_{i,k,t}$, for statistical arbitrageurs. 

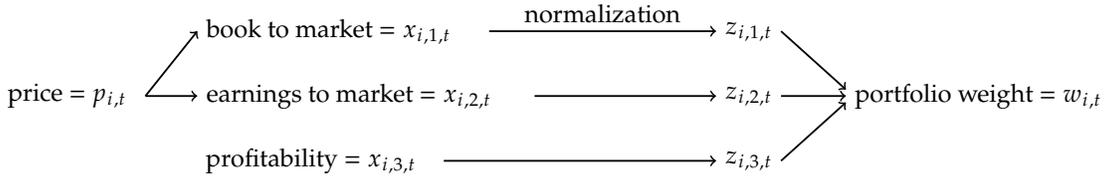
\begin{figure}[!t] \centering
    \resizebox{0.9\textwidth}{!}{\begin{tikzpicture}
\node[text centered] at (0,0) {price $= p_{i,t}$};
\draw[->,thick] (1.2, 0) -- (2, 1) node [right] {book to market $= x_{i,1,t}$};
\draw[->,thick] (1.2, 0) -- (2, 0) node [right] {earnings to market $= x_{i,2,t}$};
\node[align=left] at (2,-1) [right] {profitability $= x_{i,3,t}$};
\draw[->,thick] (6.5, 1) -- (10, 1) node [right] {$z_{i,1,t}$};
\node[text centered, above] at (8.25, 1) {normalization};
\draw[->,thick] (7.2, 0) -- (10, 0) node [right] {$z_{i,2,t}$};
\draw[->,thick] (5.8, -1) -- (10, -1) node [right] {$z_{i,3,t}$};
\draw[->,thick] (11, 1) -- (12, 0.1);
\draw[->,thick] (11, 0) -- (12, 0) node [right] {portfolio weight $= w_{i,t}$};
\draw[->,thick] (11, -1) -- (12, -0.1);
\end{tikzpicture}}
    \vspace{4mm}
    \caption{\textbf{Demand Mapping.} Example of demand for statistical arbitrageurs as a function of prices and predictors. Stock prices ($p_{i,t}$) use three example predictors: (1) book-to-market ratio, (2) earnings-to-market ratio, and (3) profitability ratio. Profitability (per \cite{famafrench15}) is not a function of prices. Predictors map to normalized asset predictors ($z_{i,k,t}$), with portfolio weights based on these. }
    \label{fig:demand}
\end{figure}




Transforming the raw predictors, $x_{i,k,t}$, into the normalized predictors, $\udot{z}_{i,k,t}$, allows me to form predictor-weighted portfolios. The dot underneath $z_{i,k,t}$ denotes this specific normalization discussed below. As \cite{shrinking} and \cite{kelly} discuss, using portfolios weighted with these cross-sectional percentiles produces clean long-short portfolios with a long tradition in asset pricing, while using portfolios based directly on raw predictors, $x_{i,k,t}$, is full of noise. 



Let $x_{k,t}$ be the $N$-dimensional vector of the raw predictor $k$ known at time $t$ that is filled with the raw predictor $x_{i,k,t}$ values. Let $\udot{z}_{k,t}$ and $\udot{z}_{i,k,t}$ be the corresponding normalized predictors. The standard normalization, used for example in \cite{shrinking} and \cite{kelly}, is to transform the raw predictors to be:
\begin{equation} \label{eq:z_transform}
    \udot{z}_{i,k,t} = \frac{\text{rank} (x_{k,t})_i-1}{N-1} - 0.5,
\end{equation}
where $\text{rank}(\cdot)$ is a function that reads in a vector and outputs a vector of values $1$ through $N$ corresponding to the relative rank of the input values ($N$ is assigned the largest value, $1$ the smallest), and $\text{rank}(\cdot)_i$ is the $i^{th}$ element of this output vector. This transforms predictors to be in the range of $[-0.5, 0.5]$. For example, if an asset has a book-to-market ratio in the twentieth percentile in the cross-section of assets, then the resulting normalized value would be $-0.3$ ($= 0.2 - 0.5$).

This is a discontinuous function of the predictor $x_{i,k,t}$. If $x_{i,k,t}$ is the book-to-market ratio, for example, then any demand function that is a function of $x_{i,k,t}$ is then a discontinuous function of the price. This complicates elasticity estimations. In other words, every step between prices and portfolio weights in Figure \ref{fig:demand} should be continuous, including having $\udot{z}_{i,k,t}$ be a continuous function of $x_{i,k,t}$. This allows the calculation of $\partial \udot{z}_{i,k,t} / \partial \log (p_{i,t})$ for each predictor (see Appendix \ref{subsec:kernel} for the closed-form formula for these derivatives). Thus, I create a continuous version of this normalization, defined as
\begin{equation} \label{eq:z_weights}
    \udot{z}_{i,k,t} = \K(\text{ArcSinh} (x_{k,t}))_i - 0.5,
\end{equation}
where $\text{ArcSinh} (\cdot)$ is the standard inverse hyperbolic sine function and $\K(\cdot)$ is the cumulative distribution function (cdf) of a standard kernel density function. Similar to above, $\K(\cdot)_i$ is the $i^{th}$ element of $\K(\cdot)$. I use the standard well-known normal cdf kernel function, described in detail in Appendix \ref{subsec:kernel}. 

This transformation above yields very similar results to its discontinuous version, but this function is of course continuous. The $\text{ArcSinh} (\cdot)$ function is used because it is similar to the log function in that it shrinks extreme values towards zero, but $\text{ArcSinh} (\cdot)$ is defined at zero and negative input values.\footnote{In fact, the approximation $\text{ArcSinh} (x) \approx \log(x) + \log(2)$ holds for large values of $x$.} Since many predictors can have negative values, this is necessary. The final kernel minus one half function is necessary to transform the value into the $[-0.5, 0.5]$ interval. 

\subsection{Forming Linear Predictor-Based Factor Portfolios}

Some models use linear predictor-weighted portfolios, defined here. Following \cite{shrinking}, I rescale these predictors in every period to avoid leverage in the predictor-weighted portfolios fluctuating across time with the number of assets available. Note that the sum of the absolute values of the predictors in any given period is about $N / 4$ (i.e., $\sum_{i=1}^N | \udot{z}_{i,k,t} | \approx N / 4$). For any $x$, I define $\breve x \equiv (4 x) / N$ and $\invbreve x \equiv (N x) / 4$. Thus, I can write $\sum_{i=1}^N | \breve{\udot{z}}_{i,k,t} | \approx 1$. This allows me to define predictor-weighted portfolio returns. Let $F_{k,t+1}^c$ be the zero-cost portfolio return associated with predictor $k$:\footnote{Note that these are zero-cost portfolios even if the weights were all positive because $r_{i,t+1}$ are excess returns, not just returns. However, since the weights, before being scaled down, fall evenly between $-0.5$ and $0.5$ (i.e., weights sum to zero), these are classic long-short predictor weighted portfolios (see \cite{shrinking}).} 
\begin{equation} \label{eq:port_returns}
F_{k,t+1}^c = \sum_{i=1}^N \breve{\udot{z}}_{i,k,t} r_{i,t+1}.
\end{equation}

Factor models, such as \cite{ff3}, require a market portfolio return as well. Let $A_t^j$ be the AUM of institution $j$, and $A_t \equiv \sum_j A_t^j$.\footnote{Note that by this definition, $A_t = \sum_{i} P_{i,t}$. As in KY and discussed below, a residual institution is added so that institutions collectively hold all assets in the data. } I let the first predictor, $k = 1$, be a market portfolio return weight: $\udot{z}_{i,1,t} \equiv \invbreve P_{i,t} / A_t$. Note that this is defined in terms of $\invbreve P_{i,t}$ and not $P_{i,t}$, so that $\sum_{i=1}^N | \breve{\udot{z}}_{i,1,t} | \approx 1$, which matches the other portfolios.

\subsection{Two Types of Predictors}

Following KY, I separate predictors into two categories: (1) those that are a function of price and (2) those where price is not needed or used to calculate the predictor. I follow KY by assuming this latter type of predictors are exogenous to prices, such as historical beta, book values, asset values, investment, and profitability. 
Fifteen of the 62 predictors are a function of price. In Appendix \ref{subsec:endogenous}, I describe the classification of the predictors into these two types in more detail. 

Some predictors from the \cite{weber} data use historical prices, instead of contemporaneous prices, to create price ratio variables. For example, the book-to-market ratio is computed with stale market capitalization values in order to better match the period when the accounting book values were released. I follow KY by instead using the most recent prices to calculate these standard valuation ratios. Using historical prices instead of current prices, of course, would make the statistical arbitrageurs in this paper even less price elastic. However, this seems to violate the spirit of standard valuation ratios. Basing investment decisions on valuation ratios should make investors sensitive to prices, which is why I use contemporaneous prices. As is standard with the CRSP data, the monthly price is the closing price on the last trading day of the month. 

The most difficult predictors to sort into these two categories are the maximum daily return over the previous month and momentum. The maximum daily return is sometimes a function of today's price if the maximum daily return was the most recent day. However, if this is classified into the category of being a function of price, then this is a discontinuous function. Since this is a historical daily flow and not a stock, this classification has little effect on the measured elasticity. Similarly, short-term reversals are classified as being a function of price, but momentum is not. Momentum is defined as the return over the past year, excluding the most recent month. This means it is not a function of price, while short-term reversals are a function of price, given that it is the return over the previous month. If these variables are dropped from the dataset, the measured elasticity results are similar. I keep these variables in the data because it is simple, does not significantly change the elasticity results, and keeps the set of predictors the same as \cite{weber}. 

\subsection{Price Sensitivity of Individual Factor Investment}

Table \ref{tab:gradz} shows the 15 endogenous predictors, and importantly the median $\partial \udot{z}_{i,k,t} / \partial \log (p_{i,t})$ derivative values for each predictor $k$ across assets and months. Some predictors with the price in the numerator, like Tobin's q, tend to have a positive derivative. Other predictors, like the book-to-market ratio, have a price term in the denominator; thus, when the price increases, the predictor tends to decline (negative derivative). The 47 exogenous predictors, of course, have zero derivatives. 


\begin{table}[!t] \centering
    \resizebox{1.\textwidth}{!}{\begin{tabular}{@{\extracolsep{4pt}}lcc@{\hskip 7.5mm}lcc@{\hskip 7.5mm}lc}
\toprule
 & $\frac{\partial \udot{z}_{i,k,t}}{\partial \log (p_{i,t})}$
 & & &  $\frac{\partial \udot{z}_{i,k,t}}{\partial \log (p_{i,t})}$
 & & & $\frac{\partial \udot{z}_{i,k,t}}{\partial \log (p_{i,t})}$ \\
\\[-1.8ex]
\cline{2-2} \cline{5-5} \cline{8-8}\\[-1.8ex]
Covariates & Gradient & & Covariates & Gradient & & Covariates & Gradient \\
\hline \\[-1.8ex]
a2me &  $-0.273$  & & e2p &  $-0.248$  & & o2p &  $-0.125$ \\
 & ($-0.422$,  $-0.082$) & &  & ($-0.534$,  $0.056$) & &  & ($-0.275$,  $-0.000$)\\
\\[-1.8ex]
beme &  $-0.374$  & & ldp &  $<10^{-10}$  & & q &  $0.321$ \\
 & ($-0.551$,  $-0.083$) & &  & ($-0.296$,  $-0.000$) & &  & ($0.076$,  $0.833$)\\
\\[-1.8ex]
beme\_adj &  $-0.343$  & & lme &  $0.147$  & & rel\_to\_high\_price &  $1.28$ \\
 & ($-0.691$,  $-0.064$) & &  & ($0.057$,  $0.184$) & &  & ($0.287$,  $3.070$)\\
\\[-1.8ex]
cum\_return\_1\_0 &  $2.93$  & & lme\_adj &  $0.037$  & & roc &  $0.161$ \\
 & ($0.504$,  $5.114$) & &  & ($0.003$,  $0.192$) & &  & ($0.019$,  $0.662$)\\
\\[-1.8ex]
debt2p &  $-0.174$  & & nop &  $-0.061$  & & s2p &  $-0.274$ \\
 & ($-0.258$,  $-0.000$) & &  & ($-0.248$,  $0.072$) & &  & ($-0.389$,  $-0.062$)\\
\hline \\[-1.8ex]
 Obs. &  1,869,963 & & & 1,869,963 & & & 1,869,963 \\
 Months & 703 & & & 703 & & & 703 \\
\bottomrule
\end{tabular}}
    \vspace{4mm}
    \caption{\textbf{Predictor Price Sensitivity.} Median (across assets and months) of $\partial \udot{z}_{i,k,t} / \partial \log (p_{i,t})$, showing how predictors change with a 1\% increase in asset prices. 10$^{th}$ and 90$^{th}$ percentiles are in parentheses (not confidence intervals). Variable descriptions from \cite{weber}; briefly: a2me (asset to market equity), beme (book-to-market), beme\_adj (book-to-market minus industry mean), cum\_return\_1\_0 (short-term reversals), debt2p (debt to price), e2p (earnings to price), ldp (dividend to price), lme (size), lme\_adj (size minus industry size), nop (net payout to price), o2p (payout to price), q (Tobin’s q), rel\_to\_high\_price (52-week high ratio), roc (rents over cash), s2p (sales to price).}
    \label{tab:gradz}
\end{table}

\section{Theoretical Elasticity of Statistical Arbitrageurs} \label{sec:elas_models}

To provide context for the empirical results below, I provide a calibration of the theoretical elasticity of standard statistical arbitrage investors in this section.

Let $a_t$ be the assets under management (AUM) if the statistical arbitrageur is a fund, or the wealth if the arbitrageur is an individual investor. Let $p_{i,t}$ denote the share price of asset $i$, and $w_{i,t}$ denote the arbitrageur's portfolio weight for the asset. Then $s_{i,t} = a_t w_{i,t} / p_{i,t}$ denotes the number of shares that the arbitrageur demands. The elasticity is the percentage change in quantity demanded, $s_{i,t}$, due to a 1\% price change, ceteris paribus. An elasticity of 3.5 implies that a 1\% price increase would lead to a 3.5\% decrease in shares demanded. The elasticity is defined as:
\begin{equation} \label{eq:elasticity}
    \eta_{i,t} \equiv - \frac{\partial \log (s_{i,t})}{\partial \log (p_{i,t})} = 1 - \underbrace{\frac{\partial \log(w_{i,t})}{\partial \log (p_{i,t})}}_{\text{weight effects}}
    - \underbrace{\frac{\partial \log(a_t)}{\partial \log (p_{i,t})}}_{\text{wealth effects}}.
\end{equation}
It turns out that wealth effect tends to be approximately $w_{i,t-1}$, which tends to be quite small for the diversified models in this paper and can well-approximated with zero.\footnote{Appendix \ref{app:wealth} derives these wealth effects and shows that these wealth effects are small empirically.} With zero wealth effects, we can write:
\begin{equation} \label{eq:simple_elasticity}
     \eta_{i,t} = 1 - \frac{1}{w_{i,t}} \left( \frac{\partial w_{i,t}}{\partial \log (p_{i,t})} \right).
\end{equation}

As KY point out, value-weighted index funds are designed to avoid rebalancing for price fluctuations, and are thus perfectly inelastic. This can be seen using equation (\ref{eq:elasticity}), since a 1\% price rise leads to a 1\% increase in the asset's portfolio weights. Thus, $\partial \log (w_i) / \partial \log(p_i) = 1$, and $\eta_i = 0$ ($=1-1$).

\subsection{Elasticity of Classic Models: A Simple Calibration Exercise} \label{subsec:calibration}

Consider an investor with CARA utility that chooses a vector of portfolio weights $w_t$ at time $t$ of the $N$ risky assets. The investor chooses $w_t$ to maximize utility:
\begin{equation} \label{eq:pass_through_estimate}
    \E_t \left[ -\exp \left( -\gamma a_t (w_t' r_{t+1} + R_{f,t}) \right) \right],
\end{equation}
where $\gamma$ is the CARA risk-aversion coefficient. The classic demand, with multivariate normal returns, is then:
\begin{equation} \label{eq:CARA}
    w_t = \frac{1}{\gamma a_t} \Sigma_t^{-1} \mu_t,
\end{equation}
where $\Sigma_t = \Var_t (r_{t+1})$ is the $N \times N$ covariance matrix, $\mu_t = \E_t [r_{t+1}]$ is the vector of expected excess returns. 

Consider the following regression of the excess return of asset $i$ on the excess returns of the $N-1$ other assets, denoted by $r_{-i,t+1}$:
\begin{equation} \label{eq:regression_stevens}
r_{i,t+1} = \beta^0_{i,t} + \beta_{i,t}' r_{-i,t+1} + \epsilon_{i,t+1},
\end{equation}
where $\beta^0_{i,t}$ is the scalar intercept, $\beta_{i,t}$ is the $(N-1)$-dimensional vector of slope coefficients, and $\epsilon_{i,t+1}$ is the zero-mean residual with conditional variance $\sigma_{i,t,\epsilon}^2$. Taking the conditional expectation of both sides yields:
\begin{equation*}
\mu_{i,t} = \beta^0_{i,t} + \beta_{i,t}' \mu_{-i,t}.
\end{equation*}

As \cite{stevens1998inverse} showed, we can rewrite the demand from equation (\ref{eq:CARA}) for asset $i$ as:
\begin{equation} \label{eq:cara_stevens}
    w_{i,t} = \frac{1}{\gamma a_t} \left( \frac{\beta^0_{i,t}}{\sigma_{i,t,\epsilon}^2} \right)
    = \frac{1}{\gamma a_t} \left( \frac{\mu_{i,t} - \beta_{i,t}' \mu_{-i,t}}{\sigma_{i,t,\epsilon}^2} \right).
\end{equation}

I start by assuming an exogenous covariance matrix with price changes affecting only $\mu_{i,t}$. This assumption will be relaxed later. Under this assumption, we have:
\begin{equation} \label{eq:eta_mu}
    \eta_{i,t} = 1 + \frac{1}{w_{i,t}} \frac{\partial w_{i,t}}{\partial \mu_{i,t}} \left( -\frac{\partial \mu_{i,t}}{\partial \log (p_{i,t})} \right)
    = 1 + \frac{1}{w_{i,t}} \left( \frac{1}{\gamma a_t \sigma_{i,t,\epsilon}^2} \right) \left( -\frac{\partial \mu_{i,t}}{\partial \log (p_{i,t})} \right),
\end{equation}
where the first equation uses the chain rule and the second uses equation (\ref{eq:cara_stevens}) to calculate $\partial w_{i,t} / \partial \mu_{i,t}$. I use this key equation to calibrate the elasticity of a statistical arbitrageur, and discuss the key inputs below. 

First, I discuss the last term, $-\partial \mu_{i,t} / \partial \log (p_{i,t})$. This measures how a 1\% change in price affects the expected return. To calibrate this, I run a panel regression of excess returns on an intercept and predictor variables $\udot{z}_{i,k,t}$. The resulting derivative is then given by:
\begin{equation*}
    -\frac{\partial \mu_{i,t} }{ \partial \log (p_{i,t})}
    = -\sum_k b_k \frac{\partial \udot{z}_{i,k,t}}{\partial \log (p_{i,t})}, 
\end{equation*}
where $b_k$ are the slope coefficients of the regression. The standard error is calculated with the delta method from the coefficient standard errors, which are double clustered at the stock and month level. The estimate of the average derivative is $0.048$, with a standard error of $0.008$. I simply use 0.048 in the calibration. This means that a price drop of 1\% is predicted to increase expected returns by 4.8 basis points (bps) the next month. This is consistent with \cite{whyinelastic}, who discuss this derivative and various estimates of it in detail. 

Second, we need an estimate for residual variance, $\sigma_{i,t,\epsilon}^2$. To do this, I fit a simple model for the covariance matrix, $\Sigma_t = \beta_t \sigma_{M}^2 \beta_t' + \sigma_I^2 I$, where $\beta_t$ is the vector of ex ante market betas in the \cite{weber} data, $\sigma_{M}^2$ is the variance of the market, $\sigma_I^2$ is the average residual variance, and $I$ is the identity matrix. I can use the \cite{stevens1998inverse} formula to solve for the residual volatility term $\sigma_{i,t,\epsilon}$. This yields an average residual volatility of 17.2\% per month. Note that if a factor model with even more factors is used for the covariance matrix, this further decreases the residual volatility and increases the elasticity. Thus, this single factor covariance structure produces conservative elasticity estimates.

Finally, I simply use market weights for $w_{i,t}$. It should be noted that it does not matter what value of $\gamma a_t$ is chosen, because whatever $\gamma a_t$ is chosen is completely offset by the portfolio weights $w_{i,t}$.\footnote{By equation (\ref{eq:CARA}), the sum of portfolio weights is $\iota' w_t = (\iota' \Sigma_t^{-1} \mu_t) / (\gamma a_t)$, where $\iota$ is a vector of ones. If we use market portfolio weights, then $\mu_t$ must be proportional to $\Sigma_t w_t$. I set the average $\mu_t$ term to average excess return.} For example, if $\gamma a_t$ is doubled, then risky asset portfolio weights $w_{i,t}$ are halved, and the resulting elasticity is unchanged. Without loss of generality, I use the $\gamma a_t$ value that is consistent with risky asset weights summing to one, with a median $\gamma a_t$ value of 6.08.\footnote{The average is 6.79.} We can plug these values into equation (\ref{eq:eta_mu}) to calculate:
\begin{equation*}
    \eta_{i,t} \approx 1 + \left( \frac{1}{6.08 (0.172)^2} \right) \frac{1}{w_{i,t}} \left( 0.048 \right) \approx 1 + \frac{0.27}{w_{i,t}}
\end{equation*}

This calibration shows that investor demand is relatively elastic for most assets. The average number of stocks per period in the main sample, which corresponds to the out-of-sample period of the 360 months from February 1990 to January 2020 for the results shown below, is 3,481. A stock with an average portfolio weight of $1 / $3,481 would have an elasticity of about 900. I show in Figure \ref{fig:calibration_and_weights} a binned scatter plot of elasticities in blue as a function of portfolio weights in Panel A, and a histogram of these portfolio weights in Panel B. The elasticities range from 10 for the largest stocks to larger than 1,000,000 for the smallest stocks. 

\begin{figure}[!t]
    \centering
    \begin{minipage}{0.49\textwidth}
        \centering
        \textbf{Panel A: Calibration Elasticity}
        \includegraphics[width=\linewidth,trim={0 0 0 1cm},clip]{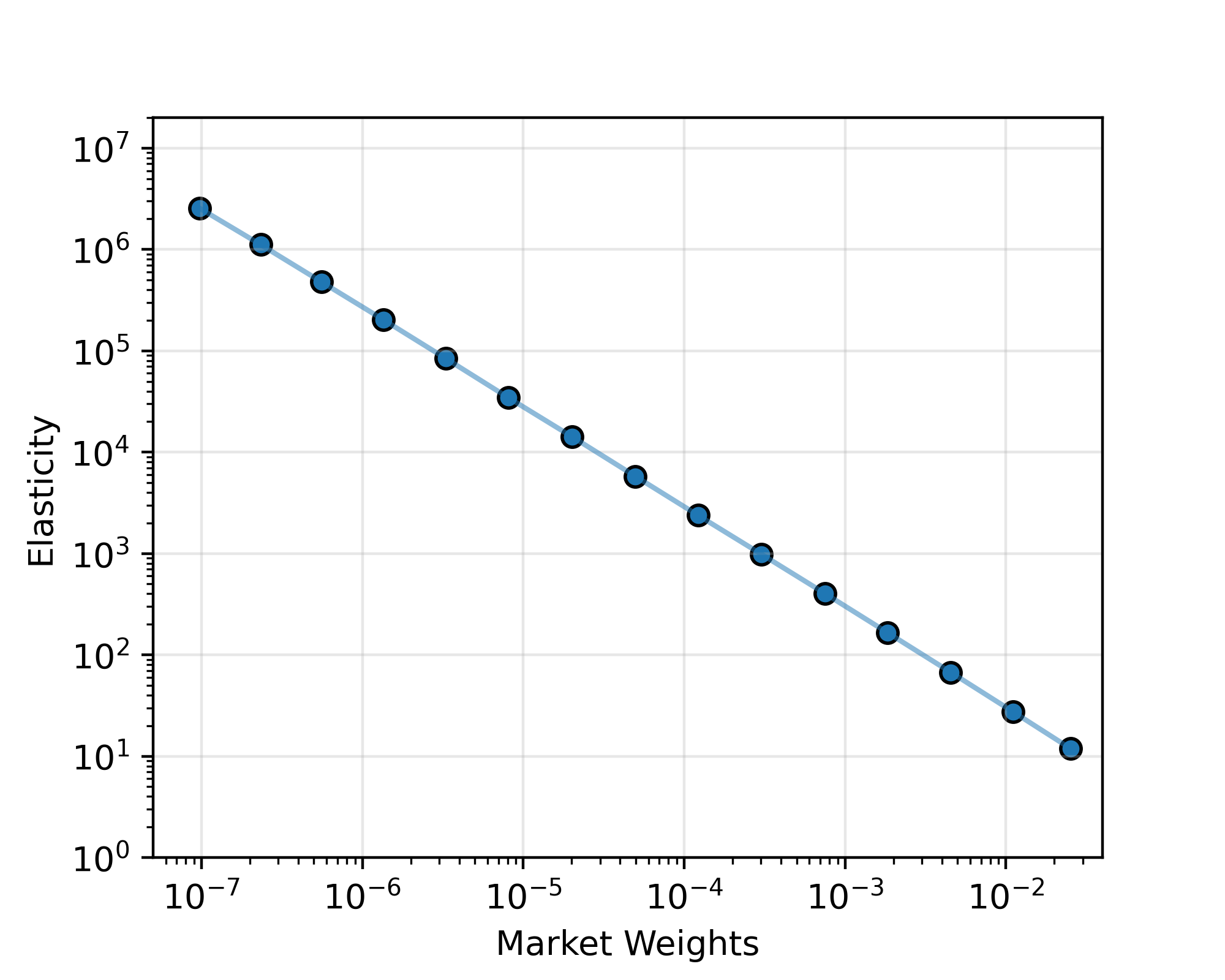}
    \end{minipage}\hfill
    \begin{minipage}{0.49\textwidth}
        \centering
        \textbf{Panel B: Histogram of Market Weights}
        \includegraphics[width=\linewidth,trim={0 0 0 1cm},clip]{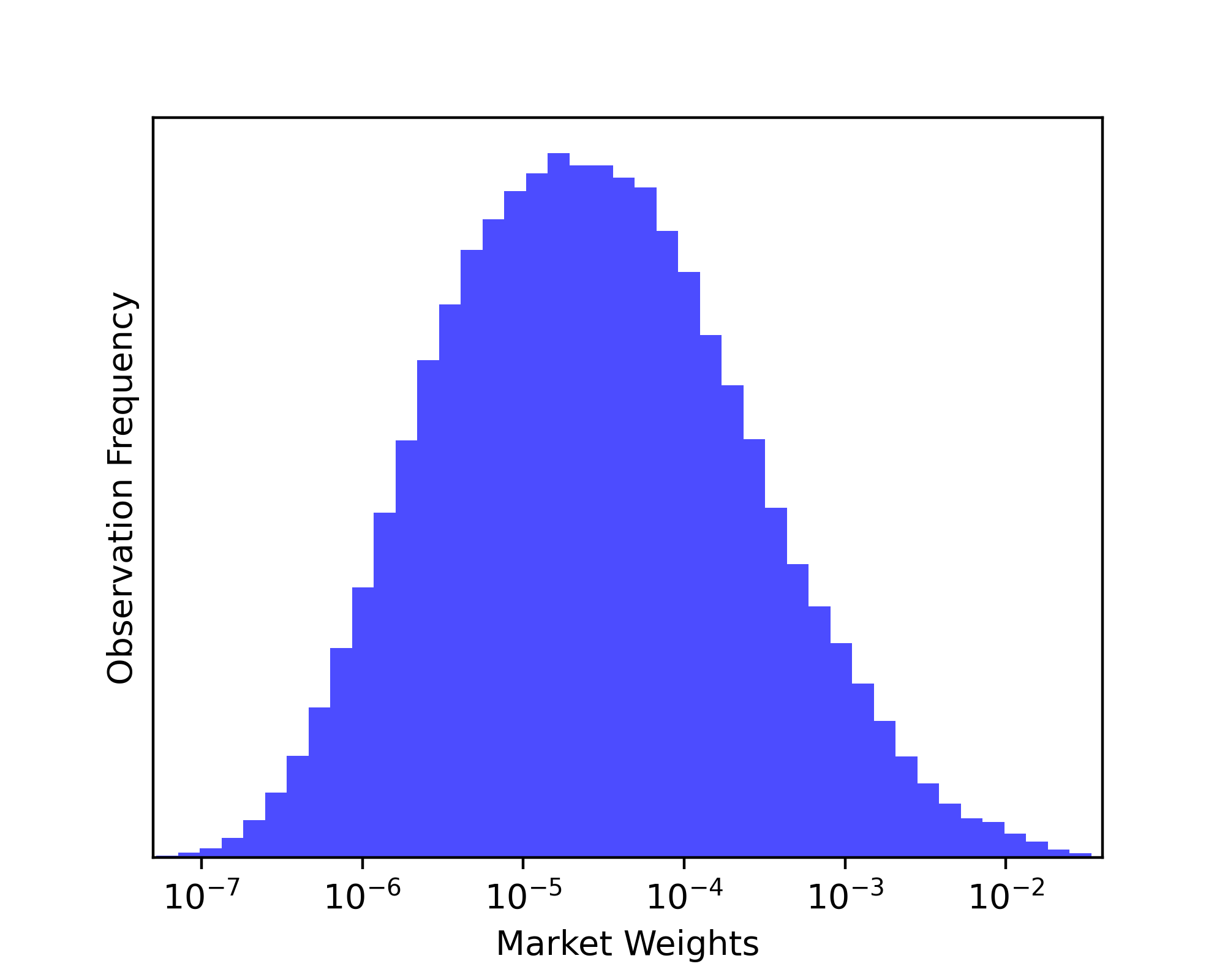}
    \end{minipage}
    \vspace{4mm}
    \caption{\textbf{Elasticity and Market Weights.} Panel A: Binned scatter plot of calibrated elasticities vs. market portfolio weights (both axes in log-scale). Panel B: Histogram of market portfolio weights (log-scale $x$-axis). Based on 1,253,175 stock-month observations from February 1990 to January 2020.}
\label{fig:calibration_and_weights}
\end{figure}


Due to the presence of \( w_{i,t} \) in the denominator of the elasticity formula (see Equation (\ref{eq:eta_mu})), assets with \( w_{i,t} \) values close to zero exhibit extremely high elasticity (approaching infinity as \( w_{i,t} \) approaches zero). This is logical since elasticity is defined in terms of percentage changes. An investor with a small initial position in an asset who significantly increases their holdings following a price drop will appear highly elastic. This explains why there is such a wide range of elasticities. A demand function with stable sensitivity to dollar changes (i.e., $\sigma_{i,t,\epsilon}^2$ and $\partial \mu_{i,t} / \partial \log(p_{i,t})$ do not vary much across assets) and varying position sizes will exhibit a range of elasticities. Although it may seem unusual for elasticity to vary so much based on the inverse of portfolio weights, this variation is a standard feature in both utility-based models and statistical arbitrageur models, as can be seen in Equation (\ref{eq:eta_mu}). 


This presents a challenge of finding a sensible way of averaging such extreme elasticities. One natural way is to value-weight the averages, with weights proportional to the market capitalization of the asset. Nearly all elasticities presented in this paper are value-weighted, unless explicitly stated. Another natural way to aggregate the elasticities is a portfolio-weighted approach, with weights proportional to the dollars invested across assets within an investor's portfolio. This is natural because the elasticity of the aggregate market for an asset is close to a value-weighted average elasticity across investors for that asset, with weights proportional to the dollars invested across investors.\footnote{Proposition \ref{prop:aggregation} states this simple result, which includes a correction for short-sellers.} Note that in the case of the current calibration we consider here, portfolio-weighting and value-weighting are equivalent since the portfolio weights of the investors are the same as the value weights of the market. I also consider various winsorization approaches, in order to look at the middle bulk of the distribution rather than an average that is disproportionately affected by large values due to the $1 / w_{i,t}$ term where weights are close to zero. In order to pair winsorization with weighting, I multiply by the weights that sum to one, apply the winsorization to the re-weighted sample, and sum the resulting terms, which yields a cross-sectional average each period. Then I average across periods.\footnote{Appendix \ref{subsec:winsorization} discusses this winsorization procedure.}

Table \ref{tab:calibration_table} shows the results of this elasticity calibration under different weighting and winsorization schemes, and demand is relatively elastic no matter how the data is cut. The results do vary substantially across methods. The first row reports equal-weighted elasticities, first without winsorization, second with winsorization at the 1$^{st}$ and 99$^{th}$ percentile, and third with winsorization at the $5^{th}$ and $95^{th}$ percentile. In every case, the elasticity is in the tens of thousands. The next row shows value-weighted elasticities, which are all around 900, similar to the number given above. Standard errors are double clustered by stock and period for these first two rows. Finally, I take the elasticity of the largest stock each period, and average the elasticities across time. This average elasticity of 10.8 is reported in the third row. In every case, the elasticities are more elastic than empirical estimates. However, equal-weighted averages are greatly affected by small portfolio weights which give massive elasticities, and value-weighted averages eliminate this effect here, but still exhibit very elastic values. 

\begin{table}[!t] \centering
    \resizebox{0.5\textwidth}{!}{\begin{tabular}{@{\extracolsep{5pt}}lccc}
\\[-1.8ex]\hline
\hline \\[-1.8ex]
& \multicolumn{3}{c}{\textit{Average Elasticity}} \
\cr \cline{2-4}
\\[-1.8ex] & (1) & (2) & (3) \\
\\[-1.8ex] Winsorization: & None & $(1^{st}, 99^{th})$ & $(5^{th}, 95^{th})$ \\
\hline \\[-1.4ex]

 Equal Weighted & 61,157.8$^{***}$ & 56,004.5$^{***}$ & 44,679.4$^{***}$ \\
& (1,966.1) & (1,740.7) & (1,287.8) \\

\\[-1.8ex]

 Value Weighted & 929.4$^{***}$ & 929.2$^{***}$ & 928.9$^{***}$ \\
& (19.1) & (19.1) & (19.1) \\

\\[-1.8ex]

Largest Stock & 10.8$^{***}$ \\
& (0.3) \\

\hline
\hline \\[-1.8ex]
\textit{Note:} & \multicolumn{3}{r}{$^{*}$p$<$0.1; $^{**}$p$<$0.05; $^{***}$p$<$0.01} \\
\end{tabular}}
    \vspace{4mm}
    \caption{\textbf{Demand Elasticity Calibration.} Calibrated elasticities under different weighting and winsorization schemes. Averages are reweighted, winsorized, then averaged cross-sectionally and over time. Rows: equal-weighted, value-weighted (market cap shares), and largest stock in each cross-section. Columns: no winsorization, 1$^{st}$-99$^{th}$ percentile winsorization, and 5$^{th}$-95$^{th}$ percentile winsorization. Based on 1,253,175 stock-month observations from February 1990 to January 2020. Standard errors (in parentheses) are double-clustered by month and stock.}
    \label{tab:calibration_table}
\end{table}

The results are robust to two additional considerations. First, Epstein-Zin demand, which nests CRRA demand, exhibits similarly elastic demand.\footnote{Appendix \ref{subsec:epstein_zin} shows that multi-asset Epstein-Zin demand elasticity is similar to mean-variance demand elasticity, with a consumption hedging term causing only slight changes in the calibration.} Second, relaxing the assumption of an exogenous covariance matrix still delivers similar and very elastic demand.\footnote{Appendix \ref{app:covariance} shows that demand is similar and very elastic allowing price effects through the covariance matrix.}


\subsection{Why is this so high relative to statistical arbitrageur models?} \label{subsec:why_inelastic}

An elasticity of 900 implies that a 1\% drop in prices leads to a 900\% increase in shares demanded. This extreme sensitivity to prices is crucial in many asset pricing models. Investors with highly elastic demand trade aggressively against price deviations, attempting to correct perceived mispricings. However, this aggressive trading is fraught with difficulty and risk, making such high elasticities rare in practice.

Price signals are relatively weak predictors of returns, both cross-sectionally and over time. For example, investing based on the \cite{ff3} model involves forming a conservative portfolio, long a diversified basket of value stocks and short a diversified basket of growth stocks, rather than making large bets on low-priced stocks. This conservative approach is due to the significant noise in both cross-sectional and time-series data. Rolling alpha estimates of value portfolios vary dramatically depending on the time span used for estimation.\footnote{This is well-known, but is also shown in Figure \ref{fig:value_alpha} in the Appendix.} Because price signals are noisy, factor model investment seeks to diversify away risk, resulting in inelastic demand. Thus, factor model investment is conservatively inelastic by design.

\section{Elasticity of Quantitative Investment} \label{sec:quant_elas}

I first briefly give an overview of the quantitative investment models, and then present the empirical elasticity results of these models. 

\subsection{Description of Portfolio Choice Models} \label{subsec:description_of_models}

To understand why portfolio choice demand may be relatively inelastic, it is important to first understand the structure of demand from these models. Figure \ref{fig:demand} shows how these demand functions work. Prices $p_{i,t}$ map into some of the raw predictors $x_{i,k,t}$, which map into normalized predictors $\udot{z}_{i,k,t}$, which map into portfolio weights $w_{i,t}$. As explained below, there is a large variety of statistical portfolio choice models, each constituting a unique function that maps the $\udot{z}_{i,k,t}$ terms into portfolio weights $w_{i,t}$. 

Stack these predictors $z_{i,k,t}$ into an $N \times K$ matrix $Z_t$ where there are $K$ predictors for each asset. These matrices, for each period $t$, serve as the critical data inputs for the models I describe. These models are described in Appendix \ref{app:stat_arbs} and of course in their original papers, but I describe the basics here. I refer to these models by the initialisms summarized in Table \ref{tab:models} in the appendix. 


Consider a model where the covariance matrix and vector of expected returns are a linear function of predictors:
\begin{equation} \label{eq:nn_basic}
    \mu_t = Z_t \Theta_\mu, \;\; \Gamma_t = Z_t \Theta_{\Gamma}, \;\;
    \Sigma_t = \Gamma_t \Gamma_t' + \Theta_{\zeta} I,
\end{equation}
where $\Theta_{\mu}$ and $\Theta_{\Gamma}$ are vectors of parameters that control the mean and covariance matrix respectively, and $\Theta_{\zeta} > 0$ is a positive scalar that controls idiosyncratic risk. This functional form actually nests the functional form used in the calibration above, and the NN (neural network) model uses a generalization of this form. 

KY show\footnote{This is shown as Proposition 1 in KY.} that this functional form delivers the BSV \citep{brandt} model:
\begin{equation} \label{eq:bsv_linear}
    w_t = \Sigma_t^{-1} \mu_t = Z_t b,
\end{equation}
for some column vector of parameters $b$. This model features individual asset weights that are a linear function of predictors. Instead of estimating the $\Theta_{\mu}$, $\Theta_{\Gamma}$, and $\Theta_{\zeta}$ parameters, we can fit $b$ directly. \cite{shrinking} discusses how estimating $b$ can be done with the two-step procedure of first forming the $K$ predictor-weighted portfolios $Z_t' r_{t+1}$, and then choosing $b$ to maximize the Sharpe ratio. In other words, one can think of $Z_t$ as both the predictors that the statistical arbitrageur trades on as in equation (\ref{eq:bsv_linear}) and also the portfolio weights when forming the portfolios $Z_t' r_{t+1}$. \cite{shrinking} and \cite{forest} discuss "shrinkage" optimization procedures, which help alleviate overfit problems when choosing $b$, and \cite{demiguel} discuss an equal-weighted strategy across a range of portfolios, which corresponds to imposing very strong shrinkage when estimating $b$. While I focus on cross-sectionally normalized predictors, which correspond to the linear predictor-weighted portfolios of \cite{shrinking} and \cite{kelly} among others, I also consider proxies of more traditional value-weighted predictors corresponding to value-weighted portfolios in Appendix \ref{subsec:value_weight}. In summary, the BPZ$_L$, BSV, and DGU models have portfolio weights $w_t = Z_t b$, where $b$ is fit by a variety of different methods that seek to obtain the highest possible Sharpe ratio out-of-sample. For the sake of consistent notation across models, I write this as $w_t = f(Z_t) b$ for the trivial identity function $f(Z_t) = Z_t$. 

The BSV model provides a foundation for describing various statistical arbitrageur models. The BPZ$_F$ model, as discussed by \cite{forest}, generates portfolio weights from predictors using decision trees, creating a matrix \(f(Z_t)\) that collapses into a single portfolio \(w_t = f(Z_t) b\). The RF model simplifies this by focusing on means and shrinking the covariance matrix. The FF3, FF6, and HXZ models form factor portfolios from a small subset of predictors (e.g., market, size and value for FF3 \citep{ff3}), which can be consolidated into a single portfolio using similar logic. Other models, such as CRW, GKX, and KPS, combine predictors into optimal Sharpe ratio portfolios through a matrix \(\Phi\) and further collapse with \(b\) into weights \(w_t = Z_t \Phi b\). Statistical arbitrageurs here aim to maximize Sharpe ratios and price the cross-section of returns, integrating both risk-based and pure-alpha strategies. Models like CRW and KPS also produce pure-alpha portfolios. I evaluate pure-alpha portfolios separate from the factor portfolios.

In summary, each model has portfolio weights $w_t = f(Z_t) b$. This involves first fitting potential parameters associated with $f$, and then fitting the parameters $b$. The models are described in detail in Appendix \ref{app:stat_arbs}. I describe both the portfolio collapse methods used to create $b$, and as well as the portfolio weight functions $f(Z_t)$.

The scaling of the $f(Z_t) b$ weights is arbitrary in these models. In other words, the models are designed to maximize the Sharpe ratio, but not to necessarily pick the level of risk and return. Thus, to create comparable returns results, the amount of leverage and risk must be chosen.\footnote{The amount of chosen leverage does not change the elasticities of the statistical arbitrageur, because elasticity is measured in percentage terms. } I simply follow \cite{shrinking}, who scale portfolio weights so that the portfolio volatility matches the market volatility. 

Let $\nabla_{{z}_{i,k,t}} (\cdot)$ denote the derivative with respect to ${z}_{i,k,t}$. Then the elasticity, defined only for assets with positive weights following KY, can be written as:
\begin{equation} \label{eq:statistical arbitrageur_elasticity}
    {\eta}_{i,t} = 1 - \frac{1}{{w}_{i,t}} \left( \frac{\partial f(Z_t)_i b}{\partial \log (p_{i,t})} \right)
    = 1 - \frac{1}{{w}_{i,t}} \left( \sum_{k=1}^K \nabla_{{z}_{i,k,t}} (f(Z_t)_i b) \frac{\partial {z}_{i,k,t}}{\partial \log (p_{i,t})} \right),
\end{equation}
where $f(Z_t)_i$ denotes the row vector of $f(Z_t)$ corresponding to asset $i$. 

For the predictors that are not a function of prices, which I refer to as exogenous predictors, these derivatives are zero (i.e., $\partial {z}_{i,k,t} / \partial \log(p_{i,t}) = 0$). The other predictors affect the demand elasticity, because the portfolio weights change as prices change. In other words, every asset's elasticity is a function of how much portfolio weights change as prices change (e.g., a valuation ratio like book-to-market changing as prices change) and how much the statistical arbitrageur invests according to that predictor (e.g., how much they invest in value stocks). 

In Appendix \ref{app:stat_arb_derivs}, I show the closed form solution for the elasticity of each model. For the two random forest based models, $f(Z_t)$ is discontinuous, and I describe in the Appendix how the derivatives are numerically calculated.


\subsection{Returns and Elasticity of Factor Investment} \label{subsec:statistical arbitrageur_returns}

I first describe how the statistical arbitrageur model parameters are chosen. The returns and elasticity results are out-of-sample. Most models have a set of hyperparameters. Table \ref{tab:hypers} in the appendix shows the hyperparameters for the models. While some hyperparameters are chosen with a standard rule-of-thumb, most are selected using a four-fold cross-validation design with the sample before 1990. 

The statistical arbitrageur parameters for each model, including $b$ and those associated with $f(Z_t)$, are fitted every decade (120 months). The entire available sample is used to fit the model, employing an expanding window approach \citep[e.g.,][]{kelly} rather than a rolling window approach. The out-of-sample period is the three decades (360 months) from February 1990 to January 2020, inclusive.\footnote{As described below, the counterfactual experiments start in January 1990. The first fully endogenous return is in February 1990, since both the initial price (January 1990) and final price (February 1990) are endogenous.} All parameters and portfolio weights are fit using strictly ex ante information to avoid any look-ahead bias. All results reported from hereon are results from this out-of-sample period. I follow \cite{shrinking} by scaling the portfolios so that their volatility matches the volatility of the market during the same period.

Table \ref{tab:returns} shows the annualized CAPM alphas, betas, and annualized Sharpe ratios of the thirteen factor models during the out-of-sample period. Table \ref{tab:returns} shows that many of these models have Sharpe ratios well above one out-of-sample. For some of the models, the annualized alphas are around 30\%. Thus, as a whole, these models perform well during this out-of-sample period. While these are large alphas, this is well in line with out-of-sample CAPM alphas and returns from the literature.\footnote{See, for example, Table 4 of \cite{gu}, which makes these results seem mild in comparison. See also \cite{mlpaper}. } 

\begin{table}[!t] \centering
    \resizebox{1.\textwidth}{!}{\begin{tabular}{@{\extracolsep{5pt}}lccccccccccccc}
\\[-1.8ex]\hline
\hline \\[-1.8ex]
\\[-1.8ex] & \multicolumn{1}{c}{BPZ$_F$} & \multicolumn{1}{c}{BPZ$_L$} & \multicolumn{1}{c}{BSV} & \multicolumn{1}{c}{CRW} & \multicolumn{1}{c}{DGU} & \multicolumn{1}{c}{FF3} & \multicolumn{1}{c}{FF6} & \multicolumn{1}{c}{GKX} & \multicolumn{1}{c}{HXZ} & \multicolumn{1}{c}{KNS} & \multicolumn{1}{c}{KPS} & \multicolumn{1}{c}{NN} & \multicolumn{1}{c}{RF}  \\
\\[-1.8ex] & (1) & (2) & (3) & (4) & (5) & (6) & (7) & (8) & (9) & (10) & (11) & (12) & (13) \\
\hline \\[-1.8ex]
 $\alpha$ & 21.526$^{***}$ & 33.489$^{***}$ & 32.940$^{***}$ & 9.615$^{***}$ & 14.094$^{***}$ & 9.044$^{***}$ & 24.950$^{***}$ & 50.879$^{***}$ & 11.338$^{***}$ & 36.955$^{***}$ & 18.289$^{***}$ & 20.638$^{***}$ & 9.267$^{***}$ \\
& (2.520) & (2.592) & (2.628) & (2.598) & (2.576) & (2.379) & (2.601) & (2.620) & (2.325) & (2.625) & (2.593) & (2.373) & (2.281) \\
 $\beta$ & 0.285$^{***}$ & -0.165$^{***}$ & -0.007$^{}$ & -0.153$^{***}$ & 0.199$^{***}$ & 0.425$^{***}$ & 0.143$^{***}$ & -0.077$^{}$ & 0.466$^{***}$ & -0.049$^{}$ & 0.163$^{***}$ & 0.430$^{***}$ & 0.497$^{***}$ \\
& (0.051) & (0.052) & (0.053) & (0.052) & (0.052) & (0.048) & (0.052) & (0.053) & (0.047) & (0.053) & (0.052) & (0.048) & (0.046) \\
\hline \\[-1.8ex]
 Sharpe & 1.681 & 2.267 & 2.318 & 0.591 & 1.107 & 0.881 & 1.841 & 3.543 & 1.066 & 2.577 & 1.383 & 1.701 & 0.937 \\
 Obs. & 360 & 360 & 360 & 360 & 360 & 360 & 360 & 360 & 360 & 360 & 360 & 360 & 360 \\
\hline
\hline \\[-1.8ex]
\textit{Note:} & \multicolumn{13}{r}{$^{*}$p$<$0.1; $^{**}$p$<$0.05; $^{***}$p$<$0.01} \\
\end{tabular}
}
    \vspace{4mm}
    \caption{\textbf{Statistical Arbitrageur Returns.} Monthly CAPM annualized alpha, betas, and Sharpe ratios for the thirteen statistical arbitrageurs from February 1990 to January 2020 (out-of-sample period). Standard errors in parentheses. See Table \ref{tab:models} for model initialisms.}
    \label{tab:returns}
\end{table}

Table \ref{tab:statistical arbitrageur_elasticity} shows the main result of the paper: these statistical arbitrageur models have relatively inelastic demand. This table shows the average demand elasticity under various weighting and winsorization schemes described above.\footnote{The results with equal weighting and winsorization at the 1$^{st}$ and 99$^{th}$ percentiles are shown in Table \ref{tab:statistical arbitrageur_elasticity_full} in the Appendix.} In every case, demand appears much more inelastic than the calibrated counterparts in Table \ref{tab:calibration_table}. Since elasticity is in log terms and thus undefined for short or zero positions, I follow KY by calculating the elasticity of only the assets with strictly positive weights for each model. Some of the elasticities are actually negative, implying upward sloping \cite{stein} style demand like a classic momentum trader, at least for a subset of assets. As discussed above, the unwinsorized results are sensitive to small portfolio weights, ${w}_{i,t}$, in the denominator (see equation (\ref{eq:statistical arbitrageur_elasticity})). This can occur even in the value-weighted results, since there are assets for which ${w}_{i,t}$ is small, but the asset is a large asset by market capitalization. Thus, even the value-weighted elasticities are sensitive to these extreme values. To focus on the elasticities in the bulk of the distribution instead of these extreme tails, in the tables from hereon, all elasticities are value-weighted and winsorized at the 5$^{th}$ and 95$^{th}$ percentiles. Standard errors are shown in parentheses below the estimates, which are double clustered by month and stock. Note that by construction of the cross-sectional predictors which are between -0.5 and 0.5, these elasticities are quite stable across time. Mechanically, it is primarily cross-sectional variation that generates variation in these results. 

\begin{table}[!t] \centering
    \resizebox{1.\textwidth}{!}{\begin{tabular}{@{\extracolsep{5pt}}llccccccccccccc}
\\[-1.8ex]\hline
\hline \\[-1.8ex]
& & \multicolumn{13}{c}{\textit{Average Elasticity}} \
\cr \cline{3-15}
\\[-1.8ex] Elasticity & Wins. & \multicolumn{1}{c}{BPZ$_F$} & \multicolumn{1}{c}{BPZ$_L$} & \multicolumn{1}{c}{BSV} & \multicolumn{1}{c}{CRW} & \multicolumn{1}{c}{DGU} & \multicolumn{1}{c}{FF3} & \multicolumn{1}{c}{FF6} & \multicolumn{1}{c}{GKX} & \multicolumn{1}{c}{HXZ} & \multicolumn{1}{c}{KNS} & \multicolumn{1}{c}{KPS} & \multicolumn{1}{c}{NN} & \multicolumn{1}{c}{RF}  \\
\\[-1.8ex] Weighting & & (1) & (2) & (3) & (4) & (5) & (6) & (7) & (8) & (9) & (10) & (11) & (12) & (13) \\
\\[-1.8ex] \hline \\[-1.4ex]

 Portfolio & None & 0.285$^{***}$ & 6.428$^{***}$ & 2.275$^{***}$ & 2.034$^{***}$ & 3.209$^{***}$ & 1.911$^{***}$ & 1.553$^{***}$ & 7.394$^{***}$ & 0.843$^{***}$ & 3.630$^{***}$ & 2.951$^{***}$ & 2.496$^{***}$ & 2.185$^{***}$ \\
& & (0.094) & (0.190) & (0.093) & (0.057) & (0.437) & (0.022) & (0.017) & (0.552) & (0.012) & (0.049) & (0.504) & (0.595) & (0.037) \\
 & $(5^{th}, 95^{th})$ & -1.233$^{***}$ & 6.252$^{***}$ & 2.354$^{***}$ & 2.019$^{***}$ & 2.753$^{***}$ & 1.969$^{***}$ & 1.561$^{***}$ & 7.923$^{***}$ & 0.866$^{***}$ & 3.600$^{***}$ & 2.545$^{***}$ & 1.893$^{***}$ & 1.864$^{***}$ \\
& & (0.082) & (0.100) & (0.035) & (0.047) & (0.034) & (0.021) & (0.016) & (0.096) & (0.011) & (0.045) & (0.290) & (0.069) & (0.032) \\

\hline \\[-1.8ex]

 Value & None & 12.781$^{***}$ & 36.835$^{***}$ & 7.105$^{***}$ & 17.635$^{**}$ & 25.920$^{**}$ & 3.003$^{***}$ & 0.290$^{}$ & 76.991$^{***}$ & -1.208$^{***}$ & 13.109$^{***}$ & 27.588$^{***}$ & 13.516$^{***}$ & 2.101$^{***}$ \\
& & (1.291) & (3.281) & (0.886) & (8.063) & (11.499) & (0.704) & (0.594) & (25.222) & (0.114) & (1.301) & (4.917) & (2.954) & (0.142) \\
 & $(5^{th}, 95^{th})$ & 0.142$^{**}$ & 7.168$^{***}$ & 1.833$^{***}$ & 1.964$^{***}$ & 2.203$^{***}$ & 0.858$^{***}$ & 0.706$^{***}$ & 9.433$^{***}$ & -0.004$^{}$ & 3.253$^{***}$ & 3.510$^{***}$ & 2.163$^{***}$ & 0.955$^{***}$ \\
& & (0.059) & (0.266) & (0.071) & (0.085) & (0.087) & (0.026) & (0.021) & (0.311) & (0.007) & (0.117) & (0.476) & (0.152) & (0.033) \\

\hline
\hline \\[-1.8ex]

\textit{Note:} & \multicolumn{14}{r}{$^{*}$p$<$0.1; $^{**}$p$<$0.05; $^{***}$p$<$0.01} \\
\end{tabular}

}
    \vspace{4mm}
    \caption{\textbf{Statistical Arbitrageur Price Elasticity.} Average elasticities of the thirteen models. Averages are reweighted, winsorized, then averaged cross-sectionally and over time. Weighting schemes: portfolio weights and value weights (market cap). Results are either not winsorized or winsorized at the 5$^{th}$-95$^{th}$ percentiles. Based on 1,253,175 stock-month observations from February 1990 to January 2020. Standard errors (in parentheses) are double-clustered by month and stock. See Table \ref{tab:models} for model initialisms.}
    \label{tab:statistical arbitrageur_elasticity}
\end{table}

Although mathematically undefined, I can extend the concept of demand elasticity to short positions. I modify equation (\ref{eq:simple_elasticity}) with an absolute value to obtain:
\begin{equation} \label{eq:pos_neg_elasticity}
    \eta_{i,t}^{\pm} = 1 - \frac{1}{|w_{i,t}|} \left( \frac{\partial w_{i,t}}{\partial \log (p_{i,t})} \right).
\end{equation}
This elasticity keeps the same interpretation of the shape of the demand curve, but is measured in percentage changes relative to the absolute value of the position, rather than the value of the position. Table \ref{tab:statistical arbitrageur_elasticity posall} compares the elasticity of only long positions to all long and short positions. In most cases, the results are similar but slightly higher than the long-only elasticities. It turns out this is because most positions are long positions and short positions are more clustered around zero, which means that percentage changes are slightly higher on average when the entire sample is included. In summary, the results are robust to considering short positions. 

\begin{table}[!t] \centering
    \resizebox{1.\textwidth}{!}{\begin{tabular}{@{\extracolsep{5pt}}lccccccccccccc}
\\[-1.8ex]\hline
\hline \\[-1.8ex]
& \multicolumn{13}{c}{\textit{Average Elasticity}} \
\cr \cline{2-14}
\\[-1.8ex] Elasticity & \multicolumn{1}{c}{BPZ$_F$} & \multicolumn{1}{c}{BPZ$_L$} & \multicolumn{1}{c}{BSV} & \multicolumn{1}{c}{CRW} & \multicolumn{1}{c}{DGU} & \multicolumn{1}{c}{FF3} & \multicolumn{1}{c}{FF6} & \multicolumn{1}{c}{GKX} & \multicolumn{1}{c}{HXZ} & \multicolumn{1}{c}{KNS} & \multicolumn{1}{c}{KPS} & \multicolumn{1}{c}{NN} & \multicolumn{1}{c}{RF}  \\
\\[-1.8ex] Sample & (1) & (2) & (3) & (4) & (5) & (6) & (7) & (8) & (9) & (10) & (11) & (12) & (13) \\
\\[-1.8ex] \hline \\[-1.4ex]

Positive & 0.142$^{**}$ & 7.168$^{***}$ & 1.833$^{***}$ & 1.964$^{***}$ & 2.203$^{***}$ & 0.858$^{***}$ & 0.706$^{***}$ & 9.433$^{***}$ & -0.004$^{}$ & 3.253$^{***}$ & 3.510$^{***}$ & 2.163$^{***}$ & 0.955$^{***}$ \\
& (0.059) & (0.266) & (0.071) & (0.085) & (0.087) & (0.026) & (0.021) & (0.311) & (0.007) & (0.117) & (0.476) & (0.152) & (0.033) \\

All & 0.447$^{***}$ & 7.946$^{***}$ & 2.074$^{***}$ & 2.027$^{***}$ & 3.005$^{***}$ & 1.217$^{***}$ & 0.923$^{***}$ & 10.060$^{***}$ & 0.168$^{***}$ & 3.812$^{***}$ & 3.535$^{***}$ & 1.737$^{***}$ & 0.959$^{***}$ \\
& (0.031) & (0.263) & (0.071) & (0.088) & (0.094) & (0.028) & (0.022) & (0.308) & (0.005) & (0.117) & (0.497) & (0.113) & (0.032) \\

\hline
\hline \\[-1.8ex]
\textit{Note:} & \multicolumn{13}{r}{$^{*}$p$<$0.1; $^{**}$p$<$0.05; $^{***}$p$<$0.01} \\
\end{tabular}}
    \vspace{4mm}
    \caption{\textbf{Statistical Arbitrageur Price Elasticity By Sample.} Average price elasticity of the thirteen models for two samples: (1) positive weights and (2) all assets, using equation (\ref{eq:pos_neg_elasticity}). Results are value-weighted and winsorized (5$^{th}$-95$^{th}$ percentiles). Sample: 1,253,175 stock-month observations, February 1990 to January 2020. Standard errors (in parentheses) are double-clustered by month and stock. See Table \ref{tab:models} for model initialisms.
}
    \label{tab:statistical arbitrageur_elasticity posall}
\end{table}

To visually show how much more inelastic these models are compared to the calibrated demand elasticity, I present a binned scatter plot of the elasticities as a function of market portfolio weights in Panel A of Figure \ref{fig:ml_model_elasticity}. The results are winsorized at the 5$^{th}$ and 95$^{th}$ percentiles in each bin, and the short positions are included so that the sample of stocks is exactly the same across models. While for the largest stocks there is some overlap, for the vast majority of the stocks, the calibrated elasticities are much higher than the statistical arbitrageur elasticities. Panel B is similar, but the stocks are first portfolio-weighted. Specifically, $N_t v_{i,t} \eta_{i,t}$ is plotted, where $v_{i,t}$ is the absolute value of the portfolio weight, normalized so that the cross-sectional sum of $v_{i,t}$ is one. Note that $N_t v_{i,t}$ is one on average in each cross-section, and thus the cross-sectional equal-weighted average of $N_t v_{i,t} \eta_{i,t}$ is the $v_{i,t}$-weighted average of $\eta_{i,t}$.\footnote{Mathematically, the weighted average is given by:
\begin{equation*}
    \sum_i v_{i,t} \eta_{i,t} = \frac{1}{N_t} \sum_i N_t v_{i,t} \eta_{i,t}.
\end{equation*}
} The calibrated elasticity is around 900, but statistical arbitrageur model demand elasticities are around 10 or below. 

\begin{figure}[!t]
    \centering
    \begin{minipage}{0.49\textwidth}
        \centering
        \textbf{Panel A: Elasticity Across Models}
        \includegraphics[width=\linewidth,trim={0 0 0 1cm},clip]{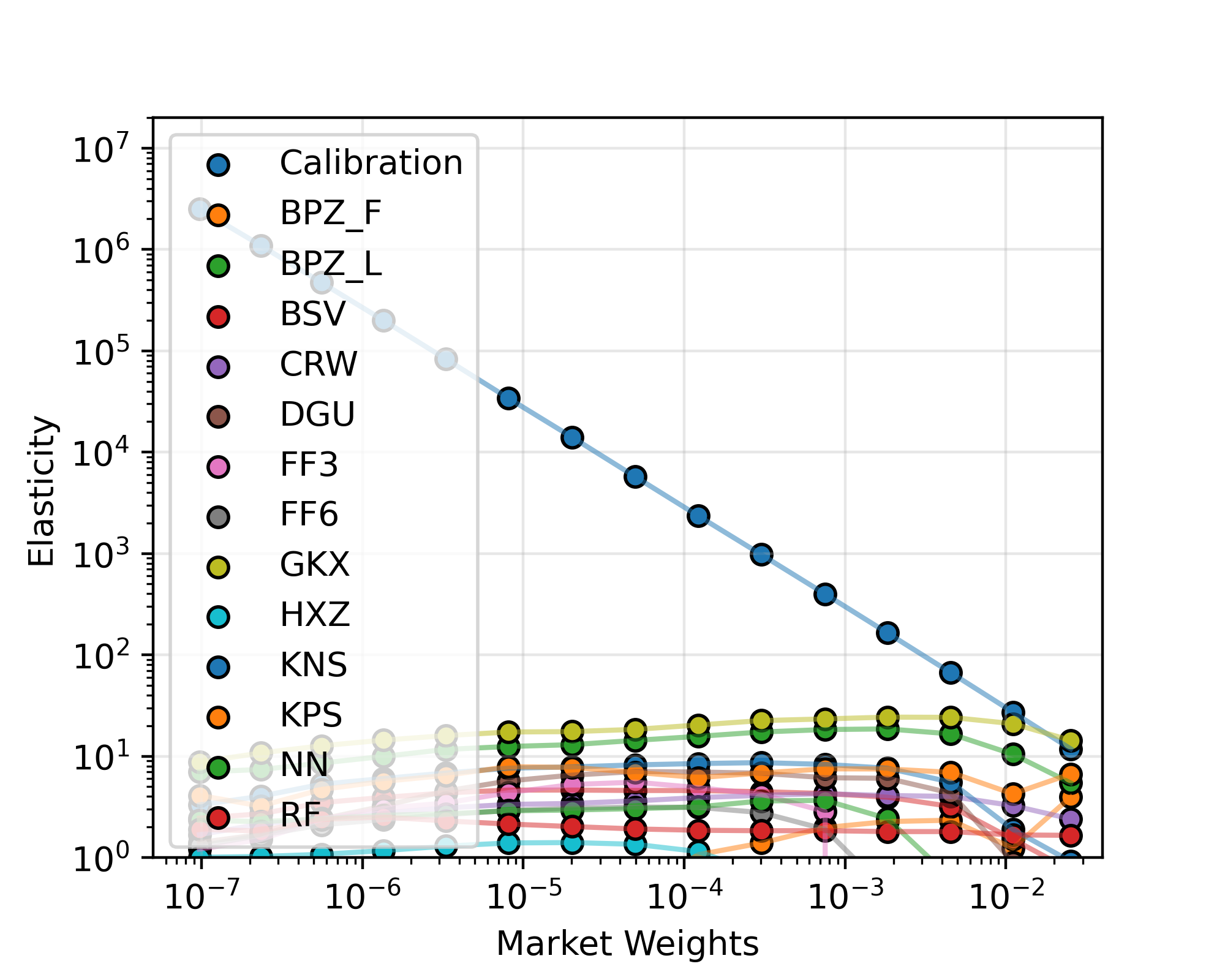}
    \end{minipage}\hfill
    \begin{minipage}{0.49\textwidth}
        \centering
        \textbf{Panel B: Portfolio-Weighted Elasticity}
        \includegraphics[width=\linewidth,trim={0 0 0 1cm},clip]{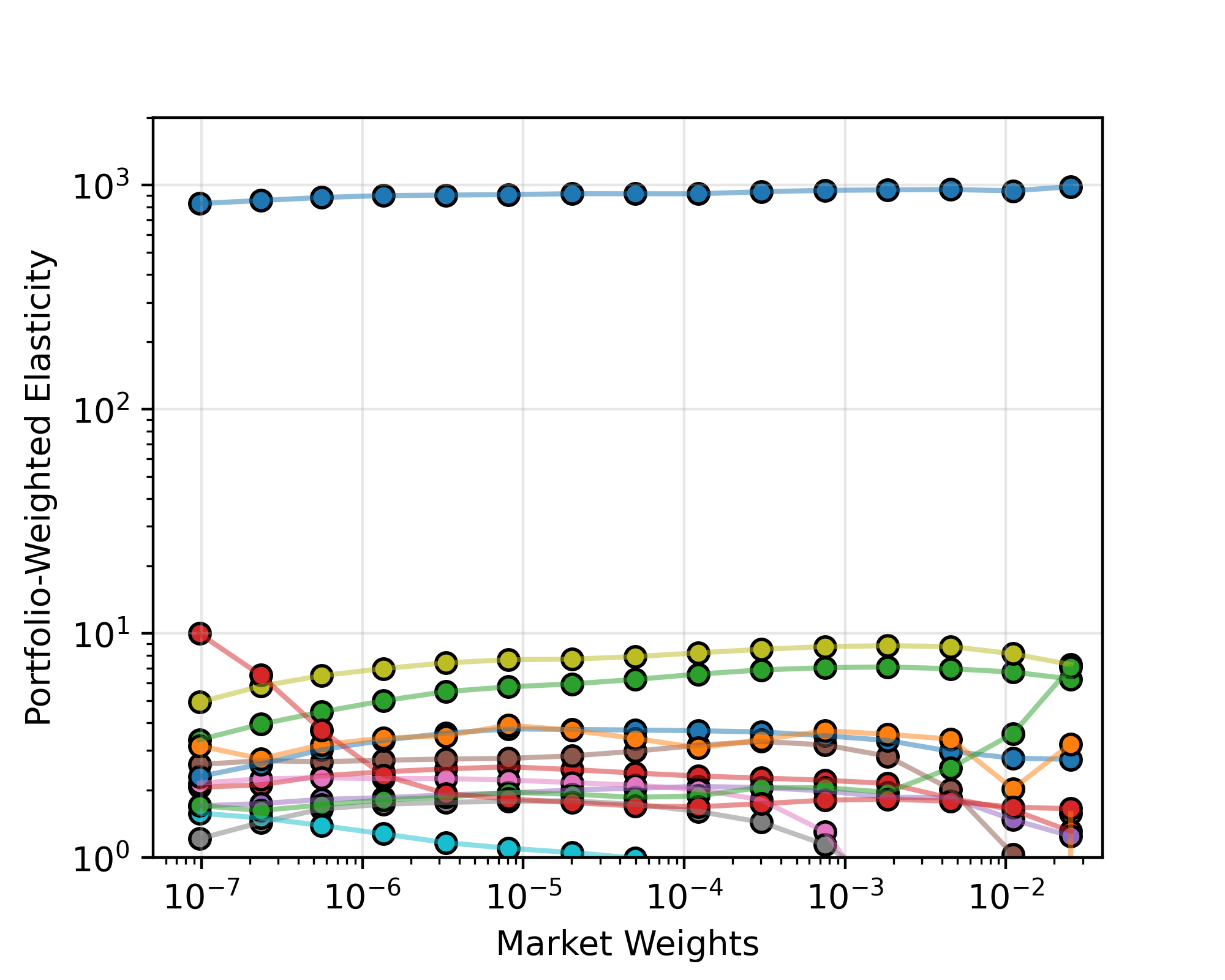}
    \end{minipage}
    \vspace{4mm}
    \caption{\textbf{Statistical Arbitrageur versus Calibrated Elasticity.} Binned scatter plots of elasticities vs. market portfolio weights (log-scale axes), winsorized in each bin (5$^{th}$-95$^{th}$ percentiles). Panel A: elasticity; Panel B: portfolio-weighted elasticity. Sample: 1,253,175 stock-month observations, February 1990 to January 2020. All stocks included using equation (\ref{eq:pos_neg_elasticity}). See Table \ref{tab:models} for model initialisms.
}
    \label{fig:ml_model_elasticity}
\end{figure}


The NN model uniquely separates mean and covariance effects (equation  (\ref{eq:nn_basic})), unlike the BSV and similar models, which do not. This allows us to decompose the elasticity of the NN model. The NN model's elasticity is about 2.2, with 0.8 due to mean effects alone.\footnote{This is shown in Table \ref{tab:cov_decomp} in the Appendix. This aligns with the extension of the calibration shown in Appendix \ref{app:covariance}, where covariance effects increase the elasticity beyond mean effects alone.} While the BSV and similar models don't explicitly separate these effects at the individual asset level, covariance effects still matter in these models \citep{shrinking}.\footnote{In Appendix \ref{app:which_assets}, I explore in more detail which assets have elastic demand across models.}

\section{Variations on Elasticity} \label{sec:variations}

I explore two key variations of the ceteris paribus elasticity of Sharpe-maximizing portfolios. First, I examine the elasticity of statistical arbitrageur strategies using pure-alpha investment strategies that hedge out factors. Second, I analyze systematic elasticity, where I measure the elasticity of stocks in response to shocks in the price of common factors.\footnote{In Appendix \ref{subsec:cost_optimization}, I also show that cost optimization can substantially decrease the elasticity of these statistical arbitrageur models.}

\subsection{Pure-Alpha Strategy Demand Elasticities} \label{sec:hedged_elasticity}

Often, statistical arbitrageurs are thought to trade on only alpha and hedge out systematic risk. What is the elasticity of these kinds of strategies? In this subsection, I address this question. These findings below are significant for two reasons. First, after hedging out risk, we observe that elasticity tends to decrease rather than increase. This indicates that the inelasticity is not simply due to unhedged risk; even with pure-alpha strategies, demand remains relatively inelastic.

Second, the inelastic demand response is not due to the absence of alpha but rather the difficulty of trading against price movements. This is evident because the pure-alpha elasticity remains close to the overall elasticity, meaning that a substantial portion of the elastic response comes from the alpha component. Thus, the inelasticity observed in these models arises from the challenge inherent in acting on perceived alpha, not from a lack of alpha. 

I present a general framework, where I consider any statistical arbitrageur strategy discussed in Section \ref{sec:elas_models}, hedge out some factors, and examine the performance and elasticity of the resulting hedged strategy.\footnote{Appendix \ref{app:pure_alpha} examines some specific pure-alpha strategies and their elasticity, and these special cases are similar. For example, the CRW and KPS strategies produce their own pure-alpha portfolios that I show have inelastic demand.} In order to disentangle the pure-alpha elasticity from the overall elasticity, I need to calculate the portfolio weights of a strategy associated with just alpha, after hedging out factor returns.\footnote{It is well-known that an arbitrary factor model may well contain alpha relative to another factor model along with traditional systematic risk-based returns \citep{hubermankandel, famaprize, interpreting}. } In particular, let $w_t$ be the $N$-dimensional vector of an arbitrary trading strategy. Let $r_{p,t+1} = w_t' r_{t+1}$ denote the corresponding portfolio excess return. Let $F_{t+1}$ be a column vector of excess factor returns (or long-short returns), with a matrix of portfolio weights, $W_t$, which has $N$ rows and the number of columns corresponding to the number of factors. Thus, we can write $F_{t+1} = W_t' r_{t+1}$. 

We can write:
\begin{equation}
    r_{p,t} = \alpha + \beta' F_t + \epsilon_t,
\end{equation}
where $\beta$ is the vector of standard constant factor betas, $\alpha$ is the standard asset pricing alpha, and $\epsilon_t$ is the error term. We seek an $N$-dimensional vector of portfolio weights $w_{\alpha,t}$ such that:
\begin{equation} \label{eq:requirement1}
    r_{\alpha,t+1} \equiv w_{\alpha,t}' r_{t+1} = \alpha + \epsilon_{t+1},
\end{equation}
for all $t$ and 
\begin{equation} \label{eq:requirement2}
    \text{Cov} (r_{\alpha,t+1}, F_{t+1}) = 0.
\end{equation}

In other words, the portfolio corresponding to $w_{\alpha,t}$ is essentially the portfolio associated with $w_t$ except the factors are hedged out. The return on this hedged portfolio, $r_{\alpha,t+1}$, is just the alpha and the idiosyncratic term. The following proposition gives a solution. 

\begin{prop} \label{prop:hedged}
    Let 
    \begin{equation} \label{eq:alpha_weights}
        w_{\alpha,t} = w_t - W_t \beta,
    \end{equation}
    then this satisfies both requirements (\ref{eq:requirement1}) and (\ref{eq:requirement2}). 
    
    It immediately follows that for an individual asset with hedged portfolio weight $w_{\alpha,i,t} > 0$ defined above, we can write:
    \begin{equation} \label{eq:hedge_deriv1}
        \frac{\partial w_{\alpha,i,t}}{\partial \log (p_{i,t})} = \frac{\partial w_{i,t}}{\partial \log (p_{i,t})} - \sum_j \frac{\partial W_{i,j,t}}{\partial \log (p_{i,t})} \beta_{j},
    \end{equation}
    where $W_{i,j,t}$ are the factor weights for asset $i$ for factor $j$, and $\beta_j$ is the beta for this factor. 
    For weights $w_{\alpha,i,t} > 0$, $w_{i,t} > 0$, and $W_{i,j,t} > 0$ for all factors, then:
    \begin{equation} \label{eq:hedge_deriv2}
        \eta_{\alpha,i,t} = 1 + \frac{1}{w_{\alpha,i,t}} \left(w_{i,t} \left( \eta_{i,t} - 1 \right) - \sum_j W_{i,j,t} \left( \eta^F_{i,j,t} - 1 \right) \beta_{j}\right),
    \end{equation}
    where $\eta_{\alpha,i,t}$ is the elasticity of the hedged strategy and $\eta^F_{i,j,t}$ is the elasticity of the $j^{th}$ factor. 
\end{prop}

The proof of this is contained in Appendix \ref{app:alpha_portfolio}. 

Thus, this proposition shows that there is a tight connection between the overall demand elasticity and the pure-alpha demand elasticity. If the $\beta_j$ and $\eta_{i,j,t}^F$ terms are high, then the demand elasticity of the pure-alpha strategy is likely low as long as $w_{\alpha,i,t}$ is not too small. On the other hand, if $\beta_j$ is low, then the pure-alpha strategy elasticity is close to the overall elasticity. 


To implement this empirically and interpret the corresponding results as the pure-alpha elasticity, the $F_t$ factors should correspond to systematic risk-based factors. Instead of taking a stand about the right model of systematic risk, I consider many different factor models and show that the results are robust across a wide range of choices. I simply consider every pairing of the thirteen portfolios we have already considered, where one portfolio is the left-hand side portfolio $w_t$ and the other is the right-hand side factor that is hedged out. This is reasonable since each of the thirteen portfolios corresponds to a factor model. There are $13^2 - 13$ pairings, excluding the pairings of a model with itself. I calculate each beta with a simple regression during the out-of-sample period. With these betas in hand, I calculate hedged portfolio weights using equation (\ref{eq:alpha_weights}). With these weights in hand, the elasticity of these portfolios can be calculated using equation (\ref{eq:hedge_deriv1}). 

The annualized alphas are reported in Table \ref{tab:cross_alphas}, with right-hand-side factors across the columns, and left-hand-side portfolios along the rows. Notice that a substantial fraction of many of the returns show up as pure-alpha when compared to other models. \cite{bbd} show a similar result, where much of the out-of-sample portfolio returns with these models have surprisingly low correlations, and much of the return shows up as alpha compared to other models. 

\begin{table}[!t] \centering
    \resizebox{1.\textwidth}{!}{\begin{tabular}{@{\extracolsep{5pt}}lccccccccccccc}
\\[-1.8ex]\hline
\hline \\[-1.8ex]
& \multicolumn{13}{c}{\textit{Rows Correspond to Left-Hand Side Portfolios, Columns Correspond to Right-Hand Side Portfolios}} \
\cr \cline{2-14}
\\[-1.8ex] & \multicolumn{1}{c}{BPZ$_F$} & \multicolumn{1}{c}{BPZ$_L$} & \multicolumn{1}{c}{BSV} & \multicolumn{1}{c}{CRW} & \multicolumn{1}{c}{DGU} & \multicolumn{1}{c}{FF3} & \multicolumn{1}{c}{FF6} & \multicolumn{1}{c}{GKX} & \multicolumn{1}{c}{HXZ} & \multicolumn{1}{c}{KNS} & \multicolumn{1}{c}{KPS} & \multicolumn{1}{c}{NN} & \multicolumn{1}{c}{RF}  \\
\\[-1.8ex] & (1) & (2) & (3) & (4) & (5) & (6) & (7) & (8) & (9) & (10) & (11) & (12) & (13) \\
\hline \\[-1.8ex]

 BPZ$_F$ &  & 13.156$^{***}$ & 6.744$^{**}$ & 24.970$^{***}$ & 16.994$^{***}$ & 19.321$^{***}$ & 20.292$^{***}$ & 7.736$^{**}$ & 20.047$^{***}$ & 5.046$^{*}$ & 22.391$^{***}$ & 13.564$^{***}$ & 24.085$^{***}$ \\
&  & (2.924) & (2.667) & (2.606) & (2.451) & (2.495) & (2.910) & (3.516) & (2.627) & (2.774) & (2.785) & (2.615) & (2.686) \\
 BPZ$_L$ & 24.237$^{***}$ &  & 9.208$^{***}$ & 31.644$^{***}$ & 29.479$^{***}$ & 33.574$^{***}$ & 24.290$^{***}$ & -0.667$^{}$ & 33.630$^{***}$ & 1.337$^{}$ & 31.901$^{***}$ & 24.771$^{***}$ & 33.487$^{***}$ \\
& (2.719) &  & (2.236) & (2.626) & (2.683) & (2.658) & (2.801) & (2.811) & (2.700) & (1.739) & (2.792) & (2.751) & (2.673) \\
 BSV & 20.494$^{***}$ & 10.446$^{***}$ &  & 33.675$^{***}$ & 28.025$^{***}$ & 29.144$^{***}$ & 27.763$^{***}$ & 3.829$^{}$ & 32.569$^{***}$ & -0.193$^{}$ & 34.418$^{***}$ & 20.249$^{***}$ & 34.138$^{***}$ \\
& (2.463) & (2.221) &  & (2.619) & (2.589) & (2.553) & (2.880) & (3.029) & (2.713) & (1.377) & (2.784) & (2.462) & (2.675) \\
 CRW & 11.606$^{***}$ & 6.423$^{**}$ & 11.486$^{***}$ &  & 13.776$^{***}$ & 13.004$^{***}$ & 9.117$^{***}$ & 1.912$^{}$ & 13.383$^{***}$ & 9.857$^{***}$ & 8.822$^{***}$ & 12.256$^{***}$ & 9.195$^{***}$ \\
& (2.856) & (3.095) & (3.108) &  & (2.557) & (2.486) & (2.936) & (3.681) & (2.560) & (3.231) & (2.792) & (2.852) & (2.681) \\
 DGU & 5.322$^{**}$ & 10.229$^{***}$ & 5.535$^{*}$ & 18.585$^{***}$ &  & 6.982$^{***}$ & 2.199$^{}$ & 13.078$^{***}$ & 2.795$^{**}$ & 5.275$^{*}$ & 11.418$^{***}$ & 4.136$^{}$ & 11.304$^{***}$ \\
& (2.595) & (3.055) & (2.968) & (2.470) &  & (1.916) & (2.514) & (3.707) & (1.413) & (3.099) & (2.725) & (2.536) & (2.535) \\
 FF3 & 3.884$^{}$ & 16.151$^{***}$ & 2.656$^{}$ & 15.597$^{***}$ & 1.531$^{}$ &  & 6.683$^{**}$ & 13.953$^{***}$ & 2.437$^{}$ & 4.976$^{}$ & 9.252$^{***}$ & -2.955$^{}$ & 7.752$^{***}$ \\
& (2.688) & (3.080) & (2.979) & (2.444) & (1.949) &  & (2.864) & (3.710) & (2.026) & (3.164) & (2.754) & (2.219) & (2.510) \\
 FF6 & 22.872$^{***}$ & 16.423$^{***}$ & 19.659$^{***}$ & 26.345$^{***}$ & 17.982$^{***}$ & 23.326$^{***}$ &  & 19.849$^{***}$ & 18.027$^{***}$ & 17.682$^{***}$ & 21.410$^{***}$ & 21.020$^{***}$ & 24.136$^{***}$ \\
& (2.856) & (2.956) & (3.061) & (2.630) & (2.330) & (2.609) &  & (3.683) & (2.294) & (3.147) & (2.711) & (2.825) & (2.657) \\
 GKX & 42.620$^{***}$ & 29.258$^{***}$ & 31.246$^{***}$ & 49.178$^{***}$ & 49.435$^{***}$ & 50.618$^{***}$ & 47.005$^{***}$ &  & 51.832$^{***}$ & 28.642$^{***}$ & 46.215$^{***}$ & 41.856$^{***}$ & 49.196$^{***}$ \\
& (2.731) & (2.348) & (2.547) & (2.609) & (2.719) & (2.675) & (2.914) &  & (2.699) & (2.608) & (2.733) & (2.709) & (2.678) \\
 HXZ & 9.151$^{***}$ & 18.262$^{***}$ & 14.440$^{***}$ & 17.895$^{***}$ & 1.714$^{}$ & 6.812$^{***}$ & 1.173$^{}$ & 20.361$^{***}$ &  & 15.775$^{***}$ & 9.025$^{***}$ & 6.611$^{**}$ & 8.278$^{***}$ \\
& (2.791) & (3.085) & (3.121) & (2.482) & (1.418) & (1.998) & (2.483) & (3.692) &  & (3.233) & (2.654) & (2.704) & (2.303) \\
 KNS & 24.303$^{***}$ & 9.449$^{***}$ & 6.801$^{***}$ & 36.894$^{***}$ & 32.074$^{***}$ & 33.985$^{***}$ & 30.538$^{***}$ & 6.840$^{**}$ & 36.825$^{***}$ &  & 36.918$^{***}$ & 25.122$^{***}$ & 38.132$^{***}$ \\
& (2.473) & (1.668) & (1.329) & (2.628) & (2.609) & (2.619) & (2.858) & (2.994) & (2.713) &  & (2.792) & (2.544) & (2.668) \\
 KPS & 17.858$^{***}$ & 19.199$^{***}$ & 22.188$^{***}$ & 19.803$^{***}$ & 16.180$^{***}$ & 17.548$^{***}$ & 13.361$^{***}$ & 9.258$^{**}$ & 14.912$^{***}$ & 20.290$^{***}$ &  & 13.256$^{***}$ & 12.666$^{***}$ \\
& (2.875) & (3.100) & (3.112) & (2.630) & (2.657) & (2.639) & (2.852) & (3.632) & (2.579) & (3.233) &  & (2.787) & (2.291) \\
 NN & 13.975$^{***}$ & 14.284$^{***}$ & 6.903$^{***}$ & 25.471$^{***}$ & 16.592$^{***}$ & 16.122$^{***}$ & 18.618$^{***}$ & 6.624$^{*}$ & 18.788$^{***}$ & 6.800$^{**}$ & 18.955$^{***}$ &  & 18.445$^{***}$ \\
& (2.608) & (2.952) & (2.659) & (2.596) & (2.389) & (2.055) & (2.871) & (3.480) & (2.539) & (2.847) & (2.694) &  & (2.429) \\
 RF & 13.744$^{***}$ & 16.522$^{***}$ & 16.400$^{***}$ & 13.812$^{***}$ & 8.097$^{***}$ & 8.839$^{***}$ & 9.425$^{***}$ & 9.291$^{**}$ & 5.511$^{**}$ & 17.631$^{***}$ & 3.048$^{}$ & 2.993$^{}$ & \\
& (2.882) & (3.085) & (3.108) & (2.626) & (2.569) & (2.500) & (2.905) & (3.700) & (2.326) & (3.211) & (2.381) & (2.613) & \\

\hline
\hline \\[-1.8ex]
\textit{Note:} & \multicolumn{13}{r}{$^{*}$p$<$0.1; $^{**}$p$<$0.05; $^{***}$p$<$0.01} \\
\end{tabular}

}
    \vspace{4mm}
    \caption{\textbf{Cross Alphas.} Annualized alpha of portfolios (rows) regressed on right-hand-side single-factor models (columns). Portfolios are formed ex ante, and alphas/betas are calculated for 360 observations from February 1990 to January 2020. Standard errors (in parentheses). See Table \ref{tab:models} for model initialisms.}
    \label{tab:cross_alphas}
\end{table}

\begin{table}[!t] \centering
    \resizebox{1.\textwidth}{!}{\begin{tabular}{@{\extracolsep{5pt}}lccccccccccccc}
\\[-1.8ex]\hline
\hline \\[-1.8ex]
& \multicolumn{13}{c}{\textit{Rows Correspond to Left-Hand Side Portfolios, Columns Correspond to Right-Hand Side Portfolios}} \
\cr \cline{2-14}
\\[-1.8ex] & \multicolumn{1}{c}{BPZ$_F$} & \multicolumn{1}{c}{BPZ$_L$} & \multicolumn{1}{c}{BSV} & \multicolumn{1}{c}{CRW} & \multicolumn{1}{c}{DGU} & \multicolumn{1}{c}{FF3} & \multicolumn{1}{c}{FF6} & \multicolumn{1}{c}{GKX} & \multicolumn{1}{c}{HXZ} & \multicolumn{1}{c}{KNS} & \multicolumn{1}{c}{KPS} & \multicolumn{1}{c}{NN} & \multicolumn{1}{c}{RF}  \\
\\[-1.8ex] & (1) & (2) & (3) & (4) & (5) & (6) & (7) & (8) & (9) & (10) & (11) & (12) & (13) \\
\hline \\[-1.8ex]

 BPZ$_F$ &  & -7.196$^{***}$ & -1.548$^{***}$ & 1.748$^{***}$ & -1.904$^{***}$ & -0.754$^{***}$ & -0.268$^{***}$ & -8.626$^{***}$ & 0.942$^{***}$ & -3.675$^{***}$ & -1.331$^{***}$ & -0.269$^{***}$ & 0.988$^{***}$ \\
&  & (0.253) & (0.081) & (0.083) & (0.059) & (0.036) & (0.030) & (0.274) & (0.040) & (0.116) & (0.225) & (0.086) & (0.040) \\
 BPZ$_L$ & 7.529$^{***}$ &  & 6.500$^{***}$ & 6.996$^{***}$ & 7.324$^{***}$ & 6.972$^{***}$ & 9.798$^{***}$ & -1.739$^{***}$ & 6.761$^{***}$ & 7.192$^{***}$ & 7.103$^{***}$ & 7.034$^{***}$ & 7.147$^{***}$ \\
& (0.274) &  & (0.262) & (0.263) & (0.268) & (0.268) & (0.308) & (0.134) & (0.260) & (0.308) & (0.265) & (0.256) & (0.266) \\
 BSV & 2.004$^{***}$ & -1.528$^{***}$ &  & 2.043$^{***}$ & 1.634$^{***}$ & 2.185$^{***}$ & 2.022$^{***}$ & -3.556$^{***}$ & 1.858$^{***}$ & -0.742$^{***}$ & 2.120$^{***}$ & 1.501$^{***}$ & 1.839$^{***}$ \\
& (0.077) & (0.117) &  & (0.077) & (0.067) & (0.080) & (0.076) & (0.141) & (0.072) & (0.050) & (0.087) & (0.073) & (0.071) \\
 CRW & 1.966$^{***}$ & 1.795$^{***}$ & 2.035$^{***}$ &  & 2.087$^{***}$ & 1.793$^{***}$ & 1.933$^{***}$ & 1.207$^{***}$ & 1.645$^{***}$ & 2.009$^{***}$ & 2.027$^{***}$ & 2.099$^{***}$ & 1.971$^{***}$ \\
& (0.085) & (0.081) & (0.086) &  & (0.088) & (0.081) & (0.085) & (0.072) & (0.076) & (0.085) & (0.087) & (0.086) & (0.085) \\
 DGU & 2.351$^{***}$ & 0.658$^{***}$ & 1.184$^{***}$ & 2.612$^{***}$ &  & 4.151$^{***}$ & 4.394$^{***}$ & 1.203$^{***}$ & 6.012$^{***}$ & 0.975$^{***}$ & 0.372$^{**}$ & 2.014$^{***}$ & 2.030$^{***}$ \\
& (0.089) & (0.025) & (0.056) & (0.101) &  & (0.119) & (0.113) & (0.050) & (0.142) & (0.051) & (0.148) & (0.123) & (0.082) \\
 FF3 & 0.906$^{***}$ & 1.571$^{***}$ & 0.150$^{***}$ & 1.444$^{***}$ & -0.307$^{***}$ &  & 0.794$^{***}$ & 1.195$^{***}$ & 1.517$^{***}$ & 0.203$^{***}$ & -0.327$^{***}$ & 0.158$^{***}$ & 0.736$^{***}$ \\
& (0.028) & (0.052) & (0.021) & (0.058) & (0.023) &  & (0.023) & (0.037) & (0.048) & (0.016) & (0.096) & (0.059) & (0.023) \\
 FF6 & 0.734$^{***}$ & -0.806$^{***}$ & 0.440$^{***}$ & 0.754$^{***}$ & 0.215$^{***}$ & 0.672$^{***}$ &  & -0.509$^{***}$ & 1.118$^{***}$ & 0.204$^{***}$ & -0.432$^{***}$ & 0.495$^{***}$ & 0.671$^{***}$ \\
& (0.021) & (0.038) & (0.016) & (0.022) & (0.008) & (0.020) &  & (0.024) & (0.034) & (0.014) & (0.091) & (0.022) & (0.020) \\
 GKX & 9.659$^{***}$ & 8.457$^{***}$ & 10.512$^{***}$ & 9.255$^{***}$ & 9.522$^{***}$ & 9.346$^{***}$ & 10.691$^{***}$ &  & 9.047$^{***}$ & 11.229$^{***}$ & 8.788$^{***}$ & 9.583$^{***}$ & 9.480$^{***}$ \\
& (0.316) & (0.274) & (0.323) & (0.310) & (0.314) & (0.313) & (0.325) &  & (0.314) & (0.339) & (0.296) & (0.309) & (0.313) \\
 HXZ & -0.012$^{}$ & 0.541$^{***}$ & -0.051$^{***}$ & 0.740$^{***}$ & -1.319$^{***}$ & -0.430$^{***}$ & -0.760$^{***}$ & 1.149$^{***}$ &  & 0.059$^{***}$ & -1.894$^{***}$ & -0.399$^{***}$ & -0.235$^{***}$ \\
& (0.009) & (0.017) & (0.009) & (0.037) & (0.047) & (0.021) & (0.032) & (0.035) &  & (0.006) & (0.166) & (0.035) & (0.024) \\
 KNS & 3.510$^{***}$ & -2.228$^{***}$ & 2.602$^{***}$ & 3.364$^{***}$ & 3.064$^{***}$ & 3.526$^{***}$ & 3.758$^{***}$ & -3.814$^{***}$ & 3.222$^{***}$ &  & 3.343$^{***}$ & 2.682$^{***}$ & 3.262$^{***}$ \\
& (0.126) & (0.158) & (0.090) & (0.121) & (0.113) & (0.123) & (0.127) & (0.152) & (0.117) &  & (0.122) & (0.108) & (0.118) \\
 KPS & 3.514$^{***}$ & 3.251$^{***}$ & 4.007$^{***}$ & 3.891$^{***}$ & 3.415$^{***}$ & 3.387$^{***}$ & 4.386$^{***}$ & 0.161$^{}$ & 5.634$^{***}$ & 3.652$^{***}$ &  & 3.322$^{***}$ & 3.112$^{***}$ \\
& (0.478) & (0.488) & (0.446) & (0.462) & (0.497) & (0.561) & (0.523) & (0.559) & (0.544) & (0.468) &  & (0.423) & (0.488) \\
 NN & 2.243$^{***}$ & -0.354$^{**}$ & 0.931$^{***}$ & 2.613$^{***}$ & 1.156$^{***}$ & 1.741$^{***}$ & 1.987$^{***}$ & -3.170$^{***}$ & 2.293$^{***}$ & 0.458$^{***}$ & 0.320$^{***}$ &  & 2.121$^{***}$ \\
& (0.153) & (0.169) & (0.108) & (0.151) & (0.137) & (0.136) & (0.143) & (0.196) & (0.136) & (0.121) & (0.033) &  & (0.154) \\
 RF & 0.957$^{***}$ & 4.405$^{***}$ & 1.680$^{***}$ & 1.814$^{***}$ & -0.527$^{***}$ & 0.322$^{***}$ & 0.752$^{***}$ & -4.596$^{***}$ & 1.561$^{***}$ & 2.502$^{***}$ & -0.207$^{}$ & 0.520$^{***}$ & \\
& (0.034) & (0.205) & (0.064) & (0.079) & (0.069) & (0.051) & (0.037) & (0.289) & (0.043) & (0.099) & (0.371) & (0.057) & \\

\hline
\hline \\[-1.8ex]
\textit{Note:} & \multicolumn{13}{r}{$^{*}$p$<$0.1; $^{**}$p$<$0.05; $^{***}$p$<$0.01} \\
\end{tabular}

}
    \vspace{4mm}
    \caption{\textbf{Elasticity with Hedging.} Value-weighted and winsorized (5$^{th}$-95$^{th}$ percentiles) average elasticity of portfolios (rows) after hedging out various portfolios (columns), using weights from Proposition \ref{prop:hedged}. Sample: 1,253,175 stock-month observations, February 1990 to January 2020. Standard errors (in parentheses) are double-clustered by month and stock. See Table \ref{tab:models} for model initialisms.
}
    \label{tab:cross_elasticity}
\end{table}

Table \ref{tab:cross_elasticity} shows the elasticity of these portfolios, with a similar layout. Like above, the elasticity terms are winsorized in exactly the same way as above, and the standard errors are clustered similarly. The take-away from this table is relatively simple: these trading strategies still have relatively inelastic demand. The mean elasticity across all these pairings is 2.1, similar to, but slightly smaller than, the previous 2.6 average elasticity of the unhedged models. 

\subsection{Systematic Elasticity} \label{subsec:systematic_elasticity}

In the previous section, we considered price changes ceteris paribus, which is much like an idiosyncratic price change. It is natural to also estimate the systematic elasticity of these models, that is, the elasticity of an individual stock when the price of a systematic factor is shocked. Consider some factor $j$ with an excess return $F_{j,t}$ and a positive price $p_{j,t}^F$. Consider a simple factor relationship with an error term $\nu_{i,j,t}$:
\begin{equation} \label{eq:factor_relationship}
    r_{i,t} = \beta_{i,j,t-1}^F F_{j,t} + \nu_{i,j,t}.
\end{equation}
I could also include controls to hold fixed an alpha (intercept) term or hold fixed other factors, but for simplicity, I just consider a simple factor relationship without holding these terms fixed. The systematic elasticity is defined as $\eta_{i,t}^{j,sys} = - \partial \log (s_{i,t}) / \partial \log (p_{j,t}^F)$. In appendix \ref{app:systematic_elasticity} I derive a formula for the systematic elasticity which is used for the empirical results below, but I also derive the following approximation:
\begin{equation} \label{eq:sys_elas_approx}
    \eta_{i,t}^{j,sys} \equiv - \frac{\partial \log (s_{i,t})}{\partial \log (p_{j,t}^F)}
    \approx \underbrace{\beta_{i,j,t-1}^F
    \eta_{i,t}}_{\text{Direct}} 
    \underbrace{- \left( \sum_{n \neq i} \frac{\partial \log (w_{i,t})}{\partial \log (p_{n,t})} \beta_{n,j,t-1}^F \right)}_{\text{Cross}}
    \underbrace{- \left( \sum_n w_{n,t-1} \beta_{n,j,t-1}^F \right)}_{\text{Wealth}}.
\end{equation}
This breaks the systematic elasticity into the three components labeled above. The direct effect is approximately the ceteris paribus elasticity, scaled by the beta on the factor. The cross-term takes into account the fact that portfolio weights are a function of the entire distribution of prices. Finally, since these are systematic shocks, wealth effects are not necessarily negligible.

As long as we have the factors and betas in equation (\ref{eq:systematic_math}), we can calculate this systematic derivative. To obtain factors, I calculate the first three principal component (PC) portfolios of the predictor-weighted portfolios shown in equation (\ref{eq:port_returns}), since the first few PCs explain much of the variation across portfolios \citep[e.g.,][]{giglio2021asset}. The subsequent weights on individual stocks are often negative, but I want to consider long-only factors with unlevered prices. Therefore, for each PC portfolio, I construct two subsequent portfolios: the long end with weights scaled to sum to one, and the short end with weights similarly scaled. Thus, we have six portfolios, where $2/-$ denotes the short end of the second PC. Note that the PC portfolios are not created to have a positive factor premia, so the long and short end portfolios have a similar interpretation as one side of a systematic PC portfolio. I estimate $\beta_{j,t}^F$ by modeling it as a linear function of all normalized predictors $\udot{z}_{i,k,t}$, following the approach of \cite{kelly} and \cite{pastor2003liquidity}. To estimate this beta, I run a panel regression of stock returns on the interactions of the factor return with the normalized predictors $\udot{z}_{i,k,t}$.\footnote{See Appendix~\ref{app:covariance} for a similar approach with fewer characteristics.}

Table \ref{tab:sys_estimates} shows the value-weighted and winsorized systematic elasticities from all six systematic portfolios. These elasticity values tend to be between -1 and 1 for most stocks on average, closer to zero than the ceteris paribus elasticities. These systematic elasticities are smaller because the cross terms tend to negate the direct effects and because the portfolio tend to be somewhat hedged against the systematic shocks, meaning the betas are relatively clustered around zero. The wealth effects tend to be small because these betas are clustered around zero.\footnote{See Appendix \ref{app:systematic_elasticity} for an empirical decomposition into the three components from equation (\ref{eq:sys_elas_approx}).}


\begin{table}[!t] \centering
    \resizebox{1.\textwidth}{!}{\begin{tabular}{@{\extracolsep{5pt}}lccccccccccccc}
\\[-1.8ex]\hline
\hline \\[-1.8ex]
& \multicolumn{13}{c}{\textit{Average Systematic Elasticity}} \
\cr \cline{2-14}
\\[-1.8ex] PCA & \multicolumn{1}{c}{BPZ$_F$} & \multicolumn{1}{c}{BPZ$_L$} & \multicolumn{1}{c}{BSV} & \multicolumn{1}{c}{CRW} & \multicolumn{1}{c}{DGU} & \multicolumn{1}{c}{FF3} & \multicolumn{1}{c}{FF6} & \multicolumn{1}{c}{GKX} & \multicolumn{1}{c}{HXZ} & \multicolumn{1}{c}{KNS} & \multicolumn{1}{c}{KPS} & \multicolumn{1}{c}{NN} & \multicolumn{1}{c}{RF}  \\
\\[-1.8ex] & (1) & (2) & (3) & (4) & (5) & (6) & (7) & (8) & (9) & (10) & (11) & (12) & (13) \\
\\[-1.8ex] \hline \\[-1.4ex]

1/+& 0.199$^{***}$ & -0.991$^{***}$ & -0.005$^{}$ & 0.326$^{***}$ & -0.224$^{***}$ & -0.170$^{***}$ & -0.028$^{***}$ & 0.326$^{***}$ & -0.034$^{***}$ & -0.254$^{***}$ & -0.056$^{**}$ & -0.248$^{***}$ & 0.027$^{***}$ \\
 & (0.022) & (0.062) & (0.014) & (0.010) & (0.017) & (0.009) & (0.007) & (0.019) & (0.004) & (0.021) & (0.028) & (0.026) & (0.010) \\
 2/+ & 0.414$^{***}$ & -0.806$^{***}$ & 0.169$^{***}$ & 0.479$^{***}$ & -0.063$^{***}$ & 0.042$^{***}$ & 0.135$^{***}$ & 0.453$^{***}$ & 0.129$^{***}$ & -0.088$^{***}$ & -0.005$^{}$ & -0.192$^{***}$ & 0.145$^{***}$ \\
 & (0.027) & (0.074) & (0.017) & (0.014) & (0.017) & (0.010) & (0.010) & (0.029) & (0.007) & (0.024) & (0.046) & (0.029) & (0.013) \\
 3/+ & 0.473$^{***}$ & -0.646$^{***}$ & 0.202$^{***}$ & 0.590$^{***}$ & 0.001$^{}$ & 0.025$^{**}$ & 0.175$^{***}$ & 0.576$^{***}$ & 0.158$^{***}$ & -0.014$^{}$ & 0.063$^{}$ & -0.181$^{***}$ & 0.188$^{***}$ \\
 & (0.030) & (0.076) & (0.019) & (0.017) & (0.017) & (0.010) & (0.011) & (0.034) & (0.008) & (0.026) & (0.061) & (0.029) & (0.016) \\
 1/-& 0.608$^{***}$ & -0.637$^{***}$ & 0.251$^{***}$ & 0.680$^{***}$ & 0.055$^{**}$ & 0.099$^{***}$ & 0.227$^{***}$ & 0.618$^{***}$ & 0.242$^{***}$ & 0.010$^{}$ & 0.104$^{}$ & -0.203$^{***}$ & 0.265$^{***}$ \\
 & (0.037) & (0.093) & (0.024) & (0.020) & (0.022) & (0.012) & (0.012) & (0.042) & (0.012) & (0.033) & (0.084) & (0.033) & (0.020) \\
 2/-& 0.289$^{***}$ & -1.207$^{***}$ & -0.017$^{}$ & 0.493$^{***}$ & -0.252$^{***}$ & -0.238$^{***}$ & -0.031$^{***}$ & 0.455$^{***}$ & -0.021$^{***}$ & -0.321$^{***}$ & -0.046$^{}$ & -0.323$^{***}$ & 0.063$^{***}$ \\
 & (0.030) & (0.081) & (0.020) & (0.014) & (0.021) & (0.012) & (0.009) & (0.027) & (0.005) & (0.028) & (0.046) & (0.033) & (0.014) \\
 3/-& 0.227$^{***}$ & -1.230$^{***}$ & -0.003$^{}$ & 0.381$^{***}$ & -0.266$^{***}$ & -0.164$^{***}$ & -0.025$^{***}$ & 0.354$^{***}$ & -0.025$^{***}$ & -0.317$^{***}$ & -0.080$^{**}$ & -0.309$^{***}$ & 0.032$^{***}$ \\
 & (0.027) & (0.078) & (0.017) & (0.011) & (0.020) & (0.010) & (0.008) & (0.022) & (0.004) & (0.026) & (0.035) & (0.033) & (0.012) \\

\hline
\hline \\[-1.8ex]

\textit{Note:} & \multicolumn{13}{r}{$^{*}$p$<$0.1; $^{**}$p$<$0.05; $^{***}$p$<$0.01} \\
\end{tabular}

}
    \vspace{4mm}
    \caption{\textbf{Systematic Elasticity Estimates.} Value-weighted and winsorized (5$^{th}$-95$^{th}$ percentiles) average systematic elasticity for six PCA portfolios: $1/+$, $2/+$, $3/+$, $1/-$, $2/-$, $3/-$. Sample: February 1990 to January 2020. Standard errors (in parentheses) are double-clustered by month and stock.}
    \label{tab:sys_estimates}
\end{table}

\section{Equilibrium Model} \label{sec:model}

The purpose of this section is to demonstrate that when hypothetical statistical arbitrageurs (HSAs)---investors who follow statistical arbitrage strategies---are inserted into an equilibrium model counterfactually, demand remains inelastic. This aligns with the rest of the paper, which shows that HSA demand is inelastic. I first present the equilibrium model accommodating HSAs and compare it with the KY demand system model. Finally, I show that allocating capital to HSAs counterfactually keeps individual stock demand inelastic. Note that I use the term HSA to avoid confusion with actual statistical arbitrageurs in the market.


\subsection{Simple Demand System Model that can Accommodate HSAs} \label{subsec:demand_for_hsas}

As discussed above, the BSV and similar models have the following functional form:
\begin{equation} \label{eq:second_stage1}
w_{i,t}^j = \hat \beta_{0,t}^j + \sum_{k=1}^K \hat \beta_{k,t}^j z_{i,k,t} + \epsilon_{i,t}^j,
\end{equation} 
where the $\hat \beta_{k,t}^j$ terms are the loadings on predictors and $z_{i,k,t}$ are the predictor variables. The KY demand function is actually exponential-linear in predictors, and thus cannot accommodate these HSAs well. Therefore, I adapt the KY model somewhat to accommodate these HSAs. 

I divide the asset predictors into three sets: (1) the set containing just the single market weight predictor $z_{i,1,t} = P_{i,t} / A_t$ which I assume is the first predictor ($k = 1$), (2) the set, $\mathcal{D}$, containing the endogenous predictors (i.e., predictors that are functions of prices) other than the market weights, and (3) the set, $\mathcal{X}$, containing the exogenous predictors. In order to estimate the coefficients on endogenous price terms with transparent linear models for predictors in the endogenous set (i.e., $k \in \mathcal{D}$), I consider log-linearized predictors:
\begin{equation} \label{eq:demand_function_predictors}
    \udot{z}_{i,k,t} \approx z_{i,k,t} \equiv a_{i,1,k,t} + a_{2,k,t} \log (P_{i,t}),
\end{equation}
where $a_{i,k,t}$ is an exogenous component of the predictor and $a_{k,t}$ is constant across assets in a given period. If a predictor is exogenous, then I just define $z_{i,k,t} \equiv \udot{z}_{i,k,t}$. I also consider rescaling the endogenous predictors by $B_{i,t}^c / C_{t}$, where $B_{i,t}$ is book equity, $c$ is some positive scalar, and $C_{t} \equiv \sum_i B_{i,t}^c$.\footnote{The term $C_t$ does not matter, since changing $C_t$ just results in a change of $\hat \beta_{k,t}^j$ coefficients that offsets any change.} This modifies the BSV-style demand in equation (\ref{eq:second_stage1}) to be:
\begin{equation} \label{eq:second_stage}
    w_{i,t}^j = \hat \beta_{0,t}^j + \hat \beta^j_{1,t} \frac{P_{i,t}}{A_{t}}
    + \sum_{k \in \mathcal{D}} \hat \beta_{k,t}^j \frac{B_{i,t}^c}{C_t} \left( a_{i,1,k,t} + a_{2,k,t} \log (P_{i,t}) \right) + \sum_{k \in \mathcal{X}} \hat \beta_{k,t}^j z_{i,k,t} + \epsilon_{i,t}^j.
\end{equation}
The rescaling factor $B_{i,t}^c / C_{t}$ is discussed more below, but without it, the BSV-style demand function produces aggregate demand that has very elastic demand for small stocks and very inelastic demand for large stocks. 

This means I can write the demand function simply as:
\begin{equation} \label{eq:indiv_demand}
    w_{i,t}^j = \underbrace{\beta^j_{i,0,t}}_{\substack{\text{intercept}\\\text{term}}} + \underbrace{\beta^j_{1,t} \frac{P_{i,t}}{A_{t}}}_{\substack{\text{level}\\\text{term}}} + \underbrace{\beta_{i,2,t}^j \log (P_{i,t})}_{\substack{\text{log}\\\text{term}}},
\end{equation}
where
\begin{equation} \label{eq:demand_connection}
    \beta_{i,0,t}^j = \hat \beta_{0,t}^j + \sum_{k \in \mathcal{D}} \hat \beta_{k,t}^j 
 \frac{B_{i,t}^c}{C_{t}} a_{i,1,k,t} 
    + \sum_{k \in \mathcal{X}} \hat \beta_{k,t}^j z_{i,k,t}
    + \epsilon_{i,t}^j, \;\;\;
    \beta_{1,t}^j = \hat \beta_{1,t}^j, \;\;\;
    \text{and } \;\; \beta_{i,2,t}^j = \sum_{k \in \mathcal{D}} \hat \beta_{k,t}^j \frac{B_{i,t}^c}{C_{t}} a_{2,k,t}.
\end{equation}



The elasticity of institution $j$ for any asset $i$ with a positive portfolio weight, using equation (\ref{eq:elasticity}), can be calculated as:
\begin{equation} \label{eq:elasticity_inst}
    \eta^{j}_{i,t} = \underbrace{1 
    -  \frac{\beta_{1,t}^j P_{i,t}}{w_{i,t}^j A_t}}_{\substack{\text{level}\\\text{term}}}
    \underbrace{- \frac{\beta_{i,2,t}^j}{w_{i,t}^j}}_{\substack{\text{log}\\\text{term}}}.
\end{equation}
If $\beta{_{i,0,t}^j} = \beta{_{2,t}^j} = 0$ and $ \beta{_{1,t}^j} = 1$, then the elasticity is zero. Just as there is a level term and log term in equation (\ref{eq:indiv_demand}), this equation has the corresponding elasticity effects of these two terms. Note that the inclusion of the log term allows potentially higher elasticities than assuming that $\beta{_{i,2,t}^j} = 0$.

This reveals the importance of having both a level and log price term in the demand function, as well as the interpretation of $\beta_{1,t}^j$ and $\beta_{i,2,t}^j$. If $\beta_{i,0,t}^j = \beta_{i,2,t}^j = 0$ and $\beta{_{1,t}^j} = 1$, then fund $j$ is a market-weighted index fund. The parameter $\beta{_{1,t}^j}$ captures institution $j$'s proclivity to market-weight index \citep[e.g., see][]{benchmarking}. The parameter $\beta{_{i,2,t}^j}$ captures the investor's sensitivity to prices, aside from its proclivity to market-weight index. As an investor's sensitivity to the range of valuation ratio variables increases, $\beta{_{i,2,t}^j}$ becomes more negative and the elasticity rises. 

KY impose a constraint on their demand function estimation in order to guarantee positive elasticities (downward sloping demand), which guarantees an equilibrium. I have similar constraints. In particular, I constrain $\beta{_{1,t}^j} < 1$ and $\beta{_{2,t}^j} < 0$ in the estimation discussed below. This ensures that aggregate demand is downward sloping, which, when combined with a similar restriction on statistical arbitrageur demand, is sufficient to ensure that a unique closed-form positive equilibrium price exists for all assets. 

\subsection{Comparison to Isoelastic Demand} \label{subsec:more_elastic}

The KY exponential-linear demand function can be written as:\footnote{KY's exponential-linear demand function is essentially a classic isoelastic demand function with two exceptions. First, it allows portfolio weights to be zero for some assets. Second, it disallows leverage and forces weights to sum to unity, slightly modifying the demand by changing $\partial \log (w_{i,t}) / \partial \log (p_{i,t})$ by a factor of $1 - w_{i,t}$ (tends to be close to 1). For simplicity, I refer to KY's demand as isoelastic demand.}
\begin{equation} \label{eq:ky_demand}
    w_{i,t}^j = \exp (\hat \varrho_{i,1,t}^j + \varrho_{2,t}^j \log(P_{i,t})) 
    = \underbrace{0}_{\substack{\text{intercept}\\\text{term}}} + \underbrace{\varrho_{i,1,t}^j P_{i,t}^{\varrho_{2,t}^j}}_{\substack{\text{level}\\\text{term}}} + \underbrace{0 \cdot \log(P_{i,t})}_{\substack{\text{log}\\\text{term}}},
\end{equation}
where $\hat \varrho{_{i,1,t}^j}$ captures their exogenous predictors including their latent demand term,\footnote{See equation (10) of KY to see how $\hat \varrho{_{i,1,t}^j}$ is parameterized. In particular, it is a linear function of covariates plus the log of their latent demand term. } $\hat \varrho{_{2,t}^j}$ denotes their exponential-linear price coefficient, and $\varrho{_{i,1,t}^j} = \exp (\hat \varrho{_{i,1,t}^j})$. This is classic isoelastic demand. Comparing this to equation (\ref{eq:indiv_demand}), it is clear that the demand function has a similar price level term, $\beta{_{1,t}^j} P_{i,t} / A_t$, as the KY demand function. However, the demand function here includes the intercept term, $\beta{_{i,0,t}^j}$, and log price term, $\beta{_{2,t}^j} \log (P_{i,t})$, which allows the demand function to potentially capture higher elasticities due to this increased flexibility. In summary, isoelastic demand omits the intercept and log term present in a BSV-style model, potentially leading to different results. I refer to this BSV-style demand function model as the level-log demand model (LLDM). 


\subsection{Statistical Arbitrageur Demand}

We can approximate demand for statistical arbitrageurs using the same functional form as in equation (\ref{eq:indiv_demand}), but replacing $\beta^j$ terms with $\udot{\beta}$ terms to denote statistical arbitrageur demand. I consider a simple first order approximation:
\begin{equation} \label{eq:learner_coefficients}
    \underbrace{\udot{\beta}_{i,1,t} = A_t \nabla_{\udot{z}_{i,1,t}} (f(Z_t)_i b) \frac{\partial \udot{z}_{i,1,t}}{\partial P_{i,t}}}_{\substack{\text{level term} \\ \text{coefficient}}} \;\; \text{ and } \;\;
    \underbrace{\udot{\beta}_{i,2,t} = \sum_{k \in \mathcal{D}} \nabla_{\udot{z}_{i,k,t}} (f(Z_t)_i b) \frac{\partial \udot{z}_{i,k,t}}{\partial \log (P_{i,t})}}_{\substack{\text{log term} \\ \text{coefficient}}},
\end{equation}
where similar to incumbent demand, in the experiments the values are constrained such that $\udot{\beta}_{i,1,t} < 1$ and $\udot{\beta}_{i,2,t} < 0$. After these slope terms are calculated, the intercept in the demand function is pinned down by the level of portfolio weights output by the statistical arbitrageur function. This gives us the portfolio weights of statistical arbitrageurs:
\begin{equation} \label{eq:arb_demand_approx}
    \udot{w}_{i,t} = \udot{\beta}_{i,0,t} + \udot{\beta}_{i,1,t} \frac{P_{i,t}}{A_t} + \udot{\beta}_{i,2,t} \log (P_{i,t}).
\end{equation}

\subsection{Equilibrium} \label{subsec:equilibrium_solution}

Following KY, the demand function for fund $j$, shown in equation (\ref{eq:indiv_demand}), is specified only for the assets in the fund's investment universe, which is the investment mandate or the set of assets a fund considers. Let $\one{_{i,t}^j}$ equal one if asset $i$ is in the investment universe of fund $j$ at time $t$, and zero otherwise. The dollar holdings of the asset for this institution are $A{_t^j} w{_{i,t}^j}$. HSAs do not have constrained investment universes, and so $\one{_{i,t}^j} = 1$ for a HSA.

It follows that aggregate demand for the asset, in terms of portfolio weights, is:
\begin{align} \label{eq:agg_demand}
    w_{i,t} &\equiv \frac{1}{A_t} \sum_j \one_{i,t}^j A_t^j w_{i,t}^j \nonumber\\
    &= \left( \frac{1}{A_t} \sum_j \one_{i,t}^j A_t^j \beta^j_{i,0,t} \right) + \left( \frac{1}{A_{t}} \sum_j \one_{i,t}^j A_t^j \beta^j_{1,t} \right) \frac{P_{i,t}}{A_t} + \left( \frac{1}{A_t} \sum_j \one_{i,t}^j A_t^j \beta_{i,2,t}^j \right) \log (P_{i,t}) \nonumber\\
    & = \beta_{i,0,t} + \beta_{i,1,t} \frac{P_{i,t}}{A_t} + \beta_{i,2,t} \log (P_{i,t}),
\end{align}
where the $\beta_{i,0,t}$, $\beta_{i,1,t}$, and $\beta_{i,2,t}$ coefficients are defined as the terms in parentheses. Thus, all that is needed to fully describe the demand function for the aggregate market is these three parameters for each stock in each period. This equilibrium model has a unique positive price for each stock, with a simple closed form solution.\footnote{Appendix \ref{app:equil_model_statements} shows the formula for the unique equilibrium price, with a concomitant proof.}

The elasticity of the aggregate demand for any asset $i$ can be written as:
\begin{equation} \label{eq:agg_elasticity}
    \eta_{i,t} = \underbrace{1 
    -\beta_{i,1,t}}_{\substack{\text{level}\\\text{term}}}
    \underbrace{- \frac{\beta_{i,2,t}}{P_{i,t} / A_t}}_{\substack{\text{log}\\\text{term}}}.
\end{equation}
This aggregate elasticity is guaranteed to be positive as long as $\beta_{1,t}^j < 1$ and $\beta_{i,2,t}^j < 0$.\footnote{To calculate the aggregate elasticity, note that in equilibrium, $w_{i,t} A_t = P_{i,t}$. A positive elasticity is guaranteed since $\beta_{1,t}^j < 1$ and $\beta_{i,2,t}^j < 0$ ensures that $\beta_{i,1,t} < 1$ and $\beta_{i,2,t} < 0$. }


\subsection{Estimation Details}

The first step of estimating demand is to calculate the demand function predictors from equation (\ref{eq:demand_function_predictors}), which is described in Appendix \ref{subsec:demand_function_predictors}. In short, $a_{2,k,t}$ is calculated from the known derivative of $\udot{z}_{i,k,t}$ with respect to log prices. The exogenous intercept, $a_{i,1,k,t}$, is then calculated with a simple regression of the cross-section every period. 

The demand estimation procedure closely mirrors the KY procedure but with notable differences. The first significant change is that because of the adoption of a straightforward and flexible linear-in-predictors functional form for portfolio weights, I use a simple two-stage least squares estimation procedure. While I use the KY instrument for prices, they use a nonlinear GMM estimation procedure because of their nonlinear demand function. Second, I use a broader set of predictors from the \cite{weber} data. There are some minor differences as well. Like KY, I estimate demand with a separate regression for every institution in every period. Similar to KY, funds are grouped together based on observables if they have too few strictly positive holdings and the aggregate demand across the funds is estimated. KY groups firms together if there are less than 1,000 observations, but I ensure at least 500 observations. This is because my estimation procedure does not have the same convergence issues. More details and discussion about demand estimation is given in Appendix \ref{subsec:demand_estimation_details}.

Note from equation (\ref{eq:agg_demand}) that differences in price elasticity for any asset are driven by the differences in what funds have that asset in their investment universe. KY note "that institutions hold a small set of stocks and that the set of stocks that they have held in the recent past (e.g., over the past 3 years) hardly changes over time." The calculation of the investment universe and instrument relies on this fact. Thus, it is relatively mild to impute the investment universe, price level coefficient, and log price coefficients of funds for the two months following the quarterly 13F holdings data.\footnote{The price level coefficient is $\beta_{1,t}^j$ while the log price coefficient is $\beta_{i,2,t}^j$. The aggregate intercept term, $\beta_{i,0,t}$, is pinned down by the equilibrium condition, $A_t w_{i,t} = P_{i,t}$, combined with equation (\ref{eq:agg_demand}). Thus, only the investment universe, the price level coefficients, and the log price coefficients need to be imputed.} I impute these values, which allows analysis of portfolio returns at the standard monthly frequency with the counterfactual experiments below. This is very different from the problematic assumption that institutions hold the same portfolio weights for the following two months, which I do {not} assume. 

\subsection{Demand Estimation Results} \label{subsec:demand_estimation_results}


Panel A of Figure \ref{fig:demand_system_estimates} shows a binned scatter plot of aggregate demand elasticity estimated with the LLDM for individual assets as a function of market portfolio weights. The blue curve shows results with $c = 0$ in the scaling term ($B_{i,t}^c / C_t$), which is the same as just not having the scaling term. The graph shows a demand elasticity of almost 1,000 for the smallest stocks and an elasticity below 1 for the largest stocks. Equation (\ref{eq:agg_elasticity}) shows that this sharply declining demand elasticity is just a mechanical result from the LLDM from the log term. Panels B and C shows the level and log terms respectively, and indeed the sharp decline is mechanically driven by this BSV-style model. This echoes the discussion above about assets with the smallest portfolio weights having the highest elasticities, but in a demand system model the market clears, so the market portfolio weights are proportional to market equity. 

This is a relatively unappealing property of the LLDM. \cite{aggressive} shows a non-mechanical decreasing elasticity in stock size, but not to this degree. To reiterate, the goal in this section is not to present a demand system model that dominates the KY model, but rather, to present an extension of KY that accommodates the HSA demand functions. One way to do this with more balanced cross-sectional elasticities would be to scale $\beta_{i,2,t}$ down for smaller stocks with an exogenous term, which is the point of the $B_{i,t}^c / C_t$ factor. Figure \ref{fig:demand_system_estimates} shows results for scaling factors with $c = $ 1, 1.25, 1.5, 1.75, and 2. It turns out that $c = 1.25$ produces the most flat log term (see Panel C), and also produces the highest median $R^2$ value calculated across institutions and quarters. The $R^2$ values are shown in the first column of Table \ref{tab:r2_demand_system}. I use the LLDM model with $c = 1.25$ for the counterfactual experiments, which has an average elasticity of about 0.81, which is relatively stable across assets. 

For comparison, I estimate a KY-style model by regressing the log of market equity on their set of covariates with their instrument.\footnote{I use normalized (${z}_{i,k,t}$ from equation (\ref{eq:demand_function_predictors})) versions of their covariates/predictors: book-to-market, book equity, beta, profitability, investment, and dividends. While not shown, I also estimated the KY model using their GMM estimation which has a similar elasticity across assets, but the $R^2$ results are higher for the OLS regression version.} The elasticities are shown in Panel A, and since the entire elasticity is driven by only level effects, this is shown in Panel B as well. I also estimate a simple LLDM with no log terms or covariates other than the market, and Panel B shows that this actually produces a relatively similar elasticity to the KY model. When a LLDM is fit without a log term, but an intercept term and the exogenous KY covariates are also used as controls, the elasticity actually increases by about 0.2 as shown in Panel B. The other LLDM estimates also have higher level terms than the KY model, except for the $c = 0$ LLDM. 

\begin{figure}[!t]
    \centering
    \begin{minipage}{0.49\textwidth}
        \centering
        \textbf{Panel A: Elasticity Across Models}
        \includegraphics[width=\linewidth,trim={0 0 0 1cm},clip]{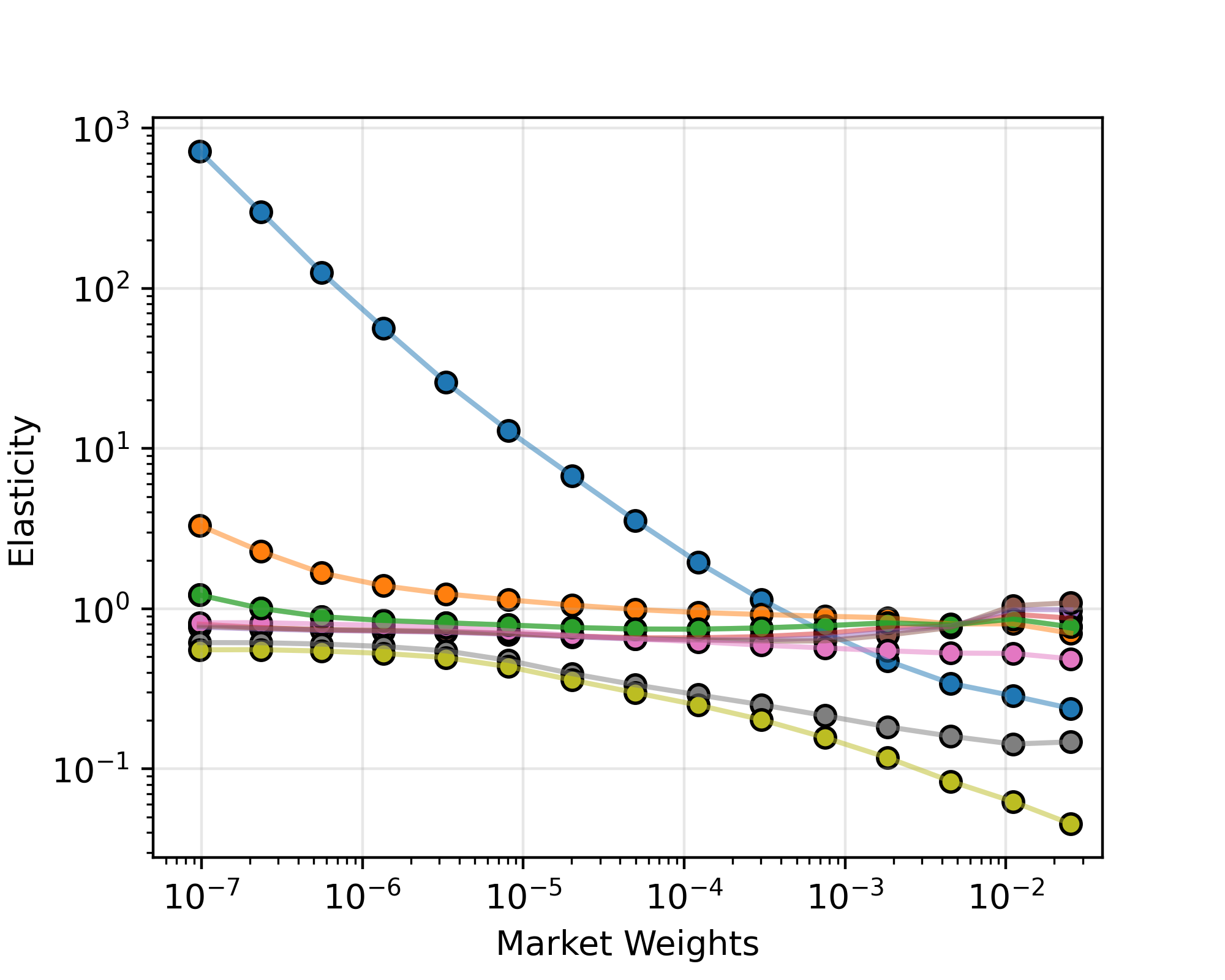}
    \end{minipage}\hfill
    \begin{minipage}{0.49\textwidth}
        \centering
        \textbf{Panel B: Level Term Elasticity Component}
        \includegraphics[width=\linewidth,trim={0 0 0 1cm},clip]{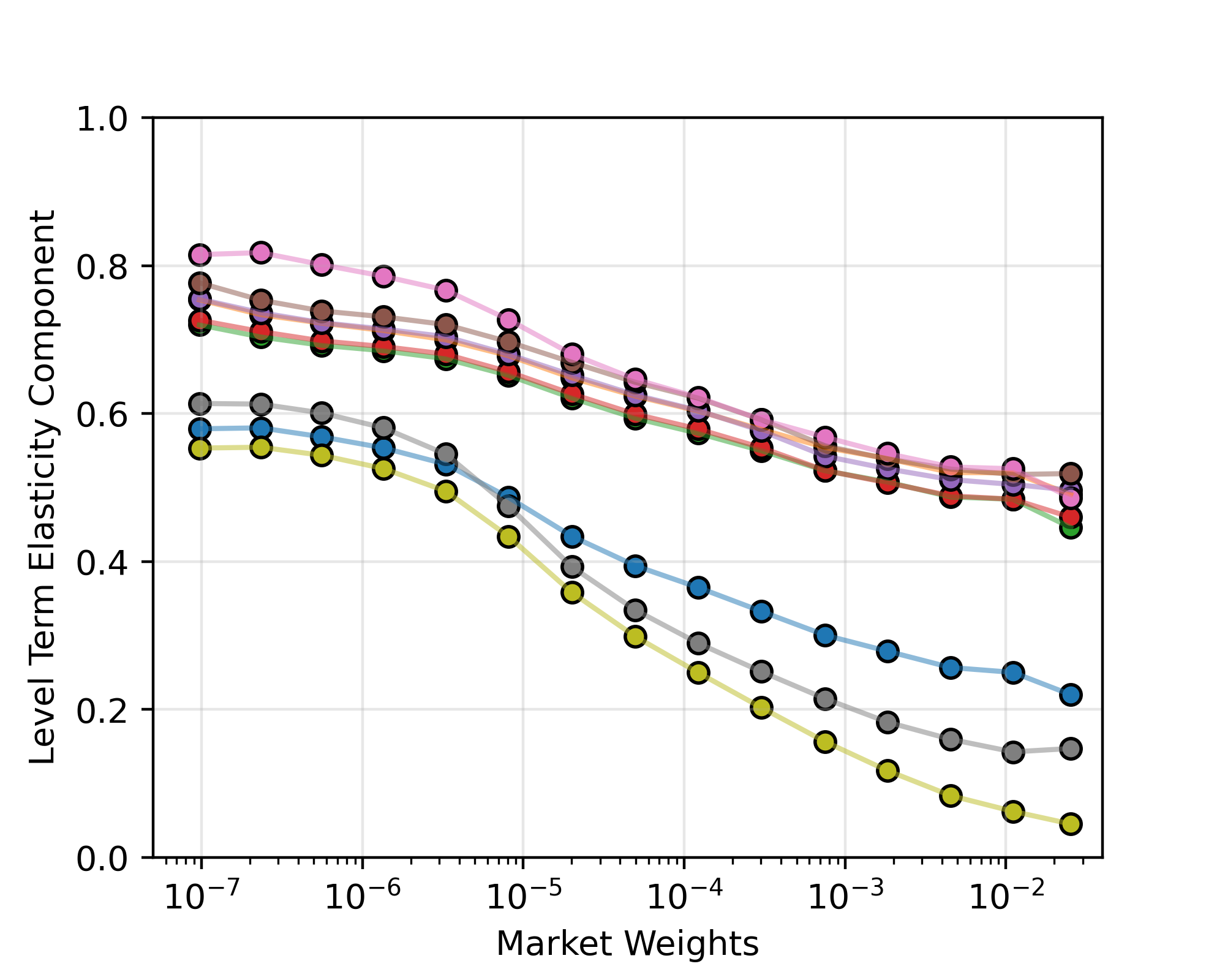}
    \end{minipage}
    \begin{minipage}{0.49\textwidth}
        \centering
        \vspace{4mm}
        \textbf{Panel C: Log Term Elasticity Component}
        \includegraphics[width=\linewidth,trim={0 0 0 1cm},clip]{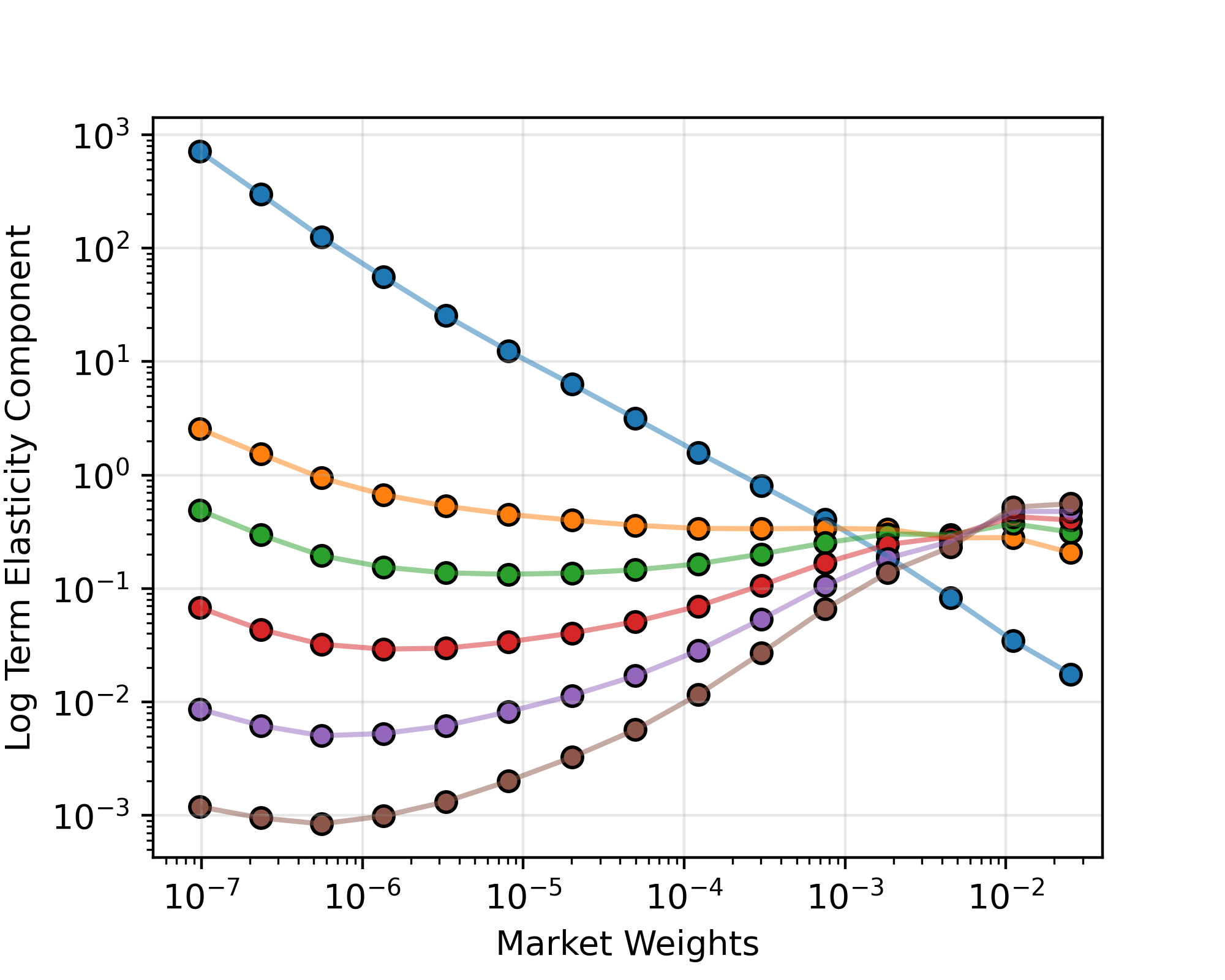}
    \end{minipage}\hfill
    \begin{minipage}{0.49\textwidth}
        \centering
        \vspace{4mm}
        \includegraphics[width=\linewidth]{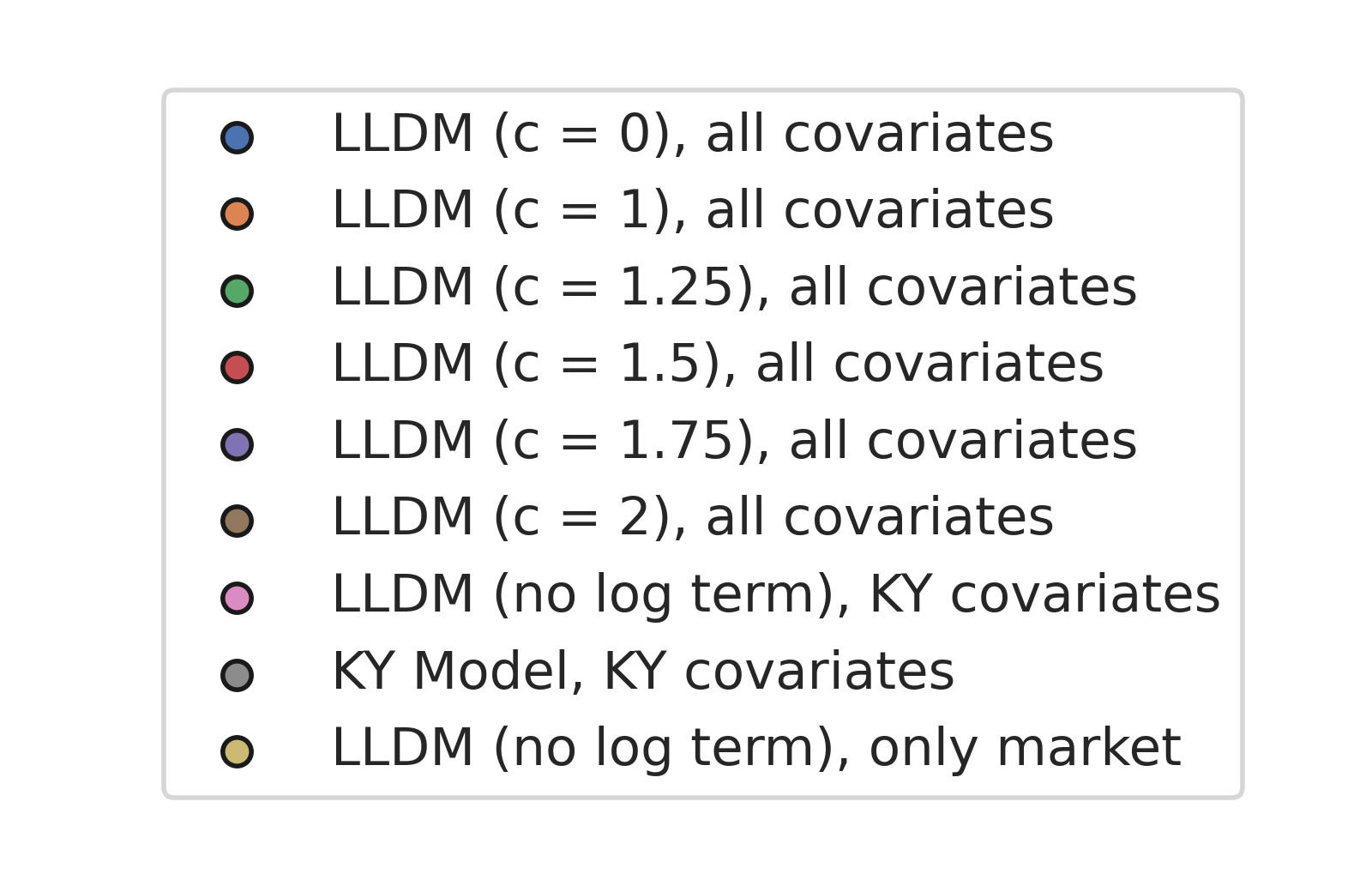}
    \end{minipage}
    \vspace{4mm}
    \caption{\textbf{Demand System Estimates.} Binned scatter plots. Panel A: overall elasticity of various models vs. market portfolio weights (log-scale). Panels B and C: level and log term components from equation (\ref{eq:agg_elasticity}). Models include LLDM ($c = 0, 1, 1.25, 1.5, 1.75, 2$), KY-style, LLDM with only KY covariates (no log terms), and LLDM with only a level term (no log terms). Winsorized in each bin (5$^{th}$-95$^{th}$ percentiles). Sample: 1,253,175 stock-month observations, February 1990 to January 2020. See Table \ref{tab:models} for model initialisms.}
    \label{fig:demand_system_estimates}
\end{figure}

The LLDM model fits the holdings data relatively well but is difficult to compare directly to the KY model. The median $R^2$, across institutions and time, of the KY holdings data in terms of log portfolio weights is $0.28$. However, since the LLDM allows short positions by design, comparing performance in terms of log portfolio weights is challenging---taking the logarithm of negative or small predicted portfolio weights leads to undefined or extremely negative values. Table~\ref{tab:r2_demand_system} presents median $R^2$ values calculated using portfolio weights rather than log portfolio weights. The LLDM with $c = 1.25$ achieves an $R^2$ of $0.47$, which is higher than the KY model's $R^2$ of $0.07$.

The higher $R^2$ values of the LLDM are primarily due to its better fit of the holdings data for the largest stocks. To assess how $R^2$ varies across the size distribution, I partition stocks each period into five size quintiles and calculate the median $R^2$ across institutions and time within each quintile, as shown in Table~\ref{tab:r2_demand_system}. Specifically, for each asset $i$, time $t$, and institution $j$, I compute 
\begin{equation*}
    (R^2)_{i,t}^j = 1 - \frac{(w_{i,t}^j - \hat{w}_{i,t}^j)^2}{ \sigma^2(w_t^j)},
\end{equation*}
where $\hat{w}_{i,t}^j$ is the fitted weight and $\sigma^2(w_t^j)$ is the variance of institution $j$'s portfolio weights at time $t$. Averaging $(R^2)_{i,t}^j$ across assets gives the standard regression $R^2$. I then calculate the average $(R^2)_{i,t}^j$ within each quintile and take the median across institutions and time within that quintile. For the bottom four quintiles, the $R^2$ for every model exceeds $0.85$ and is relatively similar. However, for the largest quintile of stocks, the $R^2$ varies significantly: the KY model has a negative $R^2$ of $-0.3$, while the LLDMs with the full set of covariates achieve $R^2$ values around $0.3$. This indicates that the LLDM's higher $R^2$ values are driven by better fitting the holdings data for the largest stocks, which is particularly challenging for a log-weight model like KY. Note that the average $R^2$ across quintiles in Table~\ref{tab:r2_demand_system} does not equal the overall $R^2$ because each institution is not equally represented in each quintile and these are median $R^2$ values. 

\begin{table}[!t] \centering
    \resizebox{0.7\textwidth}{!}{\begin{tabular}{@{\extracolsep{5pt}}llcccccc}
\\[-1.8ex]\hline
\hline \\[-1.8ex]
& & \multicolumn{6}{c}{\textit{$R^2$}} \
\cr \cline{3-8}
\\[-1.8ex] & & \multicolumn{5}{c}{Quintile} & \multicolumn{1}{c}{Overall} \\
\cline{3-7} \cline{8-8}
\\[-1.8ex] Model & Covariates & 1 & 2 & 3 & 4 & 5 & \\
\\[-2.2ex] \hline \\[-1.4ex]

KY Model & KY & 1.000 & 0.997 & 0.976 & 0.888 & -0.303 & 0.070 \\
LLDM (no log term) & market & 0.923 & 0.920 & 0.907 & 0.852 & -0.385 & 0.029 \\
LLDM (no log term) & KY & 0.975 & 0.969 & 0.954 & 0.888 & -0.023 & 0.263 \\
LLDM ($c = 0$) & all & 0.922 & 0.918 & 0.908 & 0.863 & 0.224 & 0.417 \\
LLDM ($c = 1$) & all & 0.934 & 0.928 & 0.917 & 0.870 & 0.297 & 0.465 \\
LLDM ($c = 1.25$) & all & 0.934 & 0.928 & 0.917 & 0.870 & 0.298 & 0.467 \\
LLDM ($c = 1.5$) & all & 0.932 & 0.927 & 0.916 & 0.870 & 0.295 & 0.466 \\
LLDM ($c = 1.75$) & all & 0.931 & 0.925 & 0.915 & 0.869 & 0.291 & 0.464 \\
LLDM ($c = 2$) & all & 0.929 & 0.924 & 0.915 & 0.868 & 0.288 & 0.462 \\

\hline
\hline \\[-1.8ex]

\end{tabular}}
    \vspace{4mm}
    \caption{\textbf{Goodness of Fit Measures.} Median $R^2$ values for models in Figure \ref{fig:demand_system_estimates}, reported by size quintiles and overall (last column). Quintile $R^2$ values are calculated within each size quintile (as an average $(R^2)_{i,t}^j$ within the quintile, as described in the text), then the median is reported across institutions and quarters. The final column shows the median $R^2$ across all institutions and quarters.}
    \label{tab:r2_demand_system}
\end{table}

\subsection{Counterfactual Experiments} \label{subsec:counter_experiments}

The counterfactual experiments are as follows. For some $\theta$ between zero and one, a total of $\theta A_t$ dollars is given to an atomistic mass of one type of HSA, and the rest of the investors whose demand function is estimated above are given $(1 - \theta) A_t$ dollars.\footnote{The fraction $\theta$ can be endogenized based on performance, and the results are similar. See Appendix \ref{subsec:recursive_exper} for more details.} I follow KY by assuming shares outstanding, dividends, and the other non-price-related predictor variables are exogenous as given in the actual data, while prices are endogenized. The HSA investors recursively refit the parameters in their models each period using the time series of endogenous returns. Specifically, the HSA demand function submitted in period $t$ has model parameters fit up to and including $t-1$ returns. 

The key result of the counterfactual experiments is that HSA demand barely changes the aggregate elasticity for stocks on average, if at all. I use various $\theta$ values up to 0.25 (HSAs control 25\% of the stock market), and calculate the how much the value-weighted elasticity shifts due to HSAs in the market. All elasticity changes are between a drop of 0.2 and a rise of 0.2, although generally the change in elasticity is much closer to zero.\footnote{Instead of showing a table with many near-zero values in the text, the table showing this result, Table \ref{tab:mkt_elasticity_simple}, is relegated to the Appendix.} The FF3 HSA actually decreases the elasticity of the market slightly, while the GKX model increases the elasticity of the market slightly.

An intuitive consequence of the inelastic demand of these statistical arbitrageurs is that they are unable to arbitrage away the alpha, as shown in Appendix \ref{subsec:arb_away_alpha}. By examining alpha across various anomalies, I show that alpha largely persists in the counterfactual equilibrium.


\subsection{Validity of Counterfactual Experiments} \label{subsec:validity}

A key concern in counterfactual experiments is a Lucas critique, which questions whether the estimated demand functions in the observed equilibrium remain valid under counterfactual scenarios. This issue becomes more pronounced as the counterfactual equilibrium diverges from the observed equilibrium. However, if the counterfactual equilibrium is close to the observed equilibrium, this concern is largely alleviated. 

One potential change among non-HSA investors could be a reduction in elasticity: if HSAs were to eliminate alpha, non-HSA investors might adopt passive strategies, decreasing market elasticity. However, as discussed above and shown in Appendix \ref{subsec:arb_away_alpha}, alpha largely remains in the counterfactual equilibrium, making such a shift towards more inelastic investing less likely.

Are the results robust if HSA-controlled capital evolves based on performance? In Appendix \ref{subsec:recursive_exper}, I find that the endogenous amount of capital HSAs control in equilibrium is quite small. This means that, if anything, the above results overstate the impact of HSAs on the elasticity of demand. These experiments also show that the counterfactual experiments remain close to the factual equilibrium, mitigating concerns raised by the Lucas critique.

This approach aligns well with the existing literature. \cite{aggressive} examine investors who strategically reduce their own demand elasticity when the market becomes more elastic. Their model addresses this Lucas critique and finds that for an increase in aggregate demand elasticity of 1, an investor would reduce their demand elasticity by 3. In this study, since the shift in aggregate demand elasticity is minimal, applying these estimates results in negligible changes in investor behavior. Additionally, \cite{koijen2023investors} conduct counterfactual experiments involving significant shifts, such as transitions to passive and green investments. In contrast, this study focuses on the inclusion of HSAs, which represents a much smaller change. This inherently reduces concerns that non-HSA investors would significantly alter their demand in response. Thus, both the small aggregate elasticity shift and the minor nature of the counterfactual change support the validity of the counterfactual experiments.


While I discuss ex-post alphas above, what matters for non-HSA investors' pursuit of alpha and elasticity is their ex-ante beliefs about alpha. There are two key reasons non-HSA investors would believe that ex-ante alpha will not be eliminated. First, a rational expectations argument: since counterfactual experiments show alpha remains ex-post, especially with endogenously updated capital controlled by HSAs, it is reasonable to believe ex-ante that alpha will change only slightly. Second, \cite{jensen2023there} show that most classic alpha-generating strategies have largely persisted over time, suggesting that non-HSA investors would reasonably expect alpha to continue. While I cannot rule out the belief that ex-ante alpha will disappear, the most reasonable expectation is that alpha would largely persist even with HSA investors.

We could also consider non-HSA investors becoming more elastic, by exploiting opportunities created by HSA trading. I double the aggregate demand elasticity of non-HSA investors in counterfactual experiments, detailed in Appendix \ref{subsec:double_elas}, and find that even with significant elasticity increases, the HSA investor effects on equilibrium are relatively small. Since the LLDM used in the counterfactual experiments has an average elasticity of 0.8, a doubling puts the elasticity at around 1.6, which is the upper range of the \cite{GK} elasticities for individual stocks across the reviewed literature. In both cases, because the counterfactual equilibrium does not stray from the factual, the concerns raised by a Lucas critique are largely mitigated.

One final potential concern about the counterfactual experiments is that the LLDM, as estimated, may be unable to elastically respond to high-frequency price variation. Equation (\ref{eq:eta_mu}) shows that demand is higher when prices pass through more into near-term returns (i.e., higher $-\partial \mu_{i,t} / \partial \log (p_{i,t})$). The KY instrument uses low-frequency variation, but HSA investors trade on both high-frequency (e.g., short-term reversals) and low-frequency price variation. Thus, the KY instrument may produce demand that is too inelastic for this setting, by missing high-frequency price variation. This is actually another reason why the double elasticity counterfactual experiments discussed above were performed. This raises non-HSA investor elasticity to the top end of the \cite{GK} range, and in the counterfactual experiments, this of course includes price changes of various frequencies. The results are similar with more elastic non-HSA investors. 

Given the complexity of validating counterfactual experiments, is there an alternative method to simply demonstrate that HSA investors should have minimal impact on the aggregate price elasticity of the market? The following proposition aims to address this question. Let $D_{i,t}^j = w_{i,t}^j A_t^j$ be the total dollar amount invested by investor $j$ in asset $i$ at time $t$. Define $D_{i,t} \equiv \sum_j | D_{i,t}^j|$. Define aggregate demand, in terms of total shares demanded, as $ S_{i,t} = \sum_j s_{i,t}^j$. From these definitions, I derive the following proposition:

\begin{prop} \label{prop:aggregation}
    Assume that for all investors with zero positions, i.e., $w_{i,t}^j = 0$, we have:
    \begin{equation} \label{eq:zero_response_for_zero}
        \frac{\partial w_{i,t}^j}{\partial \log (p_{i,t})} = 0.
    \end{equation}
    Let $J_i$ be the set of investors with non-zero positions for asset $i$. Then we have the following result:
    \begin{equation} \label{eq:aggregation_result}
    \eta_{i,t}^{agg} \equiv -\frac{\partial \log (S_{i,t})}{\partial \log (p_{i,t})}
    = 1 + \frac{D_{i,t}}{P_{i,t}} \left( \bar \eta_{i,t} - 1 \right), \; \text{ where } \; \bar \eta_{i,t} = \frac{1}{D_{i,t}} \sum_{j \in J_i} \left|D_{i,t}^j\right| \eta_{i,t}^{j,\pm}.
    \end{equation}
\end{prop}

The proof is given in Appendix \ref{app:aggregation}. Condition (\ref{eq:zero_response_for_zero}) assumes investors with zero positions are not price responsive. I assume this because elasticity is not defined for zero positions, and while this assumption can be relaxed, it complicates the formula slightly. Note that $D_{i,t} = P_{i,t}$ if and only if there are no short sellers. 

This proposition shows that the elasticity of the aggregate market is close to the position-size-weighted elasticity of the investors. Since HSA investors are inelastic and non-HSA investors are inelastic, proposition \ref{prop:aggregation} indicates that the market should be inelastic with HSA investors. It should be noted that all quantities in equation (\ref{eq:aggregation_result}) are endogenous, but this provides some transparent evidence that HSA investors should have little effect on elasticity.

\section{Conclusion} \label{sec:conclusion}

A classic feature of many asset pricing models is the no-arbitrage assumption, which relies on arbitrageurs, or statistical arbitrageurs in some cases, to create market demand that is highly elastic and absorbs non-fundamental flows \citep{huberman1982simple}. This stands in contradiction to elasticity estimates, such as those in KY and surveyed more broadly by \cite{GK}. This large quantitative difference raises an interesting question: what is the elasticity of a friction-free statistical arbitrageur from the literature, and how would these statistical arbitrageurs affect market elasticity? This paper shows that the demand of these statistical arbitrageurs is inelastic compared to classic models. Furthermore, in an equilibrium model, these statistical arbitrageurs affect the market demand elasticity for individual stocks very little if counterfactually given capital to invest. 

\nocite{CRW}
\nocite{carhart}
\newpage
\clearpage
\singlespacing
\bibliography{references.bib}

\doublespacing

\section*{Appendix}

\setcounter{table}{0}
\renewcommand{\thetable}{A.\arabic{table}}
\setcounter{figure}{0}
\renewcommand{\thefigure}{A.\arabic{figure}}

\singlespacing
\begin{table}[H] \centering
    \resizebox{.8\textwidth}{!}{\begin{tabular}{@{\extracolsep{5pt}}lll}
    \\[-1.8ex]\hline
\hline \\[-2.4ex]
Initialism & Factor Method $f(Z_t)$ & MVE Collapse Method $b$ \\
\hline \\[-2.4ex]
    BPZ$_F$ & \cite{forest} forest & \cite{forest} \\
    BPZ$_L$ & linear characteristic-weighted portfolios & \cite{forest} \\
    BSV & linear characteristic-weighted portfolios & \cite{brandt} \\
    CRW & linear \cite{CRW} portfolios & \cite{shrinking} \\
    DGU & linear characteristic-weighted portfolios & \cite{demiguel} \\
    FF3 & \cite{ff3} & \cite{shrinking} \\
    FF6 & \cite{famafrench15} with \cite{carhart} momentum  & \cite{shrinking} \\
    GKX & \cite{gu} & \cite{forest}  \\
    HXZ & \cite{zhang} & \cite{shrinking}  \\
    KNS & linear characteristic-weighted portfolios & \cite{shrinking} \\
    KPS & \cite{kelly} & \cite{forest} \\
    NN & neural network & $b = 1$ \\
    RF & random forest & $b = 1$ 
    \\\\[-2.5ex]
    \hline \hline \\[-1.8ex]
\end{tabular}}
    \vspace{4mm}
    \caption{\textbf{Statistical Arbitrageur Demand Models} Initialisms and sources for the thirteen statistical arbitrageurs studied. Departures from the original descriptions are noted with justifications. Each model includes: (1) a factor model, $f(Z_t)$, and (2) a vector of factor weights, $b$, collapsing factors into a single MVE portfolio.}
    \label{tab:models}
\end{table}
\doublespacing

\clearpage
\numberwithin{equation}{section}
\setcounter{table}{0}
\renewcommand{\thetable}{IA.\arabic{table}}
\setcounter{figure}{0}
\renewcommand{\thefigure}{IA.\arabic{figure}}
\setcounter{footnote}{0} 

\appendix
\onehalfspacing

\renewcommand{\thesection}{IA.\Alph{section}}
\renewcommand{\thesubsection}{\thesection.\arabic{subsection}}
\renewcommand{\thesubsubsection}{\thesubsection.\arabic{subsubsection}}

\numberwithin{equation}{section}
\renewcommand{\theequation}{IA.\Alph{section}.\arabic{equation}}

\addcontentsline{toc}{section}{Appendix}
\part{Internet Appendix}

\parttoc 
\doublespacing


\newpage

\section{Main Proofs and Additional Propositions}

\subsection{Proof of Proposition 1} \label{app:alpha_portfolio}

\begin{proof}
    I first calculate:
    \begin{equation} \label{eq:r_alpha}
        r_{\alpha,t+1} \equiv w_{\alpha,t}' r_{t+1}
        = (w_t - W_t \beta)' r_{t+1} 
        = w_t' r_{t+1} - \beta' W_t' r_{t+1}
        = r_{p,t+1} - \beta' F_{t+1} 
    \end{equation}
    Using this, I can calculate further:
    \begin{equation*}
    r_{\alpha,t+1}
    = r_{p,t+1} - \beta' F_{t+1}
        = (\alpha + \beta' F_{t+1} + \epsilon_{t+1}) - \beta' F_{t+1}
        = \alpha + \epsilon_{t+1}
    \end{equation*}
    This proves that (\ref{eq:requirement1}) is satisfied. 

    Note that by definition, I have:
    \begin{equation*}
        \beta = \text{Var}^{-1} (F_{t+1}) \text{Cov} (r_{p,t+1}, F_{t+1})
    \end{equation*}
    where $\text{Var} (F_{t+1})$ is the covariance matrix of factors, and $\text{Cov} (r_{p,t+1}, F_{t+1})$ is a column vector of covariance terms. Thus, I can write:
    \begin{equation} \label{eq:cov_beta}
        \text{Cov} (r_{p,t+1}, F_{t+1}) = \text{Var} (F_{t+1}) \beta
    \end{equation}
    
    I can also calculate, using (\ref{eq:r_alpha}) and (\ref{eq:cov_beta}):
    \begin{equation*} 
        \text{Cov} (r_{\alpha,t+1}, F_{t+1}) 
        = \text{Cov} (r_{p,t+1} - \beta' F_{t+1}, F_{t+1})
        = \text{Cov} (r_{p,t+1}, F_{t+1}) - \text{Cov} ( F_{t+1}, F_{t+1}) \beta
    \end{equation*}
    \begin{equation*}
        = \text{Var} (F_{t+1}) \beta - \text{Var} (F_{t+1}) \beta = 0
    \end{equation*}
    This proves that (\ref{eq:requirement2}) is satisfied. Note that I can write $\text{Cov} (\beta' F_{t+1}, F_{t+1}) = \text{Cov} ( F_{t+1}, F_{t+1}) \beta$ rather than $\beta' \text{Cov} ( F_{t+1}, F_{t+1})$ so that I have a column vector instead of a row vector in order to be conformable with our (assumed by convention) column vector $\text{Cov} (r_{p,t+1}, F_{t+1})$ in the equation. 

    Now note that I can write the elasticity of this pure-alpha strategy given the definition, as well as the $j^{th}$ factor, as
    \begin{equation*}
        \eta_{\alpha,i,t} = 1 - \frac{1}{w_{\alpha,i,t}} \frac{\partial w_{\alpha,i,t}}{\partial \log (p_{i,t})}
    \end{equation*}
    \begin{equation*}
        \eta^F_{i,j,t} = 1 - \frac{1}{W_{i,j,t}} \frac{\partial W_{i,j,t}}{\partial \log (p_{i,t})}
    \end{equation*}
    Using these two equations, and assuming positive weights, I can write:
    \begin{equation*}
        \frac{\partial w_{\alpha,i,t}}{\partial \log (p_{i,t})} = \frac{\partial w_{i,t}}{\partial \log (p_{i,t})} - \sum_j \frac{\partial W_{i,j,t}}{\partial \log (p_{i,t})} \beta_{j}
    \end{equation*}
    \begin{equation*}
        \implies \eta_{\alpha,i,t} = 1 + \frac{1}{w_{\alpha,i,t}} \left( \sum_j \frac{\partial W_{i,j,t}}{\partial \log (p_{i,t})} \beta_{j} - \frac{\partial w_{i,t}}{\partial \log (p_{i,t})} \right),
    \end{equation*}
    \begin{equation*}
        \implies \eta_{\alpha,i,t} = 1 + \frac{1}{w_{\alpha,i,t}} \left( \sum_j \frac{\partial W_{i,j,t}}{\partial \log (p_{i,t})} \beta_{j} + w_{i,t} \left( \eta_{i,t} - 1 \right) \right)
    \end{equation*}
    \begin{equation*}
        \implies \eta_{\alpha,i,t} = 1 + \frac{1}{w_{\alpha,i,t}} \left(w_{i,t} \left( \eta_{i,t} - 1 \right) - \sum_j W_{i,j,t} \left( \eta^F_{i,j,t} - 1 \right) \beta_{j}\right),
    \end{equation*}
\end{proof}

\subsection{Proof of Proposition \ref{prop:aggregation}} \label{app:aggregation}

\begin{proof}
This proof is simply a series of direct calculations. 
    \begin{equation*}
    S_{i,t} = \sum_j s_{i,t}^j = \sum_j \frac{A_t^j w_{i,t}^j}{p_{i,t}}.
\end{equation*}
\begin{equation*}
    \implies \log (S_{i,t}) = -\log(p_{i,t}) + \log \left( \sum_j A_t^j w_{i,t}^j \right).
\end{equation*}
\begin{equation*}
    \implies -\frac{\partial \log (S_{i,t})}{\partial \log (p_{i,t})}
    = 1 - \frac{1}{P_{i,t}} \left( \sum_j A_t^j \frac{\partial w_{i,t}^j}{\partial \log (p_{i,t})} \right)
    = 1 - \frac{1}{P_{i,t}} \left( \sum_{j \in J_i} A_t^j \frac{\partial w_{i,t}^j}{\partial \log (p_{i,t})} \right),
\end{equation*}
where above, I use the fact that in equilibrium, $P_{i,t} = \sum_{j} D_{i,t}^j = \sum A_t^j w_{i,t}^j$ as well as the assumption in equation (\ref{eq:zero_response_for_zero}). We can then calculate:
\begin{equation*}
    -\frac{\partial \log (S_{i,t})}{\partial \log (p_{i,t})}
    = 1 - \frac{1}{P_{i,t}} \left( \sum_{j \in J_i} A_t^j \frac{\left| w_{i,t}^j \right|}{\left| w_{i,t}^j \right|} \frac{\partial w_{i,t}^j}{\partial \log (p_{i,t})} \right).
\end{equation*}
\begin{equation*}
    \implies -\frac{\partial \log (S_{i,t})}{\partial \log (p_{i,t})}
    = 1 - \frac{1}{P_{i,t}} \left( \sum_{j \in J_i} \left|D_{i,t}^j\right| \left( \frac{1}{\left| w_{i,t}^j \right|} \frac{\partial w_{i,t}^j}{\partial \log (p_{i,t})} \right) \right).
\end{equation*}
\begin{equation*}
    \implies -\frac{\partial \log (S_{i,t})}{\partial \log (p_{i,t})}
    = 1 + \frac{1}{P_{i,t}} \left( \sum_{j \in J_i} \left|D_{i,t}^j\right| \left( -\frac{1}{\left| w_{i,t}^j \right|} \frac{\partial w_{i,t}^j}{\partial \log (p_{i,t})} \right) \right).
\end{equation*}
\begin{equation*}
    \implies -\frac{\partial \log (S_{i,t})}{\partial \log (p_{i,t})}
    = 1 + \frac{1}{P_{i,t}} \left( \sum_{j \in J_i} \left|D_{i,t}^j\right| \left( 1 -\frac{1}{\left| w_{i,t}^j \right|} \frac{\partial w_{i,t}^j}{\partial \log (p_{i,t})} -1\right) \right).
\end{equation*}
\begin{equation*}
    \implies -\frac{\partial \log (S_{i,t})}{\partial \log (p_{i,t})}
    = 1 + \frac{1}{P_{i,t}} \left( \sum_{j \in J_i} \left|D_{i,t}^j\right| \left( 1 -\frac{1}{\left| w_{i,t}^j \right|} \frac{\partial w_{i,t}^j}{\partial \log (p_{i,t})} \right)  - D_{i,t} \right).
\end{equation*}
\begin{equation*}
    \implies -\frac{\partial \log (S_{i,t})}{\partial \log (p_{i,t})}
    = \left(1 - \frac{D_{i,t}}{P_{i,t}} \right) + \frac{1}{P_{i,t}} \left( \sum_{j \in J_i} \left|D_{i,t}^j\right| \eta_{i,t}^{j,\pm} \right).
\end{equation*}
\begin{equation*}
    \implies -\frac{\partial \log (S_{i,t})}{\partial \log (p_{i,t})}
    = \left(1 - \frac{D_{i,t}}{P_{i,t}} \right) + \frac{D_{i,t}}{P_{i,t}} \left( \frac{1}{D_{i,t}} \sum_{j \in J_i} \left|D_{i,t}^j\right| \eta_{i,t}^{j,\pm} \right).
\end{equation*}
\begin{equation*}
    \implies \eta_{i,t}^{agg} = -\frac{\partial \log (S_{i,t})}{\partial \log (p_{i,t})}
    = 1 - \frac{D_{i,t}}{P_{i,t}} + \frac{D_{i,t}}{P_{i,t}} \bar \eta_{i,t}.
\end{equation*}
\end{proof}

Note that in equilibrium, $P_{i,t} = D_{i,t}$ if and only if there are no investors with short positions. Otherwise, $D_{i,t} > P_{i,t}$. Also note that $\bar \eta_{i,t}$ is just an (absolute value) dollar weighted elasticity. If there are no short sellers, then $\eta_{i,t}^{agg} = \bar \eta_{i,t}$, but the short-sellers can cause the aggregate elasticity to deviate somewhat from $\bar \eta_{i,t}$. Also, the effect of short sellers does not unambiguously increase the aggregate market elasticity. The presence of inelastic short sellers can actually decrease the aggregate elasticity. 

Lastly, it should be noted that this proposition can be extended to relax the assumption given in equation (\ref{eq:zero_response_for_zero}). However, elasticity is not defined for zero quantities, and relaxing this assumption forces us to deal with these zero-quantity investors in a way that makes the math more complicated with little, if any, more economic insight. 

\subsection{Statement of Equilibrium Model Propositions} \label{app:equil_model_statements}

In this section, I use the fraction $\theta$ as in Section \ref{subsec:counter_experiments}. While I could just fold HSA investors in with other investors, as in equation (\ref{eq:agg_demand}), I explicitly write out the demand for both HSA and non-HSA investors separately here. 

In the base case, aggregate demand, in dollars, with HSAs is given by:
\begin{equation} \label{eq:combined_equilibrium}
    (1 - \theta_t) A_t w_{i,t} + \theta_t A_t \udot{w}_{i,t} 
    = \zeta_{i,0,t} + \zeta_{i,1,t} P_{i,t} + \zeta_{i,2,t} \log (P_{i,t})
\end{equation}
where 
\begin{equation*}
    \zeta_{i,0,t} = (1 - \theta_t) A_t \beta_{i,0,t} + \theta_t A_t \udot{\beta}_{i,0,t}
\end{equation*}
\begin{equation} \label{eq:zeta_demand}
    \zeta_{i,1,t} = (1 - \theta_t) \beta_{i,1,t} + \theta_t \udot{\beta}_{i,1,t}, \;
    \text{and } \;
    \zeta_{i,2,t} = (1 - \theta_t) A_t \beta_{i,2,t} + \theta_t A_t \udot{\beta}_{i,2,t}.
\end{equation}
This is demand written in dollar terms. Dividing by $p_{i,t}$ generates demand in units of shares. The total supply is $S_{i,t}$ shares, which, following KY, is assumed to be exogenous. Thus, demand is characterized by the following:
\begin{equation} \label{eq:share_equilibrium}
    \frac{\zeta_{i,0,t} + \zeta_{i,1,t} P_{i,t} + \zeta_{i,2,t} \log (P_{i,t})}{p_{i,t}} = S_{i,t} \; \iff \;
    \zeta_{i,0,t} + \zeta_{i,1,t} P_{i,t} + \zeta_{i,2,t} \log (P_{i,t}) = P_{i,t},
\end{equation}
where the left-hand side equation is written in share units, and the right-hand side equation is written in dollar units. The equilibrium solution is given below. 

\begin{prop} \label{prop:equilibrium}
This equilibrium described by equation (\ref{eq:share_equilibrium}) has a unique positive price for each asset with a closed-form solution given by:\footnote{Note that many of the equations are written as a function of market equity while others are written as a function of share prices. It is obviously trivial to go from one to the other. }
\begin{equation} \label{eq:solution}
    P_{i,t} = \frac{\zeta_{i,2,t}}{\zeta_{i,1,t}-1} W \left( \frac{\zeta_{i,1,t}-1}{\zeta_{i,2,t}} \exp \left( - \frac{\zeta_{i,0,t}}{\zeta_{i,2,t}} \right) \right)
\end{equation}
where $W(\cdot)$ is the Lambert W function.
\end{prop}

The proof is given in Appendix \ref{subsec:equilibrium}. It should be noted that the equilibrium price has a closed-form solution, as opposed to KY where equilibrium prices are more difficult to calculate. I then state the final proposition, which describes the elasticity of the market for a stock. 

\begin{prop} \label{prop:market_elasticity}
If equilibrium is given by equation (\ref{eq:share_equilibrium}), with the left-hand side representing demand, then the elasticity of the market at equilibrium prices is given by:
\begin{equation} \label{eq:equil_elasticity}
    \eta_{i,t} = \underbrace{1-\zeta_{i,1,t}}_{\substack{\text{level}\\\text{term}}}  \underbrace{-\frac{\zeta_{i,2,t}}{P_{i,t}}}_{\substack{\text{log}\\\text{term}}}.
\end{equation}
On the other hand, if market demand (not individual investor demand) in terms of portfolio weights was given by the isoelastic demand curve $w_{i,t} = \varrho_{i,0,t} P_{i,t}^{\varrho_{i,1,t}}$ and equilibrium was given by $A_t w_{i,t} = P_{i,t}$, then the elasticity of the market at equilibrium prices is given by:
\begin{equation}
    \eta_{i,t} = \underbrace{1-\varrho_{i,1,t}}_{\substack{\text{level}\\\text{term}}}.
\end{equation}
\end{prop}
The proof is given in Appendix \ref{app:proof_market_elas}. This proposition shows that the elasticity of the demand function has a more general and flexible form than the elasticity of an isoelastic demand function.\footnote{To nest isoelastic demand (rather than just nesting the elasticity), we could slightly generalize equation (\ref{eq:indiv_demand}) or equation (\ref{eq:arb_demand_approx}) as:
\begin{equation*}
    w_{i,t}^j = \beta^j_{i,0,t} + \beta^j_{i,1,t} \frac{P_{i,t}^{\beta^j_{i,3,t}}}{A_{t}} + \beta_{i,2,t}^j \log (P_{i,t}),
\end{equation*}
where the $\beta^j_{i,3,t}$ exponent allows this functional form to clearly nest isoelastic demand. Assuming $\beta^j_{i,3,t} = 1$ still generates a flexible model that can fit the statistical arbitrageur demand well. In fact, these mild restrictions naturally arise from the set of common predictors we consider and demand with a functional forms consistent with equations (\ref{eq:CARA}) and (\ref{eq:bsv_linear}). However, if $\beta^j_{i,3,t} = 2$ for example, then the investor chooses portfolio weights based on the square of market equity, which is a relatively odd predictor to consider. This additional parameter can also generate multiple equilibrium when solving for prices and breaks the aggregation properties of this demand function, thus creating undue complexity without a concomitant economic motivation. 
}

Proposition \ref{prop:market_elasticity} concisely shows that if statistical arbitrageur demand is inelastic, and incumbent demand is inelastic, then the market will likely have inelastic demand. The terms $\zeta_{i,1,t}$ and $\zeta_{i,2,t}$ are simply weighted averages of the investors' demand function parameters.

\subsection{Proof of Proposition \ref{prop:equilibrium}} \label{subsec:equilibrium}

This gives the equilibrium solution for the model. 

\begin{proof}
Incumbent demand for asset $i$, written as the number of shares is:
\begin{equation} \label{eq:incumbent_demand}
    \frac{(1 - \theta_t) A_t w_{i,t}}{p_{i,t}}.
\end{equation}
Recall that the portfolio weight of asset $i$ is $w_{i,t}$, and the assets under management for all incumbents is $(1 - \theta_t) A_t$. Thus, the dollar amount that incumbents demand to invest in the asset is $(1 - \theta_t) A_t w_{i,t}$. Thus, dividing by the share price delivers demand in terms of the number of shares as in (\ref{eq:incumbent_demand}). 

Similarly, demand for statistical arbitrageurs in aggregate for asset $i$, written in terms of shares, is:
\begin{equation} \label{eq:statistical arbitrageur_demand}
    \frac{\theta_t A_t \udot{w}_{i,t}}{p_{i,t}}.
\end{equation}

Thus, total market demand is the sum of (\ref{eq:incumbent_demand}) and (\ref{eq:statistical arbitrageur_demand}). Supply is $S_{i,t}$, which is the number of shares outstanding. Thus, equilibrium is defined by the equation:
\begin{equation}
    \frac{(1 - \theta_t) A_t w_{i,t}}{p_{i,t}}
    + \frac{\theta_t A_t \udot{w}_{i,t}}{p_{i,t}} = S_{i,t}.
\end{equation}
Multiplying both sides by $p_{i,t}$ delivers the main equilibrium equation---equation (\ref{eq:combined_equilibrium}). I can write this as in equation (\ref{eq:share_equilibrium}):
\begin{equation}
    \zeta_{i,0,t} + \zeta_{i,1,t} P_{i,t} + \zeta_{i,2,t} \log (P_{i,t}) = P_{i,t}.
\end{equation}
This can be rewritten as
\begin{equation}
    P_{i,t} \frac{\zeta_{i,1,t}-1}{\zeta_{i,2,t}} + \log (P_{i,t}) = - \frac{\zeta_{i,0,t}}{\zeta_{i,2,t}}.
\end{equation}
Taking the exponential of both sides yields:
\begin{equation}
    P_{i,t} \exp \left( P_{i,t} \frac{\zeta_{i,1,t}-1}{\zeta_{i,2,t}} \right)
    = \exp \left( - \frac{\zeta_{i,0,t}}{\zeta_{i,2,t}} \right).
\end{equation}

Multiplying both sides by $(\zeta_{i,1,t}-1) /\zeta_{i,2,t}$:
\begin{equation}
    P_{i,t} \frac{\zeta_{i,1,t}-1}{\zeta_{i,2,t}} \exp \left( P_{i,t} \frac{\zeta_{i,1,t}-1}{\zeta_{i,2,t}} \right)
    = \frac{\zeta_{i,1,t}-1}{\zeta_{i,2,t}} \exp \left( - \frac{\zeta_{i,0,t}}{\zeta_{i,2,t}} \right).
\end{equation}
Notice that the left-hand side has the form of $w e^w$, where $w = P_{i,t} \frac{\zeta_{i,1,t}-1}{\zeta_{i,2,t}}$. Since $\zeta_{i,1,t} < 1$ and $\zeta_{i,2,t} < 0$, $\frac{\zeta_{i,1,t}-1}{\zeta_{i,2,t}}$ is positive. Thus, I can equivalently write:
\begin{equation}
    P_{i,t} \frac{\zeta_{i,1,t}-1}{\zeta_{i,2,t}} = W \left( \frac{\zeta_{i,1,t}-1}{\zeta_{i,2,t}} \exp \left( - \frac{\zeta_{i,0,t}}{\zeta_{i,2,t}} \right) \right). 
\end{equation}
where $W(\cdot)$ is the Lambert W function. Thus, I can write:
\begin{equation} 
    P_{i,t} = \frac{\zeta_{i,2,t}}{\zeta_{i,1,t}-1} W \left( \frac{\zeta_{i,1,t}-1}{\zeta_{i,2,t}} \exp \left( - \frac{\zeta_{i,0,t}}{\zeta_{i,2,t}} \right) \right),
\end{equation}
which is equation (\ref{eq:solution}). Recall that positive values of the Lambert W function are unique; thus this is the unique equilibrium. Also, it is easy to see that the price is guaranteed to be positive. This is close to the KY equilibrium, except that because their demand function is a bit more complicated, they have a unique positive price for each asset without a closed-form solution like I have here. 
\end{proof}

\subsection{Proof of Proposition \ref{prop:market_elasticity}} \label{app:proof_market_elas}

\begin{proof}
    Demand is given by:
    \begin{equation}
        Q_{i,t} = \frac{\zeta_{i,0,t} + \zeta_{i,1,t} P_{i,t} + \zeta_{i,2,t} \log (P_{i,t})}{p_{i,t}}
    \end{equation}
    where this demand is written in terms of shares. Thus, I can write, for positive shares demanded:
    \begin{equation}
        \log (Q_{i,t}) = \log (\zeta_{i,0,t} + \zeta_{i,1,t} P_{i,t} + \zeta_{i,2,t} \log (P_{i,t})) - \log(p_{i,t}).
    \end{equation}
    Thus I can calculate the elasticity:
    \begin{equation}
        \eta_{i,t} \equiv - \frac{\log (Q_{i,t})}{\log (P_{i,t})} = 1 - \frac{\zeta_{i,1,t} P_{i,t} + \zeta_{i,2,t}}{\zeta_{i,0,t} + \zeta_{i,1,t} P_{i,t} + \zeta_{i,2,t} \log (P_{i,t})}.
    \end{equation}
    In equilibrium, the denominator term is just $P_{i,t}$, so I have:
    \begin{equation}
        \eta_{i,t} = 1 - \frac{\zeta_{i,1,t} P_{i,t} + \zeta_{i,2,t}}{P_{i,t}}
        = 1 - \zeta_{i,1,t} - \frac{\zeta_{i,2,t}}{P_{i,t}}.
    \end{equation}

    Now consider the case where $Q_{i,t} = (A_t / p_{i,t}) \varrho_{i,0,t} P_{i,t}^{\varrho_{i,1,t}}$. Then 
    \begin{equation}
        \log (Q_{i,t}) = -\log(p_{i,t}) + \log(A_t) + \log (\varrho_{i,0,t}) + \varrho_{i,1,t} \log(P_{i,t}). 
    \end{equation}
    Then I can write:
    \begin{equation}
    \eta_{i,t} = 1 - \varrho_{i,1,t}.
\end{equation}
\end{proof}

\section{Statistical Arbitrage Models}  \label{app:stat_arbs}

\subsection{General Framework for Statistical Arbitrage Models}

In this section, I describe the general framework for the asset pricing models. Each model is described below, and of course in their original papers, but I describe the basics here. I refer to these models by the initialisms summarized in Table \ref{tab:models}. 

I assume a sample of dates $\tau = 1, ..., t-1$, which are used to form portfolio weights at time $t$. Let $Z_t$ be the $N \times K$ matrix filled with $\udot{z}_{i,k,t}$ values. Every model has portfolio weights that equal $f(Z_t) b$, where $f$ is a function that reads in predictors and outputs an $N \times M$ matrix of $M$ factor portfolio weights and $b$ is an $M \times 1$ vector that collapses the factor weights down to a single (attempted) MVE portfolio. Let $f(Z_t)_i$ denote this $1 \times M$ row of factor weights specific to asset $i$. Then portfolio weights for asset $i$ can be written as:
\begin{equation} \label{eq:fz_weights}
    \udot{w}_{i,t} = f(Z_t)_i b.
\end{equation}
Various models use different $f(\cdot)$ functions and methods of estimating $b$. Some factor models do not specify $b$. For example, the \cite{ff3} three-factor model does not specify the fractions of capital (portfolio weights) that should be invested in the market, value, and size portfolios in order to achieve mean-variance efficiency. In these cases, I use standard methods from the literature, such as \cite{forest} or \cite{shrinking} to collapse these factors down to a single set of portfolio weights. 

Each model has three steps:
\begin{enumerate}
    \item First, I form the factor model weights, $f(Z_t)$. 
    \item Second, I form factor excess returns from weights: $F_{\tau+1} = f(Z_{\tau})' r_{\tau+1}$. Then compute the mean vector and covariance matrix using these returns:
    \begin{align} \label{moments}
    \bar F_{t-1} = \frac{1}{t-1} \sum_{\tau=1}^{t-1} F_t \;\;\text{ and }\;\;
    \bar \Omega_{t-1} = \frac{1}{t-1} \sum_{\tau=1}^{t-1} (F_\tau - \bar F_{t-1}) (F_\tau - \bar F_{t-1})'.
\end{align}
    \item Third, I estimate $b$ as a function of $\bar F_{t-1}$ and $\bar \Omega_{t-1}$ using one of the methods discussed below. I refer to this as the "collapse" step, since it collapses the factors down to a single portfolio and demand function. Portfolio weights are set such that $\udot{w}_{i,t} = f(Z_t) b$. 
\end{enumerate}

\subsection{Portfolio Collapse Methods that Generate $b$}

I describe the \cite{brandt} (BSV), \cite{shrinking} (KNS), \cite{forest} (BPZ), and then the \cite{demiguel} (DGU) methods of generating $b$ from $\bar F_{t-1}$ and $\bar \Omega_{t-1}$.

\subsubsection{\cite{brandt} - BSV} \label{subsec:bsv}

The main objective of \cite{brandt} is to optimize portfolios of many assets, and the authors assume that optimal weights are linear in a set of predictors. My implementation of the BSV model similarly assumes that  $f(\cdot)$ is linear in predictors---i.e., $f(Z_t) = \breve Z_t$.   

While \cite{brandt} discuss a variety of objective functions to estimate $b$, I focus on their model that maximizes the empirical Sharpe ratio. The maximal empirical Sharpe ratio estimate of $b$ of these authors can be written as:
\begin{equation}\label{eq:BSVof}
b = \underset{\tilde b}{\text{argmin}} \; (\bar F_{t-1} - \bar \Omega_{t-1} \tilde b)' \bar \Omega_{t-1}^{-1} (\bar F_{t-1} - \bar \Omega_{t-1} \tilde b) \;\;\text{ which implies }\;\;
b = \bar \Omega_{t-1}^{-1} \bar F_{t-1}.
\end{equation}

\subsubsection{\cite{shrinking} - KNS}\label{KNSsection}

The KNS method first performs a standard principal components analysis (PCA) dimension reduction on the factors generated by the linear function $f(Z_t) = \breve Z_t$. The number of principal components is a hyperparameter of the model. This produces a smaller set of factors $F_{t-1}$, which are then reduced to a single portfolio with the vector $b$. 

The \cite{shrinking} estimator of $b$ is given by
\begin{equation} \label{eq:kns}
b = \underset{\tilde b}{\text{argmin}} (\bar F_{t-1} - \bar \Omega_{t-1} \tilde b)' \bar \Omega_{t-1}^{-1} (\bar F_{t-1} - \bar \Omega_{t-1} \tilde b) + \lambda_1 || \tilde b ||_1 + \lambda_2 || \tilde b ||_2,
\end{equation}
where $||\cdot||_1$ and $||\cdot||_2$ are the $L^1$ and $L^2$ norms respectively, and $\lambda_1$ and $\lambda_2$ are additional penalty hyperparameters that shrink the $b$ values towards zero. This has similar intuition as a standard elastic net. If $\lambda_1 = \lambda_2 = 0$, this collapses down to a BSV portfolio. 

\subsubsection{\cite{forest} - BPZ$_L$} \label{subsec:bpz}

\cite{forest} have two methodological contributions in their paper I discuss: first their random forest techniques and second their method of obtaining $b$. I discuss the latter here, which is given by:
\begin{equation}\label{BPZLest}
b = \underset{\tilde b}{\text{argmin}} ((\bar F_{t-1} + \lambda_0 \iota_M) - \bar \Omega_{t-1} \tilde b)' \bar \Omega_{t-1}^{-1} ((\bar F_{t-1} + \lambda_0 \iota_M) - \bar \Omega_{t-1} \tilde b) + \lambda_1 || \tilde b ||_1 + \lambda_2 || \tilde b ||_2,
\end{equation}
where $\iota_M$ is an $M$ dimensional column vector of ones. This extra term is used to shrink the mean of portfolio returns, potentially eliminating estimation error. 

As further discussed below in this Appendix, $b$ in Equation (\ref{BPZLest}) can either be estimated with a classic Lasso regression with an $L^1$ penalty parameter $\lambda_1$, or can be calculated with a LARS regression where the number of nonzero terms in $\hat b\tau$ is explicitly chosen. Choosing the number of nonzero weights effectively chooses the $\lambda_1$ parameter. I follow \cite{forest} by using LARS regression with the number of nonzero values of $b$ as a hyperparameter. This nests both the KNS and BSV methods. 

\subsubsection{\cite{demiguel} - DGU} \label{subsec:DGU_model}

\cite{demiguel} choose $b$ by setting:
\begin{equation}
    b = \underset{\substack{\lambda_2 / \lambda_0 = M \\ \lambda_0, \lambda_2 \to \infty}}{\lim} 
    \underset{b}{\text{argmin}} ((\bar F_{t-1} + \lambda_0 \text{sign} (\bar F_{t-1})) - \bar \Omega_{t-1} b)' \bar \Omega_{t-1}^{-1} ((\bar F_{t-1} + \lambda_0 \text{sign} (\bar F_{t-1})) - \bar \Omega_{t-1} b) + \lambda_2 || b ||_2,
\end{equation}
which implies that
\begin{equation}
    b = \frac{1}{M} \text{sign} (\bar F_{t-1}).
\end{equation}
In other words, this gives equal weight to each portfolio. It is an extreme form of shrinking, where the shrinking parameters go to infinity in the limit. This is shown in their paper to yield a high Sharpe ratio, in many cases better than other proposed statistical arbitrageurs. Thus, I use this as a contender statistical arbitrageur. 

\subsection{Factor Weight Function $f(Z_t)$}

In this subsection, I summarize the various factor weight functions $f(Z_t)$ other than the statistical arbitrageurs that use the linear predictor weight functions---which were discussed above.

\subsubsection{Classic Asset Pricing Factor Models: FF3, FF6, and HXZ}\label{fz1}

The FF3, FF6, and HXZ are classic asset pricing factor models. Each of these models use only a subset of the 62 predictors from the data, which I describe here. The FF3 uses only the market ($k = 1$ predictor), size (lme), and the book-to-market ratio (beme). The FF6 uses these same predictors, in addition to investment (investment), profitability (prof), and momentum (cum\_return\_12\_2). Finally, the HXZ model uses the market ($k = 1$), investment (investment), return on equity (roe), and size (lme). Thus, for these models, $f(Z_t)$ is the subset of $\breve Z_t$ of just these predictors. 

These models use different weights on these portfolios than the linear predictor portfolios. These bins create demand functions that are very discontinuous, while continuous demand functions are needed for this analysis. More importantly, using different weighting schemes decreases the comparability of these models and the models that use all predictors. Thus, the linear predictor-weighted portfolios should be thought of as approximations to these models' different portfolio weighting schemes. 

\subsubsection{\cite{kelly} - KPS} \label{subsec:kps}

\cite{kelly} consider the model:
\begin{equation*}
    r_{\tau+1} = \beta_\tau F_{\tau+1} + \epsilon_{\tau+1} \;\; \text{ where } \;\; \beta_\tau = Z_\tau \Gamma_\beta , 
     \;\;
    F_{\tau+1} = \left( \Gamma_\beta' Z_\tau' Z_\tau \Gamma_\beta \right)^{-1} \Gamma_\beta' Z_\tau' r_{\tau+1},
\end{equation*}
and $\Gamma_{\beta}$ is $K \times M$ dimensional. They choose $\Gamma_{\beta}$ to minimize:
\begin{equation*}
    \sum_{\tau=1}^{t-1} (r_{\tau} - \beta_{\tau-1} F_{\tau})' (r_{\tau} - \beta_{\tau-1} F_{\tau}).
\end{equation*}
Thus, in this model, 
\begin{equation}\label{kellyweights}
    f(Z_t) = Z_t \Gamma_\beta \left( \Gamma_\beta' Z_t' Z_t \Gamma_\beta \right)^{-1}.
\end{equation}
I estimate $b$ for the \cite{kelly} model using the \cite{shrinking} method described above. The number of factors is a hyperparameter. 

\cite{kelly} also describe in their paper how to estimate an additional vector of parameters associated with model alpha, which is written as $\Gamma_{\alpha}$. As described by the authors, the portfolio weights associated with this alpha portfolio, or an "arbitrage" portfolio as they describe it, are simply:
\begin{equation}\label{kellyweights_alpha}
f(Z_t) = Z_t (Z_t' Z_t)^{-1} \Gamma_{\alpha}. 
\end{equation}

Thus, there are two potential sets of portfolios associated with this model: the beta portfolio in equation (\ref{kellyweights}) and the alpha portfolio in equation (\ref{kellyweights_alpha}). In most of the paper, I simply use the beta portfolio, as the authors argue that this prices the cross-section and should therefore be the maximum Sharpe ratio portfolio, as opposed to the alpha portfolio that the authors argue has relatively low and statistically insignificant returns in their original paper. 

\subsubsection{\cite{CRW} - CRW} \label{app:crw}

I implement a simple linear version of \cite{CRW}. This is similar to \cite{kelly}, where I have:
\begin{equation}
r_{t+1} = Z_t \Gamma_{\beta} f_{t+1} + \epsilon_{t+1},
\end{equation}
where $\Gamma_{\beta}$ is a matrix of parameters. This is solved in a two-step procedure the authors refer to as "regressed PCA".

First, $\Gamma_{\beta}$ is calculated by performing standard PCA (with standard demeaning) on $Y_{t+1} = (Z_t' Z_t)^{-1} Z_t' r_{t+1}$ in the time series across a range of $t$ values. Note that using the standard PCA normalization, this means that $\Gamma_{\beta}' \Gamma_{\beta} = I$, where $I$ is the identity matrix. Second, $f_{t+1}$ is calculated as 
\begin{equation*}
f_{t+1} = (\Gamma_{\beta}' \Gamma_{\beta})^{-1} \Gamma_{\beta}' (Z_t' Z_t)^{-1} Z_t' r_{t+1}
= \Gamma_{\beta}' (Z_t' Z_t)^{-1} Z_t' r_{t+1}.
\end{equation*}
Thus, factor portfolio weights are then formed using the simple formula:
\begin{equation} \label{eq:crw_beta}
f(Z_t) = Z_t (Z_t' Z_t)^{-1} \Gamma_{\beta}. 
\end{equation}

The authors also discuss how to estimate a vector of alpha intercept terms, which I can denote as $\Gamma_{\alpha}$. With these estimates in hand, the authors give the formula for alpha portfolio weights:
\begin{equation} \label{eq:crw_alpha}
f(Z_t) = Z_t (Z_t' Z_t)^{-1} \Gamma_{\alpha}.
\end{equation}

As in the KPS model above, there are two sets of portfolios associated with this model: the beta portfolios from equation (\ref{eq:crw_beta}) and the alpha portfolio associated with equation (\ref{eq:crw_alpha}). Like with the KPS model above, I keep these sets of portfolios separate, but I do consider both in the paper. 

\subsubsection{KKN} \label{app:kkn}

This method \citep{kkn} is unique compared to the other statistical arbitrageur models. I will highlight the main three differences:
\begin{enumerate}
    \item This method allows only a time-invariant $Z$, rather than a time-varying $Z_t$. This forces users of this method to employ very short estimation sample periods. 
    \item This method does not form factor portfolios meant to price the cross-section. This method only provides a portfolio that is meant to be a pure-alpha arbitrage portfolio. 
    \item This method must be trained on a balanced sample. This forces us to use short rolling sample periods to train the model parameters. 
\end{enumerate}

I let $R = [r_1 \; r_2 \; ... \; r_T]$ be the $N \times T$ matrix. Let $\bar R$ be the $N$-dimensional column vector of average returns. Let $\tilde R = R - \bar R \iota_T'$ be a matrix of demeaned returns, where $\iota_T$ is a column vector of $T$ ones. 

I can take the following steps:
\begin{enumerate}
    \item Let $\hat R = Z (Z' Z)^{-1} Z' \tilde R$. Let $G_{\beta} (Z)$ be the $N \times M$ matrix of the first $M$ eigenvectors of $\hat R \hat R' / N$. 
    \item Regress $\bar R$ on $G_{\beta} (Z)$, with regression coefficients $\bar F$. 
    \item Regress $\bar R - G_{\beta} (Z) \bar F$ on $Z$, with regression coefficients $\theta$. 
\end{enumerate}

Arbitrage portfolio weights are then:
\begin{equation*}
    w_{\alpha,t} = Z_t \theta.
\end{equation*}

\subsubsection{\cite{gu} - GKX}

\cite{gu} propose a latent factor model in which the portfolio weights are
\begin{equation}
    f(Z_t) = Z_t (Z_t' Z_t)^{-1} \Gamma_A,
\end{equation}
where $\Gamma_A$ is a $K \times M$ matrix of parameters. In the language of neural networks, $\Gamma_A$ is a matrix of weights. \cite{gu} also allows the betas of these factors to be have a neural network structure as a function of $Z_t$, $\beta (Z_t)$, which is $N \times M$ dimensional. Model parameters are essentially estimated to achieve a low sum of squared errors:\footnote{There are also penalties in the objective function and other details typical of neural network backpropagation. See \cite{gu} for more details}
\begin{equation}\label{gkxweights}
    \sum_{t=0}^{T-1} (r_{t+1} - \beta (Z_t) F_{t+1})' (r_{t+1} - \beta (Z_t) F_{t+1}).
\end{equation}
I estimate $b$ for the \cite{gu} model using the \cite{forest} method described above.

Similar to the model of \cite{kelly}, this model requires the number of factors $M$, as well as $\lambda_0$ and $\lambda_2$ as hyperparameters. The standard hyperparameters associated with neural networks, including the number of layers and nodes, are hyperparameters in this model. Table \ref{tab:hypers} displays various amounts of hidden layers and nodes that I use during cross-validation. For example, $(25, 5)$ signifies two hidden layers, with the first hidden layer containing 25 nodes and the second containing 5 nodes. The value $()$ signifies zero hidden layers. As described in \cite{gu}, I also use an $L^1$ penalty for the weights in the neural network, which adds another hyperparameter. 

\subsubsection{\cite{forest} - BPZ$_F$} 

In their effort to create a set of test assets that spans the mean-variance efficient frontier, \cite{forest} describe how they create a forest of decision trees to map a set of $K$ predictors to the portfolio weights, $f(Z_t)$, of $M$ test assets such that the mean-variance combination of the portfolios has the highest Sharpe ratio possible.  For a given $K$ and tree depth, $d$, the authors advocate creating trees based on all possible combinations of conditional sorts, where at each level of a given tree from 1 to $d$, stocks are mechanically sorted into two portfolios based on the median value of some predictor across stocks. Portfolios at every intermediate and final node of every tree are equally weighted. The authors then prune the trees by choosing weights to maximize the Sharpe ratio using their approach described in Equation (\ref{BPZLest}). The total number of portfolios formed is given by $K^d 2^d$, which for large values of $K$ becomes untenable. 

I employ an alternative method based on \cite{forest}. I stack
\begin{equation} \label{eq:xy}
    X_\tau = \begin{bmatrix}
    Z_0 \\ Z_1 \\ \vdots \\ Z_{\tau-1}
    \end{bmatrix}, \;\;\;
    Y_\tau = \begin{bmatrix}
    r_1 \\ r_2 \\ \vdots \\ r_{\tau}
    \end{bmatrix},
\end{equation}
and fit a random forest of regression trees with feature matrix $X$ to predict the target vector $Y$. The number of trees in the forest, the maximum depth of each tree, and the number of features to consider at each tree are all hyperparameters of this model. 

Each intermediate node and end node of each tree represents a portfolio. In other words, each column of $f(Z_t)$ represents an equally weighted portfolio from a node in a tree of the fitted random forest.\footnote{For example, assume that each tree is fitted to a depth of 4. Then there are $2^4 = 16$ end nodes, and $1 + 2^1 + 2^2 + 2^3 = 15$ intermediate nodes. Therefore, there are $15 + 16 = 31$ total portfolios in this tree. If there are 100 trees in the forest, then there are $31 * 100 = 3,100$ total portfolios. Since the base node portfolios are always the same, there are at least $100 - 1 = 99$ duplicate portfolios. Thus, if there are no other duplicate portfolios, then there are $3,001$ unique portfolios in this forest. Thus, a model with these hyperparameters assumes that $M = 3,001$. Thus, $F_{t+1} = f(Z_t)' r_{t+1}$ is a potentially large dimensional vector of portfolio returns.}

In estimating $b$ as described using the BPZ method described above, \cite{forest} recommend rescaling these forest portfolios to take into account estimation noise given the varying number of assets in each portfolio. I do this as they recommend. 

The hyperparameters required for this model are the number of trees in the forest, the maximum depth of the trees, the maximum number of features considered at each tree, as well as the three BPZ hyperparameters required to calculate $b$. 

\subsubsection{Neural Network - NN}

This is a neural network model original to this paper. Let $g: \R^{N} \times \R^{K} \mapsto \R^{N} \times \R^{J+1}$ denote a standard neural network with $J+1 \geq 2$ output nodes. Let $\mu_t$ denote the first column of $g(Z_t)$, and $\Gamma_t$ denote the remaining $J$ columns of $g(Z_t)$. Thus, $\Gamma_t$ is an $N \times J$ matrix, while $\mu_t$ is $N \times 1$. Thus, I can write
\begin{equation} \label{eq:nn_split}
    \begin{bmatrix}
    \underbrace{\mu_t}_{N \times 1} & \underbrace{\Gamma_t}_{N \times J}
    \end{bmatrix}
     = \underbrace{g(Z_t)}_{N \times (J+1)}
\end{equation}

The model assumes a KY style covariance matrix:
\begin{equation} \label{eq:nn_sigma}
    \Sigma_t = \Gamma_t \Gamma_t' + \zeta I_N
\end{equation}
where $\zeta$ is a scalar parameter and $I_N$ is an $N \times N$ identity matrix. Finally, the model assumes
\begin{equation}
    r_{t+1} \sim N (\mu_t, \Sigma_t)
\end{equation}
and the neural network parameters associated with $g$, as well as the scalar $\zeta$, are estimated with a typical neural network maximum a posterior (MAP) estimation procedure with a multivariate normal log-likelihood function. More details about this SDF, including equations to efficiently compute the log-likelihood, are found below in this Appendix. 

For this SDF, $f$ is defined such that
\begin{equation}
    f(Z_t) = \Sigma_t^{-1} \mu_t
\end{equation}
where $\Sigma_t$ and $\mu_t$ are functions of $Z_t$ defined in equations (\ref{eq:nn_split}) and (\ref{eq:nn_sigma}). Note that $f(Z_t)$ is actually already an $N \times 1$ matrix. So the collapse method is trivially set to $b = 1$. 

The hyperparameters required for this algorithm are the number of layers and nodes in each layer, as well as a $L^2$ penalty regularization parameter commonly used in neural networks. 

\subsubsection{Random Forest - RF}

I also consider a random forest model where the means of stocks are predicted with a random forest. The covariance matrix is just assumed to be proportional to the identity matrix, which potentially reduces estimation error in the covariance matrix following the logic of \cite{demiguel}. 

For this statistical arbitrageur, consider a fitted regression random forest that used features $X$ to fit $Y$ as defined in equation (\ref{eq:xy}) above. Let $f(Z_t)$ be the random forest excess return predictions for each of the rows of $Z_t$. Thus, $f(Z_t)$ is an $N \times 1$ dimensional vector of excess return predictions. I take these to be estimates of the expected excess returns. Since the covariance matrix is assumed to be proportional to the diagonal matrix, then $f(Z_t)$ are the MVE weights insofar as $f(Z_t)$ are the expected excess returns. Like the above model, this model delivers a single factor, meaning that trivially $b = 1$. 

The hyperparameters required for this model are the number of trees in the forest, the maximum depth of the trees, and the maximum number of features considered at each tree.

\subsection{PCA Reduction - \cite{shrinking}}

This subsection describes in detail how PCA can be used to reduce the dimension of $f$ before $b$ is computed, which is used for the KNS method. 

Let $f^1: \R^{N} \times \R^{K} \mapsto \R^{N} \times \R^{M^1}$ denote any given $f$ from Section \ref{sec:elas_models}. Note that I substitute the notation of $f$ and $M$ to $f^1$ and $M^1$ respectively. Define:
\begin{equation}
    F_{t+1}^1 \equiv f^1(Z_t)' r_{t+1},
\end{equation}
\begin{equation}
    \bar F^1_T = \frac{1}{T} \sum_{t=1}^T F_t^1, \text{ and }
\end{equation}
\begin{equation}
    \bar \Omega^1 = \frac{1}{T} \sum_{t=1}^T (F_T - \bar \Omega^1) (F_T - \bar \Omega^1)'.
\end{equation}

Let $M \leq M^1$. Define $Q$ to be the $M^1 \times M$ matrix whose $m^{th}$ column contains the eigenvector corresponding to the $m^{th}$ largest eigenvalue of $\bar \Omega^1$. 

Then define
\begin{equation}
    f(Z_t) = f^1 (Z_t) Q
\end{equation}

Note that $f$ now maps from $\R^N \times \R^K$ to $\R^N \times \R^M$. The matrix $Q$ acts as a dimension reduction method. This imposes the prior on the data that the eigenvectors associated with the smallest $M^1 - M$ eigenvalues are not important for delivering risk prices in the SDF. 

Now that $f$ is in hand, $b$ can be calculated according to \cite{brandt}, \cite{shrinking}, or \cite{forest} as described above. 

\subsection{Efficient Computation of KNS - \cite{shrinking}}

I follow \cite{forest} in the efficient calculation of $b$, which can be applied to the KNS version of $b$ trivially. 

For simplicity, I drop the $T$ subscripts on $\bar F \equiv \bar F_T$ and $\bar \Omega = \bar \Omega_T$. 

I can write
\begin{equation}
    \bar \Omega = Q D Q',
\end{equation}
where $Q$ is an $M \times M$ matrix where the $j^{th}$ column is the eigenvector associated with the $j^{th}$ largest eigenvalue of $\bar \Omega$. $D$ is the corresponding diagonal matrix of eigenvalues. Define $m = \min \{ T, M, J \}$ where $J$ is the number of strictly positive eigenvalues of $\bar \Omega$. Let $\tilde Q$ denote the $M \times m$ matrix of the first $m$ columns of $Q$. Let $\tilde D$ denote the $m \times m$ diagonal matrix of the first $m$ terms along the diagonal of $D$. Then define
\begin{equation}
    \tilde \Omega^{1/2} =  \tilde Q \tilde D^{1/2} \tilde Q',\;\;\;
    \tilde \Omega^{-1/2} =  \tilde Q \tilde D^{-1/2} \tilde Q', \; \text{ and } \;
    \tilde \omega = \tilde \Omega^{-1/2} \bar \Omega.
\end{equation}

Let 
\begin{equation}
    X = \begin{bmatrix}
    \tilde \Omega^{1/2} \\
    \sqrt{\lambda_2} I_{M} \\
    \end{bmatrix}, \; \text{ and } \;
    y = \begin{bmatrix}
    \tilde \omega \\
    0_M \\
    \end{bmatrix}.
\end{equation}

Then the problem is equivalent to a lasso/LARS regression that minimizes:
\begin{equation}
    (y - X b)' (y - X b) + \lambda_1 || b ||_1.
\end{equation}

\subsection{Efficient Computation of BPZ - \cite{forest}}

Efficient computation of $b$ for the BPZ method is identical to the above subsection, except that $\tilde \omega$ should be replaced with 
\begin{equation}
    \tilde \omega = \tilde \Omega^{-1/2} (\bar \Omega + \lambda_0 \iota).
\end{equation}

\subsection{Efficient Computation of the Neural Network (NN)}

In order to calculate the maximum likelihood function, $\Sigma_t^{-1}$ and the determinant of $\Sigma_t$, denoted as $| \Sigma_t|$ need to be calculated efficiently, even with large $N$. This section shows formulas that can be used to do this. 

By the Woodbury Matrix Identity:
\begin{equation}
    \Sigma_t^{-1} = \frac{1}{\zeta} \left( I_N - \Gamma_t (\zeta I_J + \Gamma_t' \Gamma_t)^{-1} \Gamma_t' \right).
\end{equation}
Note that if $J = 1$, then $\zeta I_J + \Gamma_t' \Gamma_t$ is a scalar. If $J$ is much smaller than $N$, then $(\zeta I_J + \Gamma_t' \Gamma_t)^{-1}$ is much easier to numerically calculate than $\Sigma_t^{-1}$.

By Sylvester's matrix identity:
\begin{equation}
    |\Sigma_t| = \zeta^N \left| I_J + \frac{1}{\zeta} \Gamma_t' \Gamma_t \right|.
\end{equation}
Note that if $J = 1$, $I_J + \frac{1}{\zeta} \Gamma_t' \Gamma_t$ is a scalar. If $J$ is much smaller than $N$, then the determinant of $I_J + \frac{1}{\zeta} \Gamma_t' \Gamma_t$ is much easier to numerically calculate than the determinant of $\Sigma_t$. 

Using these formulas, it is trivial to calculate the multivariate normal log-likelihood function. Denote this as $l(\theta)$, where $\theta$ is a vector of parameters from the neural network $g$ and the scalar $\zeta$. These parameters are chosen by setting $\theta^*$ to
\begin{equation}
    \theta^* = \text{argmax}_{\theta} \; l(\theta) - \lambda || \theta ||_2,
\end{equation}
where $\lambda$ is a hyperparameter. This is a standard maximum a posteriori (MAP) estimator commonly used with neural networks and in other areas of statistics and econometrics.

\section{Extensions} \label{app:extensions}

\subsection{Wealth Effects} \label{app:wealth}

We first consider the wealth effects labeled above. If there are no inflows or outflows for the fund, then $a_t$ is given by:
\begin{equation} \label{eq:wealth}
    a_t = a_{t-1} \left( w_{t-1}' (r_{t} + R_{f,t} \iota) + (1 - \iota' w_{t-1}) R_{f,t} \right),
\end{equation} 
where $r_t$ is the $N$-dimensional vector of returns in excess of the risk-free rate, $R_{f,t}$ is the gross risk-free rate, and $\iota$ is a vector of ones. In this section, we want to consider price changes ceteris paribus, which gives a simple term for the wealth effect:
\begin{equation} \label{eq:wealth_effects_equa}
    -\frac{\partial \log(a_t)}{\partial \log (p_{i,t})}
    = -\left( \frac{a_{t-1}}{a_t} \right) \left( \frac{p_{i,t}}{p_{i,t-1}} \right) w_{i,t-1},
\end{equation}
where the first two terms are close to one as long as there are no large fluctuations in AUM or the price of the stock, and the last term, $w_{i,t-1}$, is small in a diversified portfolio.\footnote{This can be calculated by noting that the excess return of asset $i$ is
\begin{equation*}
    r_{i,t} = \frac{p_{i,t} + d_{i,t}}{p_{i,t-1}} - R_{f,t}
    = \frac{\exp(\log(p_{i,t})) + d_{i,t}}{p_{i,t-1}} - R_{f,t},
\end{equation*}
where $d_{i,t}$ is the dividend.} I show below that the wealth effects are quite small and well-approximated with zero. The wealth effects are assumed to be zero in KY as well. 

Table \ref{tab:idio_wealth} shows the wealth effects using equation (\ref{eq:wealth_effects_equa}), and these effects are indeed quite small, usually decreasing the elasticity by just 0.001. I use only long positions, and as can be seen from equation (\ref{eq:wealth_effects_equa}), if $w_{i,t}$ is the same sign as $w_{i,t-1}$, then wealth effects are bounded above by zero. 

\begin{table}[!t] \centering
    \resizebox{1.\textwidth}{!}{\begin{tabular}{@{\extracolsep{5pt}}llccccccccccccc}
\\[-1.8ex]\hline
\hline \\[-1.8ex]
& & \multicolumn{13}{c}{\textit{Elasticity Impact from Wealth Effects}} \
\cr \cline{3-15}
\\[-1.8ex] Elasticity & Wins. & \multicolumn{1}{c}{BPZ$_F$} & \multicolumn{1}{c}{BPZ$_L$} & \multicolumn{1}{c}{BSV} & \multicolumn{1}{c}{CRW} & \multicolumn{1}{c}{DGU} & \multicolumn{1}{c}{FF3} & \multicolumn{1}{c}{FF6} & \multicolumn{1}{c}{GKX} & \multicolumn{1}{c}{HXZ} & \multicolumn{1}{c}{KNS} & \multicolumn{1}{c}{KPS} & \multicolumn{1}{c}{NN} & \multicolumn{1}{c}{RF}  \\
\\[-1.8ex] Weighting & & (1) & (2) & (3) & (4) & (5) & (6) & (7) & (8) & (9) & (10) & (11) & (12) & (13) \\
\\[-1.8ex] \hline \\[-1.4ex]

 Equal & None &   -0.0002 & -0.0006 & -0.0016 & -0.0007 & -0.0005 & -0.0007 & -0.0009 & -0.0009 & -0.0009 & -0.0009 & 0.0000 & -0.0011 & -0.0001 \\
 & $(1^{st}, 99^{th})$ &   -0.0001 & -0.0005 & -0.0016 & -0.0007 & -0.0005 & -0.0006 & -0.0009 & -0.0009 & -0.0008 & -0.0009 & 0.0000 & -0.0011 & -0.0001 \\
 & $(5^{th}, 95^{th})$ &   -0.0001 & -0.0005 & -0.0016 & -0.0007 & -0.0005 & -0.0006 & -0.0008 & -0.0009 & -0.0008 & -0.0009 & 0.0000 & -0.0011 & -0.0001 \\

\hline \\[-1.8ex]

Portfolio & None &   -0.0008 & -0.0008 & -0.0024 & -0.0013 & -0.0007 & -0.0016 & -0.0017 & -0.0014 & -0.0019 & -0.0013 & 0.0000 & -0.0017 & -0.0001 \\
 & $(1^{st}, 99^{th})$ &   -0.0005 & -0.0008 & -0.0023 & -0.0012 & -0.0007 & -0.001 & -0.0013 & -0.0013 & -0.0013 & -0.0012 & 0.0000 & -0.0016 & -0.0001 \\
 & $(5^{th}, 95^{th})$ &   -0.0001 & -0.0007 & -0.0021 & -0.0011 & -0.0006 & -0.0007 & -0.0011 & -0.0012 & -0.001 & -0.0011 & 0.0000 & -0.0015 & -0.0001 \\

\hline \\[-1.8ex]

Value & None &   -0.0003 & -0.001 & -0.0026 & -0.0007 & -0.0008 & -0.0036 & -0.0036 & -0.0011 & -0.0042 & -0.0022 & 0.0000 & -0.0022 & -0.0001 \\
 & $(1^{st}, 99^{th})$ &   -0.0001 & -0.0005 & -0.0016 & -0.0004 & -0.0003 & -0.0016 & -0.0017 & -0.0006 & -0.0019 & -0.0012 & 0.0000 & -0.0012 & -0.0001 \\
 & $(5^{th}, 95^{th})$ &   0.0000 & -0.0002 & -0.0007 & -0.0003 & -0.0001 & -0.0003 & -0.0004 & -0.0003 & -0.0004 & -0.0004 & 0.0000 & -0.0005 & -0.0001 \\

\hline
\hline \\[-1.8ex]

\end{tabular}

}
    \vspace{4mm}
    \caption{\textbf{Wealth Effects on Elasticity.} This table shows the valued-weighted and winsorized average wealth component of elasticity of the thirteen models (under three winsorization schemes similar to Table \ref{tab:statistical arbitrageur_elasticity}) during the out-of-sample period from February 1990 to January 2020, inclusive. Below the average values in parentheses, the standard errors, double clustered by month and stock, are shown. Table \ref{tab:models} displays the initialisms of these models. There are 1,253,175 month stock observations these values are calculated from.}
    \label{tab:idio_wealth}
\end{table}

\subsection{Consumption Hedging with Epstein-Zin Demand} \label{subsec:epstein_zin}

Mean-variance demand exhibits constant absolute risk aversion, which has often been criticized in the literature, but I just briefly note here that Epstein-Zin demand exhibits similarly elastic behavior. To see this, consider the case of Epstein-Zin multivariate demand for the $N$ risky assets from \cite{campbell:chan:viceira}. Let $y_{t}$ be the $N$-dimensional vector of log returns minus the log risk-free return. They show, after log-linearization, portfolio weights are given by the approximation:
\begin{equation} \label{eq:ez_weights}
    w_t = \frac{1}{\tilde \gamma} \tilde \Sigma_t^{-1} \left[ \tilde \mu_{t} - \frac{\theta}{\psi} \sigma_{c-w,t} \right],
\end{equation}
where $\tilde \mu_t = \E_t [y_{t+1}] + \frac{1}{2} \sigma_t^2$, $\tilde \Sigma_t$ is the $N \times N$ conditional covariance matrix of $y_{t+1}$, $\sigma_t^2$ is the $N$-dimensional vector containing the diagonal elements of $\tilde \Sigma_t$, $\sigma_{c-w,t}$ is the $N$-dimensional vector of the conditional covariance of the log consumption to wealth ratio and $y_{t+1}$, $\tilde \gamma > 0$ is the relative risk aversion coefficient, $\psi > 0$ is the elasticity of intertemporal substitution, and $\theta \equiv (1 - \tilde \gamma) / (1 - \psi^{-1})$.\footnote{See equation (20) of \cite{campbell:chan:viceira}. Note that in their equation, there are some additional terms because they also consider $y_t$ to be log return of the asset minus a benchmark with potential covariance terms. I consider just the risk-free rate, which eliminates some of these extra terms.} Recall that Epstein-Zin demand nests CRRA demand. The main difference lies in the consumption hedging term $\sigma_{c-w,t}$. 

We can use the \cite{stevens1998inverse} formula for the Epstein-Zin modification of equation (\ref{eq:cara_stevens}):
\begin{equation*}
    w_{i,t}
    \approx \frac{1}{\tilde \gamma} \left( \frac{\mu_{i,t} - \frac{\theta}{\psi} \sigma_{i,c-w,t} - \beta_{i,t}' \left( \mu_{-i,t} - \frac{\theta}{\psi} \sigma_{-i,c-w,t} \right)}{\sigma_{i,t,\epsilon}^2} \right),
\end{equation*}
where the $\sigma_{i,c-w,t}$ and $\sigma_{-i,c-w,t}$ notation is analogous to the $\mu_{i,t}$ and $\mu_{-i,t}$ notation but for the consumption hedging vector $\sigma_{c-w,t}$. Note here that I use the $\mu_{i,t}$, $\beta_{i,t}$, and $\sigma_{i,t,\epsilon}^2$ terms from equation (\ref{eq:cara_stevens}) instead of log-return counterparts, which matters little quantitatively. To be clear, the beta term, $\beta_{i,t}$, is still a regression beta of the asset $i$ return on the $N-1$ other assets without including consumption hedging terms.\footnote{This can be seen from the \cite{stevens1998inverse} block matrix formula.}

Then equation (\ref{eq:eta_mu}) is extended to become:
\begin{equation}
    \eta_{i,t} 
    = 1 + \frac{1}{w_{i,t}} \left( \frac{1}{\tilde \gamma \sigma_{i,t,\epsilon}^2} \right) \left( -\frac{\partial \mu_{i,t}}{\partial \log (p_{i,t})} 
    + \frac{\theta}{\psi} \frac{\partial \sigma_{i,c-w,t}}{\partial \log (p_{i,t})}\right).
\end{equation}
If $\partial \sigma_{i,c-w,t} / \partial \log (p_{i,t})$ is small, then this is the same as equation (\ref{eq:eta_mu}) except with $\tilde \gamma$ instead of $\gamma a_t$. In this case, a change in $\tilde \gamma$ does not affect the elasticity, just as described in the text for $\gamma a_t$. Thus, the main modification is the consumption hedging term. 

It turns out that the derivative, $\partial \sigma_{i,c-w,t} / \partial \log (p_{i,t})$, is indeed small empirically, when measured via the method described below. In order to measure the conditional covariance of asset returns and consumption, I need both unexpected returns (returns minus the conditional expected return) and unexpected consumption to wealth log ratios. For unexpected consumption to wealth log ratios, I use cay from \cite{lettau2001consumption}.\footnote{I downloaded this data from Amit Goyal's website, since the cay data is updated there through 2023.} This data is quarterly, so I need to use quarterly returns data to match. For unexpected returns, I simply use the residuals of a panel regression of quarterly returns on the set of lagged normalized predictors. Let $r_{i,t+1}^{unex}$ denote the unexpected return for asset $i$ and let $cay_{t+1}$ denote cay. Then we have $\sigma_{i,c-w,t} = E_t [r_{i,t+1}^{unex} cay_{t+1}]$. To model a conditional expectation, I simply run a regression of $r_{i,t+1}^{unex} cay_{t+1}$ on all normalized predictors at time $t$, $\udot{z}_{i,k,t}$. Denote the regression coefficients from this regression as $b_k^c$. Then the derivative is calculated as:
\begin{equation*}
    \frac{\partial \sigma_{i,c-w,t} }{ \partial \log (p_{i,t})}
    = \sum_k b_k^c \frac{\partial \udot{z}_{i,k,t}}{\partial \log (p_{i,t})}.
\end{equation*}
Note that this is the same way the derivative, $\partial \mu_{i,t} / \partial \log (p_{i,t})$, is measured in the paper. The standard error is calculated with the delta method from the coefficient standard errors, which are double clustered at the stock and month level. The estimate is 0.0012 with a standard error of 0.0018. To scale this down to a montly level, we can divide by 3 to get 0.0004. This is quite small compared to the measured derivative, $-\partial \mu_{i,t} / \partial \log (p_{i,t})$, which is about 0.048 as described in the main text. For large elasticities where the 1 in equation (\ref{eq:eta_mu}) matters little, this modifies elasticities by less than 1\%. Even for smaller elasticities this is relatively small. 

In summary, Epstein-Zin demand delivers similarly elastic demand to mean-variance demand, where the consumption hedging term is relatively small quantitatively. This is sensible economically as well, since the consumption to wealth ratio is an aggregate quantity, and a change of 1\% in the price of an individual stock is unlikely to change the asset's overall ability to hedge consumption significantly. 

\subsection{Pure-alpha Strategy Elasticity Special Cases} \label{app:pure_alpha}

I first calculate the elasticity of the pure-alpha strategies from the CRW, KKN \citep{kkn}, and KPS models. These are separate portfolios from the factor-based strategies shown in Tables \ref{tab:returns} and \ref{tab:statistical arbitrageur_elasticity}. The KKN model is not used in the main paper, and is thus missing from Table \ref{tab:models}, because it is a method that only produces a pure-alpha strategy without a more general Sharpe optimizing strategy. The KKN model uniquely requires short training sample periods, as it is not designed for long samples. To handle this, the parameters are estimated using a rolling window approach with the past 12 months to fit the parameters, and these parameters are refit every month. The hyperparameters likewise cannot use a four-fold cross-validation design as the other models because these sample lengths are too long. To choose hyperparameters, I fit the parameters based on every year in the training sample, and calculate performance over the subsequent year. Hyperparameters are chosen to maximize the Sharpe ratio based on this in-sample performance. This is a simple cross-validation design to deal with short training sample periods. More details about this model are given in Appendix \ref{app:kkn}. 

The out-of-sample alphas and average elasticities are shown for these three hedged statistical arbitrageur models in Table \ref{tab:hedged_alpha_ports}. The hedged CRW model performs poorly, while the other two models have positive and statistically significant CAPM alphas. The hedged KPS model has a lower elasticity than its corresponding factor strategy (about 1.7 versus 3.5). The KKN model produces slightly upward-sloping demand on average, with an elasticity of about -0.2. 

\begin{table}[!t] \centering
    \resizebox{0.4\textwidth}{!}{\begin{tabular}{@{\extracolsep{5pt}}lccc}
\\[-1.8ex]\hline
\hline \\[-1.8ex]
\\[-1.8ex] & \multicolumn{1}{c}{CRW} & \multicolumn{1}{c}{KKN} & \multicolumn{1}{c}{KPS}  \\
\\[-1.8ex] & (1) & (2) & (3) \\
\hline \\[-1.8ex]
 $\alpha$ & -4.907$^{*}$ & 22.310$^{***}$ & 27.740$^{***}$ \\
& (2.562) & (2.621) & (2.625) \\

\hline
\hline \\[-1.8ex]

 Elasticity & 2.817$^{***}$ & -0.208$^{}$ & 1.730$^{***}$ \\
& (0.106) & (0.151) & (0.096) \\

\hline
\hline \\[-1.8ex]
\textit{Note:} & \multicolumn{3}{r}{$^{*}$p$<$0.1; $^{**}$p$<$0.05; $^{***}$p$<$0.01} \\
\end{tabular}}
    \vspace{4mm}
    \caption{\textbf{Hedged Statistical Arbitrageur CAPM Alpha and Price Elasticity.} This table shows annualized CAPM alphas and the average (across assets and months) price elasticity of the three pure-alpha or hedged portfolios during the out-of-sample period from February 1990 to January 2020, inclusive. Standard errors are shown in parentheses below the estimates. The standard errors for the elasticity estimates are double clustered by month and stock. Table \ref{tab:models} displays the initialisms of these models. There are 1,253,175 month stock observations these values are calculated from. }
    \label{tab:hedged_alpha_ports}
\end{table}


\subsection{Covariance Matrix Effects} \label{app:covariance}

I also consider relaxing the assumption of an exogenous covariance matrix. In other words, I consider an elasticity where a price change also affects not just $\mu_{i,t}$ but also $\beta_{i,t}$ and $\sigma_{i,t,\epsilon}^2$ in equation (\ref{eq:cara_stevens}). To do this, I use the same covariance matrix as above, except the market betas, $\beta_{i,t}^M$, are modified to be a function of prices. Specifically, I assume an asset's beta has the following form:
\begin{equation} \label{eq:beta_in_characteristics}
    \beta_{i,t}^M = \Gamma_0 + \Gamma_1 (\text{cum\_return\_1\_0}_{i,t}) + \Gamma_2 (\text{beta}_{i,t}), 
\end{equation}
where $\Gamma_0$, $\Gamma_1$, and $\Gamma_2$ are parameters, and $\text{cum\_return\_1\_0}_{i,t}$ and $\text{beta}_{i,t}$ are the lagged return and the ex-ante market beta raw predictors, $x_{i,k,t}$. \cite{kelly} and \cite{pastor2003liquidity} also model betas that are linear in predictors. I assume returns follow a standard factor model with the market portfolio return:
\begin{equation} \label{eq:market_model}
    r_{i,t+1} = \beta_{i,t}^M r_{t+1}^M + \epsilon_{i,t+1}^M,
\end{equation}
where $r_{t+1}^M$ is the market portfolio return and $\epsilon_{i,t+1}^M$ is the error term. Combining equations (\ref{eq:beta_in_characteristics}) and (\ref{eq:market_model}) yields the following:
\begin{equation} \label{eq:combined_interaction_model}
    r_{i,t+1} = \Gamma_0 r_{t+1}^M + \Gamma_1 (\text{cum\_return\_1\_0}_{i,t}) r_{t+1}^M + \Gamma_2 (\text{beta}_{i,t}) r_{t+1}^M + \epsilon_{i,t+1}^M,
\end{equation}
which implies that a panel regression of stock returns at time $t+1$ on the market excess return at time $t+1$ and the market excess return interacted with $\text{cum\_return\_1\_0}_{i,t}$ and $\text{beta}_{i,t}$ is the natural way to estimate $\Gamma_0$, $\Gamma_1$, and $\Gamma_2$. I estimate this regression. 

\begin{table}[!t] \centering
    \resizebox{0.5\textwidth}{!}{\begin{tabular}{@{\extracolsep{5pt}}lccc}
\\[-1.8ex]\hline
\hline \\[-1.8ex]
& \multicolumn{3}{c}{\textit{Dependent variable: Excess Return}} \
\cr \cline{2-4}
\\[-1.8ex] & Intercept $\times$ Mkt & cum\_return\_1\_0 $\times$ Mkt & beta $\times$ Mkt \\
\\[-1.8ex] & (1) & (2) & (3) \\
\hline \\[-1.8ex]
 Estimates: & 0.445$^{***}$ & -0.802$^{***}$ & 0.713$^{***}$ \\
& (0.062) & (0.222) & (0.065) \\
\hline
\hline \\[-1.8ex]
\textit{Note:} & \multicolumn{3}{r}{$^{*}$p$<$0.1; $^{**}$p$<$0.05; $^{***}$p$<$0.01} \\
\end{tabular}}
    \vspace{4mm}
    \caption{\textbf{Covariance Effects Regression.} This table shows the three parameter estimates of a panel regression of monthly excess stocks returns regressed on the market excess return, as well as the market excess return interacted with the previous month stock return (cum\_return\_1\_0) and the ex ante market beta. These estimates come from the period from February 1990 to January 2020, inclusive. Below the average values in parentheses, the standard errors, double clustered by month and stock, are shown. There are 1,253,175 month stock observations in this regression. }
    \label{tab:cal_regression}
\end{table}

Table \ref{tab:cal_regression} shows the estimates of the three regression parameters. A drop in prices by 1\% increases a stock's market beta by about 0.008 (e.g., from 1 to 1.008). This increased correlation with the market also increases $\beta_{i,t}$ and decreases $\sigma_{i,t,\epsilon}$, where the latter effect is large enough to increase the elasticity of the investor.\footnote{The formula for the derivatives necessarily to calculate these covariance effects are given in Appendix \ref{subsec:nn_model_derivs}.} Figure \ref{fig:calibration_and_weights2} shows this higher calibrated elasticity that allows both covariance and expected return effects. 

\begin{figure}[!t]
    \centering
    \includegraphics[width=0.6\linewidth,trim={0 0 0 1cm},clip]{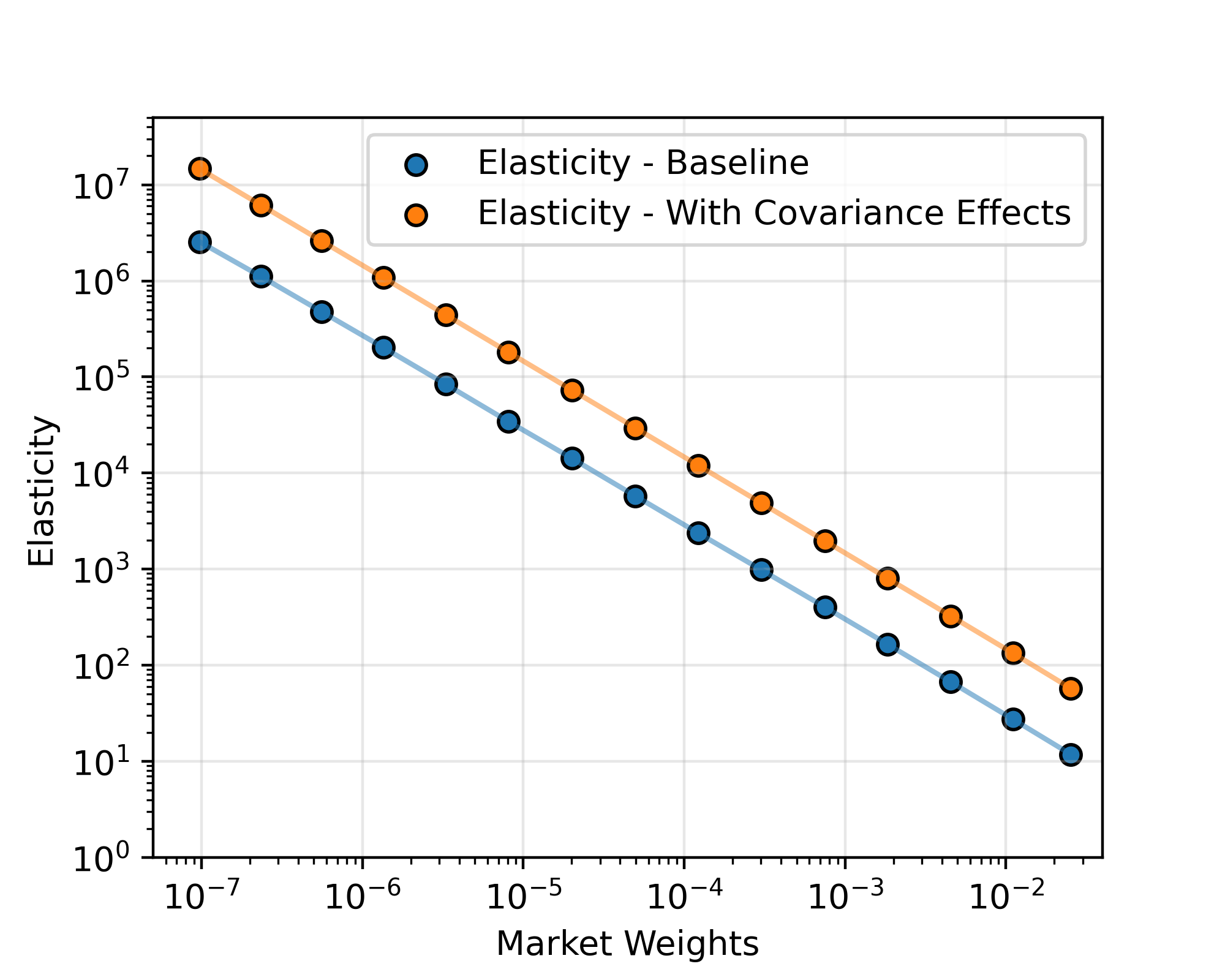}
    \vspace{4mm}
    \caption{\textbf{Elasticity and Market Weights.} This is similar to Panel A of Figure \ref{fig:calibration_and_weights}, except a curve with the elasticity with covariance matrix effects is added. The blue points show the calibrated elasticities with only mean effects, while the orange include covariance effects. }
\label{fig:calibration_and_weights2}
\end{figure}

\subsection{Value Weight Portfolios} \label{subsec:value_weight}

The FF3, FF6, and HXZ, in their original papers, consider value-weighted portfolios, where we instead consider cross-sectionally normalized predictors above. In this section, we consider using a proxy for value-weighted predictors for the BPZ$_L$, BSV, FF3, FF6, HXZ, and KNS portfolios since these models use predictor-weighted portfolios. In other words, using value-weighted predictors corresponds to using value-weighted portfolios with these models. 

A very simple way to create value-weighted predictors would be to use predictors of the following form:
\begin{equation}
    \tilde z_{i,k,t}^{\infty} \equiv \frac{\invbreve P_{i,t} \text{sign} (\udot{z}_{i,k,t})}{\sum_j | P_{j,t} \text{sign} ( \udot{z}_{j,k,t}) |},
\end{equation}
which means that if $x_{i,k,t}$ is above (below) median, then the portfolio weight is proportional to market equity, meaning that it is value-weighted, but on the long (short) side of the long-short portfolio. The long versus short sides can be flipped by multiplying by negative one; the process where $b$ is chosen to collapse the predictors into a single portfolio chooses both the sign and magnitude of weights to place on the predictors. 

If these predictors are used, then it is easy to see that for a positive value of the predictor and for an asset that makes up a relatively small share of the market:
\begin{equation*}
\frac{\partial \log (\tilde z_{i,k,t}^{\infty})}{\partial \log (p_{i,t})} \approx 1,
\end{equation*}
and thus the elasticity of these models, using equation (\ref{eq:elasticity}), with these predictors is mechanically zero ($= 1 - 1$). This does not mean that these statistical arbitrageurs would not react to price changes, but simply that the reaction to price changes occurs at these discontinuous jumps between the long and short end of the portfolios. It is important to understand how the elasticity of these models changes with value-weighted portfolios, but avoiding this discontinuity problem requires a proxy for value-weighted portfolios that smooths the discontinuity between the long and short ends of the portfolio. To do this, I define:
\begin{equation} \label{eq:value_weight_analogues}
    h_{\Xi} (x) \equiv \frac{\tan^{-1} (\Xi x)}{\tan^{-1} (\Xi / 2)} \;\; \text{  and  } \;\;
    \tilde z_{i,k,t}^{\Xi} = \frac{\invbreve P_{i,t} h_{\Xi} (\udot{z}_{i,k,t})}{\sum_j | P_{j,t} h_{\Xi} (\udot{z}_{j,k,t}) |}.
\end{equation}

\begin{table}[!t] \centering
    \resizebox{0.7\textwidth}{!}{\begin{tabular}{@{\extracolsep{5pt}}lcccccc}
\\[-1.8ex]\hline
\hline \\[-1.8ex]
\\[-1.8ex] $\Xi$ & \multicolumn{1}{c}{BPZ$_L$} & \multicolumn{1}{c}{BSV} & \multicolumn{1}{c}{FF3} & \multicolumn{1}{c}{FF6} & \multicolumn{1}{c}{HXZ} & \multicolumn{1}{c}{KNS}  \\
\\[-1.8ex] & (1) & (2) & (3) & (4) & (5) & (6) \\
\hline \\[-1.8ex]
 1 & 11.658$^{***}$ & 7.172$^{***}$ & -1.201$^{}$ & 7.172$^{***}$ & 3.008$^{}$ & 9.897$^{***}$ \\
& (2.423) & (2.502) & (0.796) & (2.337) & (2.184) & (2.521) \\
\\[-1.8ex]

 10 & 8.627$^{***}$ & 7.744$^{***}$ & -0.815$^{}$ & 6.526$^{***}$ & 2.387$^{}$ & 10.092$^{***}$ \\
& (2.115) & (2.528) & (0.674) & (2.192) & (1.994) & (2.597) \\
\\[-1.8ex]

 100 & 9.101$^{***}$ & 6.527$^{**}$ & -0.241$^{}$ & 6.478$^{***}$ & 2.056$^{}$ & 9.079$^{***}$ \\
& (1.983) & (2.580) & (0.282) & (2.108) & (1.777) & (2.625) \\
\\[-1.8ex]

 $10^9$ & 8.743$^{***}$ & 6.993$^{***}$ & -0.164$^{}$ & 6.254$^{***}$ & 2.001$^{}$ & 9.298$^{***}$ \\
& (2.104) & (2.569) & (0.197) & (2.094) & (1.710) & (2.616) \\
\\[-1.8ex]

\hline
\hline \\[-1.8ex]
\textit{Note:} & \multicolumn{6}{r}{$^{*}$p$<$0.1; $^{**}$p$<$0.05; $^{***}$p$<$0.01} \\
\end{tabular}}
    \vspace{4mm}
    \caption{\textbf{Alphas with Value Weighted Portfolios.} This table shows the monthly CAPM annualized alphas of the six statistical arbitrageurs, with different analogues of value-weighted portfolios, during the out-of-sample period: February 1990 to January 2020, inclusive. Portfolio weights are calculated using equation (\ref{eq:value_weight_analogues}), with the different values of $\Xi$ shown along the rows. As $\Xi$ increases, these analogues of value-weighted portfolios approach true value-weighted long-short portfolios. Standard errors are shown in parentheses below the estimates. Table \ref{tab:models} describes the initialisms of these models. Note that these six models are chosen simply because these are the models where the portfolio weights are not dictated by the model, and thus we can consider value-weighted portfolios. }
    \label{tab:value_weight_alpha}
\end{table}

Note that as $\Xi$ goes to infinity, $h_{\Xi} (x)$ converges point-wise to $\text{sign} (x)$. This smooths the sign function, creating a reasonable analogue for value-weighted predictors and portfolio weights. Appendix \ref{subsec:value_weights} provides some further discussion as well as the formula for the derivative of this predictor with respect to the log of prices, showing that it allows $\udot{z}_{j,k,t}$ values that are a function of prices to smoothly affect the elasticity. 

Table \ref{tab:value_weight_alpha} presents the out-of-sample CAPM alphas of these portfolios, with $\Xi$ values of 1, 10, 100, and $10^9$ (one billion). The BPZ$_L$, BSV, FF6, and KNS models generally performed well during this period, while the FF3 and HXZ performed more poorly. Table \ref{tab:value_weight_elas} presents the elasticity of these statistical arbitrageurs. Note that the elasticities do not monotonically decline towards zero as $\Xi$ increases, which is simply because these models are trained using these various $\Xi$ values, which generates different loadings within $b$, and thus different portfolio elasticities. The right-most column display the elasticities, and there are two important points to discuss. First, the average elasticity of 1.7 across models even with $\Xi = 1$ is smaller than the average 2.6 elasticity of these six models without value weights. Second, the average elasticity does decline as $\Xi$ increases and unsurprisingly, the elasticity with $\Xi = 10^9$ is quite close to zero. This does not mean that using value-weighted predictors generates demand that is completely non-reactive to price changes, but rather that these reactions occur at discontinuous jumps between portfolios that cannot be captured by derivatives that are inherently local. 

\begin{table}[!t] \centering
    \resizebox{0.8\textwidth}{!}{\begin{tabular}{@{\extracolsep{5pt}}lccccccc}
\\[-1.8ex]\hline
\hline \\[-1.8ex]
& \multicolumn{6}{c}{\textit{Elasticity}} \
\cr \cline{2-7} \cline{8-8}
\\[-1.8ex] $\Xi$ & \multicolumn{1}{c}{BPZ$_L$} & \multicolumn{1}{c}{BSV} & \multicolumn{1}{c}{FF3} & \multicolumn{1}{c}{FF6} & \multicolumn{1}{c}{HXZ} & \multicolumn{1}{c}{KNS} & Average \\
\\[-1.8ex] & (1) & (2) & (3) & (4) & (5) & (6) & \\
\hline \\[-1.8ex]

1 & 7.517$^{***}$ & 0.452$^{***}$ & 0.474$^{***}$ & 0.082$^{***}$ & -0.0004$^{***}$ & 1.494$^{***}$ & 1.670 \\
& (0.285) & (0.024) & (0.035) & (0.004) & (0.000) & (0.054) \\
\\[-1.8ex]

10 & 2.097$^{***}$ & 0.818$^{***}$ & 0.459$^{***}$ & 0.063$^{***}$ & -0.001$^{***}$ & 1.721$^{***}$ & 0.859 \\
& (0.106) & (0.028) & (0.040) & (0.003) & (0.000) & (0.060) \\
\\[-1.8ex]

100 & 1.663$^{***}$ & 0.648$^{***}$ & 0.026$^{***}$ & 0.021$^{***}$ & -0.001$^{***}$ & 1.039$^{***}$ & 0.566 \\
& (0.059) & (0.019) & (0.002) & (0.001) & (0.000) & (0.032) \\
\\[-1.8ex]

$10^9$ & 0.00061$^{***}$ & -0.00006$^{***}$ & -0.00002$^{***}$ & 0.00022$^{***}$ & 0.00025$^{***}$ & 0.00009$^{***}$ & 0.00018 \\
& (0.00003) & (0.00000) & (0.00000) & (0.00001) & (0.00001) & (0.00001) \\

\hline
\hline \\[-1.8ex]
\textit{Note:} & \multicolumn{7}{r}{$^{*}$p$<$0.1; $^{**}$p$<$0.05; $^{***}$p$<$0.01} \\
\end{tabular}}
    \vspace{4mm}
    \caption{\textbf{Elasticity with Value Weighted Portfolios.} This table shows the value-weighted and winsorized (5$^{th}$ and 95$^{th}$ percentiles) average (across assets and months) price elasticity of the six statistical arbitrageurs, with different analogues of value-weighted portfolios, during the out-of-sample period: February 1990 to January 2020, inclusive. Portfolio weights are calculated using equation (\ref{eq:value_weight_analogues}), with the different values of $\Xi$ shown along the rows. As $\Xi$ increases, these analogues of value-weighted portfolios approach true value-weighted long-short portfolios. Below the average values in parentheses, the standard errors, double clustered by month and stock, are shown. Table \ref{tab:models} describes the initialisms of these models. Note that these six models are chosen simply because these are the models where the portfolio weights are not dictated by the model, and thus we can consider value-weighted portfolios. There are 1,253,175 month stock observations these values are calculated from.}
    \label{tab:value_weight_elas}
\end{table}

\subsection{Trading Cost Optimization} \label{subsec:cost_optimization}

In practice, a statistical arbitrageur would be sensitive to trading costs. In an extreme case, a statistical arbitrageur would take an initial investment position and then never trade against price movements, generating completely inelastic demand. In this subsection, I consider a statistical arbitrageur that trades relatively conservatively due to large perceived trading costs, and show that this does indeed substantially reduce the elasticity of statistical arbitrageurs. 

I follow a slightly modified \cite{tradingcosts} trading costs optimizer. Let $w_t^*$ be the portfolio weight after optimizing for trading costs from their model, which has the form:
\begin{equation} \label{eq:trading_cost_weights}
    w_t^* = \left( 1 - s_{\text{aim}} \right) w_{t \leftarrow t-1} + (s_{\text{aim}}) \text{aim}_t
\end{equation}
where $w_{t \leftarrow t-1}$ is the $N$ dimension vector of weights from period $t-1$ that are changed only passively due to the price changes experienced at time $t$, $s_{\text{aim}}$ is a scalar which dictates the weighting, and 
\begin{equation} \label{eq:aim_weights}
    \text{aim}_t = \left( \sum_{\tau=0}^{\mathcal{T} - 1} \left(1 - \rho_{\text{aim}} \right) \left( \rho_{\text{aim}} \right)^{\tau} \E_t [w_{t+\tau}] \right)
    + \left( \rho_{\text{aim}} \right)^{\mathcal{T}} \E_t [w_{t + \mathcal{T}}].
\end{equation}
The intuition is straightforward. The cost-optimized statistical arbitrageur weights are a linear combination of the completely passive weights and the "aim" portfolio weights, where the aim portfolio is a linear combination of unoptimized weights now as well as future weights. As \cite{tradingcosts} discuss, a trading costs model should take into account future weights, and by doing this it can reduce trading costs now. As discussed in more detail in Appendix \ref{app:trading_costs}, I consider a finite $\mathcal{T}$ value, while \cite{tradingcosts} can consider the limit as $\mathcal{T}$ goes to infinity because of their simple trading strategy. The two key parameters are $s_{\text{aim}}$ and $\rho_{\text{aim}}$. As $s_{\text{aim}}$ approaches zero, the elasticity approaches zero and the statistical arbitrageur becomes increasingly passive. As $\rho_{\text{aim}}$ approaches one, the statistical arbitrageur discounts future expected portfolio weights less, implying less willingness to trade against short-term price movements. Thus, low values of $s_{\text{aim}}$ and high values of $\rho_{\text{aim}}$ generate more inelastic demand. As described in more detail in Appendix \ref{app:trading_costs}, I choose conservative parameters chosen from \cite{tradingcosts}, which yields $s_{\text{aim}} = 0.0303$ and $\rho_{\text{aim}} = 0.9681$. I emphasize that these are "conservative" values taken from \cite{tradingcosts} in order to obtain a lower bound on statistical arbitrageur elasticities. I use $\mathcal{T} = 24$, or two years.\footnote{The results are similar for larger $\mathcal{T}$ values.}

How does a statistical arbitrageur form expectations $\E_t [w_{t+\tau}]$? All of these conditional expectations are subjective expectations from the perspective of the statistical arbitrageur. The statistical arbitrageur needs some expected evolution of asset predictors, $Z_t$, to do this. To model this, I just fit a single AR(1) model for each predictor using the same data the statistical arbitrageur parameters uses to fit its parameters, in order to create true ex-ante predictions. I use a standard AR(1) predictions for $\tau$ periods out to create some matrix of predictors, denoted as $Z_{t+\tau | t}$. Then future portfolios are formed as $\E_t [w_{t+\tau}] \equiv f(Z_{t+\tau | t}) b$. Note that a more general VAR structure for predictors delivers similar results. Appendix \ref{app:trading_costs} also provides the formulas for the derivatives used to calculate the elasticities. 

Table \ref{tab:cost_optimizer} displays the cost-optimized portfolio results. All portfolios have positive and statistically significant out-of-sample alphas, except BPZ$_F$ and RF, as shown in the first row. The second row shows the average elasticities, and these conservative parameters values indeed generate inelastic demand. The average elasticity across models is 0.03. Thus, a market with large transaction costs and statistical arbitrageurs that take this into account indeed produces quite inelastic demand. 

\begin{table}[!t] \centering
    \resizebox{1.\textwidth}{!}{\begin{tabular}{@{\extracolsep{5pt}}lccccccccccccc}
\\[-1.8ex]\hline
\hline \\[-1.8ex]
\\[-1.8ex] & \multicolumn{1}{c}{BPZ$_F$} & \multicolumn{1}{c}{BPZ$_L$} & \multicolumn{1}{c}{BSV} & \multicolumn{1}{c}{CRW} & \multicolumn{1}{c}{DGU} & \multicolumn{1}{c}{FF3} & \multicolumn{1}{c}{FF6} & \multicolumn{1}{c}{GKX} & \multicolumn{1}{c}{HXZ} & \multicolumn{1}{c}{KNS} & \multicolumn{1}{c}{KPS} & \multicolumn{1}{c}{NN} & \multicolumn{1}{c}{RF}  \\
\\[-1.8ex] & (1) & (2) & (3) & (4) & (5) & (6) & (7) & (8) & (9) & (10) & (11) & (12) & (13) \\
\hline \\[-1.8ex]
 $\alpha$ & -0.918$^{}$ & 6.835$^{***}$ & 13.088$^{***}$ & 9.430$^{***}$ & 5.414$^{**}$ & 8.889$^{***}$ & 12.960$^{***}$ & -4.430$^{*}$ & 6.164$^{***}$ & 11.329$^{***}$ & 6.468$^{**}$ & 6.722$^{***}$ & 2.124$^{}$ \\
& (2.525) & (2.478) & (2.553) & (2.626) & (2.626) & (2.397) & (2.619) & (2.548) & (2.351) & (2.567) & (2.607) & (2.580) & (1.877) \\
\hline
\hline \\[-1.8ex]

Elasticity & -0.036$^{***}$ & 0.056$^{***}$ & 0.038$^{***}$ & 0.102$^{***}$ & 0.043$^{***}$ & 0.029$^{***}$ & 0.026$^{***}$ & 0.014$^{***}$ & -0.005$^{***}$ & 0.055$^{***}$ & 0.006$^{***}$ & 0.022$^{***}$ & 0.017$^{***}$ \\
& (0.003) & (0.003) & (0.002) & (0.008) & (0.002) & (0.001) & (0.001) & (0.001) & (0.000) & (0.002) & (0.002) & (0.001) & (0.001) \\

\hline
\hline \\[-1.8ex]
\textit{Note:} & \multicolumn{13}{r}{$^{*}$p$<$0.1; $^{**}$p$<$0.05; $^{***}$p$<$0.01} \\
\end{tabular}}
    \vspace{4mm}
    \caption{\textbf{Elasticity with Cost Optimization.} This table shows the monthly CAPM annualized alphas, as well as the value-weighted and winsorized (5$^{th}$ and 95$^{th}$ percentiles) average (across assets and months) price elasticity of the thirteen cost-optimized models during the out-of-sample period from February 1990 to January 2020, inclusive. Below the estimates in parentheses, the standard errors, are shown. Standard errors are double clustered by month and stock for the average elasticity estimates. Table \ref{tab:models} displays the initialisms of these models. There are 1,253,175 month stock observations these values are calculated from. Cost optimized portfolio weights are calculated as shown in equations (\ref{eq:trading_cost_weights}) and (\ref{eq:aim_weights}).}
    \label{tab:cost_optimizer}
\end{table}

\subsection{Systematic Elasticity Derivations and Results} \label{app:systematic_elasticity}

The systematic elasticity is defined as:
\begin{equation} \label{eq:systematic_elasticity}
    \eta_{i,t}^{j,sys} \equiv - \frac{\partial \log (s_{i,t})}{\partial \log (p_{j,t}^F)}
    = \frac{\partial \log(p_{i,t})}{\partial \log (p_{j,t}^F)}
    - \frac{\partial \log (w_{i,t})}{\partial \log (p_{j,t}^F)} - \frac{\partial \log (a_t)}{\partial \log (p_{j,t}^F)}.
\end{equation}

We can use the chain rule and our factor relationship in equation (\ref{eq:factor_relationship}) to calculate each component individually:\footnote{All derivatives are straightforward, except perhaps \(\partial \log (p_{n,t}) / \partial r_{n,t} = p_{n,t-1} / p_{n,t}\). Since we vary \(p_{n,t}\) while holding \(p_{n,t-1}\) and the dividend \(d_{n,t}\) fixed, \(p_{n,t}\) is an invertible function of \(r_{n,t}\). Thus, we can write:
\begin{equation*}
    \frac{\partial \log (p_{n,t}) }{ \partial r_{n,t}}
    = \left( \frac{\partial r_{n,t}}{\partial \log (p_{n,t})} \right)^{-1}
    = \left( \frac{\partial}{\partial \log (p_{n,t})} \left( \frac{\exp (\log (p_{n,t})) + d_{n,t}}{p_{n,t-1}} - R_{f,t} \right) \right)^{-1}
    = \left( \frac{p_{n,t}}{p_{n,t-1}} \right)^{-1}
    = \frac{p_{n,t-1}}{p_{n,t}}.
\end{equation*}}
{\small \begin{equation}
    \frac{\partial \log(p_{i,t})}{\partial \log (p_{j,t}^F)}
    = \frac{\partial \log(p_{i,t})}{\partial r_{i,t}}
    \frac{\partial r_{i,t}}{\partial F_{j,t}}
     \frac{\partial F_{j,t}}{\partial \log (p_{j,t}^F)}
     = \frac{p_{i,t-1}}{p_{i,t}} \beta_{i,j,t-1}^F
    \frac{p_{j,t}^F}{p_{j,t-1}^F}.
\end{equation}
\begin{equation}
    - \frac{\partial \log (w_{i,t})}{\partial \log (p_{j,t}^F)}
    = - \left( \sum_n \frac{\partial \log (w_{i,t})}{\partial \log (p_{n,t})} \frac{\partial \log (p_{n,t})}{\partial r_{n,t}} \frac{\partial r_{n,t}}{\partial F_{j,t}} \frac{\partial F_{j,t}}{\partial \log (p_{j,t}^F)} \right)
    = - \left( \sum_n \frac{\partial \log (w_{i,t})}{\partial \log (p_{n,t})} \frac{p_{n,t-1}}{p_{n,t}} \beta_{n,j,t-1}^F \frac{p_{j,t}^F}{p_{j,t-1}^F} \right).
\end{equation}
\begin{equation}
    - \frac{\partial \log (a_t)}{\partial \log (p_{j,t}^F)}
    = - \left( \sum_n \frac{\partial \log (a_t)}{\partial r_{n,t}} \frac{\partial r_{n,t}}{\partial F_{j,t}} \frac{\partial F_{j,t}}{\partial \log (p_{j,t}^F)} \right)
    = - \left( \sum_n \frac{a_{t-1}}{a_t} w_{n,t-1} \beta_{n,j,t-1}^F \frac{p_{j,t}^F}{p_{j,t-1}^F} \right).
\end{equation}}
Putting this together, we can write this out as three components:
{\small \begin{align} \label{eq:systematic_math}
    \eta_{i,t}^{j,sys} &= \underbrace{\frac{p_{i,t-1}}{p_{i,t}} \beta_{i,j,t-1}^F
    \frac{p_{j,t}^F}{p_{j,t-1}^F} \left( 1 - \frac{\partial \log (w_{i,t})}{\partial \log (p_{i,t})} \right)}_{\text{Direct}} \nonumber \\
    & \underbrace{- \left( \sum_{n \neq i} \frac{\partial \log (w_{i,t})}{\partial \log (p_{n,t})} \frac{p_{n,t-1}}{p_{n,t}} \beta_{n,j,t-1}^F \frac{p_{j,t}^F}{p_{j,t-1}^F} \right)}_{\text{Cross}}
    \underbrace{- \left( \sum_n \frac{a_{t-1}}{a_t} w_{n,t-1} \beta_{n,j,t-1}^F \frac{p_{j,t}^F}{p_{j,t-1}^F} \right)}_{\text{Wealth}}.
\end{align}}
Since $p_{i,t-1} / p_{i,t} \approx 1$ and $p_{j,t}^F / p_{j,t-1}^F \approx 1$, the direct effect is approximately $\beta_{i,j,t-1}^F \eta_{i,t}$. It is straightforward to see how equation (\ref{eq:systematic_math}) delivers the approximation in equation (\ref{eq:sys_elas_approx}). As mentioned in the text, the wealth effects can conceptually matter here, although empirically they tend to be small even with these systematic shocks.\footnote{With CARA demand as in equation (\ref{eq:cara_stevens}), total dollar holdings for asset $i$ are given by:
\begin{equation*}
    D_{i,t} \equiv a_t w_{i,t} = a_t \frac{1}{\gamma a_t} \left( \frac{\mu_{i,t} - \beta_{i,t}' \mu_{-i,t}}{\sigma_{i,t,\epsilon}^2} \right) = \frac{1}{\gamma} \left( \frac{\mu_{i,t} - \beta_{i,t}' \mu_{-i,t}}{\sigma_{i,t,\epsilon}^2} \right).
\end{equation*}
Thus, CARA demand exhibits the well-known feature that wealth effects, even for systematic shocks, do not affect demand. If we motivate demand with Epstein-Zin or CRRA demand as in Appendix \ref{subsec:epstein_zin}, then the $a_t$ term does not drop out when calculating $D_{i,t}$, and these wealth effects remain. In the empirical exercise, I assume $w_{i,t}$ is not a function of $a_t$, wealth effects remain in $D_{i,t}$, and thus wealth effects are allowed to be nonzero.}

Tables \ref{tab:sys_estimates5}, \ref{tab:sys_estimates1}, and \ref{tab:sys_estimates0} show the value-weighted three components in equation (\ref{eq:systematic_math}) and overall systematic elasticity for each of the six portfolios for various winsorization schemes. Note that for Tables \ref{tab:sys_estimates5} and \ref{tab:sys_estimates1}, the systematic elasticity does not necessarily equal the sum of the three components since the sum of winsorized terms does not necessarily equal the winsorized sum of terms. As discussed in the text, these tables show that the cross components tend to negate the direct components and the wealth components are small. Also, the direct components tend to be smaller than the ceteris paribus elasticity because betas tend to be less than one. 

\begin{table}[!t] \centering
    \resizebox{1.\textwidth}{!}{\begin{tabular}{@{\extracolsep{5pt}}llccccccccccccc}
\\[-1.8ex]\hline
\hline \\[-1.8ex]
& & \multicolumn{13}{c}{\textit{Average Systematic Elasticity}} \
\cr \cline{3-15}
\\[-1.8ex] Estimate & PCA & \multicolumn{1}{c}{BPZ$_F$} & \multicolumn{1}{c}{BPZ$_L$} & \multicolumn{1}{c}{BSV} & \multicolumn{1}{c}{CRW} & \multicolumn{1}{c}{DGU} & \multicolumn{1}{c}{FF3} & \multicolumn{1}{c}{FF6} & \multicolumn{1}{c}{GKX} & \multicolumn{1}{c}{HXZ} & \multicolumn{1}{c}{KNS} & \multicolumn{1}{c}{KPS} & \multicolumn{1}{c}{NN} & \multicolumn{1}{c}{RF}  \\
\\[-1.8ex] & & (1) & (2) & (3) & (4) & (5) & (6) & (7) & (8) & (9) & (10) & (11) & (12) & (13) \\
\\[-1.8ex] \hline \\[-1.4ex]

Elasticity & 1/+& 0.199$^{***}$ & -0.991$^{***}$ & -0.005$^{}$ & 0.326$^{***}$ & -0.224$^{***}$ & -0.170$^{***}$ & -0.028$^{***}$ & 0.326$^{***}$ & -0.034$^{***}$ & -0.254$^{***}$ & -0.056$^{**}$ & -0.248$^{***}$ & 0.027$^{***}$ \\
 & & (0.022) & (0.062) & (0.014) & (0.010) & (0.017) & (0.009) & (0.007) & (0.019) & (0.004) & (0.021) & (0.028) & (0.026) & (0.010) \\
 & 2/+ & 0.414$^{***}$ & -0.806$^{***}$ & 0.169$^{***}$ & 0.479$^{***}$ & -0.063$^{***}$ & 0.042$^{***}$ & 0.135$^{***}$ & 0.453$^{***}$ & 0.129$^{***}$ & -0.088$^{***}$ & -0.005$^{}$ & -0.192$^{***}$ & 0.145$^{***}$ \\
 & & (0.027) & (0.074) & (0.017) & (0.014) & (0.017) & (0.010) & (0.010) & (0.029) & (0.007) & (0.024) & (0.046) & (0.029) & (0.013) \\
 & 3/+ & 0.473$^{***}$ & -0.646$^{***}$ & 0.202$^{***}$ & 0.590$^{***}$ & 0.001$^{}$ & 0.025$^{**}$ & 0.175$^{***}$ & 0.576$^{***}$ & 0.158$^{***}$ & -0.014$^{}$ & 0.063$^{}$ & -0.181$^{***}$ & 0.188$^{***}$ \\
 & & (0.030) & (0.076) & (0.019) & (0.017) & (0.017) & (0.010) & (0.011) & (0.034) & (0.008) & (0.026) & (0.061) & (0.029) & (0.016) \\
 & 1/-& 0.608$^{***}$ & -0.637$^{***}$ & 0.251$^{***}$ & 0.680$^{***}$ & 0.055$^{**}$ & 0.099$^{***}$ & 0.227$^{***}$ & 0.618$^{***}$ & 0.242$^{***}$ & 0.010$^{}$ & 0.104$^{}$ & -0.203$^{***}$ & 0.265$^{***}$ \\
 & & (0.037) & (0.093) & (0.024) & (0.020) & (0.022) & (0.012) & (0.012) & (0.042) & (0.012) & (0.033) & (0.084) & (0.033) & (0.020) \\
 & 2/-& 0.289$^{***}$ & -1.207$^{***}$ & -0.017$^{}$ & 0.493$^{***}$ & -0.252$^{***}$ & -0.238$^{***}$ & -0.031$^{***}$ & 0.455$^{***}$ & -0.021$^{***}$ & -0.321$^{***}$ & -0.046$^{}$ & -0.323$^{***}$ & 0.063$^{***}$ \\
 & & (0.030) & (0.081) & (0.020) & (0.014) & (0.021) & (0.012) & (0.009) & (0.027) & (0.005) & (0.028) & (0.046) & (0.033) & (0.014) \\
 & 3/-& 0.227$^{***}$ & -1.230$^{***}$ & -0.003$^{}$ & 0.381$^{***}$ & -0.266$^{***}$ & -0.164$^{***}$ & -0.025$^{***}$ & 0.354$^{***}$ & -0.025$^{***}$ & -0.317$^{***}$ & -0.080$^{**}$ & -0.309$^{***}$ & 0.032$^{***}$ \\
 & & (0.027) & (0.078) & (0.017) & (0.011) & (0.020) & (0.010) & (0.008) & (0.022) & (0.004) & (0.026) & (0.035) & (0.033) & (0.012) \\
\hline \\[-1.8ex] Direct & 1/+& -0.220$^{***}$ & 2.586$^{***}$ & 0.728$^{***}$ & 0.772$^{***}$ & 0.882$^{***}$ & 0.354$^{***}$ & 0.270$^{***}$ & 3.536$^{***}$ & 0.033$^{***}$ & 1.217$^{***}$ & 1.797$^{***}$ & 0.880$^{***}$ & 0.394$^{***}$ \\
 & & (0.052) & (0.081) & (0.026) & (0.031) & (0.031) & (0.010) & (0.007) & (0.104) & (0.002) & (0.039) & (0.195) & (0.060) & (0.012) \\
 & 2/+ & -0.359$^{***}$ & 4.921$^{***}$ & 1.348$^{***}$ & 1.441$^{***}$ & 1.594$^{***}$ & 0.627$^{***}$ & 0.495$^{***}$ & 6.681$^{***}$ & 0.040$^{***}$ & 2.286$^{***}$ & 3.191$^{***}$ & 1.685$^{***}$ & 0.738$^{***}$ \\
 & & (0.083) & (0.158) & (0.048) & (0.058) & (0.057) & (0.018) & (0.013) & (0.201) & (0.004) & (0.075) & (0.359) & (0.117) & (0.023) \\
 & 3/+ & -0.446$^{***}$ & 6.083$^{***}$ & 1.646$^{***}$ & 1.741$^{***}$ & 1.948$^{***}$ & 0.774$^{***}$ & 0.601$^{***}$ & 8.232$^{***}$ & 0.035$^{***}$ & 2.797$^{***}$ & 3.845$^{***}$ & 1.990$^{***}$ & 0.882$^{***}$ \\
 & & (0.100) & (0.205) & (0.060) & (0.072) & (0.072) & (0.022) & (0.016) & (0.255) & (0.006) & (0.094) & (0.447) & (0.138) & (0.028) \\
 & 1/-& -0.466$^{***}$ & 7.936$^{***}$ & 2.127$^{***}$ & 2.282$^{***}$ & 2.523$^{***}$ & 1.003$^{***}$ & 0.792$^{***}$ & 10.814$^{***}$ & 0.043$^{***}$ & 3.617$^{***}$ & 4.911$^{***}$ & 2.659$^{***}$ & 1.148$^{***}$ \\
 & & (0.120) & (0.264) & (0.077) & (0.094) & (0.093) & (0.028) & (0.022) & (0.334) & (0.007) & (0.120) & (0.577) & (0.186) & (0.037) \\
 & 2/-& -0.318$^{***}$ & 4.335$^{***}$ & 1.185$^{***}$ & 1.255$^{***}$ & 1.439$^{***}$ & 0.576$^{***}$ & 0.441$^{***}$ & 5.907$^{***}$ & 0.041$^{***}$ & 1.996$^{***}$ & 2.845$^{***}$ & 1.429$^{***}$ & 0.633$^{***}$ \\
 & & (0.074) & (0.142) & (0.042) & (0.051) & (0.052) & (0.016) & (0.012) & (0.178) & (0.004) & (0.065) & (0.323) & (0.099) & (0.020) \\
 & 3/-& -0.241$^{***}$ & 3.120$^{***}$ & 0.881$^{***}$ & 0.938$^{***}$ & 1.067$^{***}$ & 0.428$^{***}$ & 0.329$^{***}$ & 4.301$^{***}$ & 0.042$^{***}$ & 1.478$^{***}$ & 2.197$^{***}$ & 1.098$^{***}$ & 0.480$^{***}$ \\
 & & (0.061) & (0.095) & (0.030) & (0.037) & (0.037) & (0.012) & (0.009) & (0.125) & (0.003) & (0.047) & (0.235) & (0.075) & (0.014) \\
\hline \\[-1.8ex] Cross & 1/+& 0.348$^{***}$ & -3.974$^{***}$ & -0.886$^{***}$ & -0.469$^{***}$ & -1.059$^{***}$ & -0.396$^{***}$ & -0.237$^{***}$ & -3.111$^{***}$ & 0.251$^{***}$ & -1.659$^{***}$ & -1.762$^{***}$ & -1.082$^{***}$ & -0.197$^{***}$ \\
 & & (0.063) & (0.151) & (0.040) & (0.026) & (0.043) & (0.013) & (0.008) & (0.092) & (0.014) & (0.063) & (0.200) & (0.086) & (0.011) \\
 & 2/+ & 0.710$^{***}$ & -6.070$^{***}$ & -1.302$^{***}$ & -0.889$^{***}$ & -1.573$^{***}$ & -0.531$^{***}$ & -0.315$^{***}$ & -5.997$^{***}$ & 0.471$^{***}$ & -2.482$^{***}$ & -3.030$^{***}$ & -1.648$^{***}$ & -0.320$^{***}$ \\
 & & (0.099) & (0.230) & (0.058) & (0.049) & (0.062) & (0.019) & (0.011) & (0.182) & (0.026) & (0.094) & (0.366) & (0.130) & (0.018) \\
 & 3/+ & 0.856$^{***}$ & -6.938$^{***}$ & -1.465$^{***}$ & -1.098$^{***}$ & -1.762$^{***}$ & -0.596$^{***}$ & -0.338$^{***}$ & -7.356$^{***}$ & 0.570$^{***}$ & -2.792$^{***}$ & -3.628$^{***}$ & -1.851$^{***}$ & -0.365$^{***}$ \\
 & & (0.118) & (0.264) & (0.067) & (0.060) & (0.070) & (0.021) & (0.012) & (0.228) & (0.032) & (0.107) & (0.463) & (0.145) & (0.020) \\
 & 1/-& 1.039$^{***}$ & -8.746$^{***}$ & -1.826$^{***}$ & -1.453$^{***}$ & -2.243$^{***}$ & -0.747$^{***}$ & -0.420$^{***}$ & -9.774$^{***}$ & 0.701$^{***}$ & -3.513$^{***}$ & -4.593$^{***}$ & -2.298$^{***}$ & -0.466$^{***}$ \\
 & & (0.145) & (0.333) & (0.084) & (0.079) & (0.089) & (0.026) & (0.015) & (0.303) & (0.039) & (0.135) & (0.603) & (0.180) & (0.026) \\
 & 2/-& 0.509$^{***}$ & -5.916$^{***}$ & -1.290$^{***}$ & -0.783$^{***}$ & -1.561$^{***}$ & -0.578$^{***}$ & -0.335$^{***}$ & -5.189$^{***}$ & 0.391$^{***}$ & -2.433$^{***}$ & -2.770$^{***}$ & -1.580$^{***}$ & -0.301$^{***}$ \\
 & & (0.088) & (0.224) & (0.058) & (0.043) & (0.062) & (0.019) & (0.011) & (0.157) & (0.022) & (0.093) & (0.337) & (0.124) & (0.017) \\
 & 3/-& 0.379$^{***}$ & -4.853$^{***}$ & -1.081$^{***}$ & -0.573$^{***}$ & -1.303$^{***}$ & -0.483$^{***}$ & -0.291$^{***}$ & -3.828$^{***}$ & 0.299$^{***}$ & -2.032$^{***}$ & -2.161$^{***}$ & -1.320$^{***}$ & -0.246$^{***}$ \\
 & & (0.076) & (0.184) & (0.048) & (0.032) & (0.052) & (0.016) & (0.009) & (0.112) & (0.017) & (0.077) & (0.241) & (0.105) & (0.014) \\
\hline \\[-1.8ex] Wealth & 1/+& -0.019$^{}$ & 0.088$^{***}$ & 0.097$^{***}$ & 0.061$^{***}$ & -0.067$^{***}$ & -0.142$^{***}$ & -0.069$^{***}$ & 0.039$^{***}$ & -0.241$^{***}$ & 0.085$^{***}$ & -0.001$^{***}$ & -0.072$^{***}$ & -0.138$^{***}$ \\
 & & (0.014) & (0.006) & (0.009) & (0.002) & (0.003) & (0.007) & (0.006) & (0.003) & (0.011) & (0.007) & (0.000) & (0.010) & (0.008) \\
 & 2/+ & -0.042$^{**}$ & 0.062$^{***}$ & 0.096$^{***}$ & 0.015$^{***}$ & -0.050$^{***}$ & -0.125$^{***}$ & -0.074$^{***}$ & 0.046$^{***}$ & -0.255$^{***}$ & 0.069$^{***}$ & -0.001$^{***}$ & -0.190$^{***}$ & -0.203$^{***}$ \\
 & & (0.017) & (0.007) & (0.009) & (0.002) & (0.003) & (0.007) & (0.007) & (0.004) & (0.012) & (0.008) & (0.000) & (0.013) & (0.012) \\
 & 3/+ & -0.056$^{***}$ & 0.060$^{***}$ & 0.036$^{***}$ & 0.058$^{***}$ & -0.073$^{***}$ & -0.239$^{***}$ & -0.129$^{***}$ & 0.058$^{***}$ & -0.297$^{***}$ & 0.023$^{***}$ & -0.002$^{***}$ & -0.237$^{***}$ & -0.237$^{***}$ \\
 & & (0.018) & (0.006) & (0.008) & (0.002) & (0.004) & (0.012) & (0.009) & (0.003) & (0.013) & (0.008) & (0.000) & (0.014) & (0.014) \\
 & 1/-& -0.072$^{***}$ & 0.023$^{***}$ & -0.015$^{**}$ & 0.009$^{***}$ & -0.078$^{***}$ & -0.257$^{***}$ & -0.188$^{***}$ & 0.049$^{***}$ & -0.312$^{***}$ & -0.026$^{***}$ & -0.002$^{***}$ & -0.402$^{***}$ & -0.293$^{***}$ \\
 & & (0.020) & (0.006) & (0.007) & (0.002) & (0.004) & (0.012) & (0.010) & (0.003) & (0.014) & (0.007) & (0.000) & (0.017) & (0.017) \\
 & 2/-& -0.037$^{**}$ & 0.087$^{***}$ & 0.050$^{***}$ & 0.096$^{***}$ & -0.104$^{***}$ & -0.254$^{***}$ & -0.152$^{***}$ & 0.033$^{***}$ & -0.340$^{***}$ & 0.047$^{***}$ & -0.002$^{***}$ & -0.159$^{***}$ & -0.206$^{***}$ \\
 & & (0.018) & (0.007) & (0.009) & (0.003) & (0.005) & (0.012) & (0.009) & (0.003) & (0.015) & (0.008) & (0.000) & (0.012) & (0.012) \\
 & 3/-& -0.023$^{}$ & 0.097$^{***}$ & 0.123$^{***}$ & 0.058$^{***}$ & -0.071$^{***}$ & -0.137$^{***}$ & -0.071$^{***}$ & 0.038$^{***}$ & -0.275$^{***}$ & 0.102$^{***}$ & -0.001$^{***}$ & -0.110$^{***}$ & -0.166$^{***}$ \\
 & & (0.016) & (0.007) & (0.010) & (0.002) & (0.004) & (0.007) & (0.007) & (0.003) & (0.012) & (0.009) & (0.000) & (0.011) & (0.010) \\

\hline
\hline \\[-1.8ex]

\textit{Note:} & \multicolumn{14}{r}{$^{*}$p$<$0.1; $^{**}$p$<$0.05; $^{***}$p$<$0.01} \\
\end{tabular}

}
    \vspace{4mm}
    \caption{\textbf{Systematic Elasticity Estimates Full Results.} This is similar to Table \ref{tab:sys_estimates}, except that the winsorized (5$^{th}$ and 95$^{th}$ percentiles) average of the three terms of the systematic elasticity from equation (\ref{eq:systematic_math}) with the same labels are also shown.}
    \label{tab:sys_estimates5}
\end{table}

\begin{table}[!t] \centering
    \resizebox{1.\textwidth}{!}{\begin{tabular}{@{\extracolsep{5pt}}llccccccccccccc}
\\[-1.8ex]\hline
\hline \\[-1.8ex]
& & \multicolumn{13}{c}{\textit{Average Systematic Elasticity}} \
\cr \cline{3-15}
\\[-1.8ex] Estimate & PCA & \multicolumn{1}{c}{BPZ$_F$} & \multicolumn{1}{c}{BPZ$_L$} & \multicolumn{1}{c}{BSV} & \multicolumn{1}{c}{CRW} & \multicolumn{1}{c}{DGU} & \multicolumn{1}{c}{FF3} & \multicolumn{1}{c}{FF6} & \multicolumn{1}{c}{GKX} & \multicolumn{1}{c}{HXZ} & \multicolumn{1}{c}{KNS} & \multicolumn{1}{c}{KPS} & \multicolumn{1}{c}{NN} & \multicolumn{1}{c}{RF}  \\
\\[-1.8ex] & & (1) & (2) & (3) & (4) & (5) & (6) & (7) & (8) & (9) & (10) & (11) & (12) & (13) \\
\\[-1.8ex] \hline \\[-1.4ex]

Elasticity & 1/+& -0.339$^{***}$ & -2.415$^{***}$ & -0.205$^{***}$ & 0.522$^{***}$ & -0.781$^{***}$ & -0.240$^{***}$ & -0.025$^{*}$ & 0.580$^{***}$ & -0.064$^{***}$ & -0.654$^{***}$ & -0.135$^{**}$ & -0.592$^{***}$ & -0.014$^{}$ \\
 & & (0.090) & (0.160) & (0.034) & (0.026) & (0.061) & (0.015) & (0.015) & (0.045) & (0.008) & (0.050) & (0.061) & (0.057) & (0.017) \\
 & 2/+ & 0.086$^{}$ & -2.304$^{***}$ & 0.052$^{}$ & 0.761$^{***}$ & -0.467$^{***}$ & 0.085$^{***}$ & 0.238$^{***}$ & 0.789$^{***}$ & 0.216$^{***}$ & -0.415$^{***}$ & -0.108$^{}$ & -0.458$^{***}$ & 0.160$^{***}$ \\
 & & (0.088) & (0.197) & (0.041) & (0.037) & (0.057) & (0.028) & (0.028) & (0.065) & (0.016) & (0.059) & (0.099) & (0.060) & (0.022) \\
 & 3/+ & 0.212$^{**}$ & -1.988$^{***}$ & 0.121$^{***}$ & 0.926$^{***}$ & -0.271$^{***}$ & 0.053$^{**}$ & 0.300$^{***}$ & 1.015$^{***}$ & 0.265$^{***}$ & -0.270$^{***}$ & -0.093$^{}$ & -0.454$^{***}$ & 0.235$^{***}$ \\
 & & (0.096) & (0.188) & (0.042) & (0.044) & (0.050) & (0.024) & (0.030) & (0.079) & (0.019) & (0.057) & (0.132) & (0.063) & (0.025) \\
 & 1/-& 0.421$^{***}$ & -2.374$^{***}$ & 0.097$^{*}$ & 1.029$^{***}$ & -0.240$^{***}$ & 0.169$^{***}$ & 0.365$^{***}$ & 0.988$^{***}$ & 0.403$^{***}$ & -0.368$^{***}$ & -0.316$^{*}$ & -0.479$^{***}$ & 0.318$^{***}$ \\
 & & (0.107) & (0.245) & (0.050) & (0.049) & (0.057) & (0.032) & (0.034) & (0.090) & (0.028) & (0.070) & (0.190) & (0.073) & (0.033) \\
 & 2/-& -0.468$^{***}$ & -3.113$^{***}$ & -0.313$^{***}$ & 0.769$^{***}$ & -0.931$^{***}$ & -0.333$^{***}$ & -0.034$^{*}$ & 0.788$^{***}$ & -0.047$^{***}$ & -0.864$^{***}$ & -0.283$^{***}$ & -0.788$^{***}$ & 0.021$^{}$ \\
 & & (0.123) & (0.215) & (0.045) & (0.037) & (0.074) & (0.019) & (0.018) & (0.063) & (0.010) & (0.066) & (0.100) & (0.071) & (0.024) \\
 & 3/-& -0.491$^{***}$ & -3.051$^{***}$ & -0.263$^{***}$ & 0.605$^{***}$ & -0.967$^{***}$ & -0.239$^{***}$ & -0.019$^{}$ & 0.629$^{***}$ & -0.050$^{***}$ & -0.835$^{***}$ & -0.217$^{***}$ & -0.718$^{***}$ & -0.027$^{}$ \\
 & & (0.109) & (0.205) & (0.042) & (0.030) & (0.076) & (0.019) & (0.019) & (0.052) & (0.009) & (0.063) & (0.077) & (0.069) & (0.021) \\
\hline \\[-1.8ex] Direct & 1/+& 1.109$^{***}$ & 4.388$^{***}$ & 1.170$^{***}$ & 1.336$^{***}$ & 1.475$^{***}$ & 0.451$^{***}$ & 0.317$^{***}$ & 6.741$^{***}$ & -0.064$^{***}$ & 1.900$^{***}$ & 3.874$^{***}$ & 1.553$^{***}$ & 0.601$^{***}$ \\
 & & (0.132) & (0.162) & (0.050) & (0.064) & (0.057) & (0.015) & (0.010) & (0.251) & (0.007) & (0.070) & (0.414) & (0.119) & (0.024) \\
 & 2/+ & 1.927$^{***}$ & 8.364$^{***}$ & 2.146$^{***}$ & 2.494$^{***}$ & 2.691$^{***}$ & 0.788$^{***}$ & 0.573$^{***}$ & 12.665$^{***}$ & -0.151$^{***}$ & 3.537$^{***}$ & 6.928$^{***}$ & 2.980$^{***}$ & 1.119$^{***}$ \\
 & & (0.187) & (0.312) & (0.092) & (0.121) & (0.106) & (0.027) & (0.018) & (0.476) & (0.014) & (0.130) & (0.764) & (0.233) & (0.045) \\
 & 3/+ & 2.276$^{***}$ & 10.362$^{***}$ & 2.610$^{***}$ & 3.040$^{***}$ & 3.300$^{***}$ & 0.963$^{***}$ & 0.690$^{***}$ & 15.657$^{***}$ & -0.199$^{***}$ & 4.309$^{***}$ & 8.318$^{***}$ & 3.503$^{***}$ & 1.345$^{***}$ \\
 & & (0.240) & (0.402) & (0.114) & (0.150) & (0.134) & (0.033) & (0.022) & (0.608) & (0.018) & (0.164) & (0.945) & (0.276) & (0.055) \\
 & 1/-& 3.008$^{***}$ & 13.324$^{***}$ & 3.309$^{***}$ & 3.882$^{***}$ & 4.293$^{***}$ & 1.252$^{***}$ & 0.908$^{***}$ & 20.035$^{***}$ & -0.256$^{***}$ & 5.532$^{***}$ & 10.317$^{***}$ & 4.642$^{***}$ & 1.703$^{***}$ \\
 & & (0.306) & (0.502) & (0.141) & (0.188) & (0.173) & (0.043) & (0.028) & (0.748) & (0.023) & (0.205) & (1.183) & (0.368) & (0.069) \\
 & 2/-& 1.737$^{***}$ & 7.221$^{***}$ & 1.864$^{***}$ & 2.151$^{***}$ & 2.413$^{***}$ & 0.727$^{***}$ & 0.514$^{***}$ & 11.032$^{***}$ & -0.117$^{***}$ & 3.068$^{***}$ & 5.989$^{***}$ & 2.489$^{***}$ & 0.947$^{***}$ \\
 & & (0.173) & (0.269) & (0.079) & (0.104) & (0.095) & (0.024) & (0.016) & (0.410) & (0.012) & (0.113) & (0.668) & (0.193) & (0.037) \\
 & 3/-& 1.338$^{***}$ & 5.274$^{***}$ & 1.410$^{***}$ & 1.609$^{***}$ & 1.781$^{***}$ & 0.550$^{***}$ & 0.388$^{***}$ & 8.095$^{***}$ & -0.074$^{***}$ & 2.294$^{***}$ & 4.682$^{***}$ & 1.932$^{***}$ & 0.726$^{***}$ \\
 & & (0.150) & (0.189) & (0.059) & (0.075) & (0.068) & (0.018) & (0.012) & (0.294) & (0.009) & (0.083) & (0.495) & (0.148) & (0.028) \\
\hline \\[-1.8ex] Cross & 1/+& -1.797$^{***}$ & -7.204$^{***}$ & -1.572$^{***}$ & -0.910$^{***}$ & -2.088$^{***}$ & -0.423$^{***}$ & -0.202$^{***}$ & -6.161$^{***}$ & 0.471$^{***}$ & -2.758$^{***}$ & -3.981$^{***}$ & -2.025$^{***}$ & -0.385$^{***}$ \\
 & & (0.214) & (0.324) & (0.081) & (0.056) & (0.098) & (0.032) & (0.023) & (0.227) & (0.031) & (0.117) & (0.425) & (0.169) & (0.022) \\
 & 2/+ & -2.210$^{***}$ & -11.045$^{***}$ & -2.303$^{***}$ & -1.725$^{***}$ & -3.079$^{***}$ & -0.455$^{***}$ & -0.166$^{***}$ & -11.801$^{***}$ & 0.882$^{***}$ & -4.134$^{***}$ & -6.962$^{***}$ & -3.080$^{***}$ & -0.613$^{***}$ \\
 & & (0.261) & (0.495) & (0.119) & (0.104) & (0.142) & (0.054) & (0.042) & (0.438) & (0.058) & (0.174) & (0.790) & (0.257) & (0.035) \\
 & 3/+ & -2.362$^{***}$ & -12.630$^{***}$ & -2.577$^{***}$ & -2.142$^{***}$ & -3.437$^{***}$ & -0.443$^{***}$ & -0.120$^{**}$ & -14.525$^{***}$ & 1.064$^{***}$ & -4.652$^{***}$ & -8.414$^{***}$ & -3.490$^{***}$ & -0.694$^{***}$ \\
 & & (0.300) & (0.567) & (0.135) & (0.127) & (0.157) & (0.065) & (0.051) & (0.556) & (0.070) & (0.197) & (1.001) & (0.289) & (0.040) \\
 & 1/-& -2.864$^{***}$ & -15.980$^{***}$ & -3.226$^{***}$ & -2.800$^{***}$ & -4.403$^{***}$ & -0.565$^{***}$ & -0.157$^{**}$ & -18.885$^{***}$ & 1.310$^{***}$ & -5.866$^{***}$ & -10.739$^{***}$ & -4.334$^{***}$ & -0.880$^{***}$ \\
 & & (0.377) & (0.720) & (0.170) & (0.164) & (0.201) & (0.080) & (0.062) & (0.706) & (0.086) & (0.249) & (1.305) & (0.360) & (0.052) \\
 & 2/-& -2.610$^{***}$ & -10.759$^{***}$ & -2.285$^{***}$ & -1.517$^{***}$ & -3.079$^{***}$ & -0.573$^{***}$ & -0.253$^{***}$ & -10.174$^{***}$ & 0.731$^{***}$ & -4.052$^{***}$ & -6.260$^{***}$ & -2.975$^{***}$ & -0.577$^{***}$ \\
 & & (0.275) & (0.484) & (0.119) & (0.091) & (0.143) & (0.048) & (0.036) & (0.378) & (0.048) & (0.172) & (0.711) & (0.247) & (0.033) \\
 & 3/-& -2.256$^{***}$ & -8.813$^{***}$ & -1.924$^{***}$ & -1.096$^{***}$ & -2.576$^{***}$ & -0.529$^{***}$ & -0.258$^{***}$ & -7.472$^{***}$ & 0.562$^{***}$ & -3.384$^{***}$ & -4.854$^{***}$ & -2.464$^{***}$ & -0.477$^{***}$ \\
 & & (0.253) & (0.396) & (0.099) & (0.067) & (0.121) & (0.038) & (0.028) & (0.269) & (0.037) & (0.144) & (0.514) & (0.207) & (0.027) \\
\hline \\[-1.8ex] Wealth & 1/+& -0.037$^{}$ & 0.155$^{***}$ & 0.159$^{***}$ & 0.092$^{***}$ & -0.203$^{***}$ & -0.239$^{***}$ & -0.115$^{***}$ & 0.065$^{***}$ & -0.395$^{***}$ & 0.149$^{***}$ & -0.002$^{***}$ & -0.136$^{***}$ & -0.215$^{***}$ \\
 & & (0.023) & (0.013) & (0.016) & (0.004) & (0.020) & (0.017) & (0.012) & (0.005) & (0.026) & (0.015) & (0.000) & (0.017) & (0.014) \\
 & 2/+ & -0.076$^{***}$ & 0.113$^{***}$ & 0.168$^{***}$ & 0.023$^{***}$ & -0.161$^{***}$ & -0.209$^{***}$ & -0.124$^{***}$ & 0.076$^{***}$ & -0.418$^{***}$ & 0.130$^{***}$ & -0.002$^{***}$ & -0.287$^{***}$ & -0.316$^{***}$ \\
 & & (0.028) & (0.013) & (0.018) & (0.003) & (0.018) & (0.015) & (0.014) & (0.007) & (0.028) & (0.017) & (0.000) & (0.023) & (0.021) \\
 & 3/+ & -0.099$^{***}$ & 0.112$^{***}$ & 0.064$^{***}$ & 0.087$^{***}$ & -0.213$^{***}$ & -0.401$^{***}$ & -0.211$^{***}$ & 0.098$^{***}$ & -0.487$^{***}$ & 0.050$^{***}$ & -0.004$^{***}$ & -0.370$^{***}$ & -0.369$^{***}$ \\
 & & (0.031) & (0.013) & (0.015) & (0.005) & (0.021) & (0.028) & (0.018) & (0.007) & (0.032) & (0.014) & (0.000) & (0.027) & (0.025) \\
 & 1/-& -0.124$^{***}$ & 0.048$^{***}$ & -0.020$^{}$ & 0.014$^{***}$ & -0.220$^{***}$ & -0.430$^{***}$ & -0.308$^{***}$ & 0.082$^{***}$ & -0.511$^{***}$ & -0.034$^{**}$ & -0.004$^{***}$ & -0.597$^{***}$ & -0.456$^{***}$ \\
 & & (0.033) & (0.011) & (0.013) & (0.003) & (0.021) & (0.030) & (0.023) & (0.006) & (0.034) & (0.014) & (0.000) & (0.032) & (0.031) \\
 & 2/-& -0.068$^{**}$ & 0.155$^{***}$ & 0.075$^{***}$ & 0.145$^{***}$ & -0.300$^{***}$ & -0.425$^{***}$ & -0.251$^{***}$ & 0.055$^{***}$ & -0.555$^{***}$ & 0.083$^{***}$ & -0.004$^{***}$ & -0.278$^{***}$ & -0.321$^{***}$ \\
 & & (0.031) & (0.014) & (0.016) & (0.007) & (0.028) & (0.029) & (0.020) & (0.006) & (0.037) & (0.015) & (0.000) & (0.024) & (0.022) \\
 & 3/-& -0.045$^{*}$ & 0.172$^{***}$ & 0.203$^{***}$ & 0.088$^{***}$ & -0.221$^{***}$ & -0.230$^{***}$ & -0.120$^{***}$ & 0.063$^{***}$ & -0.450$^{***}$ & 0.180$^{***}$ & -0.002$^{***}$ & -0.191$^{***}$ & -0.258$^{***}$ \\
 & & (0.027) & (0.015) & (0.020) & (0.004) & (0.023) & (0.017) & (0.014) & (0.006) & (0.030) & (0.018) & (0.000) & (0.019) & (0.017) \\

\hline
\hline \\[-1.8ex]

\textit{Note:} & \multicolumn{14}{r}{$^{*}$p$<$0.1; $^{**}$p$<$0.05; $^{***}$p$<$0.01} \\
\end{tabular}

}
    \vspace{4mm}
    \caption{\textbf{Systematic Elasticity Estimates - Alternative Winsorization.} This is similar to Table \ref{tab:sys_estimates5} except results are winsorized at the 1$^{st}$ and 99$^{th}$ percentiles.}
    \label{tab:sys_estimates1}
\end{table}

\begin{table}[!t] \centering
    \resizebox{1.\textwidth}{!}{\begin{tabular}{@{\extracolsep{5pt}}llccccccccccccc}
\\[-1.8ex]\hline
\hline \\[-1.8ex]
& & \multicolumn{13}{c}{\textit{Average Systematic Elasticity}} \
\cr \cline{3-15}
\\[-1.8ex] Estimate & PCA & \multicolumn{1}{c}{BPZ$_F$} & \multicolumn{1}{c}{BPZ$_L$} & \multicolumn{1}{c}{BSV} & \multicolumn{1}{c}{CRW} & \multicolumn{1}{c}{DGU} & \multicolumn{1}{c}{FF3} & \multicolumn{1}{c}{FF6} & \multicolumn{1}{c}{GKX} & \multicolumn{1}{c}{HXZ} & \multicolumn{1}{c}{KNS} & \multicolumn{1}{c}{KPS} & \multicolumn{1}{c}{NN} & \multicolumn{1}{c}{RF}  \\
\\[-1.8ex] & & (1) & (2) & (3) & (4) & (5) & (6) & (7) & (8) & (9) & (10) & (11) & (12) & (13) \\
\\[-1.8ex] \hline \\[-1.4ex]

Elasticity & 1/+& -3.709$^{***}$ & -5.333$^{}$ & -1.666$^{***}$ & 47.709$^{}$ & -13.196$^{}$ & -0.816$^{***}$ & -0.615$^{***}$ & -2.083$^{}$ & -0.085$^{***}$ & -2.736$^{***}$ & -2.039$^{}$ & -2.377$^{***}$ & -0.087$^{***}$ \\
 & & (0.551) & (3.673) & (0.333) & (46.869) & (10.039) & (0.168) & (0.230) & (2.394) & (0.025) & (0.502) & (1.294) & (0.587) & (0.033) \\
 & 2/+ & -3.550$^{***}$ & -1.842$^{}$ & -1.537$^{***}$ & 66.129$^{}$ & -17.235$^{}$ & -0.411$^{**}$ & -0.425$^{}$ & -3.212$^{}$ & 0.157$^{***}$ & -2.309$^{***}$ & -3.458$^{**}$ & -1.359$^{*}$ & 0.157$^{***}$ \\
 & & (0.722) & (7.379) & (0.381) & (65.106) & (14.181) & (0.164) & (0.310) & (4.204) & (0.040) & (0.672) & (1.723) & (0.735) & (0.046) \\
 & 3/+ & -2.993$^{***}$ & 0.095$^{}$ & -1.267$^{***}$ & 56.606$^{}$ & -14.497$^{}$ & -0.400$^{**}$ & -0.238$^{}$ & -3.277$^{}$ & 0.202$^{***}$ & -1.791$^{***}$ & -3.674$^{**}$ & -1.407$^{*}$ & 0.250$^{***}$ \\
 & & (0.716) & (8.980) & (0.390) & (55.251) & (12.381) & (0.171) & (0.270) & (4.517) & (0.044) & (0.478) & (1.616) & (0.788) & (0.050) \\
 & 1/-& -4.252$^{***}$ & 0.368$^{}$ & -1.524$^{***}$ & 69.942$^{}$ & -18.880$^{}$ & -0.257$^{*}$ & -0.303$^{}$ & -3.586$^{}$ & 0.281$^{***}$ & -2.125$^{***}$ & -5.764$^{***}$ & -1.080$^{}$ & 0.290$^{***}$ \\
 & & (1.188) & (11.621) & (0.463) & (68.543) & (16.197) & (0.147) & (0.355) & (5.088) & (0.055) & (0.601) & (1.987) & (0.875) & (0.075) \\
 & 2/-& -5.227$^{***}$ & -6.455$^{}$ & -2.133$^{***}$ & 52.248$^{}$ & -16.369$^{}$ & -1.078$^{***}$ & -0.774$^{***}$ & -2.839$^{}$ & -0.097$^{***}$ & -3.413$^{***}$ & -3.330$^{**}$ & -3.042$^{***}$ & -0.087$^{*}$ \\
 & & (0.817) & (5.798) & (0.432) & (51.177) & (12.580) & (0.224) & (0.295) & (3.266) & (0.033) & (0.559) & (1.571) & (0.721) & (0.048) \\
 & 3/-& -4.984$^{***}$ & -6.652$^{}$ & -2.106$^{***}$ & 57.693$^{}$ & -17.810$^{}$ & -0.924$^{***}$ & -0.780$^{**}$ & -2.936$^{}$ & -0.084$^{***}$ & -3.507$^{***}$ & -2.820$^{*}$ & -2.617$^{***}$ & -0.115$^{***}$ \\
 & & (0.728) & (4.607) & (0.408) & (56.867) & (13.728) & (0.189) & (0.304) & (3.217) & (0.029) & (0.709) & (1.587) & (0.689) & (0.043) \\
\hline \\[-1.8ex] Direct & 1/+& 5.093$^{***}$ & 15.881$^{***}$ & 2.805$^{***}$ & 9.540$^{*}$ & 7.578$^{***}$ & 1.432$^{***}$ & 0.226$^{}$ & 20.736$^{***}$ & -0.423$^{***}$ & 5.163$^{***}$ & 11.715$^{***}$ & 6.370$^{***}$ & 0.867$^{***}$ \\
 & & (0.474) & (3.842) & (0.350) & (5.645) & (2.387) & (0.281) & (0.174) & (4.070) & (0.043) & (0.531) & (1.958) & (1.291) & (0.058) \\
 & 2/+ & 9.027$^{***}$ & 30.947$^{***}$ & 5.104$^{***}$ & 17.357$^{*}$ & 13.691$^{***}$ & 2.546$^{***}$ & 0.417$^{}$ & 45.745$^{***}$ & -0.851$^{***}$ & 9.697$^{***}$ & 21.422$^{***}$ & 12.208$^{***}$ & 1.599$^{***}$ \\
 & & (0.783) & (7.854) & (0.605) & (10.266) & (4.340) & (0.539) & (0.304) & (13.862) & (0.083) & (1.030) & (3.614) & (2.476) & (0.104) \\
 & 3/+ & 11.037$^{***}$ & 37.626$^{***}$ & 6.283$^{***}$ & 17.858$^{*}$ & 18.762$^{***}$ & 2.998$^{***}$ & 0.420$^{}$ & 55.514$^{***}$ & -1.039$^{***}$ & 11.711$^{***}$ & 25.419$^{***}$ & 13.946$^{***}$ & 1.869$^{***}$ \\
 & & (0.955) & (9.571) & (0.771) & (9.236) & (6.668) & (0.595) & (0.412) & (15.880) & (0.099) & (1.195) & (4.246) & (2.938) & (0.118) \\
 & 1/-& 13.454$^{***}$ & 48.155$^{***}$ & 8.101$^{***}$ & 22.925$^{*}$ & 23.552$^{***}$ & 3.950$^{***}$ & 0.638$^{}$ & 75.598$^{***}$ & -1.345$^{***}$ & 15.109$^{***}$ & 31.574$^{***}$ & 18.044$^{***}$ & 2.307$^{***}$ \\
 & & (1.086) & (12.302) & (0.971) & (12.165) & (8.012) & (0.828) & (0.502) & (24.689) & (0.129) & (1.547) & (5.307) & (3.790) & (0.140) \\
 & 2/-& 7.889$^{***}$ & 25.367$^{***}$ & 4.541$^{***}$ & 12.875$^{*}$ & 13.521$^{***}$ & 2.240$^{***}$ & 0.339$^{}$ & 36.632$^{***}$ & -0.691$^{***}$ & 8.280$^{***}$ & 17.853$^{***}$ & 9.854$^{***}$ & 1.313$^{***}$ \\
 & & (0.680) & (6.066) & (0.564) & (6.787) & (4.709) & (0.429) & (0.302) & (9.021) & (0.069) & (0.821) & (3.013) & (2.110) & (0.083) \\
 & 3/-& 6.010$^{***}$ & 19.392$^{***}$ & 3.375$^{***}$ & 12.111$^{}$ & 8.664$^{***}$ & 1.788$^{***}$ & 0.321$^{*}$ & 26.071$^{***}$ & -0.510$^{***}$ & 6.234$^{***}$ & 14.159$^{***}$ & 8.075$^{***}$ & 1.050$^{***}$ \\
 & & (0.542) & (4.813) & (0.413) & (7.446) & (2.528) & (0.374) & (0.195) & (6.154) & (0.053) & (0.655) & (2.382) & (1.641) & (0.071) \\
\hline \\[-1.8ex] Cross & 1/+& -8.738$^{***}$ & -21.401$^{***}$ & -4.658$^{***}$ & 38.051$^{}$ & -20.508$^{*}$ & -1.968$^{***}$ & -0.708$^{**}$ & -22.901$^{***}$ & 0.798$^{***}$ & -8.075$^{***}$ & -13.751$^{***}$ & -8.586$^{***}$ & -0.693$^{***}$ \\
 & & (0.878) & (2.050) & (0.567) & (41.304) & (12.135) & (0.425) & (0.290) & (6.382) & (0.059) & (0.864) & (2.240) & (1.611) & (0.059) \\
 & 2/+ & -12.461$^{***}$ & -32.928$^{***}$ & -6.844$^{***}$ & 48.743$^{}$ & -30.716$^{*}$ & -2.712$^{***}$ & -0.699$^{}$ & -49.052$^{***}$ & 1.496$^{***}$ & -12.164$^{***}$ & -24.877$^{***}$ & -13.235$^{***}$ & -1.057$^{***}$ \\
 & & (1.266) & (3.128) & (0.834) & (54.976) & (18.232) & (0.646) & (0.465) & (17.961) & (0.109) & (1.316) & (4.074) & (2.508) & (0.086) \\
 & 3/+ & -13.884$^{***}$ & -37.670$^{***}$ & -7.629$^{***}$ & 38.636$^{}$ & -32.983$^{*}$ & -2.928$^{***}$ & -0.413$^{}$ & -58.914$^{***}$ & 1.807$^{***}$ & -13.565$^{***}$ & -29.089$^{***}$ & -14.923$^{***}$ & -1.171$^{***}$ \\
 & & (1.390) & (3.547) & (0.954) & (46.221) & (18.693) & (0.737) & (0.548) & (20.236) & (0.130) & (1.399) & (4.592) & (2.813) & (0.092) \\
 & 1/-& -17.527$^{***}$ & -47.850$^{***}$ & -9.607$^{***}$ & 46.999$^{}$ & -42.149$^{*}$ & -3.704$^{***}$ & -0.585$^{}$ & -79.287$^{***}$ & 2.220$^{***}$ & -17.199$^{***}$ & -37.333$^{***}$ & -18.440$^{***}$ & -1.464$^{***}$ \\
 & & (1.886) & (4.502) & (1.205) & (56.596) & (23.855) & (0.928) & (0.696) & (29.559) & (0.160) & (1.811) & (5.971) & (3.500) & (0.114) \\
 & 2/-& -13.008$^{***}$ & -32.009$^{***}$ & -6.760$^{***}$ & 39.187$^{}$ & -29.500$^{*}$ & -2.820$^{***}$ & -0.821$^{*}$ & -39.541$^{***}$ & 1.240$^{***}$ & -11.791$^{***}$ & -21.177$^{***}$ & -12.569$^{***}$ & -1.009$^{***}$ \\
 & & (1.296) & (3.034) & (0.834) & (44.529) & (16.944) & (0.630) & (0.442) & (12.182) & (0.091) & (1.216) & (3.332) & (2.360) & (0.086) \\
 & 3/-& -10.917$^{***}$ & -26.250$^{***}$ & -5.720$^{***}$ & 45.470$^{}$ & -26.184$^{*}$ & -2.442$^{***}$ & -0.961$^{***}$ & -29.087$^{***}$ & 0.951$^{***}$ & -9.955$^{***}$ & -16.976$^{***}$ & -10.468$^{***}$ & -0.851$^{***}$ \\
 & & (1.102) & (2.534) & (0.688) & (49.512) & (15.916) & (0.522) & (0.364) & (9.297) & (0.071) & (1.095) & (2.791) & (1.955) & (0.074) \\
\hline \\[-1.8ex] Wealth & 1/+& -0.063$^{*}$ & 0.187$^{***}$ & 0.187$^{***}$ & 0.118$^{***}$ & -0.265$^{***}$ & -0.280$^{***}$ & -0.134$^{***}$ & 0.082$^{***}$ & -0.461$^{***}$ & 0.177$^{***}$ & -0.003$^{***}$ & -0.162$^{***}$ & -0.261$^{***}$ \\
 & & (0.038) & (0.019) & (0.022) & (0.009) & (0.033) & (0.025) & (0.016) & (0.008) & (0.038) & (0.021) & (0.000) & (0.022) & (0.021) \\
 & 2/+ & -0.116$^{**}$ & 0.139$^{***}$ & 0.204$^{***}$ & 0.029$^{***}$ & -0.210$^{***}$ & -0.245$^{***}$ & -0.143$^{***}$ & 0.095$^{***}$ & -0.487$^{***}$ & 0.158$^{***}$ & -0.003$^{***}$ & -0.332$^{***}$ & -0.384$^{***}$ \\
 & & (0.048) & (0.019) & (0.025) & (0.004) & (0.029) & (0.022) & (0.018) & (0.010) & (0.041) & (0.023) & (0.000) & (0.030) & (0.032) \\
 & 3/+ & -0.147$^{***}$ & 0.139$^{***}$ & 0.079$^{***}$ & 0.112$^{***}$ & -0.276$^{***}$ & -0.470$^{***}$ & -0.245$^{***}$ & 0.123$^{***}$ & -0.567$^{***}$ & 0.064$^{***}$ & -0.005$^{***}$ & -0.430$^{***}$ & -0.448$^{***}$ \\
 & & (0.053) & (0.018) & (0.019) & (0.009) & (0.035) & (0.041) & (0.024) & (0.011) & (0.047) & (0.018) & (0.001) & (0.037) & (0.037) \\
 & 1/-& -0.180$^{***}$ & 0.063$^{***}$ & -0.019$^{}$ & 0.017$^{***}$ & -0.283$^{***}$ & -0.503$^{***}$ & -0.356$^{***}$ & 0.103$^{***}$ & -0.595$^{***}$ & -0.034$^{**}$ & -0.005$^{***}$ & -0.684$^{***}$ & -0.553$^{***}$ \\
 & & (0.057) & (0.015) & (0.016) & (0.004) & (0.034) & (0.043) & (0.032) & (0.010) & (0.049) & (0.017) & (0.000) & (0.046) & (0.046) \\
 & 2/-& -0.108$^{**}$ & 0.188$^{***}$ & 0.086$^{***}$ & 0.185$^{***}$ & -0.389$^{***}$ & -0.497$^{***}$ & -0.291$^{***}$ & 0.069$^{***}$ & -0.646$^{***}$ & 0.098$^{***}$ & -0.005$^{***}$ & -0.327$^{***}$ & -0.390$^{***}$ \\
 & & (0.051) & (0.021) & (0.020) & (0.014) & (0.047) & (0.043) & (0.028) & (0.008) & (0.053) & (0.019) & (0.001) & (0.033) & (0.032) \\
 & 3/-& -0.077$^{*}$ & 0.207$^{***}$ & 0.239$^{***}$ & 0.112$^{***}$ & -0.290$^{***}$ & -0.269$^{***}$ & -0.140$^{***}$ & 0.079$^{***}$ & -0.525$^{***}$ & 0.214$^{***}$ & -0.003$^{***}$ & -0.223$^{***}$ & -0.314$^{***}$ \\
 & & (0.045) & (0.022) & (0.027) & (0.008) & (0.037) & (0.024) & (0.017) & (0.009) & (0.044) & (0.025) & (0.000) & (0.025) & (0.026) \\

\hline
\hline \\[-1.8ex]

\textit{Note:} & \multicolumn{14}{r}{$^{*}$p$<$0.1; $^{**}$p$<$0.05; $^{***}$p$<$0.01} \\
\end{tabular}

}
    \vspace{4mm}
    \caption{\textbf{Systematic Elasticity Estimates - No Winsorization.} This is similar to Table \ref{tab:sys_estimates5} except results are not winsorized.}
    \label{tab:sys_estimates0}
\end{table}

\clearpage

\subsection{Which Assets have Elastic Demand} \label{app:which_assets}

In this section, I discuss which assets have elastic demand from the various statistical arbitrageur models. Equation (\ref{eq:statistical arbitrageur_elasticity}) shows the elasticity of each statistical arbitrageur comes from the different predictor variables. The key finding from this analysis here is that the demand elasticity varies widely across different models. Some statistical arbitrageur models, such as FF3 and BSV, exhibit elastic demand for value stocks for example, whereas others, like GKX, show particularly high elasticity for stocks with short-term reversals. Generally, most statistical arbitrageur models derive their elastic demand from a combination of both high-frequency price variations (such as short-term reversals) and low-frequency price variations (like value stocks). However, there is significant heterogeneity in how different models achieve this elasticity.

I perform two main analyses. The first examines directly the partial effects, while the second examines how the portfolios correlate with predictor-based portfolio returns. This second analysis is able to examine which stocks have elastic demand overall unconditionally, rather than holding fixed many other partial effects. 

We can directly measure which predictors give more elastic demand across models, and I show the value-weighted winsorized components of this equation for each predictor $k$ in the heatmap in Figure \ref{fig:which_assets_direct}. I show only the price-related predictors. The darker the blue in the square, the larger the elasticity component is in equation (\ref{eq:statistical arbitrageur_elasticity}). The squares with the green dots means that those predictors represent components that positively contribute to the elasticity of the model, while the red dots indicate predictors that negatively contribute to the elasticity of the asset. Since many statistical arbitrageur portfolios have positive market betas, they act like market indexers to some small degree, which tends to decrease the overall elasticity as discussed above. The short-term-reversals predictor (cum\_return\_1\_0) is an important elasticity contributor for many models, but it negatively contributes to the elasticity for the BPZ$_{\text{F}}$ and the KPS model. Note that the classic factor models, like FF3, have zero loadings (no green or red dot) on many predictors by construction. Overall, the models are quite different, creating a rich diversity of sources of elasticity. Which assets have elastic demand? It turns out that it depends quite a bit on the model. 

It should be noted that the BPZ$_{\text{F}}$ and RF models have many zero dots as well. These models are discontinuous functions of prices, if a large price movement is required to reach a discontinuous breakpoint, then the numerical derivatives used to calculate the sensitivity to price picks up zero sensitivity to prices through a given predictor. To deal with this, I consider numerical derivatives with reasonable log-price changes, but these still often pick up zero effects. This is quite analogous to investing in the \cite{ff3} value portfolios with bins sorted on discontinuous breakpoints. The elasticity of this investment is zero, but the strategies are still somewhat sensitive to prices across the portfolio bins. See Appendix \ref{subsec:value_weights} for more discussion on investing with bin-style investment strategies. 

\begin{figure}[!t] \centering
    \resizebox{0.9\textwidth}{!}{\includegraphics[trim={0 0 0 0.3cm},clip]{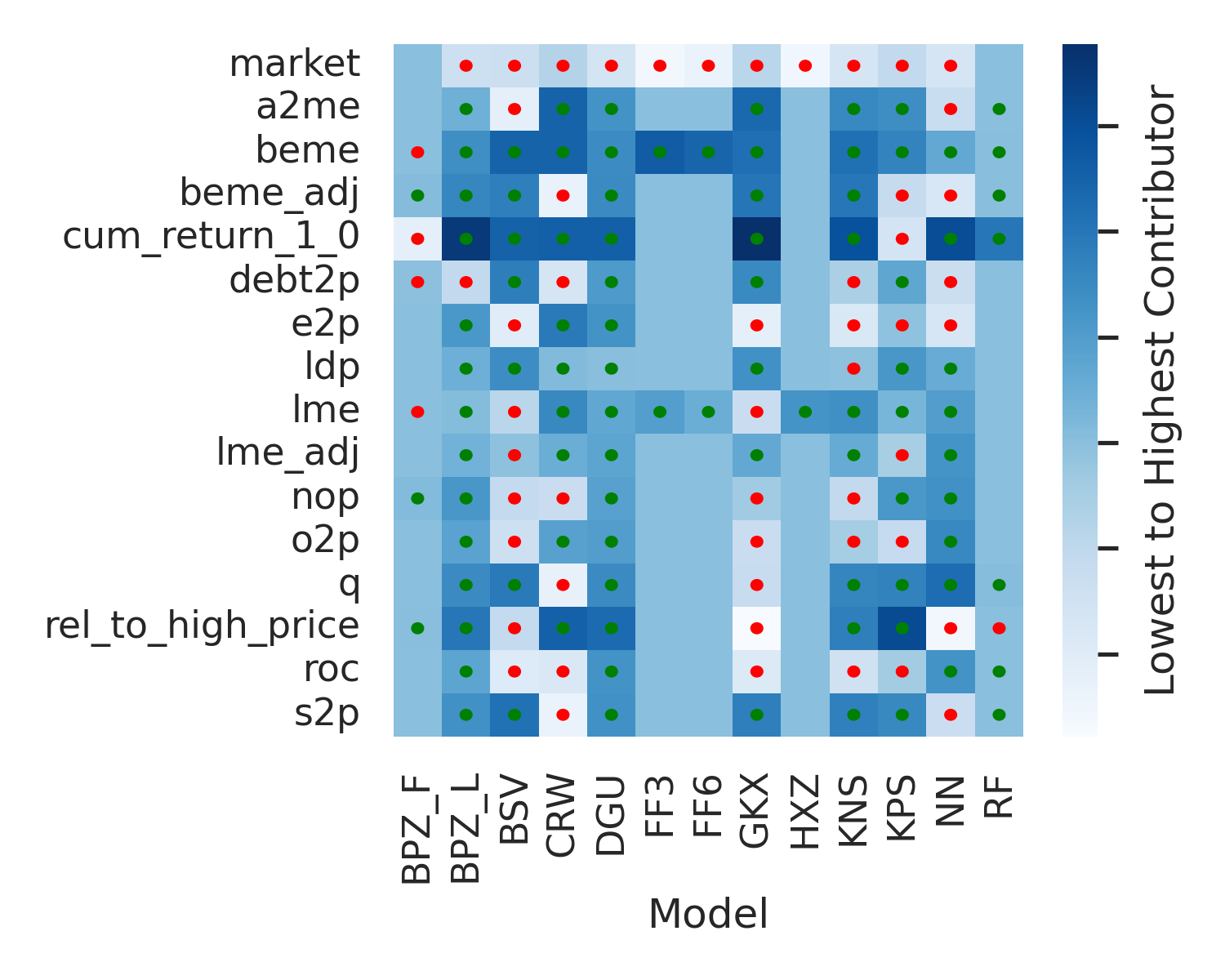}}
    \vspace{4mm}
    \caption{\textbf{Which Assets have Elastic Demand? Partial Effects.} This heatmap shows, for each statistical arbitrageur model, the value-weighted average contributions to elasticity in equation (\ref{eq:statistical arbitrageur_elasticity}). The statistical arbitrageur models are labeled along the $x$-axis, while the predictor variables that are functions of prices are shown along the $y$-axis. If the term is positive (contributes positively to elasticity), then there is a green dot on the square for the given predictor and model. If the term is negative, then a red dot is put on the square. If the term is zero, then no dot is put on the square.}
    \label{fig:which_assets_direct}
\end{figure}

The above analysis using equation (\ref{eq:statistical arbitrageur_elasticity}) depends on partial effects, i.e., holding fixed all predictors and determining the sensitivity to only a single predictor. One may also care about which assets have elastic demand, not in this sense of holding all else equal, but in an overall sense. In order to understand which assets have more elasticity, we can just correlate the statistical arbitrageur portfolios to classic long-short predictor based portfolios. Specifically, for each statistical arbitrageur model, I estimate the correlation between the price-related predictor-weighted factor portfolios described in equation (\ref{eq:port_returns}) and statistical arbitrageur portfolio returns. 

Figure \ref{fig:which_assets} shows a heatmap of these correlations. Unsurprisingly, this shows a different picture where the price relative to the past year's highest price (rel\_to\_high\_price) is correlated with many statistical arbitrageur portfolios. While some predictors negatively contribute to the elasticity for some models in Figure \ref{fig:which_assets_direct}, in some cases the sign flips when looking at the unconditional correlation of returns. Again, this heatmap also shows quite a bit of variation across models, indicating that models are related differently to different price-related variables. 

\begin{figure}[!t] \centering
    \resizebox{0.9\textwidth}{!}{\includegraphics[trim={0 0 0 0.3cm},clip]{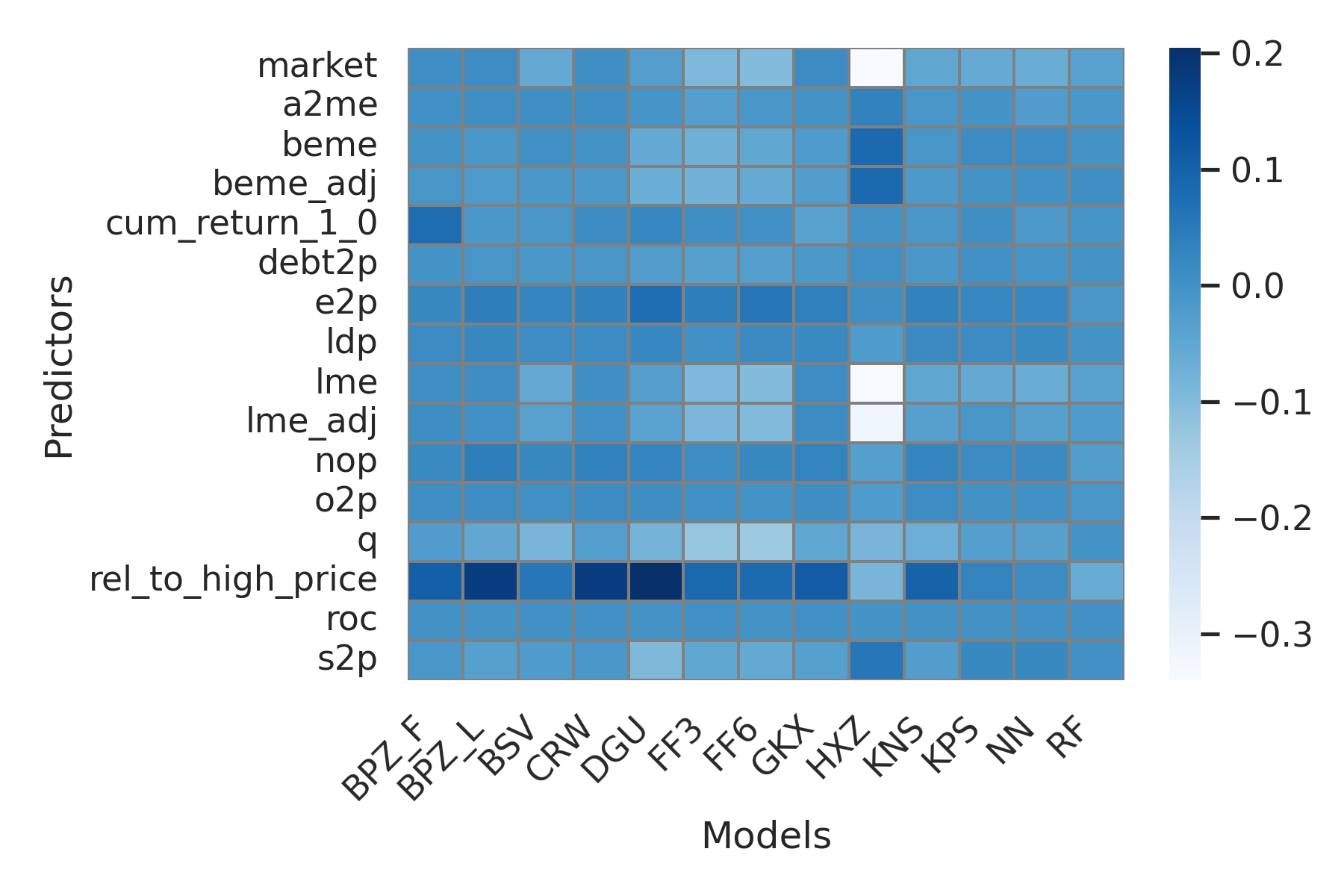}}
    \vspace{4mm}
    \caption{\textbf{Which Assets have Elastic Demand? Overall Correlations.} This heatmap shows, for each statistical arbitrageur model, the correlation between predictor-weighted portfolios based on price-related predictors (see equation (\ref{eq:port_returns})) and statistical arbitrageur returns. The statistical arbitrageur models are labeled along the $x$-axis, while the predictor variables that are functions of prices are shown along the $y$-axis.}
    \label{fig:which_assets}
\end{figure}

\subsection{Similarity Between Statistical Arbitrageurs Models (HSAs) and Systematic Hedge Funds} \label{app:statarb_similarity}

In this section, I assess how similar the HSAs are compared to systematic hedge funds, a subset of which are perhaps the most likely candidates for actual statistical arbitrageurs in the market. I follow the KY methodology to correct investors types in the 13F data, and define a systematic hedge fund as an "investment advisor" that holds at least 500 stocks. There are many hedge funds that hold only 10 or 20 stocks, instead of investing in systematic strategies, and these are obviously not acting like the HSAs I consider. With this filter, I capture some well-known hedge funds that use quantitative models to pursue systematic equities trading strategies, including AQR, Renaissance, Tower Research, Citadel, Two Sigma, S.A.C Capital, D.E. Shaw, Roll and Ross Asset Management, Paradigm Asset Management, Quantitative Systematic Strategies, Bridgewater Associates, Cubist Systematic Strategies, Martingale Asset Management, and Quantitative Systematic Strategies. 

\begin{figure}[!t] \centering
    \resizebox{0.9\textwidth}{!}{\includegraphics[trim={0 0 0 0.3cm},clip]{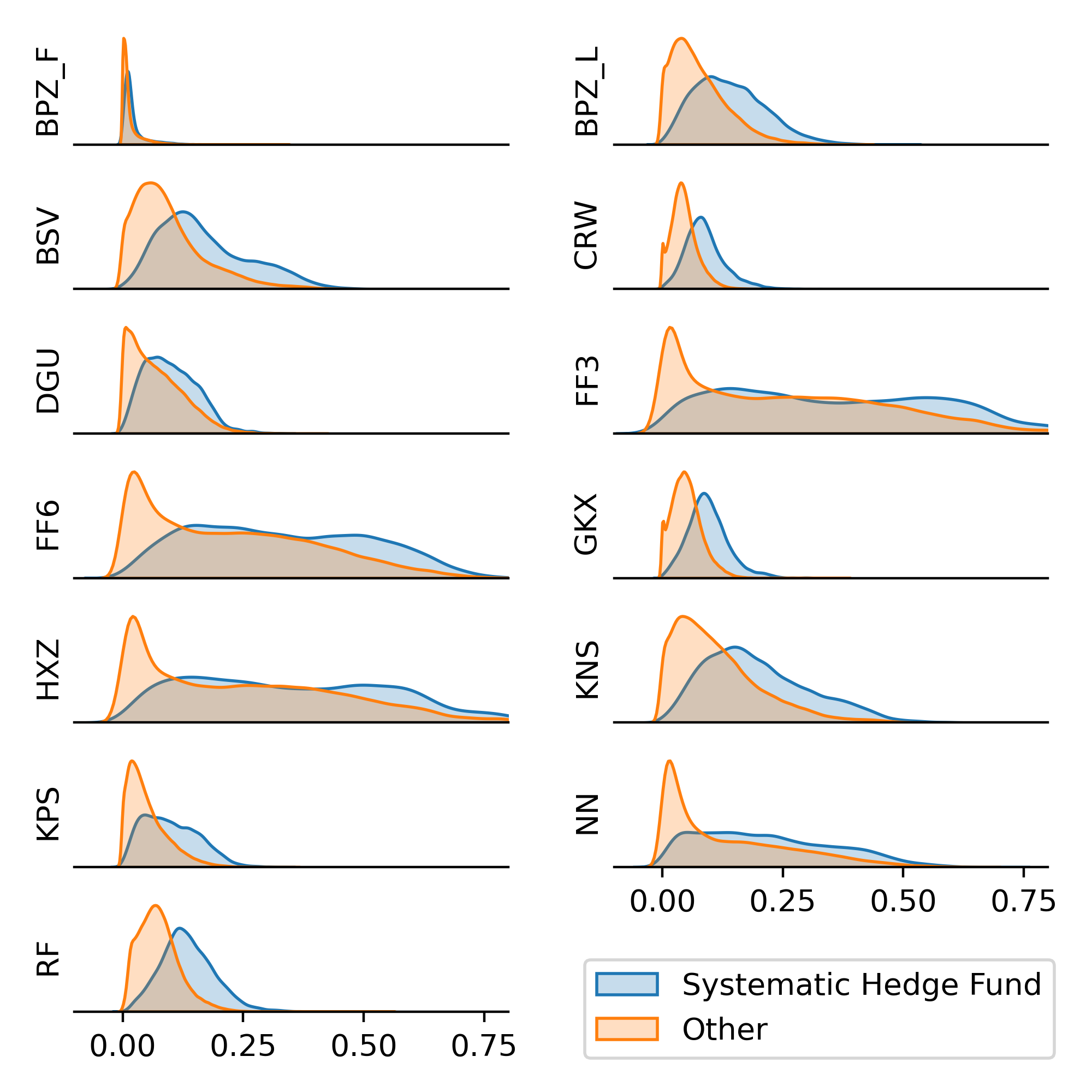}}
    \vspace{4mm}
    \caption{\textbf{Compare Statistical Arbitrage Portfolios with Holdings} This figure show kernel density plots of the cosine similarity between statistical arbitrageur long-only portfolio weights and 13F investor portfolio weights, for systematic equity hedge funds and all other 13F investors. The distributions are calculated across quarters and investors in the 13F data. I define a systematic equity hedge fund as an "investment advisor" that holds at least 500 stocks.
    }
    \label{fig:compare_arbs_with_holdings}
\end{figure}

Figure \ref{fig:compare_arbs_with_holdings} shows the distribution, across time and institutions, of the cosine similarity between the vectors of weights of a given HSA compared to the systematic hedge fund. The 13F data reports only long positions, and so for comparison, I set the short positions of the HSA to zero and then normalize HSA weights to sum to one just as the 13F investor weights. The cosine similarity between HSA and all other 13F investors is also shown. 

While the cosine similarity between portfolio weights is higher on average for the systematic hedge funds, the underlying populations of course still have quite a large degree of overlap. The systematic strategies are actually surprisingly uncorrelated to each other, as can be seen from the large fraction of model returns that are alphas relative to each other (see Table \ref{tab:cross_alphas}). Also, some investors identified as systematic hedge funds are likely not behaving as statistical arbitrageurs while other 13F investors, like some mutual fund investors, may actually be pursuing statistical arbitrage strategies. 

\subsection{Do Statistical Arbitrageurs (HSAs) Actually Arbitrage Away the Alpha?} \label{subsec:arb_away_alpha}

Do HSAs arbitrage away the alpha in the counterfactual experiments? In this subsection, I present the asset pricing implications of the equilibrium model. I first present these results using the counterfactual experiments discussed in the main text of the paper. Then I consider a Lucas critique about the aggregate market potentially having a different demand function in response to counterfactual HSA investors. 

\subsubsection{Alpha with HSA Investors}

Financial economists have allocated quite a bit of attention to producing these HSA models, as can be seen by looking at the list of models in Table \ref{tab:models}. The high out-of-sample alphas and Sharpe ratios in the literature and also shown in Tables \ref{tab:returns} and \ref{tab:cross_alphas} may lead one to believe that statistical arbitrageurs can easily identify alpha and arbitrage it away. However, in this section, I show that these statistical arbitrageurs actually struggle to arbitrage away alpha. When alpha is reduced, the majority of alpha remains and the remaining alpha is always statistically significant. Why do these statistical arbitrageurs struggle to arbitrage away the alpha? This is for two reasons, discussed below. 

First, statistical arbitrageur models fail to identify and separate alpha and systematic risk in the cross-section. Table \ref{tab:cross_alphas} shows that these models actually tend to disagree about what is alpha and what is systematic risk. If all models were only minor variations on each other, then they would effectively price the portfolio corresponding to other models, delivering statistically insignificant alphas. However, this is not the case. Given that these models are not able to price their test assets, these models cannot price the cross-section. They fail to separately identify alpha and systematic risk. \cite{bbd} show the same result, and even show that these models fail to do this conditionally as well. 

Second, statistical arbitrageur models are relatively inelastic, limiting their ability to sufficiently trade against mispricing. This has been discussed at length above, where I showed that these models produce relatively inelastic demand and that when incorporated into the market with an equilibrium model, the elasticity of the market for stocks changes little. While the high out-of-sample alphas may appear promising, the reality of inelastic demand pierces through this smokescreen, providing clarity.

In summary, given both the level and slope of the statistical arbitrageur demand function, it is not surprising that statistical arbitrageurs cannot arbitrage away all alpha, and in some cases it may even increase. This is what I show in this subsection, and it illustrates why we may care about the elasticity of statistical arbitrageurs and the elasticity of the market with these statistical arbitrageurs. 

There are three full decades (120-month periods) in the out-of-sample period: February 1990 - January 2000, February 2000 - January 2010, and February 2010 - January 2020. It is perhaps not surprising that alpha is not arbitraged away immediately when the statistical arbitrageurs enter the market, so the analysis here is focused on the last decade, after the statistical arbitrageurs have had two full decades of participation in the market. In particular, in each counterfactual experiment, we have 62 predictor-based portfolios, $F_{k,t+1}^c$, calculated from equation (\ref{eq:port_returns}). For the CAPM, \cite{ff3}, \cite{famafrench15} model with momentum, and the \cite{zhang} models, I use the 62 portfolios not in the given model as test assets. This means there are 62 test assets for the CAPM model (with 1 right-hand side factor), 60 test assets for the \cite{ff3} model (with 3 right-hand side factors), 57 test assets for the \cite{famafrench15} model with momentum (with 6 right-hand side factors), and 59 test assets for the \cite{zhang} model (with 4 right-hand side factors). I then calculate the GRS \citep{grs} test for this decade across test assets for a given model, both for the every counterfactual experiment and for the portfolio returns observed in the data (corresponding to $\theta_t = 0$). To summarize how the GRS test results change relative to the case without statistical arbitrageurs, which I refer to as the baseline case, I take the log ratio of GRS test statistics. These results are shown in Table \ref{tab:grs_baseline}. In the baseline data, the GRS test statistics are statistically different from zero alpha for all models. In order for the alphas in the counterfactual experiments to fall low enough to fail to reject the null hypothesis of zero alpha, the log GRS ratio needs to fall to about -1.9 at the 5\% significance level. There are no cases where this log GRS ratio falls even half this far. This is true across all four factor models. In a minority of cases, the log ratio is actually positive, implying that mispricing relative to the respective factor model actually increased. This not surprising given the inelastic demand of statistical arbitrageurs. Note that Table \ref{tab:grs_2010} shows similar results during the February 2000 - January 2010 period. 

In summary, statistical arbitrageurs from the literature, when actually inserted into an equilibrium model, are ultimately unable to arbitrage away the alpha. Given that these models are inelastic and struggle to separate alpha and systematic risk, this is not surprising. Statistical arbitrageurs, at least those from the literature, fall short of statistical arbitrageurs used in classic theoretical models. 

\subsection{Double Elasticity Experiments} \label{subsec:double_elas} 

In this section of the Appendix, I describe the counterfactual experiments where the non-HSA investors have an elasticity that is doubled. The counterfactual equilibrium is calculated just as described above. In order the double the elasticity of the LLDM investors, I have to choose whether to increase the log component of the elasticity, the level component, or both in equation (\ref{eq:agg_elasticity}). For simplicity, I just scale the coefficients in order to increase both components by the same amount equally. 

Table \ref{tab:mkt_elasticity_double} shows the elasticity with HSA investors minus the elasticity of the market with only non-HSA investors. This is similar to Table \ref{tab:mkt_elasticity_simple} and has a similar takeaway: HSA investors have little impact on the elasticity of the market for individual assets. Note that due to the higher elasticity of non-HSA investors, HSA investors actually tend to \textit{decrease} the elasticity of the market. This certainly stands in contrast to the typical idea of statistical arbitrageurs providing an elastic source of demand, making the market more elastic. 

\begin{table}[!t] \centering
    \resizebox{1\textwidth}{!}{\begin{tabular}{@{\extracolsep{5pt}}lccccccccccccc}
\\[-1.8ex]\hline
\hline \\[-1.8ex]
 & \multicolumn{13}{c}{\textit{Change in Aggregate 
 Counterfactual Elasticity for Individual Stocks with Different Statistical Arbitrageurs}} \
\cr \cline{2-14}
\\[-1.8ex] AF $\theta_t$ & \multicolumn{1}{c}{BPZ$_F$} & \multicolumn{1}{c}{BPZ$_L$} & \multicolumn{1}{c}{BSV} & \multicolumn{1}{c}{CRW} & \multicolumn{1}{c}{DGU} & \multicolumn{1}{c}{FF3} & \multicolumn{1}{c}{FF6} & \multicolumn{1}{c}{GKX} & \multicolumn{1}{c}{HXZ} & \multicolumn{1}{c}{KNS} & \multicolumn{1}{c}{KPS} & \multicolumn{1}{c}{NN} & \multicolumn{1}{c}{RF}  \\
\\[-1.8ex] & (1) & (2) & (3) & (4) & (5) & (6) & (7) & (8) & (9) & (10) & (11) & (12) & (13) \\
\hline \\[-1.8ex]

 0.0 & 0.000 & 0.000 & 0.000 & 0.000 & 0.000 & 0.000 & 0.000 & 0.000 & 0.000 & 0.000 & 0.000 & 0.000 & 0.000 \\
 0.0001 &   0.000 & 0.001 & 0.001 & 0.000 & 0.000 & 0.000 & 0.000 & 0.000 & 0.000 & 0.001 & 0.000 & 0.000 & 0.000 \\
 0.0005 &   0.000 & 0.004 & 0.006 & 0.000 & 0.002 & 0.001 & 0.001 & 0.000 & 0.000 & 0.006 & 0.000 & 0.000 & 0.000 \\
 0.001 &   0.000 & 0.005 & 0.005 & 0.000 & 0.003 & 0.002 & 0.002 & 0.000 & 0.000 & 0.011 & 0.000 & 0.000 & 0.000 \\
 0.005 &   -0.001 & 0.001 & 0.014 & -0.001 & 0.009 & 0.001 & -0.003 & 0.000 & -0.003 & -0.003 & -0.001 & -0.001 & -0.001 \\
 0.01 &   -0.003 & 0.002 & 0.027 & -0.003 & 0.011 & -0.003 & -0.006 & 0.000 & -0.007 & -0.006 & -0.002 & -0.003 & -0.003 \\
 0.05 &   -0.013 & 0.003 & 0.019 & -0.013 & 0.007 & -0.036 & -0.03 & -0.003 & -0.033 & -0.029 & -0.013 & -0.013 & -0.013 \\
 0.1 &   -0.027 & 0.016 & 0.001 & -0.026 & -0.002 & -0.073 & -0.059 & -0.014 & -0.069 & -0.057 & -0.026 & -0.026 & -0.027 \\
 0.25 &   -0.067 & 0.018 & -0.047 & -0.067 & -0.035 & -0.174 & -0.148 & -0.033 & -0.17 & -0.138 & -0.066 & -0.067 & -0.067 \\

\hline
\hline \\[-1.8ex]
\end{tabular}
}
    \vspace{4mm}
    \caption{\textbf{Elasticity of the Market with Statistical Arbitrageurs with Double Elasticity Experiments.} This table is similar to table \ref{tab:mkt_elasticity_simple}, except for two differences. First, this table shows results from the counterfactual experiments where the elasticity for the aggregate 13F LLDM investors is doubled for each stock, as described in the text. Second, only the winsorized results are shown (winsorized at the 5$^{th}$ and 95$^{th}$ percentiles).}
    \label{tab:mkt_elasticity_double}
\end{table}

Note that Table \ref{tab:grs_elasmult2} shows that alpha remains in these double elasticity counterfactual experiments.

\subsection{Recursive Experiments} \label{subsec:recursive_exper}

In subsection \ref{subsec:recursive_learner_fraction} below, I describe how the arbitrageur fraction is endogenized to have both wealth effects and fund flow effects. In subsection \ref{subsec:recursive_experiments} I describe the results of these recursive experiments. 

\subsubsection{Recursive Arbitrageur Fraction} \label{subsec:recursive_learner_fraction}

In this subsection, I consider recursive experiments where the arbitrageur fraction $\theta_t$ is updated through time and the arbitrageurs re-estimate their models. 

I follow \cite{berk:green}, who have a simple fund flow relationship similar to the one used here:
\begin{equation} \label{eq:fund_flow}
    \theta_t = \Bigg( \underbrace{\udot{\phi} \frac{1}{11} \sum_{\tau=t-12}^{t-1} \udot{r}_{\tau}}_{\substack{\text{fund} \\ \text{flow term}}} + \underbrace{\udot{r}_{t-1}}_{\substack{\text{wealth} \\ \text{effect}}} \Bigg) \theta_{t-1}
\end{equation}
where $\udot{\phi}$ is an exogenous fund flow parameter based on the previous year's return excluding the previous month, and $\udot{r}_{\tau}$ are the statistical arbitrageur excess returns. The wealth effect captures the simple dynamics that if the statistical arbitrageur funds earn more, they have more money to manage. In other words, $\theta_{t+1} A_{t+1} - \theta_t A_t$ is not a typical fund flow measure, but rather a change in AUM. Thus, this wealth effect should be added as shown, and without it the assumption would be that extra returns are withdrawn by investors. The fund flow term comes from \cite{berk:green}, where flows are a function of the previous year's returns. Like \cite{berk:green}, there is no intercept in this equation. 

\begin{table}[!t] \centering
    \resizebox{0.6\textwidth}{!}{\begin{tabular}{@{\extracolsep{5pt}}lc}
\\[-1.8ex]\hline
\hline \\[-1.8ex]
& \multicolumn{1}{c}{\textit{Dependent variable: CRSP Fund Flows}} \
\cr \cline{2-2}
\\[-2ex]
\hline \\[-1.8ex]
 past year average excess returns & 0.382$^{***}$ \\
 (excluding the previous month) & (0.009) \\
\hline \\[-1.8ex]
 Observations & 5,451,860 \\
 $R^2$ & 0.002 \\
\hline
\hline \\[-1.8ex]
\textit{Note:} & \multicolumn{1}{r}{$^{*}$p$<$0.1; $^{**}$p$<$0.05; $^{***}$p$<$0.01} \\
\end{tabular}
}
    \vspace{4mm}
    \caption{\textbf{Flows.} This table shows the regression: $\text{(fund flows)}^j_t = \udot{\phi} \frac{1}{11} \sum_{\tau=t-12}^{t-1} r_{\tau}^j + e_{t}^j$, where $\text{(fund flows)}^j_t$ is the flows of fund $j$ at time $t$, $r_{t}^j$ is the excess return of fund $j$ at time $t$, $\udot{\phi}$ is the slope coefficient, and $e_t^j$ is the residual. The data used here are the CRSP monthly fund flows data. In other words, this shows the results of a regression of fund flows on average lagged returns.
    }
    \label{tab:flows}
\end{table}

Note that $\theta_{t+1}$ is not a function of $\udot{r}_{t}$, but only earlier period returns. Thus, both the fund flow term and wealth effects are in some sense lagged or delayed. This is because $\udot{r}_{t}$ is a function of prices at time $t$, meaning that $\theta_t$ would be a function of prices at time $t$ as well. This would invalidate the KY assumption that AUM are exogenous to time $t$ prices, as well as make the exercise less tractable. Frictions that slow the deployment of capital immediately, as in  \cite{duffie:afa} for example, are enough to justify this.

I use the CRSP fund flow data to estimate the parameter $\udot{\phi}$, by simply running a regression of monthly flows in the previous average excess return of the fund excluding the previous month. I estimate the regression without an intercept, since there is no intercept in equation (\ref{eq:fund_flow}). Table \ref{tab:flows} shows the results, giving a value of $\udot{\phi} = 0.38$ to plug into the counterfactual experiments. The CRSP monthly fund flows data start January 1960 and end March 2022.

In order to reduce extreme flows in the experiments, I make the maximum possible rise in $\theta_{t} - \theta_t$ to be 0.01 per month. A rise of 1\% of the entire stock market in a single month for a group of asset managers would constitute extremely large flows, and thus this is a relatively conservative restriction. 

\subsubsection{Recursive Experiment Results} \label{subsec:recursive_experiments}

In the recursive experiments, I set the initial arbitrageur fraction (IAF), denoted as $\theta_0$, to a range of values. Like in the text, I focus on the results for the last decade in the sample, since this is after the statistical arbitrageurs have been trading for two entire decades. Table \ref{tab:lf_table} shows the average arbitrageur fraction ($\theta_t$) during this period. It is notable that with endogenous flows, these fractions tend to be relatively low. As shown below, this of course strengthens the results in the text that these statistical arbitrageurs affect the demand elasticity of the market only slightly. 

\begin{table}[!t] \centering
    \resizebox{1\textwidth}{!}{\begin{tabular}{@{\extracolsep{5pt}}lccccccccccccc}

\\[-1.8ex]\hline
\hline \\[-1.8ex]
\\[-1.8ex] 
& \multicolumn{13}{c}{Average Learning Fraction $(\theta_t)$, February 2010 - January 2020} \\ \cline{2-14}
\\[-1.8ex] 
\\[-1.8ex] IAF $\theta_0$ & \multicolumn{1}{c}{BPZ$_F$} & \multicolumn{1}{c}{BPZ$_L$} & \multicolumn{1}{c}{BSV} & \multicolumn{1}{c}{CRW} & \multicolumn{1}{c}{DGU} & \multicolumn{1}{c}{FF3} & \multicolumn{1}{c}{FF6} & \multicolumn{1}{c}{GKX} & \multicolumn{1}{c}{HXZ} & \multicolumn{1}{c}{KNS} & \multicolumn{1}{c}{KPS} & \multicolumn{1}{c}{NN} & \multicolumn{1}{c}{RF}  \\
\\[-1.8ex]
\hline \\[-1.8ex]

0.0001 &   0.0001 & 0.0001 & 0.0001 & 0.0001 & 0.0004 & 0.0003 & 0.0001 & 0.0001 & 0.0017 & 0.0003 & 0.0001 & 0.0001 & 0.0001 \\
\\[-1.8ex]
0.0005 &   0.0001 & 0.0001 & 0.0001 & 0.0001 & 0.0004 & 0.0005 & 0.0001 & 0.0001 & 0.0014 & 0.0002 & 0.0005 & 0.0001 & 0.0004 \\
\\[-1.8ex]
0.001 &   0.0001 & 0.0001 & 0.0001 & 0.0001 & 0.0005 & 0.0004 & 0.0001 & 0.0001 & 0.001 & 0.0002 & 0.0009 & 0.0001 & 0.0008 \\
\\[-1.8ex]
0.005 &   0.0001 & 0.0001 & 0.0001 & 0.0001 & 0.0011 & 0.0007 & 0.0001 & 0.0001 & 0.0008 & 0.0003 & 0.0001 & 0.0001 & 0.0001 \\
\\[-1.8ex]
0.01 &   0.0001 & 0.0001 & 0.0001 & 0.0001 & 0.0023 & 0.0006 & 0.0001 & 0.0001 & 0.0006 & 0.0005 & 0.0001 & 0.0001 & 0.0001 \\
\\[-1.8ex]
0.05 &   0.0001 & 0.0001 & 0.0001 & 0.0001 & 0.0122 & 0.0006 & 0.0001 & 0.0001 & 0.0005 & 0.0005 & 0.0001 & 0.0001 & 0.0001 \\
\\[-1.8ex]
0.1 &   0.0001 & 0.0001 & 0.0001 & 0.0001 & 0.0111 & 0.0006 & 0.0002 & 0.0001 & 0.0005 & 0.0006 & 0.0001 & 0.0001 & 0.0001 \\
\\[-1.8ex]
0.25 &   0.0001 & 0.0001 & 0.0001 & 0.0001 & 0.0001 & 0.0006 & 0.0002 & 0.0001 & 0.0006 & 0.0006 & 0.0001 & 0.0001 & 0.0001 \\

\hline
\hline 
 & \multicolumn{13}{r}{\textit{Note:} $^{*}$p$<$0.1; $^{**}$p$<$0.05; $^{***}$p$<$0.01} \\

\end{tabular}
}
    \vspace{4mm}
    \caption{\textbf{Arbitrageur Fraction of Recursive Experiments.} This shows the arbitrageur fraction ($\theta_t$) that endogenously arises across statistical arbitrageur model (columns) and for different initial arbitrageur fraction (IAF) values in the recursive experiments. This shows the average arbitrageur fraction for the February 2010 - January 2020 decade.}
    \label{tab:lf_table}
\end{table}

Table \ref{tab:mkt_elasticity_recursive} shows the elasticity of the market with these endogenous arbitrageur fraction experiments. The conclusion is clear: statistical arbitrageurs have a minimal impact on the market's elasticity.

\begin{table}[!t] \centering
    \resizebox{1\textwidth}{!}{\begin{tabular}{@{\extracolsep{5pt}}lccccccccccccc}
\\[-1.8ex]\hline
\hline \\[-1.8ex]
 & \multicolumn{13}{c}{\textit{Change in Aggregate 
 Counterfactual Elasticity for Individual Stocks with Different Statistical Arbitrageurs}} \
\cr \cline{2-14}
\\[-1.8ex] IAF $\theta_0$ & \multicolumn{1}{c}{BPZ$_F$} & \multicolumn{1}{c}{BPZ$_L$} & \multicolumn{1}{c}{BSV} & \multicolumn{1}{c}{CRW} & \multicolumn{1}{c}{DGU} & \multicolumn{1}{c}{FF3} & \multicolumn{1}{c}{FF6} & \multicolumn{1}{c}{GKX} & \multicolumn{1}{c}{HXZ} & \multicolumn{1}{c}{KNS} & \multicolumn{1}{c}{KPS} & \multicolumn{1}{c}{NN} & \multicolumn{1}{c}{RF}  \\
\\[-1.8ex] & (1) & (2) & (3) & (4) & (5) & (6) & (7) & (8) & (9) & (10) & (11) & (12) & (13) \\
\hline \\[-1.8ex]

 0.0 &   0.000 & 0.000 & 0.000 & 0.000 & 0.000 & 0.000 & 0.000 & 0.000 & 0.000 & 0.000 & 0.000 & 0.000 & 0.000 \\
 0.0001 &   0.000 & 0.001 & 0.002 & 0.000 & 0.001 & 0.000 & 0.000 & 0.000 & 0.001 & 0.001 & 0.000 & 0.000 & 0.000 \\
 0.0005 &   0.000 & 0.001 & 0.001 & 0.000 & 0.002 & 0.000 & 0.000 & 0.000 & 0.000 & 0.000 & 0.000 & 0.000 & 0.000 \\
 0.001 &   0.000 & 0.001 & 0.001 & 0.000 & 0.002 & 0.000 & 0.000 & 0.000 & 0.000 & 0.000 & 0.000 & 0.000 & 0.000 \\
 0.005 &   0.000 & 0.001 & 0.001 & 0.000 & 0.003 & 0.000 & 0.000 & 0.001 & 0.000 & 0.000 & 0.000 & 0.000 & 0.000 \\
 0.01 &   0.000 & 0.001 & 0.001 & 0.000 & 0.004 & 0.000 & 0.000 & 0.000 & 0.000 & 0.000 & 0.000 & 0.000 & 0.000 \\
 0.05 &   0.000 & 0.003 & 0.003 & 0.000 & 0.006 & 0.000 & 0.000 & 0.006 & 0.000 & 0.002 & 0.000 & 0.000 & 0.000 \\
 0.1 &   0.000 & 0.005 & 0.004 & 0.000 & 0.003 & 0.000 & 0.001 & 0.004 & 0.000 & 0.004 & 0.000 & 0.000 & 0.000 \\
 0.25 &   0.000 & 0.01 & 0.009 & 0.000 & 0.000 & 0.000 & 0.001 & 0.004 & 0.000 & 0.011 & 0.000 & 0.000 & 0.000 \\

\hline
\hline \\[-1.8ex]
\end{tabular}
}
    \vspace{4mm}
    \caption{\textbf{Elasticity of the Market with Statistical Arbitrageurs with Recursive Experiments.} This table is similar to table \ref{tab:mkt_elasticity_simple}, except for three differences. First, this table shows results from the counterfactual experiments where the arbitrage fraction (AF, denoted as $\theta_t$) is recursively updated. Second, this table does not show AF, but rather shows IAF ($\theta_0$), which represents the initial arbitrageur fraction used in the counterfactual experiments. Third, only the winsorized results are shown (winsorized at the 5$^{th}$ and 95$^{th}$ percentiles).}
    \label{tab:mkt_elasticity_recursive}
\end{table}

Note that Table \ref{tab:grs_recursive} shows that alpha remains largely unchanged in these counterfactual experiments. This table shows that, if anything, the main counterfactual experiments of Table \ref{tab:grs_baseline} overstate the amount that alpha changes with the introduction of HSA investors. 



In summary, these recursive experiments strengthen the main take-away from the paper: statistical arbitrageurs exert only a minor influence on the market's elasticity.

\section{Additional Details} \label{app:details}

\subsection{Kernel Density Function} \label{subsec:kernel}

This kernel function $\K(\cdot)$ has an $N$-dimensional vector input, with $N$-dimensional vector output. As written above, $\K(x)_i$ is the $i^{th}$ output of $\K(x)$. It is defined as
\begin{equation}
    \K(x)_i = \frac{1}{N} \sum_{j=1}^N \Phi \left( \frac{x_i - x_j}{h} \right)
\end{equation}
where $\Phi$ is the standard normal cumulative distribution function (cdf) and $h$ is a scalar, called the bandwidth. I use the standard Silverman's rule of thumb bandwidth \citep{silverman}:
\begin{equation}
    h = 0.9 \min \left( \sigma, \frac{\text{IQR}}{1.34} \right) N^{-1/5}
\end{equation}
where $\sigma$ is the typical estimate of the standard deviation of the elements in $x$, and IQR is the interquartile range of $x$. If the interquartile range is zero\footnote{The dividend-price ratio variable, ldp, is so full of zeros that sometimes the interquartile range can be zero.} then that term is ignored and $\sigma$ is just used. 

Importantly, $\K(x)_i$ is guaranteed to be in the interval $[0, 1]$. Thus, $\K(x)_i - 0.5$ is obviously in $[-0.5, 0.5]$.

The derivative of this is needed to calculate elasticities. I can write:
\begin{equation}
    \frac{\partial \K(x)_i}{\partial x_i} = \underbrace{\frac{1}{N h} \sum_{j=1}^N \phi \left( \frac{x_i - x_j}{h} \right)}_{\text{direct effect}}
    \underbrace{
    - \underbrace{\frac{1}{N h} \phi \left( 0 \right)}_{\approx 0} 
    + \frac{\partial \K(x)_i}{\partial h}
    \underbrace{\frac{\partial h}{\partial x_i}}_{\approx 0}}_{\substack{\text{indirect effect of a change in }x_i\\\text{ on the distribution}}}.
\end{equation}
With many assets (large $N$), changing a single observation changes the kernel distribution---the indirect effect labeled above---very little. For example, $\partial h / \partial x_i$ is how much the estimate of a standard deviation or interquartile range changes when a single observation changes (adjusted by a scalar). This effect is very small, in fact this can easily be in the range of computer rounding error or absolutely zero with many assets. Thus, I can safely and accurately use the approximation:
\begin{equation}
    \frac{\partial \K(x)_i}{\partial x_i}
    = \frac{1}{N h} \sum_{j=1}^N \phi \left( \frac{x_i - x_j}{h} \right)
    \equiv \kk(x)_i,
\end{equation}
where $\phi(\cdot)$ is the standard normal probability density function (pdf). Note that $\kk(x)$ is the standard kernel density function defined above, with a similar $N$-dimensional input and $N$-dimensional output. Thus, this $\K(\cdot)$ is just a standard kernel cdf and the derivative is a standard kernel pdf. Since these indirect effects are very small, I can likewise use the approximation:
\begin{equation} \label{eq:cross_terms}
    \frac{\partial \K(x)_i}{\partial x_j} \approx 0 \text{ for } i \neq j
\end{equation}
In practice, approximating these indirect effects with zero makes very little difference, but does make the calculations more simple. Numerically, these effects are very small. 

There are two kinds of endogenous predictors (other than the level term): (1) predictors with a denominator price term and (2) predictors with a numerator price term. The denominator price predictors are of the form:
\begin{equation}
    x_{i,k,t} = \frac{x_{i,k,t}^n}{p_{i,t}} + x_{i,k,t}^a,
\end{equation}
where $x_{i,k,t}^n$ and $x_{i,k,t}^a$ are exogenous components. 

The numerator predictors have the form:
\begin{equation}
    x_{i,k,t} = \frac{p_{i,t}}{x_{i,k,t}^d} + x_{i,k,t}^a,
\end{equation}
where $x_{i,k,t}^d$ and $x_{i,k,t}^a$ are exogenous components. 

Recall that the normalized predictors have the form:
\begin{equation} 
    \udot{z}_{i,k,t} = \K(\text{ArcSinh} (x_{k,t}))_i - 0.5.
\end{equation}
Thus, I can calculate
\begin{equation}
    \frac{\partial \udot{z}_{i,k,t}}{\partial x_{i,k,t}}
    = \frac{\kk (\text{ArcSinh} (x_{k,t}))_i}{\sqrt{x_{i,k,t}^2 + 1}}.
\end{equation}

Furthermore, predictors with a numerator price term have the following derivative with respect to the log price term:
\begin{equation} \label{eq:num_deriv}
    \frac{\partial \udot{z}_{i,k,t}}{\partial \log (p_{i,t})}
    = \frac{\partial \udot{z}_{i,k,t}}{\partial x_{i,k,t}}
    \frac{\partial x_{i,k,t}}{\partial \log (p_{i,t})}
    = \frac{\kk (\text{ArcSinh} (x_{k,t}))_i}{\sqrt{x_{i,k,t}^2 + 1}} \left( \frac{p_{i,t}}{x_{i,k,t}^d} \right).
\end{equation}
Similarly, predictors with a denominator price have the following derivative:
\begin{equation} \label{eq:denom_deriv}
    \frac{\partial \udot{z}_{i,k,t}}{\partial \log (p_{i,t})}
    = - \frac{\kk (\text{ArcSinh} (x_{k,t}))_i}{\sqrt{x_{i,k,t}^2 + 1}} \left( \frac{x_{i,k,t}^n}{p_{i,t}} \right).
\end{equation}
It is the values in (\ref{eq:num_deriv}) and (\ref{eq:denom_deriv}) that are reported in Table \ref{tab:gradz}. 

\subsection{Endogenous Anomaly Predictors} \label{subsec:endogenous}

The 15 endogenous variables---along with the 47 exogenous variable descriptions---are described in detail in \cite{weber}. I describe the construction and components of these variables below using this notation, although I do not describe in detail the construction of every variable because it is described so well in \cite{weber}. Some of the descriptions come directly from \cite{weber}, and I follow them by putting CRSP and Compustat variables in parentheses. I follow their lagged timing convention for when balance sheet variables are used. 

I can consider the price ratio variables in terms of per-share values or across-all-share values. For example, if $x_{i,k,t}$ is the book to price ratio, then $x_{i,k,t}^n$ is the book value divided by shares outstanding and $x_{i,k,t}^a = 0$. Obviously, the per-share stock price divided by the book value divided by the shares outstanding yields the market to book ratio. 

\vspace{5mm}
\noindent{\large\textbf{Denominator Price Predictors:}}
\vspace{5mm}

\noindent\textbf{a2me}

\begin{itemize}
    \item Description: assets to market equity ratio.
    \item $x_{i,k,t}^n$: total assets (AT) divided by shares outstanding (SHROUT). 
    \item $x_{i,k,t}^a = 0$.
\end{itemize}

\noindent\textbf{beme}

\begin{itemize}
    \item Description: book equity to market equity ratio. 
    \item $x_{i,k,t}^n$: ratio of book value of equity to shares outstanding (SHROUT). See \cite{weber} for the construction of book value of equity. 
    \item $x_{i,k,t}^a = 0$.
\end{itemize}

\noindent\textbf{beme\_adj}

\begin{itemize}
    \item Description: book equity to market equity ratio minus the industry mean of this ratio.
    \item $x_{i,k,t}^n$: same as beme above. 
    \item $x_{i,k,t}^a$: negative of the average beme ratio within an asset's (Fama-French 48) industry. 
\end{itemize}

\noindent\textbf{debt2p}

\begin{itemize}
    \item Description: debt to market equity ratio.
    \item $x_{i,k,t}^n$: debt to shares outstanding ratio (SHROUT). See \cite{weber} for the construction of debt in the debt2p variable. 
    \item $x_{i,k,t}^a = 0$.
\end{itemize}

\noindent\textbf{e2p}

\begin{itemize}
    \item Description: earnings to price ratio.
    \item $x_{i,k,t}^n$: earnings (IB) divided by shares outstanding (SHROUT). 
    \item $x_{i,k,t}^a = 0$.
\end{itemize}

\noindent\textbf{ldp}

\begin{itemize}
    \item Description: dividends to price ratio. 
    \item $x_{i,k,t}^n$: Sum (across time and all shares) of dividends paid out over the last 12 months divided by shares outstanding (SHROUT).
    \item $x_{i,k,t}^a = 0$.
\end{itemize}

\noindent\textbf{nop}

\begin{itemize}
    \item Description: net payout ratio.  
    \item $x_{i,k,t}^n$: ratio of common dividends (DVC) plus purchase of common and
preferred stock (PRSTKC) minus the sale of common and preferred stock (SSTK) over
shares outstanding (SHROUT). 
    \item $x_{i,k,t}^a = 0$.
\end{itemize}

\noindent\textbf{o2p}

\begin{itemize}
    \item Description: payout ratio. 
    \item $x_{i,k,t}^n$: ratio of common dividends (DVC) plus purchase of common and
preferred stock (PRSTKC) minus the change in value of the net number of preferred
stocks outstanding (PSTKRV) over shares outstanding (SHROUT). 
    \item $x_{i,k,t}^a = 0$.
\end{itemize}

\noindent\textbf{s2p}

\begin{itemize}
    \item Description: sales to price ratio. 
    \item $x_{i,k,t}^n$: common dividends (DVC) plus purchase of common and
preferred stock (PRSTKC) minus the change in value of the net number of preferred
stocks outstanding (PSTKRV) over shares outstanding (SHROUT). 
    \item $x_{i,k,t}^a = 0$.
\end{itemize}

\vspace{5mm}
\noindent{\large\textbf{Numerator Price Predictors:}}
\vspace{5mm}

\noindent\textbf{cum\_return\_1\_0}

\begin{itemize}
    \item Description: return over the prior month.
    \item $x_{i,k,t}^d$: one month lagged price (PRC), adjusted for stock splits. 
    \item $x_{i,k,t}^a$: ratio of dividends paid out over the previous month to the one month lagged price, adjusted for stock splits. 
\end{itemize}

\noindent\textbf{lme}

\begin{itemize}
    \item Description: size---or market equity---defined as price (PRC) times shares outstanding (SHROUT). 
    \item $x_{i,k,t}^d$: one over shares outstanding (SHROUT). 
    \item $x_{i,k,t}^a = 0$.
\end{itemize}

\noindent\textbf{lme\_adj}

\begin{itemize}
    \item Description: size---same as lme above---minus the average (Fama-French 48) industry lme. 
    \item $x_{i,k,t}^d$: one over shares outstanding (SHROUT).
    \item $x_{i,k,t}^a$: Negative of average lme within an asset's (Fama-French 48) industry. 
\end{itemize}

\noindent\textbf{q}

\begin{itemize}
    \item Description: Tobin's Q---see \cite{weber} for more details. 
    \item $x_{i,k,t}^d$: total assets (AT) divided by shares outstanding (SHROUT). 
    \item $x_{i,k,t}^a$: total assets (AT) minus cash and short-term investments (CEQ), minus deferred taxes (TXDB) scaled by total assets (AT).
\end{itemize}

\noindent\textbf{rel\_to\_high\_price}

\begin{itemize}
    \item Description: the price to 52 week high stock price ratio. 
    \item $x_{i,k,t}^d$: 52 week high stock price.
    \item $x_{i,k,t}^a = 0$.
\end{itemize}

\noindent\textbf{roc}

\begin{itemize}
    \item Description: economic rents over cash ratio. 
    \item $x_{i,k,t}^d$: ratio of Short-Term Investments (CHE) to shares outstanding (SHROUT). 
    \item $x_{i,k,t}^a = 0$: ratio of long-term debt (DLTT) minus total assets to Cash and Short-Term Investments (CHE). 
\end{itemize}

\subsection{Winsorization} \label{subsec:winsorization}

Consider some positive weights, $v_{i,t}$, in a weighted average. Assume, as usual, that these weights sum to one. Then there are two equivalent ways of applying this three-step procedure to elasticity values $\eta_{i,t}$:
\begin{enumerate}[label=\Alph*.]
    \item The first way is the following
    \begin{enumerate}[label=\arabic*.]
        \item Multiply the elasticity by $v_{i,t}$, to get $v_{i,t} \eta_{i,t}$. The sum of $v_{i,t} \eta_{i,t}$ in the cross-section is the weighted average. 
        \item Winsorize the reweighted sample of $v_{i,t} \eta_{i,t}$ terms. 
        \item Sum the winsorized terms to get a weighted and winsorized average. 
    \end{enumerate}
    \item The second equivalent way is the following:
    \begin{enumerate}[label=\arabic*.]
        \item Multiply the elasticity by $N_t v_{i,t}$, to get $N_t v_{i,t} \eta_{i,t}$, where $N_t$ is the number of stocks in period $t$. The average of $N_t v_{i,t} \eta_{i,t}$ in the cross-section is the weighted average. 
        \item Winsorize the reweighted sample of $N_t v_{i,t} \eta_{i,t}$ terms. 
        \item Average the winsorized terms to get a weighted and winsorized average. 
    \end{enumerate}
\end{enumerate}
The fourth and final step just consists of averaging these terms across time, to get a single reweighted and winsorized average. I reweight before winsorization simply because the portfolio-reweighted elasticities tend to already down-weight the extreme elasticities in the tails. See Table \ref{tab:statistical arbitrageur_elasticity} to compare winsorization and weighting schemes across the main results of the paper. 

\subsection{Derivation of Demand Function} \label{subsec:demand_function}

The derivation of the demand function follows trivially from \cite{ky} (KY), but I describe it here in detail for the sake of completeness. Let $\mu_t = \E_t [r_{t+1}]$ be the $N$-dimensional vector of asset expected excess returns, and let $\Sigma_t = \Var_t [r_{t+1}]$ be the $N \times N$ conditional covariance matrix. Let $\mu_t^j$ and $\Sigma_t^j$ be similar, but instead let these terms represent the {subjective} beliefs of institution $j$ for the mean and covariance respectively. Like KY, I assume investors believe:
\begin{equation}
    \mu_t^j = Z_t \Phi_t^j, \;\;\; \Sigma_t^j = \Gamma_t^j (\Gamma_t^j)' + \zeta_{t}^j I, \;\;\; \text{ and } 
    \Gamma_t^j = Z_t \Psi_t^j,
\end{equation}
where $\Phi_t^j$ and $\Psi_t^j$ are $K \times 1$ vectors of parameters, $\zeta_{t}^j$ is a positive scalar, $I$ is an $N \times N$ identity matrix, and here $Z_t$ is an $N \times (K+1)$ matrix of $K$ predictors $z_{i,k,t}$ and the first column of $Z_t$ is a constant (intercept term). This equation is essentially the same as equation (\ref{eq:nn_basic}). See KY for a discussion about this standard single factor structure. KY mention, "[w]e could relax the one-factor assumption and generalize to a multifactor case, but the resulting expressions are less intuitive and less preferable for expositional purposes."  

Classic CARA demand with multivariate normal expected returns has the following functional form:
\begin{equation}
    w_t^j = \frac{1}{\gamma_t^j} \left( \Sigma_t^j \right)^{-1} \mu_t^j,
\end{equation}
where $\gamma_t^j$ is the risk aversion of the institution at time $t$. 

With these assumptions above---using similar logic to KY---I can write:
\begin{equation} \label{eq:demand_matrix}
    w_t^j = Z_t \hat \beta_t^j
\end{equation}
where the $K \times 1$ vector of parameters $\beta_t^j$ has the form:
\begin{equation}
    \beta_t^j = 
    \frac{1}{\gamma_t^j \zeta_t^j} 
    \left( \Phi_t^j - \kappa_{t}^j \Psi_t^j \right)
\end{equation}
where $\kappa_t^j$ is a scalar (common across assets):
\begin{equation}
    \kappa_t^j = \frac{(\Gamma_t^j)' \mu_t^j}{(\Gamma_t^j)' \Gamma_t^j + \zeta_t^j}
\end{equation}

The proof of this---which is just a trivial application of the Woodbury matrix identity---is so similar to the proof of KY's proposition 1 that it is excluded here. 

Importantly, KY shows that this linear demand function is still delivered {when an institution has short selling constraints on some assets}. See KY for more details. 

Equation (\ref{eq:demand_matrix}) can be written, for an individual asset $i$, as:
\begin{align} 
    w_{i,t}^j &= \hat \beta_{0,t}^j + \hat \beta_{1,t}^j \breve{\udot{z}}_{i,1,t} + \sum_{k=2}^K \hat \beta_{k,t}^j  z_{i,k,t} + \epsilon_{i,t}^j \nonumber\\
    &= \hat \beta_{0,t}^j + \hat \beta_{1,t}^j \left( \frac{P_{i,t}}{A_t^j} \right) + \sum_{k=2}^K \hat \beta_{k,t}^j \left( a_{i,1,k,t} + a_{2,k,t} \log (P_{i,t}) \right) + \epsilon_{i,t}^j,
\end{align}
where these $\hat \beta_{k,t}^j$ are the elements of $\hat \beta_t^j$. 

What are the differences between this and KY? All differences are trivial, but I list them here. First, I derive standard mean-variance demand instead of log utility demand. Since an SDF/factor model is equivalent to mean-variance demand, this result is useful (and easily obtainable). Second, KY use log returns, and I use simple excess returns.\footnote{This means their $\mu$ and $\Sigma$ are slightly different to deal with log returns. I just use standard excess returns here.} Third, I ignore short-selling constraints in the math here, but this can be trivially added back in. 

\subsection{Trading Cost Optimization} \label{app:trading_costs}

I follow the \cite{tradingcosts} method for trading costs, where $w_t^*$ is the portfolio weight after optimizing for trading costs. In \cite{tradingcosts}, this vector is:
\begin{equation}
    w_t^* = \left( 1 - \frac{a}{\lambda} \right) w_{t \leftarrow t-1} + \frac{a}{\lambda} \text{aim}_t
\end{equation}
where $w_{t \leftarrow t-1}$ is the $N$ dimension vector of weights from period $t-1$ that are changed only passively due to the returns experienced at time $t$ and
\begin{equation}
    \text{aim}_t = \sum_{\tau=0}^{\infty} \left(1 - \frac{a}{\gamma + a} \right) \left( \frac{a}{\gamma + a} \right)^{\tau} \E_t [w_{t+\tau}],
\end{equation}
the variable $\gamma$ is the mean-variance utility risk aversion coefficient, the variable $\lambda$ controls the size of transaction costs (assumed proportional to the covariance matrix), the variable $a$ is defined as
\begin{equation}
a = \frac{{-(\gamma(1 - \rho) + \lambda\rho) + \sqrt{(\gamma(1 - \rho) + \lambda\rho)^2 + 4\gamma\lambda(1 - \rho)^2}}}{{2(1 - \rho)}},
\end{equation}
and $\rho$ is a discount rate. If, in an extreme case of $a / \lambda = 0$, the portfolio optimizer never actively updates the portfolio weights and the portfolio is entirely passive after initial weights are chosen. If $a / \lambda = 1$, then the optimizer is purely forward-looking and entirely invests in the "aim" portfolio, never taking into account its existing portfolio.

Notice that:
\begin{equation} \label{eq:exact_aim}
    \sum_{\tau=0}^{\infty} \left(1 - \frac{a}{\gamma + a} \right) \left( \frac{a}{\gamma + a} \right)^{\tau} = 1
\end{equation}
Thus the $\text{aim}_t$ portfolio is a weighted average of current and future optimal but not cost optimized portfolios. Thus, this can be approximated as:
\begin{equation}
    \text{aim}_t \approx \left( \sum_{\tau=0}^{\mathcal{T} - 1} \left(1 - \frac{a}{\gamma + a} \right) \left(\frac{a}{a + \gamma}\right)^{\tau} \E_t [w_{t+\tau}] \right)
    + \left(\frac{a}{a + \gamma}\right)^{\mathcal{T}} \E_t [w_{t + \mathcal{T}}].
\end{equation}
Note that the weights similarly sum to one,\footnote{To see the weights sum to one, I can calculate:
\begin{equation*}
    \left( \sum_{\tau=0}^{\mathcal{T} - 1} \left(1 - \frac{a}{\gamma + a} \right) \left(\frac{a}{a + \gamma}\right)^{\tau} \right)
    + \left(\frac{a}{a + \gamma}\right)^{\mathcal{T}} = 
    1 - \left(\frac{a}{a + \gamma }\right)^{\mathcal{T}} + \left(\frac{a}{a + \gamma }\right)^{\mathcal{T}} 
    = 1.
\end{equation*}
} and by assuming standard convergence assumptions, this approximation can be arbitrarily close to the exact form in (\ref{eq:exact_aim}) by choosing $\mathcal{T}$ large enough. Thus, I can write:
\begin{equation}
    s_{\text{aim}} = \frac{a}{\lambda} \; \text{ and } \;
    \rho_{\text{aim}} = \frac{a}{a + \gamma},
\end{equation}
where $s_{\text{aim}}$ and $\rho_{\text{aim}}$ are shown in equations (\ref{eq:trading_cost_weights}) and (\ref{eq:aim_weights}). 

This is a relatively parsimonious cost optimizer, where only the parameters $\lambda$, $\gamma$, and $\rho$ are needed. \cite{tradingcosts} use $\rho = 1 - \exp (-0.02 / 260)$ for daily data (corresponding to a 2\% annualized rate), and I use $\rho = 1 - \exp(-0.02 / 12)$ for monthly data. Following \cite{tradingcosts}, I use a "conservative" transaction cost estimate of $\lambda = 10^{-6}$.\footnote{\cite{tradingcosts} consider commodity futures, and I consider equities markets. However, as discussed below, this provides about a 97\% weight on the passive portfolio. In other words, these parameters provide an investment strategy that trades little, which ultimately shows how much the elasticity can be reduced with a very trade-reluctant cost optimizer.} Lastly, I follow \cite{tradingcosts} by using $\gamma = 10^{-9}$. 

These parameters give a value of $a = 3.032 \times 10^{-8}$. This implies $s_{\text{aim}} = a / \lambda \approx 0.0303$ and $\rho_{\text{aim}} = a / (a + \gamma) \approx 0.9681$. Note that these are the only two parameters that are actually needed for the optimizer: (1) the weight on the aim portfolio ($s_{\text{aim}}$) and (2) the rate of decay on the weight of future expected portfolios ($\rho_{\text{aim}}$).

There are important details about the evolution of $w_{t \leftarrow t-1}$ that need to be described. Let $w_{i, t \leftarrow t-1}$ be the $i^{th}$ element of $w_{t \leftarrow t-1}$ corresponding to asset $i$. It should be the case that for positive weights:
\begin{equation*}
    \frac{\partial \log (w_{i, t \leftarrow t-1})}{\partial \log (p_{i,t})} = 1,
\end{equation*}
which means that when prices rise 1\%, the passive portfolio weights should also rise 1\%. This implies that if the weight on the aim portfolio, $a / \lambda$, is zero, then the elasticity of demand is zero for this investor. This comes directly from the definition of passive weights. The standard way to do this is simply define:
\begin{equation}
    w_{i,t \leftarrow t-1} \equiv c_t w_{i,t-1} \frac{p_{i,t}}{p_{i,t-1}},
\end{equation}
where $c_t$ is a scaling constant. This means simply that these weights float with capital gains in an entirely automatic and passive way. Dividends are re-invested in the entire portfolio (not just in the assets corresponding to the dividend issuance), and delisted stock returns are considered dividends. The $w_{i,t \leftarrow t-1}$ weights are initialized with the aim portfolio weights, and new stocks in the sample are assigned the aim portfolio weights rather than a zero weight. The constant $c_t$ is chosen to target the same volatility as the market, as described in the paper. 

\subsection{Analogue of Value-Weighted Portfolios} \label{subsec:value_weights}

Define 
\begin{equation}
    h_{\Xi} (x) = \frac{\tan^{-1} (\Xi x)}{\tan^{-1} (\Xi / 2)}.
\end{equation}
Let an analogue of value-weighted predictors be:
\begin{equation}
    \tilde z_{i,k,t}^{\Xi} = \frac{\invbreve P_{i,t} h_{\Xi} (\udot{z}_{i,k,t})}{\sum_j | P_{j,t} h_{\Xi} (\udot{z}_{j,k,t}) |} 
    = \left( \frac{N}{4} \right) \frac{P_{i,t} h_{\Xi} (\udot{z}_{i,k,t})}{\sum_j | P_{j,t} h_{\Xi} (\udot{z}_{j,k,t}) |}
    = \left( \frac{N}{4} \right) \frac{P_{i,t} \tan^{-1} (\Xi \udot{z}_{i,k,t})}{\sum_j | P_{j,t} \tan^{-1} (\Xi \udot{z}_{j,k,t}) |}
\end{equation}
Then note that
\begin{equation}
    \lim_{\Xi \to \infty} \tilde z_{i,k,t}^{\Xi} = \left( \frac{N}{4} \right) \frac{P_{i,t} \text{sign} (\udot{z}_{i,k,t})}{\sum_j | P_{j,t} \text{sign} ( \udot{z}_{j,k,t}) |}
\end{equation}

The derivative, using equation (\ref{eq:cross_terms}), can be easily calculated:
\begin{equation}
\frac{\partial \tilde z_{i,k,t}^{\Xi}}{\partial \log (p_{i,t})}
= \underbrace{\left( 1 -
\frac{4 }{N}
\left| \tilde z_{i,k,t}^{\Xi} \right|
\right)}_{\substack{\text{Normalization} \\ \text{Component}}}
\left(
\underbrace{\tilde z_{i,k,t}^{\Xi}}_{\substack{\text{Value-Weight} \\ \text{Component}}}
+ \underbrace{\left( \frac{N}{4} \right)
\frac{\Xi P_{i,t} \left(1 + ( \Xi \udot{z}_{i,k,t})^2 \right)^{-1}}{\sum_j | P_{j,t} \tan^{-1} (\Xi \udot{z}_{j,k,t}) |} \frac{\partial \udot{z}_{i,k,t}}{\partial \log (p_{i,t})}}_{\substack{\text{Price Sensitivity of the Predictor} \\ \text{Component}}} \right)
\end{equation}

There are three components to this derivative labeled above that are discussed below. 

\begin{enumerate}
\item Normalization Component: This component accounts for the price in the denominator. It acts as a slight correction factor to ensure that the derivative remains normalized. As the number of assets increases, this component tends to 1, making it a minor correction factor in portfolios with many assets. 

\item Value-Weight Component: This component represents the weight of the asset in the portfolio, akin to a value-weighted portfolio. It signifies the proportion of the asset in the portfolio, and its presence in the derivative equation indicates that the change in the portfolio weight is roughly proportional to the change in the asset's value (depending on the importance of the other two components).

\item Price Sensitivity of the Predictor Component: This component captures the sensitivity of the derivative to changes in the predictor with respect to the price. The presence of \( \Xi \) indicates that as \( \Xi \) becomes large, this component converges to zero, making the derivative less sensitive to changes in \( \udot{z}_{i,k,t} \). For example, if $\Xi$ is one billion, this is essentially a traditional discrete-bucket long-short portfolio (as opposed to continuous weighting), and small changes in the predictor do not change the discrete bucket to which the asset belongs, which essentially turns off this component in the derivative (i.e., set it equal to zero). 
\end{enumerate}

For portfolios with many assets and a large value of \( \Xi \), the equation simplifies to:
\begin{equation}
\frac{\partial \tilde z_{i,k,t}^{\Xi}}{\partial \log (p_{i,t})} \approx \tilde z_{i,k,t}^{\Xi}
\end{equation}
Furthermore, if \( \tilde z_{i,k,t}^{\Xi} > 0 \), then:
\begin{equation}
\frac{\partial \log (\tilde z_{i,k,t}^{\Xi})}{\partial \log (p_{i,t})} \approx 1
\end{equation}
This approximation shows that with many assets and large $\Xi$, a 1\% change in the price of an asset translates directly into a 1\% change in the portfolio weights (as long as the price change does not change the portfolio to which the asset belongs). In other words, the value-weight component dominates, and small price changes do not require rebalancing. 

\subsection{Demand Function Predictors} \label{subsec:demand_function_predictors}

For the demand function estimation, it is useful to separate $\udot{z}_{i,k,t}$ linearly into two components:
\begin{equation*}
    \udot{z}_{i,k,t} \approx  z_{i,k,t} \equiv a_{i,1,k,t} + a_{2,k,t} \log (P_{i,t})
\end{equation*}
where $a_{i,k,t}$ is an exogenous component of the predictor and $a_{k,t}$ is constant across assets in a given period. A first order approximation in a given period would make the intercept term $a_{i,k,t}$ clearly an endogenous function of the asset's own price. It is useful to use these demand function predictors, $ z_{i,k,t}$, because this allows us to separately instrument for the endogenous component using standard linear regression models. The results of the paper, including the counterfactual experiment results below, can be shown with similar results with these predictors as well. At times these can result with predictors that are outside of -0.5 and 0.5, potentially causing concerns that the models are not using properly normalized predictors and are thus disadvantaged. Thus, I use different predictors for demand estimation than for the statistical arbitrage models. 

For $a_{2,k,t}$, I simply set $a_{2,k,t}$ equal to the median value of $\partial \udot{z}_{i,k,t} / \partial \log (P_{i,t})$ for each month and each predictor $k$. In Appendix \ref{subsec:kernel}, I show the closed-form solution for these derivatives. It should be obvious that
\begin{equation}
    \frac{\partial \udot{z}_{i,k,t}}{\partial \log (P_{i,t})} = \frac{\partial \udot{z}_{i,k,t}}{\partial \log (p_{i,t})}. 
\end{equation}

I simply use a regression to estimate the other two parameters. In particular, in every month for every predictor, I run the following cross-sectional regression:
\begin{equation}
    \udot{z}_{i,k,t} - 
    a_{2,k,t} \log(P_{i,t}) = \underbrace{\hat a_{0,k,t} + \hat a_{1,k,t} \bar z_{i,k,t}}_{a_{i,1,k,t}} + \nu_{i,k,t},
\end{equation}
where $\hat a_{0,k,t}$ and $\hat a_{1,k,t}$ are OLS regression estimates, and $\bar z_{i,k,t}$ is described as follows. I calculate $\bar z_{i,k,t} \equiv \text{ArcSinh} (\bar x_{i,k,t})$ where $\bar x_{i,k,t}$ is the very same as the raw predictor $x_{i,k,t}$, except instead of using the market equity of the asset at time $t$, the median market equity across stocks that month is used. It can be thought of as the predictor with the market equity stripped out with only the exogenous component remaining. Following KY, where aggregate parameters across assets are exogenous,\footnote{See equation (9) of KY, which has aggregate values that are assumed exogenous. This essentially corresponds to a model where the market-wide parameters are endogenous, and cross-sectional differences are the endogenized components of the model. This paper of course similarly investigates cross-sectional pricing as well. In this sense, it is a partial equilibrium model, which is standard in cross-sectional theory and empirical papers.} this essentially strips out the endogenous price component and leaves $\bar z_{i,k,t}$ to be the exogenous predictors component of $\udot{z}_{i,k,t}$. This calculation is only done for the 15 endogenous predictors. 

These regression estimates are then plugged to give the log-linearized predictor $ z_{i,k,t}$. Importantly, this procedure is only done for the 15 predictors that are a function of price. For the other 47 predictors, excluding the market weights predictor ($k = 1$), I simply set $ z_{i,k,t} = \udot{z}_{i,k,t} = a_{i,1,k,t} + a_{2,k,t} \log(P_{i,t})$, where trivially $a_{i,1,k,t} = \udot{z}_{i,k,t}$ and $a_{2,k,t} = 0$. For the market weight predictor, I just define $ z_{i,1,t} \equiv \udot{z}_{i,1,t} = \invbreve P_{i,t} / A_t$. 

Why is $\hat a_{2,k,t}$ not estimated with the regression (i.e., put the on right-hand side of the regression)? This is for two reasons. First, the very point of calculating $\udot{z}_{i,k,t}$ with the continuous function above is so that the marginal impact of a change in prices on $\udot{z}_{i,k,t}$ is known. Thus, I do not need to do this again---the marginal effect is known. Second, this would create an endogenous regressor problem, since prices are endogenous. 

\subsection{Demand Estimation Details} \label{subsec:demand_estimation_details}

As KY discuss, there is some cause for concern that correlated demand shocks across investors will cause price impacts, making $\epsilon{_{i,t}^j}$ endogenous. I follow their approach and use their same instrument for price. In order to calculate this instrument, I follow KY and calculate the investment universe in the same way: all stocks ever held in the previous 11 quarters.\footnote{See KY for more details.} Then the instrument for the price of institution $j$ is:
\begin{equation*}
    \xi_{i,t}^j = \log \left( \sum_{m \neq j} A_t^m \frac{\one_{i,t}^m}{\sum_{n=1}^N \one_{n,t}^m} \right),
\end{equation*}
where $\xi{_{i,t}^j}$ is what the log market equity of asset $i$ would be if all other funds with asset $i$ in their investment universe held equal-weighted portfolios.\footnote{I can also define:
\begin{equation*}
    \xi_{i,t}^j = \log \left( \sum_{m \neq j} A_t^m \frac{\one_{i,t}^m}{1+\sum_{n=1}^N \one_{n,t}^m} \right),
\end{equation*}
as KY do, but I get very similar results. The extra one accounts for the "outside asset" in their setting, they need to include because of their different function form. This simpler functional form does not require this. 
}
Following KY, I exclude households from the calculation of the instrument for the institutions in the data. 

Thus, the first stage regression, which is estimated with standard OLS for every institution in every quarter, is
\begin{equation*}
    \log (P_{i,t}) 
    = \hat b_{0,t}^j + \sum_{k \in \mathcal{D}} \hat b_{k,t}^j \frac{B_{i,t}^c}{C_t} a_{i,1,k,t} + \sum_{k \in \mathcal{X}} \hat b_{k,t}^j z_{i,k,t} + \upsilon_{i,t}^j,
\end{equation*}
where the $\hat b$ terms are regression coefficients and $\upsilon{_{i,t}^j}$ is the error term. The predicted value $\widehat{\log (P_{i,t})}$ is used in the second stage regression. 

KY have an exponential linear demand function estimated using GMM. Given this simpler functional form than KY, I can also estimate the first and second stage regression with OLS, which gives a two-stage least squares (2SLS) estimate of the betas. The second stage regression is given by equation (\ref{eq:second_stage}) above, except that $\log (P_{i,t})$ is replaced by ${\widehat{\log (P_{i,t})}}$ and $P_{i,t}$ is replaced by ${\exp (\widehat{\log (P_{i,t})})}$. Note that the regression can also be estimated such that there are two first stage regressions to instrument $P_{i,t}$ and $\log (P_{i,t})$ with the instruments $\exp(\xi{_{i,t}^j})$ and $\xi{_{i,t}^j}$, but this yields similar results, is slightly more complicated, and deviates a bit more from KY.

Second, like KY, who constrain the coefficients on the log of prices to have well-behaved demand functions and equilibrium prices, I have similar constraints. In particular, I constrain $\beta{_{1,t}^j} < 1$ and $\beta{_{2,t}^j} < 0$. This ensures that $\beta_{i,1,t} < 1$ and $\beta_{i,2,t} < 0$ for all assets, which when combined with a similar restriction on statistical arbitrageur demand is sufficient to ensure that a unique closed-form positive equilibrium price exists in the counterfactual experiments and investors have downward sloping demand. Like KY, this can be estimated with generalized method of moments (GMM). However, since this is just a standard linear regression with constraints, this is the same as just estimating a constrained OLS regression for the second stage regression, which is what I estimate.

Similar to KY, funds are grouped together if they have too few strictly positive holdings and the aggregate demand across the funds is estimated. KY group firms together if there are less than 1,000 observations. My estimation procedure does not have the same convergence issues, so I group funds together if they have less than 500 strictly positive holdings. I group firms together with similar predictors, using the same procedure as KY.

The 13F data only has strictly positive holdings. I estimate the model with only these holdings for each institution, which means that the sample is truncated. In other words, the short and zero asset positions that a fund has are not in the sample. Data truncation is not a problem unless there are sample selection issues between the truncated and non-truncated sample. However, even if there are sample selection issues, the short and zero asset positions are grouped into "household" demand. Like in KY, this household demand captures short positions. However, the model aggregates nicely, and it is fine to estimate a demand function of investors with very different beliefs and preferences as long as demand can be captured with this more flexible functional form.\footnote{Note that there are some assets that were held by a fund in the previous three years, but are no longer held by the fund. These assets in the investment fund could be included as zeros in the institution-level regressions. This would require an assumption that the true asset position is actually zero or negative and that the investment universe is measured correctly for these assets. While the investment universe is not measured perfectly, as KY discuss, it is more problematic for these assets than for the instrument calculation. The instrument calculation should average some of the error out. Including these extra zeros in the regression also requires a censored data approach, such as a Tobit regression. A second stage estimated with a Tobit model or Censored Least Absolute Deviations (CLAD) model from \cite{clad} yield similar results. I present the simpler version here, where these zero asset positions are not added to the regression. }

In conclusion, I highlight the major two ways this demand estimation differs from KY. First, I have a simple and flexible linear-in-predictors functional form for the portfolio weights, rather than exponential-linear. Second, I use a broader set of predictors from the \cite{weber} data.

\subsection{Statistical Arbitrageur Log Price Derivatives} \label{app:stat_arb_derivs}

In order to calculate (\ref{eq:statistical arbitrageur_elasticity}), I have to calculate 
\begin{equation}
    \nabla_{{z}_{i,k,t}} (f(Z_t)_i b)
    \equiv \frac{\partial f(Z_t)_i b}{\partial {z}_{i,k,t}}. 
\end{equation}

In words, it is important to know how statistical arbitrageur-implied demand changes when predictors change. This is because the endogenous predictors change when prices change, and it is critical to know how these values in order to calculate statistical arbitrageur elasticity. 

First I give some general details about these calculations, and then give the formulas for each of the models. 

\subsubsection{General Derivative Calculation Details} \label{subsec:derivative_details}

Let $Y$ be an $N \times J$ matrix with elements $Y_{i,j}$ and let $X$ be an $N \times 1$ vector with elements $X_i$. Assume that each element of $Y_{i,j}$ is a differentiable function of $X_i$. Define the following:
\begin{equation}
    \nabla_{X} Y = \begin{bmatrix}
        \frac{\partial Y_{1,1}}{\partial X_{1}} & 
        \frac{\partial Y_{1,2}}{\partial X_{1}} & 
        \hdots &
        \frac{\partial Y_{1,J}}{\partial X_{1}} \\
        \frac{\partial Y_{2,1}}{\partial X_{2}} & 
        \frac{\partial Y_{2,2}}{\partial X_{2}} & 
        \hdots &
        \frac{\partial Y_{2,J}}{\partial X_{2}} \\
        \vdots & \vdots & \ddots & \vdots \\
        \frac{\partial Y_{N,1}}{\partial X_{N}} & 
        \frac{\partial Y_{N,2}}{\partial X_{N}} & 
        \hdots &
        \frac{\partial Y_{N,J}}{\partial X_{N}} \\
    \end{bmatrix}.
\end{equation}

I will assume throughout that
\begin{equation}
    \frac{\partial Y_{i,j}}{\partial X_l} = 0 \text{ for all } l \neq i.
\end{equation}

Importantly, this is not the standard Jacobian matrix. I give some identities below that are useful for calculating these derivatives with the models. These identities differ from standard Jacobian identities because this is a separate concept (although obviously it is related and borrows similar notation). 

Using this matrix notation for example, I can rewrite the elasticity of statistical arbitrageurs (equation (\ref{eq:statistical arbitrageur_elasticity})) as the $N$-dimensional vector, defined over only assets with positive weights:
\begin{equation}
    \eta_t = \iota - \diag (\udot{w}_t)^{-1} \nabla_{\log (p_t)} \udot{w}_t,
\end{equation}
where $\iota$ is a vector of ones, $\udot{w}_t$ and $p_t$ the vectorized versions of $\udot{w}_{i,t}$ and $p_{i,t}$, and $\diag(\cdot)$ is the diagonal matrix with the input along its diagonal. Recall that $\udot{w}_t = f(Z_t) b$. 

Let $z_{k,t}$ denote the $k^{th}$ column of $Z_t$. Then equation (\ref{eq:statistical arbitrageur_elasticity}) can be written in vector form for assets with positive weights as
\begin{align}
    \eta_t &= \iota - \diag (\udot{w}_t)^{-1} \sum_{k=1}^K \diag( \nabla_{z_{k,t}} (f(Z_t) b)) \nabla_{\log (p_t)} z_{k,t} \\
    &= \iota - \diag (\udot{w}_t)^{-1} \sum_{k=1}^K (\nabla_{z_{k,t}} (f(Z_t) b)) \circ (\nabla_{\log (p_t)} z_{k,t}),
\end{align}
where $\circ$ denotes the standard Hadamard (element-wise) product. This is a convenient way of writing this because Appendix \ref{subsec:endogenous} above already gives the formulas used to calculate the elements of $\nabla_{\log (p_t)} z_{k,t}$, and Appendix \ref{app:stat_arbs} gives the formulas for the calculation of $\udot{w}_t = f(Z_t) b$. Thus, all that is left is the formulas for $\nabla_{z_{k,t}} (f(Z_t) b)$ for the different statistical arbitrageurs. 

Here are some simple identities with this matrix derivative that are useful for the calculations. For all of these identities below, I define $Y$ and $X$ as above. Let $A$, $B$, $C$, $D$, and $E$ be constant matrices (not a function of the elements of $X$). 

\begin{prop}
\begin{equation}
    \nabla_X (Y B) = (\nabla_X Y) B.
\end{equation}
\end{prop}

This result is trivial, but used in these calculations below.

\begin{prop} \label{prop:yp}
Let $A$ be $N \times J$ dimensional, and $B$ be $N \times 1$ dimensional. Then:
\begin{equation}
    \nabla_X (A Y' B) = \Diag (A (\nabla_X Y)') \circ B,
\end{equation}
where $\Diag(\cdot)$ denotes the function that takes the diagonal of its matrix input and makes it a column vector (the off diagonal terms are discarded). 
\end{prop}

\begin{proof}
I can write:
\begin{equation}
    (A Y' B)_{i} = \sum_l \sum_m A_{i,l} Y'_{l,m} B_{m} = \sum_l \sum_m A_{i,l} Y_{m,l} B_{m}.
\end{equation}
Then I can write:
\begin{equation*}
    \frac{\partial (A Y' B)_{i}}{\partial X_i} 
    = \frac{\partial }{\partial X_i} \sum_l \sum_m A_{i,l} Y_{m,l} B_{m}
\end{equation*}
\begin{equation}
    = \sum_l A_{i,l} \frac{\partial Y_{i,l}}{\partial X_i} B_{i} = B_{i} \sum_l A_{i,l} \frac{\partial Y_{i,l}}{\partial X_i}.
\end{equation}
and writing this back as matrices implies that
\begin{equation}
    \nabla_X (A Y' B) = \Diag (A (\nabla_X Y)') \circ B.
\end{equation}
\end{proof}

\begin{prop} \label{prop:y}
Let $A$ be $N \times N$ and $B$ be $J \times 1$. Then
\begin{equation}
    \nabla_X (A Y B) = \Diag(A) \circ ((\nabla_X Y) B).
\end{equation}
\end{prop}

\begin{proof}
Again write:
\begin{equation}
    (A Y B)_{i} = \sum_l \sum_m A_{i,l} Y_{l,m} B_{m}
\end{equation}
Then I can write:
\begin{equation*}
    \frac{\partial (A Y B)_{i}}{\partial X_i} 
    = \frac{\partial }{\partial X_i} \sum_l \sum_m A_{i,l} Y_{l,m} B_{m}
\end{equation*}
\begin{equation}
    = \sum_m A_{i,i} \frac{\partial Y_{i,m}}{\partial X_i} B_m = A_{i,i} \sum_m \frac{\partial Y_{i,m}}{\partial X_i} B_m.
\end{equation}
Writing this back in matrix form gives:
\begin{equation}
    \nabla_X (A Y B) = \Diag(A) \circ ((\nabla_X Y) B).
\end{equation}
\end{proof}

It is a well-known result that if $g(x)$ is a differentiable matrix-valued function of the scalar $x$, then
\begin{equation}
    \frac{\partial g(x)^{-1}}{\partial x} = - g(x)^{-1} \frac{\partial g(x)}{\partial x} g(x)^{-1}.
\end{equation}
I use this result in the calculations below.

\begin{prop}
Assume the first dimension of $A$ is $N$, $E$ is a column vector, $B + CY'D$ is a square invertible matrix, and the rest of the dimensions of the matrices are such that the multiplication is well-defined. Then:
\begin{equation}
    \nabla_X (A (B + C Y' D)^{-1} E) 
    = - \Diag( A (B + C Y' D)^{-1} C (\nabla_X Y)') \circ
    (D (B + C Y' D)^{-1} E).
\end{equation}
\end{prop}

\begin{proof}
\begin{equation}
    \frac{\partial}{\partial X_i} (A (B + C Y' D)^{-1} E)
    = - A (B + C Y' D)^{-1} \left( \frac{\partial}{\partial X_i} C Y' D \right) (B + C Y' D)^{-1} E.
\end{equation}
Using this combined with Proposition \ref{prop:yp} gives the stated result. 
\end{proof}

\begin{prop}
Assume the first dimension of $A$ is $N$, $E$ is a column vector, $B + CYD$ is a square invertible matrix, and the rest of the dimensions of the matrices are such that the multiplication is well-defined. Then
\begin{equation}
    \nabla_X (A (B + C Y D)^{-1} E) 
    = - \Diag( A (B + C Y D)^{-1} C) \circ
    ((\nabla_X Y) D (B + C Y D)^{-1} E).
\end{equation}
\end{prop}

\begin{proof}
\begin{equation}
    \frac{\partial}{\partial X_i} (A (B + C Y D)^{-1} E)
    = - A (B + C Y D)^{-1} \left( \frac{\partial}{\partial X_i} C Y D \right) (B + C Y D)^{-1} E.
\end{equation}
Using this combined with Proposition \ref{prop:y} gives the stated result. 
\end{proof}

All the statistical arbitrageur derivative calculations below simply use these five propositions combined with the product rule. 

\subsubsection{Linear statistical arbitrageur Derivatives}

I use the notation and propositions from Appendix \ref{subsec:derivative_details} in this section of the Appendix. 

The linear statistical arbitrageurs---BPZ$_L$, BSV, DGU, FF3, FF6, HXZ, and KNS models---have trivial derivative calculations. For all models (not just these linear models), I can calculate:
\begin{equation}
    \grad w_t = (\grad f(Z_t)) b.
\end{equation}

However, for these models $\grad f(Z_t)$ is also trivial to calculate. The KKN model derivative is similar with this very simple linear structure. 

For example, for the BPZ$_L$, BSV, DGU, and the KNS models, $f(Z_t) = \breve Z_t$. Thus, for these models
\begin{equation}
    \grad f(Z_t) = \grad \breve Z_t = \left( \frac{4}{N} \right) \grad Z_t,
\end{equation}
where $\grad Z_t$ is simply an $N \times K$ matrix of zeros, except the $k^{th}$ column is filled with ones. 

The other linear models---FF3, FF6, and HXZ---are identical except that the derivative $\grad f(Z_t)$ equals a matrix of zeros if predictor $k$ is not included in the model at all. 

\subsubsection{CRW Model Derivatives}

I use the notation and propositions from Appendix \ref{subsec:derivative_details} in this section of the Appendix. 

In \cite{CRW}, the MVE portfolio is:
\begin{equation}
    w_{t} = f(Z_t) b.
\end{equation}
\begin{equation}
    f(Z_t) = Z_t \left( Z_t' Z_t \right)^{-1} \Gamma_\beta 
\end{equation}

I can calculate:
\begin{equation*}
    \grad w_t = (\grad Z_t) \left( Z_t' Z_t \right)^{-1} \Gamma_\beta b
\end{equation*}
\begin{equation*}
    - \Diag (Z_t \left( Z_t' Z_t \right)^{-1} (\grad Z_t)') \circ (Z_t \left( Z_t' Z_t \right)^{-1} \Gamma_\beta b)
\end{equation*}
\begin{equation}
    - \Diag (Z_t \left( Z_t' Z_t \right)^{-1} Z_t') \circ ((\grad Z_t) \left( Z_t' Z_t \right)^{-1} \Gamma_\beta b).
\end{equation}

\subsubsection{KPS Model Derivatives}

I use the notation and propositions from Appendix \ref{subsec:derivative_details} in this section of the Appendix. 

In \cite{kelly}, the MVE portfolio is:
\begin{equation}
    w_{t} = f(Z_t) b.
\end{equation}
\begin{equation}
    f(Z_t) = Z_t \Gamma_\beta \left( \Gamma_\beta' Z_t' Z_t \Gamma_\beta \right)^{-1}.
\end{equation}

I can calculate:
\begin{equation*}
    \grad w_t = (\grad Z_t) \Gamma_\beta \left( \Gamma_\beta' Z_t' Z_t \Gamma_\beta \right)^{-1} b
\end{equation*}
\begin{equation*}
    - \Diag (Z_t \Gamma_\beta \left( \Gamma_\beta' Z_t' Z_t \Gamma_\beta \right)^{-1} \Gamma_{\beta}' (\grad Z_t)') \circ (Z_t \Gamma_{\beta} \left( \Gamma_\beta' Z_t' Z_t \Gamma_\beta \right)^{-1} b)
\end{equation*}
\begin{equation}
    - \Diag (Z_t \Gamma_\beta \left( \Gamma_\beta' Z_t' Z_t \Gamma_\beta \right)^{-1} \Gamma_{\beta}' Z_t') \circ ((\grad Z_t) \Gamma_{\beta} \left( \Gamma_\beta' Z_t' Z_t \Gamma_\beta \right)^{-1} b).
\end{equation}

\subsubsection{GKX Model Derivatives}

I use the notation and propositions from Appendix \ref{subsec:derivative_details} in this section of the Appendix. 

In \cite{gu}, the portfolio weights are
\begin{equation}
    f(Z_t) = Z_t (Z_t' Z_t)^{-1} \Gamma_A + \iota_N \Gamma_B',
\end{equation}
where $\Gamma_A$ is a $K \times M$ matrix and $\Gamma_B$ is a $M$ dimensional column vector of parameters. In the language of neural networks, $\Gamma_A$ is a matrix of weights and $\Gamma_B$ is a vector of biases. 

Then I have:
\begin{equation*}
    \grad w_t = (\grad Z_t) (Z_t' Z_t)^{-1} \Gamma_A b
\end{equation*}
\begin{equation*}
    - \Diag (Z_t (Z_t' Z_t)^{-1} (\grad Z_t)') \circ (Z_t (Z_t' Z_t)^{-1} \Gamma_A b) 
\end{equation*}
\begin{equation}
    - \Diag (Z_t (Z_t' Z_t)^{-1} Z_t') \circ ((\grad Z_t) (Z_t' Z_t)^{-1} \Gamma_A b).
\end{equation}

\subsubsection{NN Model Derivatives} \label{subsec:nn_model_derivs}

I use the notation and propositions from Appendix \ref{subsec:derivative_details} in this section of the Appendix. 

I have:
\begin{equation*}
    w_t = \frac{1}{\zeta} \mu_t - \frac{1}{\zeta} \Gamma_t (\zeta I + \Gamma_t' \Gamma_t)^{-1} \Gamma_t' \mu_t.
\end{equation*}
Then I can calculate:
\begin{equation*}
    \grad w_t = \frac{1}{\zeta} (\grad \mu_t) - \frac{1}{\zeta} (\grad \Gamma_t) (\zeta I + \Gamma_t' \Gamma_t)^{-1} \Gamma_t' \mu_t
\end{equation*}
\begin{equation*}
    + \frac{1}{\zeta} \Diag (\Gamma_t (\zeta I + \Gamma_t' \Gamma_t)^{-1}
    (\grad \Gamma_t)') \circ
    (\Gamma_t (\zeta I + \Gamma_t' \Gamma_t)^{-1} \Gamma_t' \mu_t)
\end{equation*}
\begin{equation*}
    + \frac{1}{\zeta} \Diag (\Gamma_t (\zeta I + \Gamma_t' \Gamma_t)^{-1}
    \Gamma_t') \circ
    ((\grad \Gamma_t) (\zeta I + \Gamma_t' \Gamma_t)^{-1} \Gamma_t' \mu_t)
\end{equation*}
\begin{equation*}
    - \frac{1}{\zeta} \Diag (\Gamma_t (\zeta I + \Gamma_t' \Gamma_t)^{-1} (\grad \Gamma_t)') \circ
    \mu_t
\end{equation*}
\begin{equation} \label{eq:cov_nn_effects}
    - \frac{1}{\zeta} \Diag (\Gamma_t (\zeta I + \Gamma_t' \Gamma_t)^{-1} \Gamma_t') \circ
    (\grad \mu_t).
\end{equation}

In the case of zero hidden layers, then
\begin{equation}
    \mu_t = Z_t \beta_{\mu}, \;\text{ and }\;
    \Gamma_t = Z_t \beta_{\Sigma},
\end{equation}
for some parameter vectors $\beta_{\mu}$ and $\beta_{\Sigma}$. 

Then
\begin{equation}
    \grad \mu_t = (\grad Z_t) \beta_{\mu}, \;\;\;
    \grad \Gamma_t = (\grad Z_t) \beta_{\Gamma}.
\end{equation}

In the more general case, 
\begin{equation}
    \mu_{i,t} = g_{\mu} (Z_{i,t}), \; \text{ and }\;
    \Gamma_{i,t} = g_{\Sigma} (Z_{i,t}).
\end{equation}
where $g_\mu : \R^K \to \R$ and $g_\Sigma : \R^K \to V$ are the standard neural network functions. Then let $D g_{\mu}$ be the $1 \times K$ Jacobian matrix of derivatives of $g_\mu$ and $D g_\Sigma$ be the $V \times K$ Jacobian matrix of $g_\Sigma$. Then I can write:
\begin{equation}
    \grad \mu_{i,t} = D g_\mu (\grad Z_{i,t}), \;\;\;
    \grad \Gamma_{i,t} = D g_{\Sigma} (\grad Z_{i,t}).
\end{equation}

In Figures \ref{fig:calibration_and_weights} and Tables \ref{tab:cov_decomp}, I separately decompose the elasticity that includes and excludes covariance matrix effects. Equation (\ref{eq:cov_nn_effects}) includes covariance effects. The derivative that excludes covariance effects is given by:
\begin{equation} \label{eq:mean_nn_effects}
    \grad w_t = \frac{1}{\zeta} (\grad \mu_t)
    - \frac{1}{\zeta} \Diag (\Gamma_t (\zeta I + \Gamma_t' \Gamma_t)^{-1} \Gamma_t') \circ
    (\grad \mu_t).
\end{equation}

\subsubsection{BPZ$_F$ and RF Model Derivatives} \label{subsec:bpz_derivs}

These models are not continuous functions of $Z_{t}$. Ignoring the break points severely attenuates any measure of how much $f(Z_t)$ varies when $Z_t$ varies. In order to have a continuous demand function (necessary for equilibrium to exist) and to have reasonable estimates of the elasticity of these statistical arbitrageurs, I am forced to approximate these derivatives with numerical derivatives. This is done as I describe below. 

I calculate $\nabla_{\log (p_t)} f(Z_t) b$ directly. Here I use $\nabla_{\log (p_t)} Z_t$ (see Appendix \ref{subsec:derivative_details} to see how this is calculated). In particular, for these two models, I approximate:
\begin{equation}
    \nabla_{\log (p_t)} f(Z_t) b \approx \frac{f(Z_t + h \nabla_{\log (p_t)} Z_t)b - f(Z_t - h \nabla_{\log (p_t)} Z_t)b}{2h}.
\end{equation}

As $h$ goes to zero, this becomes a better and better approximation of the derivative (where it is defined). However, the actual derivative understates how much $f(Z_t)b$ changes when prices change because it ignores the large breakpoint jumps. Thus, I calculate this with values of $h$ that are not too small and not too large. 

I calculate this with $h \in \{0.0025, 0.005, 0.0075, 0.01\}$, and then average over these four numerical derivatives. This obviously corresponds to a 0.25\%, 0.5\%, 0.75\%, and 1\% price change, which seems like a reasonable degree of price change to approximate how much $f(Z_t)b$ reacts to changes in prices.

\clearpage
\newpage
\section{Additional Tables and Figures} \label{app:tables:figures}


\begin{figure}[!htpb]
    \centering
    \resizebox{.9\textwidth}{!}{\includegraphics{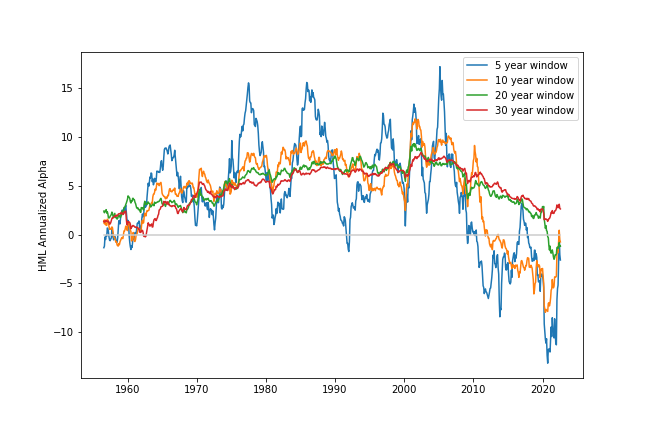}}
    \vspace{3mm}
    \caption{\textbf{Value $\alpha$.} This plot shows the $\alpha$ of value (HML), when value is regressed on the market with various rolling windows, labeled in the legend. The monthly factor returns data is downloaded from Kenneth French's website.}
    \label{fig:value_alpha}
\end{figure}

\begin{table}[!t] \centering
    \resizebox{0.8\textwidth}{!}{\begin{tabular}{@{\extracolsep{5pt}}llll}
    \\[-1.8ex]\hline
\hline \\[-2.4ex]
Model & Hyperparameter & Search Values & Selected \\
\hline \\[-2.4ex] \hline \\[-2.4ex]
    BPZ$_F$ & Trees & & 100 \\
     & Max depth & & 4 \\
     & Max features & & $\approx \log_2(K)$ \\
    & Mean penalty - $\lambda_0$ & $0$, $10^{-5}$, $10^{-3}$, $10^{-1}$, ..., $10^{5}$ & $10^{-1}$ \\
    & LARS nonzero coefficients & 5, 10, 15, ..., 50 & 10 \\
    & $L^2$ penalty - $\lambda_2$ & $0$, $10^{-5}$, $10^{-3}$, $10^{-1}$, ..., $10^{5}$ & $10^{3}$ \\
    \hline \\[-2.4ex]

    BPZ$_L$ & No. of PCA factors & 10, 25, 40, 55, $K$ (no PCA) & 40 \\
    & Mean penalty - $\lambda_0$ & 0, $10^{-10}$, $10^{-9}$, $10^{-8}$, ..., $10^{0}$ & $10^{-1}$ \\
    & LARS nonzero coefficients & 2, 4, 6, ..., 20 & 12 \\
    & $L^2$ penalty - $\lambda_2$ & 0, $10^{-10}$, $10^{-9}$, $10^{-8}$, ..., $10^{0}$ & 0 \\
    \hline \\[-2.4ex]

    CRW & No. of factors & 1, 2, 3, ..., 10 & 7 \\
    & Mean penalty - $\lambda_0$ & $0$, $10^{-5}$, $10^{-4}$, $10^{-3}$, ..., $10^{5}$ & $10^{-1}$ \\
    & $L^2$ penalty - $\lambda_2$ & $0$, $10^{-5}$, $10^{-4}$, $10^{-3}$, ..., $10^{5}$ & $10^{4}$ \\
    \hline \\[-2.4ex]

    FF3 & $L^1$ penalty - $\lambda_1$ & $0$, $10^{-10}$, $10^{-9}$, $10^{-8}$, ..., $10^{0}$ & $10^{-1}$ \\
    & $L^2$ penalty - $\lambda_2$ & $0$, $10^{-10}$, $10^{-9}$, $10^{-8}$, ..., $10^{0}$ & $10^{0}$\\
    \hline \\[-2.4ex]

    FF6 & $L^1$ penalty - $\lambda_1$ & $0$, $10^{-10}$, $10^{-9}$, $10^{-8}$, ..., $10^{0}$ & $10^{-1}$ \\
    & $L^2$ penalty - $\lambda_2$ & $0$, $10^{-10}$, $10^{-9}$, $10^{-8}$, ..., $10^{0}$ & $10^{0}$ \\
    \hline \\[-2.4ex]

    GKX & NN nodes and layers & $()$, $(32)$, $(32, 16)$, $(32, 16, 8)$ & $()$ \\
    & $L^1$ penalty of NN weights & $10^{-5}$, $10^{-3}$ & $10^{-5}$ \\
    & Learning Rate & $0.01$, $0.001$ & $0.01$ \\
    & No. of Factors & 1, 3, 5, 10 & 10 \\
    & Mean penalty - $\lambda_0$ & $0$, $10^{-1}$, $10^{-4}$, $10^{-3}$, ..., $10^{5}$ & $10^{-1}$ \\
    & $L^2$ penalty - $\lambda_2$ & $0$, $10^{-5}$, $10^{-4}$, $10^{-3}$, ..., $10^{5}$ & $0$ \\
    \hline \\[-2.4ex]

    HXZ & $L^1$ penalty - $\lambda_1$ & $0$, $10^{-10}$, $10^{-9}$, $10^{-8}$, ..., $10^{0}$ & $10^{-1}$ \\
    & $L^2$ penalty - $\lambda_2$ & $0$, $10^{-10}$, $10^{-9}$, $10^{-8}$, ..., $10^{0}$ & $10^{0}$ \\
    \hline \\[-2.4ex]

    KKN & No. of factors & 1, 2, 3, ..., 10 & 7 \\
    \hline \\[-2.4ex]

    KNS & No. of PCA factors & 10, 25, 40, 55, $K$ (no PCA) & $K$ \\
    & $L^1$ penalty - $\lambda_1$ & $0$, $10^{-10}$, $10^{-9}$, $10^{-8}$, ..., $10^{0}$ & $10^{-2}$ \\
    & $L^2$ penalty - $\lambda_2$ & $0$, $10^{-10}$, $10^{-9}$, $10^{-8}$, ..., $10^{0}$ & $10^{-1}$\\
    \hline \\[-2.4ex]

    KPS & No. of factors & 1, 2, 3, ..., 10 & 10 \\
    & Mean penalty - $\lambda_0$ & $0$, $10^{-5}$, $10^{-4}$, $10^{-3}$, ..., $10^{5}$ & $10^{5}$ \\
    & $L^2$ penalty - $\lambda_2$ & $0$, $10^{-5}$, $10^{-4}$, $10^{-3}$, ..., $10^{5}$ & $10^{5}$ \\
    \hline \\[-2.4ex]

    NN & NN nodes and layers & $()$, $(5)$, $(25)$, $(25, 5)$ & $()$ \\
    & $L^2$ penalty of NN weights & 0, 50, 100, ..., 1000 & 600 \\
    \hline \\[-2.4ex]

     RF & Trees & & 100 \\
     & Max depth & & 4 \\
     & Max features & & $\approx \log_2(K)$ \\
    \\[-2.5ex]
    \hline \hline \\[-1.8ex]
\end{tabular}}
    \vspace{4mm}
    \caption{\textbf{Hyperparameters.} This table shows the hyperparameter search values and selected values from the four-fold cross-validation design. }
    \label{tab:hypers}
\end{table}

\begin{table}[!t] \centering
    \resizebox{1.\textwidth}{!}{\begin{tabular}{@{\extracolsep{5pt}}llccccccccccccc}
\\[-1.8ex]\hline
\hline \\[-1.8ex]
& & \multicolumn{13}{c}{\textit{Average Elasticity}} \
\cr \cline{3-15}
\\[-1.8ex] Elasticity & Wins. & \multicolumn{1}{c}{BPZ$_F$} & \multicolumn{1}{c}{BPZ$_L$} & \multicolumn{1}{c}{BSV} & \multicolumn{1}{c}{CRW} & \multicolumn{1}{c}{DGU} & \multicolumn{1}{c}{FF3} & \multicolumn{1}{c}{FF6} & \multicolumn{1}{c}{GKX} & \multicolumn{1}{c}{HXZ} & \multicolumn{1}{c}{KNS} & \multicolumn{1}{c}{KPS} & \multicolumn{1}{c}{NN} & \multicolumn{1}{c}{RF}  \\
\\[-1.8ex] Weighting & & (1) & (2) & (3) & (4) & (5) & (6) & (7) & (8) & (9) & (10) & (11) & (12) & (13) \\
\\[-1.8ex] \hline \\[-1.4ex]

 Equal & None & 11.070$^{***}$ & 48.596$^{***}$ & 25.324$^{**}$ & 16.628$^{***}$ & 39.299$^{**}$ & 14.829$^{***}$ & 7.274$^{***}$ & 65.984$^{***}$ & 1.318$^{***}$ & 26.389$^{***}$ & 26.471$^{***}$ & 9.020$^{***}$ & 3.155$^{***}$ \\
& & (0.490) & (5.094) & (10.031) & (4.153) & (16.257) & (2.807) & (0.895) & (13.268) & (0.134) & (3.791) & (5.180) & (1.077) & (0.208) \\
 & $(1^{st}, 99^{th})$ & 6.313$^{***}$ & 19.630$^{***}$ & 5.709$^{***}$ & 4.806$^{***}$ & 7.455$^{***}$ & 5.405$^{***}$ & 3.303$^{***}$ & 25.786$^{***}$ & 1.073$^{***}$ & 10.210$^{***}$ & 9.114$^{***}$ & 3.973$^{***}$ & 2.243$^{***}$ \\
& & (0.353) & (0.367) & (0.128) & (0.166) & (0.140) & (0.069) & (0.042) & (0.367) & (0.014) & (0.168) & (1.156) & (0.219) & (0.048) \\
 & $(5^{th}, 95^{th})$ & 2.003$^{***}$ & 13.916$^{***}$ & 4.255$^{***}$ & 3.642$^{***}$ & 5.442$^{***}$ & 4.064$^{***}$ & 2.610$^{***}$ & 18.137$^{***}$ & 1.061$^{***}$ & 7.400$^{***}$ & 6.394$^{***}$ & 3.083$^{***}$ & 2.016$^{***}$ \\
& & (0.112) & (0.247) & (0.086) & (0.114) & (0.092) & (0.046) & (0.029) & (0.245) & (0.011) & (0.114) & (0.786) & (0.151) & (0.040) \\

\hline \\[-1.8ex]

 Portfolio & None & 0.285$^{***}$ & 6.428$^{***}$ & 2.275$^{***}$ & 2.034$^{***}$ & 3.209$^{***}$ & 1.911$^{***}$ & 1.553$^{***}$ & 7.394$^{***}$ & 0.843$^{***}$ & 3.630$^{***}$ & 2.951$^{***}$ & 2.496$^{***}$ & 2.185$^{***}$ \\
& & (0.094) & (0.190) & (0.093) & (0.057) & (0.437) & (0.022) & (0.017) & (0.552) & (0.012) & (0.049) & (0.504) & (0.595) & (0.037) \\
 & $(1^{st}, 99^{th})$ & -0.364$^{***}$ & 6.252$^{***}$ & 2.358$^{***}$ & 2.019$^{***}$ & 2.754$^{***}$ & 1.915$^{***}$ & 1.556$^{***}$ & 7.908$^{***}$ & 0.846$^{***}$ & 3.609$^{***}$ & 2.552$^{***}$ & 1.903$^{***}$ & 1.998$^{***}$ \\
& & (0.098) & (0.100) & (0.035) & (0.047) & (0.035) & (0.022) & (0.016) & (0.096) & (0.012) & (0.045) & (0.290) & (0.069) & (0.032) \\
 & $(5^{th}, 95^{th})$ & -1.233$^{***}$ & 6.252$^{***}$ & 2.354$^{***}$ & 2.019$^{***}$ & 2.753$^{***}$ & 1.969$^{***}$ & 1.561$^{***}$ & 7.923$^{***}$ & 0.866$^{***}$ & 3.600$^{***}$ & 2.545$^{***}$ & 1.893$^{***}$ & 1.864$^{***}$ \\
& & (0.082) & (0.100) & (0.035) & (0.047) & (0.034) & (0.021) & (0.016) & (0.096) & (0.011) & (0.045) & (0.290) & (0.069) & (0.032) \\

\hline \\[-1.8ex]

 Value & None & 12.781$^{***}$ & 36.835$^{***}$ & 7.105$^{***}$ & 17.635$^{**}$ & 25.920$^{**}$ & 3.003$^{***}$ & 0.290$^{}$ & 76.991$^{***}$ & -1.208$^{***}$ & 13.109$^{***}$ & 27.588$^{***}$ & 13.516$^{***}$ & 2.101$^{***}$ \\
& & (1.291) & (3.281) & (0.886) & (8.063) & (11.499) & (0.704) & (0.594) & (25.222) & (0.114) & (1.301) & (4.917) & (2.954) & (0.142) \\
 & $(1^{st}, 99^{th})$ & 3.010$^{***}$ & 12.925$^{***}$ & 3.006$^{***}$ & 3.638$^{***}$ & 3.871$^{***}$ & 0.986$^{***}$ & 0.756$^{***}$ & 19.438$^{***}$ & -0.276$^{***}$ & 5.031$^{***}$ & 8.363$^{***}$ & 3.694$^{***}$ & 1.525$^{***}$ \\
& & (0.218) & (0.589) & (0.147) & (0.198) & (0.174) & (0.037) & (0.025) & (0.886) & (0.022) & (0.214) & (1.098) & (0.306) & (0.071) \\
 & $(5^{th}, 95^{th})$ & 0.142$^{**}$ & 7.168$^{***}$ & 1.833$^{***}$ & 1.964$^{***}$ & 2.203$^{***}$ & 0.858$^{***}$ & 0.706$^{***}$ & 9.433$^{***}$ & -0.004$^{}$ & 3.253$^{***}$ & 3.510$^{***}$ & 2.163$^{***}$ & 0.955$^{***}$ \\
& & (0.059) & (0.266) & (0.071) & (0.085) & (0.087) & (0.026) & (0.021) & (0.311) & (0.007) & (0.117) & (0.476) & (0.152) & (0.033) \\

\hline
\hline \\[-1.8ex]

\textit{Note:} & \multicolumn{14}{r}{$^{*}$p$<$0.1; $^{**}$p$<$0.05; $^{***}$p$<$0.01} \\
\end{tabular}

}
    \vspace{4mm}
    \caption{\textbf{Statistical Arbitrageur Price Elasticity.}
    This is similar to Table \ref{tab:statistical arbitrageur_elasticity}, except equal weighting and winsorization at the 1$^{st}$ and 99$^{th}$ percentiles are shown. }
    \label{tab:statistical arbitrageur_elasticity_full}
\end{table}

\begin{table}[!t] \centering
    \resizebox{0.3\textwidth}{!}{\begin{tabular}{@{\extracolsep{5pt}}lcc}
\\[-1.8ex]\hline
\hline \\[-1.8ex]
& \multicolumn{2}{c}{\textit{NN Decomposition}} \
\cr \cline{2-3}
\\[-1.8ex] & Elasticity & Mean Component \\
\\[-1.8ex] & (1) & (2) \\
\hline \\[-1.8ex]
 NN & 2.163$^{***}$ & 0.750$^{***}$ \\
& (0.152) & (0.095) \\
\hline
\hline \\[-1.8ex]
\textit{Note:} & \multicolumn{2}{r}{$^{*}$p$<$0.1; $^{**}$p$<$0.05; $^{***}$p$<$0.01} \\
\end{tabular}}
    \vspace{4mm}
    \caption{\textbf{Mean and Covariance Effect Decomposition.} This table shows the value-weighted winsorized (5$^{th}$ and 95$^{th}$ percentiles) average price elasticity of the NN model in the first column, and of the NN model with only mean effects in the second column, during the out-of-sample period from February 1990 to January 2020, inclusive. Below the average values in parentheses, the standard errors, double clustered by month and stock, are shown. Table \ref{tab:models} displays the initialisms of these models. There are 1,253,175 month stock observations these values are calculated from.}
    \label{tab:cov_decomp}
\end{table}



\begin{table}[!t] \centering
    \resizebox{1\textwidth}{!}{\begin{tabular}{@{\extracolsep{5pt}}llccccccccccccc}
\\[-1.8ex]\hline
\hline \\[-1.8ex]
 & & \multicolumn{13}{c}{\textit{Change in Aggregate 
 Counterfactual Elasticity for Individual Stocks with Different Statistical Arbitrageurs}} \
\cr \cline{3-15}
\\[-1.8ex] Wins & AF $\theta_t$ & \multicolumn{1}{c}{BPZ$_F$} & \multicolumn{1}{c}{BPZ$_L$} & \multicolumn{1}{c}{BSV} & \multicolumn{1}{c}{CRW} & \multicolumn{1}{c}{DGU} & \multicolumn{1}{c}{FF3} & \multicolumn{1}{c}{FF6} & \multicolumn{1}{c}{GKX} & \multicolumn{1}{c}{HXZ} & \multicolumn{1}{c}{KNS} & \multicolumn{1}{c}{KPS} & \multicolumn{1}{c}{NN} & \multicolumn{1}{c}{RF}  \\
\\[-1.8ex] & & (1) & (2) & (3) & (4) & (5) & (6) & (7) & (8) & (9) & (10) & (11) & (12) & (13) \\
\hline \\[-1.8ex]

None & 0.0001 &   0.000 & 0.001 & 0.002 & 0.000 & 0.000 & 0.000 & 0.000 & 0.000 & 0.000 & 0.001 & 0.000 & 0.000 & 0.000 \\
& 0.0005 &   0.000 & 0.001 & 0.002 & 0.000 & 0.002 & 0.000 & 0.000 & 0.000 & 0.000 & 0.000 & 0.000 & 0.000 & 0.000 \\
& 0.001 &   0.000 & 0.001 & 0.004 & 0.000 & 0.004 & 0.000 & 0.000 & 0.001 & 0.000 & 0.000 & 0.000 & 0.000 & 0.000 \\
& 0.005 &   0.001 & 0.003 & 0.009 & 0.001 & 0.009 & -0.003 & -0.002 & 0.003 & -0.003 & -0.002 & 0.001 & 0.001 & 0.001 \\
& 0.01 &   0.002 & 0.006 & 0.013 & 0.002 & 0.012 & -0.006 & -0.004 & 0.004 & -0.006 & -0.005 & 0.002 & 0.002 & 0.002 \\
& 0.05 &   0.009 & 0.03 & 0.025 & 0.009 & 0.025 & -0.03 & -0.023 & 0.019 & -0.029 & -0.023 & 0.009 & 0.009 & 0.009 \\
& 0.1 &   0.018 & 0.05 & -0.004 & 0.018 & 0.037 & -0.06 & -0.044 & 0.025 & -0.059 & -0.042 & 0.019 & 0.02 & 0.018 \\
& 0.25 &   0.046 & 0.122 & -0.018 & 0.046 & 0.069 & -0.165 & -0.11 & 0.078 & -0.147 & -0.148 & 0.047 & 0.046 & 0.046 \\

\hline \\[-1.8ex]

($5^{th}$, $95^{th}$) & 0.0001 &   0.000 & 0.001 & 0.001 & 0.000 & 0.000 & 0.000 & 0.000 & 0.000 & 0.000 & 0.001 & 0.000 & 0.000 & 0.000 \\
& 0.0005 &   0.000 & 0.001 & 0.002 & 0.000 & 0.002 & 0.000 & 0.000 & 0.000 & 0.000 & 0.000 & 0.000 & 0.000 & 0.000 \\
& 0.001 &   0.000 & 0.001 & 0.004 & 0.000 & 0.003 & 0.000 & 0.000 & 0.001 & 0.000 & 0.000 & 0.000 & 0.000 & 0.000 \\
& 0.005 &   0.001 & 0.003 & 0.009 & 0.001 & 0.008 & -0.001 & -0.001 & 0.002 & -0.002 & -0.001 & 0.001 & 0.001 & 0.001 \\
& 0.01 &   0.002 & 0.005 & 0.014 & 0.001 & 0.01 & -0.003 & -0.002 & 0.003 & -0.003 & -0.002 & 0.001 & 0.001 & 0.001 \\
& 0.05 &   0.007 & 0.028 & 0.03 & 0.007 & 0.02 & -0.014 & -0.01 & 0.017 & -0.014 & -0.008 & 0.007 & 0.007 & 0.007 \\
& 0.1 &   0.013 & 0.044 & 0.007 & 0.013 & 0.03 & -0.028 & -0.019 & 0.02 & -0.028 & -0.015 & 0.014 & 0.015 & 0.013 \\
& 0.25 &   0.033 & 0.109 & 0.017 & 0.033 & 0.053 & -0.079 & -0.048 & 0.066 & -0.069 & -0.065 & 0.034 & 0.033 & 0.033 \\

\hline
\hline \\[-1.8ex]
\end{tabular}
}
    \vspace{4mm}
    \caption{\textbf{Elasticity of the Market with Statistical Arbitrageurs.} This table shows the value-weighted average (across assets and months) change in price elasticity of the aggregate market for individual stocks, with HSA investors. This is calculated simply as the elasticity in the counterfactual experiment minus the elasticity with no HSA investors. AF (arbitrageur fraction or $\theta_t$) represents the fraction of capital the HSA is granted in the counterfactual experiments. The results are either not winsorized (labeled None) or winsorized at the 5$^{th}$ and 95$^{th}$ percentiles (labeled (5$^{th}$, 95$^{th}$)). Each column represents a different statistical arbitrageur model, where Table \ref{tab:models} describes the initialisms of these models. There are 1,253,175 month stock observations these values are calculated from. The sample period is February 1990 to January 2020, inclusive.}
    \label{tab:mkt_elasticity_simple}
\end{table}

\begin{table}[!t] \centering
    \resizebox{1\textwidth}{!}{\begin{tabular}{@{\extracolsep{5pt}}lccccccccccccc}

\\[-1.8ex]\hline
\hline \\[-1.8ex]
\\[-1.8ex] 
& \multicolumn{13}{c}{Panel A: Log CAPM GRS Ratio, February 2010 - January 2020} \\ \cline{2-14}
\\[-1.8ex] 
\\[-1.8ex] AF ($\theta_t$) & \multicolumn{1}{c}{BPZ$_F$} & \multicolumn{1}{c}{BPZ$_L$} & \multicolumn{1}{c}{BSV} & \multicolumn{1}{c}{CRW} & \multicolumn{1}{c}{DGU} & \multicolumn{1}{c}{FF3} & \multicolumn{1}{c}{FF6} & \multicolumn{1}{c}{GKX} & \multicolumn{1}{c}{HXZ} & \multicolumn{1}{c}{KNS} & \multicolumn{1}{c}{KPS} & \multicolumn{1}{c}{NN} & \multicolumn{1}{c}{RF}  \\
\\[-1.8ex]
\hline \\[-1.8ex]

0.0001 &   $-0.141^{***}$ & $-0.434^{***}$ & $-0.096^{***}$ & $-0.019^{***}$ & $-0.062^{***}$ & $-0.114^{***}$ & $-0.247^{***}$ & $0.001^{***}$ & $-0.059^{***}$ & $0.066^{***}$ & $-0.003^{***}$ & $0.026^{***}$ & $-0.087^{***}$\\
0.0005 &   $-0.002^{***}$ & $0.165^{***}$ & $-0.09^{***}$ & $0.005^{***}$ & $-0.055^{***}$ & $-0.001^{***}$ & $-0.097^{***}$ & $0.000^{***}$ & $-0.054^{***}$ & $-0.031^{***}$ & $0.011^{***}$ & $0.000^{***}$ & $-0.042^{***}$\\
0.001 &   $0.001^{***}$ & $0.001^{***}$ & $-0.045^{***}$ & $0.000^{***}$ & $0.889^{***}$ & $-0.01^{***}$ & $-0.005^{***}$ & $0.001^{***}$ & $-0.061^{***}$ & $-0.088^{***}$ & $-0.145^{***}$ & $0.001^{***}$ & $-0.04^{***}$\\
0.005 &   $0.004^{***}$ & $0.004^{***}$ & $0.178^{***}$ & $0.007^{***}$ & $1.986^{***}$ & $0.002^{***}$ & $-0.001^{***}$ & $0.004^{***}$ & $-0.003^{***}$ & $-0.002^{***}$ & $-0.074^{***}$ & $0.004^{***}$ & $0.002^{***}$\\
0.01 &   $0.007^{***}$ & $0.007^{***}$ & $0.058^{***}$ & $0.003^{***}$ & $1.921^{***}$ & $-0.263^{***}$ & $0.001^{***}$ & $0.007^{***}$ & $-0.014^{***}$ & $-0.271^{***}$ & $-0.156^{***}$ & $0.007^{***}$ & $0.007^{***}$\\
0.05 &   $0.034^{***}$ & $0.028^{***}$ & $0.218^{***}$ & $0.007^{***}$ & $2.058^{***}$ & $0.012^{***}$ & $0.01^{***}$ & $0.027^{***}$ & $-0.077^{***}$ & $0.275^{***}$ & $0.088^{***}$ & $0.015^{***}$ & $0.028^{***}$\\
0.1 &   $0.005^{***}$ & $0.006^{***}$ & $0.116^{***}$ & $0.007^{***}$ & $2.516^{***}$ & $0.015^{***}$ & $0.037^{***}$ & $0.007^{***}$ & $-0.081^{***}$ & $0.737^{***}$ & $-0.236^{***}$ & $0.015^{***}$ & $0.006^{***}$\\
0.25 &   $-0.155^{***}$ & $-0.086^{***}$ & $0.325^{***}$ & $-0.085^{***}$ & $2.433^{***}$ & $0.026^{***}$ & $0.007^{***}$ & $-0.073^{***}$ & $0.04^{***}$ & $-0.03^{***}$ & $-0.147^{***}$ & $-0.087^{***}$ & $-0.086^{***}$\\

\hline
\hline \\[-1.8ex]
\\[-1.8ex] 
& \multicolumn{13}{c}{Panel B: Log GRS FF3 Ratio, February 2010 - January 2020} \\ \cline{2-14}
\\[-1.8ex] 
AF ($\theta_t$) & \multicolumn{1}{c}{BPZ$_F$} & \multicolumn{1}{c}{BPZ$_L$} & \multicolumn{1}{c}{BSV} & \multicolumn{1}{c}{CRW} & \multicolumn{1}{c}{DGU} & \multicolumn{1}{c}{FF3} & \multicolumn{1}{c}{FF6} & \multicolumn{1}{c}{GKX} & \multicolumn{1}{c}{HXZ} & \multicolumn{1}{c}{KNS} & \multicolumn{1}{c}{KPS} & \multicolumn{1}{c}{NN} & \multicolumn{1}{c}{RF}  \\
\\[-1.8ex] 
\hline \\[-1.8ex]

0.0001 &   $-0.164^{***}$ & $-0.474^{***}$ & $-0.11^{***}$ & $-0.019^{***}$ & $-0.058^{***}$ & $-0.118^{***}$ & $-0.29^{***}$ & $0.000^{***}$ & $-0.026^{***}$ & $0.054^{***}$ & $-0.003^{***}$ & $0.029^{***}$ & $-0.075^{***}$\\
0.0005 &   $-0.002^{***}$ & $0.162^{***}$ & $-0.079^{***}$ & $0.006^{***}$ & $-0.071^{***}$ & $-0.001^{***}$ & $-0.124^{***}$ & $0.000^{***}$ & $-0.033^{***}$ & $-0.032^{***}$ & $0.015^{***}$ & $0.000^{***}$ & $-0.04^{***}$\\
0.001 &   $0.001^{***}$ & $0.001^{***}$ & $-0.04^{***}$ & $0.000^{***}$ & $1.05^{***}$ & $-0.017^{***}$ & $-0.004^{***}$ & $0.001^{***}$ & $-0.091^{***}$ & $-0.086^{***}$ & $-0.144^{***}$ & $0.001^{***}$ & $-0.027^{***}$\\
0.005 &   $0.005^{***}$ & $0.005^{***}$ & $0.172^{***}$ & $0.008^{***}$ & $1.766^{***}$ & $0.001^{***}$ & $0.019^{***}$ & $0.006^{***}$ & $15.499^{***}$ & $-0.01^{***}$ & $-0.102^{***}$ & $0.006^{***}$ & $0.004^{***}$\\
0.01 &   $0.01^{***}$ & $0.01^{***}$ & $0.067^{***}$ & $0.005^{***}$ & $1.542^{***}$ & $-0.293^{***}$ & $0.008^{***}$ & $0.01^{***}$ & $-0.007^{***}$ & $-0.297^{***}$ & $-0.184^{***}$ & $0.009^{***}$ & $0.01^{***}$\\
0.05 &   $0.042^{***}$ & $0.036^{***}$ & $0.267^{***}$ & $0.015^{***}$ & $1.152^{***}$ & $0.036^{***}$ & $0.018^{***}$ & $0.035^{***}$ & $-0.082^{***}$ & $0.292^{***}$ & $0.093^{***}$ & $0.024^{***}$ & $0.036^{***}$\\
0.1 &   $0.016^{***}$ & $0.017^{***}$ & $0.126^{***}$ & $0.018^{***}$ & $1.836^{***}$ & $0.019^{***}$ & $0.041^{***}$ & $0.018^{***}$ & $-0.111^{***}$ & $0.713^{***}$ & $-0.264^{***}$ & $0.026^{***}$ & $0.017^{***}$\\
0.25 &   $-0.134^{***}$ & $-0.061^{***}$ & $0.449^{***}$ & $-0.059^{***}$ & $1.508^{***}$ & $0.029^{***}$ & $0.013^{***}$ & $-0.046^{***}$ & $0.044^{***}$ & $-0.025^{***}$ & $-0.126^{***}$ & $-0.061^{***}$ & $-0.061^{***}$\\

\hline
\hline \\[-1.8ex]
\\[-1.8ex] 
& \multicolumn{13}{c}{Panel C: Log GRS FF6 Ratio, February 2010 - January 2020} \\ \cline{2-14}
\\[-1.8ex] 
AF ($\theta_t$) & \multicolumn{1}{c}{BPZ$_F$} & \multicolumn{1}{c}{BPZ$_L$} & \multicolumn{1}{c}{BSV} & \multicolumn{1}{c}{CRW} & \multicolumn{1}{c}{DGU} & \multicolumn{1}{c}{FF3} & \multicolumn{1}{c}{FF6} & \multicolumn{1}{c}{GKX} & \multicolumn{1}{c}{HXZ} & \multicolumn{1}{c}{KNS} & \multicolumn{1}{c}{KPS} & \multicolumn{1}{c}{NN} & \multicolumn{1}{c}{RF}  \\
\\[-1.8ex] 
\hline \\[-1.8ex]

0.0001 &   $-0.223^{***}$ & $-0.42^{***}$ & $-0.154^{***}$ & $-0.02^{***}$ & $0.05^{***}$ & $-0.171^{***}$ & $-0.289^{***}$ & $-0.01^{***}$ & $-0.119^{***}$ & $0.043^{***}$ & $-0.009^{***}$ & $0.032^{***}$ & $-0.04^{***}$\\
0.0005 &   $-0.002^{***}$ & $0.134^{***}$ & $-0.124^{***}$ & $0.004^{***}$ & $-0.001^{***}$ & $-0.002^{***}$ & $46.05^{***}$ & $0.000^{***}$ & $-0.148^{***}$ & $-0.026^{***}$ & $-0.015^{***}$ & $0.000^{***}$ & $-0.029^{***}$\\
0.001 &   $0.002^{***}$ & $0.001^{***}$ & $-0.047^{***}$ & $0.002^{***}$ & $0.983^{***}$ & $-0.012^{***}$ & $-0.01^{***}$ & $0.001^{***}$ & $-0.18^{***}$ & $-0.059^{***}$ & $-0.204^{***}$ & $0.001^{***}$ & $-0.058^{***}$\\
0.005 &   $0.005^{***}$ & $0.005^{***}$ & $0.152^{***}$ & $0.009^{***}$ & $1.237^{***}$ & $0.006^{***}$ & $-0.019^{***}$ & $0.006^{***}$ & $-0.034^{***}$ & $-0.009^{***}$ & $-0.141^{***}$ & $0.006^{***}$ & $0.004^{***}$\\
0.01 &   $0.012^{***}$ & $0.011^{***}$ & $0.069^{***}$ & $0.002^{***}$ & $1.116^{***}$ & $-0.315^{***}$ & $-0.007^{***}$ & $0.011^{***}$ & $-0.003^{***}$ & $-0.447^{***}$ & $-0.308^{***}$ & $0.01^{***}$ & $0.011^{***}$\\
0.05 &   $0.067^{***}$ & $0.058^{***}$ & $0.189^{***}$ & $0.025^{***}$ & $0.698^{***}$ & $-0.007^{***}$ & $0.006^{***}$ & $0.056^{***}$ & $-0.1^{***}$ & $0.114^{***}$ & $0.025^{***}$ & $0.034^{***}$ & $0.058^{***}$\\
0.1 &   $0.057^{***}$ & $0.061^{***}$ & $0.094^{***}$ & $0.066^{***}$ & $1.161^{***}$ & $0.019^{***}$ & $0.069^{***}$ & $0.062^{***}$ & $-0.191^{***}$ & $0.586^{***}$ & $-0.291^{***}$ & $0.079^{***}$ & $0.061^{***}$\\
0.25 &   $-0.202^{***}$ & $-0.109^{***}$ & $0.33^{***}$ & $-0.107^{***}$ & $1.002^{***}$ & $0.053^{***}$ & $0.07^{***}$ & $-0.091^{***}$ & $-0.105^{***}$ & $-0.007^{***}$ & $-0.198^{***}$ & $-0.11^{***}$ & $-0.109^{***}$\\

\hline
\hline \\[-1.8ex]
\\[-1.8ex] 
& \multicolumn{13}{c}{Panel D: Log GRS HXZ Ratio, February 2010 - January 2020} \\ \cline{2-14}
\\[-1.8ex] 
AF ($\theta_t$) & \multicolumn{1}{c}{BPZ$_F$} & \multicolumn{1}{c}{BPZ$_L$} & \multicolumn{1}{c}{BSV} & \multicolumn{1}{c}{CRW} & \multicolumn{1}{c}{DGU} & \multicolumn{1}{c}{FF3} & \multicolumn{1}{c}{FF6} & \multicolumn{1}{c}{GKX} & \multicolumn{1}{c}{HXZ} & \multicolumn{1}{c}{KNS} & \multicolumn{1}{c}{KPS} & \multicolumn{1}{c}{NN} & \multicolumn{1}{c}{RF}  \\
\\[-1.8ex] 
\hline \\[-1.8ex]

0.0001 &   $-0.195^{***}$ & $-0.458^{***}$ & $-0.081^{***}$ & $-0.021^{***}$ & $-0.014^{***}$ & $-0.159^{***}$ & $-0.265^{***}$ & $0.003^{***}$ & $0.077^{***}$ & $0.078^{***}$ & $-0.009^{***}$ & $0.016^{***}$ & $-0.041^{***}$\\
0.0005 &   $-0.004^{***}$ & $0.126^{***}$ & $-0.171^{***}$ & $0.006^{***}$ & $0.039^{***}$ & $-0.001^{***}$ & $-0.094^{***}$ & $-0.001^{***}$ & $0.078^{***}$ & $0.017^{***}$ & $0.011^{***}$ & $0.000^{***}$ & $-0.035^{***}$\\
0.001 &   $0.001^{***}$ & $0.000^{***}$ & $-0.049^{***}$ & $0.001^{***}$ & $1.074^{***}$ & $-0.007^{***}$ & $0.01^{***}$ & $0.000^{***}$ & $-0.056^{***}$ & $-0.039^{***}$ & $-0.12^{***}$ & $0.000^{***}$ & $-0.019^{***}$\\
0.005 &   $0.000^{***}$ & $0.000^{***}$ & $0.133^{***}$ & $0.003^{***}$ & $1.668^{***}$ & $0.003^{***}$ & $-0.007^{***}$ & $0.002^{***}$ & $0.009^{***}$ & $-0.079^{***}$ & $-0.109^{***}$ & $0.001^{***}$ & $-0.002^{***}$\\
0.01 &   $0.001^{***}$ & $0.000^{***}$ & $0.059^{***}$ & $-0.006^{***}$ & $1.526^{***}$ & $-0.265^{***}$ & $-0.004^{***}$ & $0.000^{***}$ & $-0.025^{***}$ & $-0.181^{***}$ & $-0.181^{***}$ & $0.000^{***}$ & $0.000^{***}$\\
0.05 &   $0.009^{***}$ & $0.001^{***}$ & $0.265^{***}$ & $-0.022^{***}$ & $1.105^{***}$ & $0.007^{***}$ & $0.001^{***}$ & $-0.001^{***}$ & $0.032^{***}$ & $0.227^{***}$ & $0.067^{***}$ & $-0.008^{***}$ & $0.001^{***}$\\
0.1 &   $-0.024^{***}$ & $-0.019^{***}$ & $0.122^{***}$ & $-0.017^{***}$ & $1.883^{***}$ & $0.006^{***}$ & $0.009^{***}$ & $-0.019^{***}$ & $0.038^{***}$ & $0.894^{***}$ & $-0.251^{***}$ & $-0.008^{***}$ & $-0.019^{***}$\\
0.25 &   $-0.23^{***}$ & $-0.145^{***}$ & $0.415^{***}$ & $-0.143^{***}$ & $1.486^{***}$ & $0.006^{***}$ & $-0.014^{***}$ & $-0.132^{***}$ & $0.065^{***}$ & $-0.025^{***}$ & $-0.165^{***}$ & $-0.146^{***}$ & $-0.145^{***}$\\

\hline
\hline 
 & \multicolumn{13}{r}{\textit{Note:} $^{*}$p$<$0.1; $^{**}$p$<$0.05; $^{***}$p$<$0.01} \\

\end{tabular}
}
    \vspace{4mm}
    \caption{\textbf{Simple Experiments GRS.} This table shows the log ratio of the GRS test $F$ statistics of the counterfactual experiment anomaly portfolio outcomes to the baseline GRS test statistic. Thus, positive values imply that alpha across anomaly portfolios increased relative to baseline, while negative values implies that alpha decreased. Small values (close to zero) are approximately percentage changes. AF ($\theta_t$) represents the arbitrageur fraction used in the counterfactual experiments. Each column represents a different statistical arbitrageur model, where Table \ref{tab:models} describes the initialisms of these models. The GRS statistics are calculated using the last decade of the sample after the statistical arbitrageur has participated in the market for decades: February 2010 - January 2020. The panels correspond to different right-hand side factors, and the remaining portfolios are the test assets. Panel A represents results with the CAPM right-hand side factors, Panel B represents results with the \cite{ff3} right-hand side factors, Panel C represents results with the \cite{famafrench15} right-hand side factors in addition to momentum, and Panel D represents results with the \cite{zhang} right-hand side factors. Stars are given to represent whether the alpha remaining counterfactually is different from zero (alpha remained statistically significant), {not} whether baseline and counterfactual alpha are statistically different from one another. The critical values for the log GRS ratio for the CAPM model $-1.984$, $-1.905$, and $-1.759$ for 10\%, 5\%, and 1\% significance respectively. The critical values for the other models are relatively similar. }
    \label{tab:grs_baseline}
\end{table}

\begin{table}[!t] \centering
    \resizebox{0.9\textwidth}{!}{\begin{tabular}{@{\extracolsep{5pt}}lccccccccccccc}

\\[-1.8ex]\hline
\hline \\[-1.8ex]
\\[-1.8ex] 
& \multicolumn{13}{c}{Panel A: Log CAPM GRS Ratio, February 2000 - January 2010} \\ \cline{2-14}
\\[-1.8ex] 
\\[-1.8ex] AF ($\theta_t$) & \multicolumn{1}{c}{BPZ$_F$} & \multicolumn{1}{c}{BPZ$_L$} & \multicolumn{1}{c}{BSV} & \multicolumn{1}{c}{CRW} & \multicolumn{1}{c}{DGU} & \multicolumn{1}{c}{FF3} & \multicolumn{1}{c}{FF6} & \multicolumn{1}{c}{GKX} & \multicolumn{1}{c}{HXZ} & \multicolumn{1}{c}{KNS} & \multicolumn{1}{c}{KPS} & \multicolumn{1}{c}{NN} & \multicolumn{1}{c}{RF}  \\
\\[-1.8ex]
\hline \\[-1.8ex]

0.0001 &   $-0.082^{***}$ & $-0.3^{***}$ & $-0.431^{***}$ & $-0.001^{***}$ & $0.011^{***}$ & $-0.109^{***}$ & $-0.136^{***}$ & $0.001^{***}$ & $-0.118^{***}$ & $-0.033^{***}$ & $-0.007^{***}$ & $0.009^{***}$ & $-0.036^{***}$\\
0.0005 &   $0.000^{***}$ & $0.168^{***}$ & $-0.024^{***}$ & $-0.001^{***}$ & $0.523^{***}$ & $-0.001^{***}$ & $0.209^{***}$ & $0.000^{***}$ & $-0.198^{***}$ & $-0.013^{***}$ & $-0.035^{***}$ & $-0.004^{***}$ & $-0.026^{***}$\\
0.001 &   $-0.001^{***}$ & $-0.001^{***}$ & $0.003^{***}$ & $-0.001^{***}$ & $0.827^{***}$ & $-0.014^{***}$ & $0.02^{***}$ & $-0.001^{***}$ & $-0.224^{***}$ & $0.009^{***}$ & $-0.018^{***}$ & $-0.001^{***}$ & $-0.018^{***}$\\
0.005 &   $-0.004^{***}$ & $0.39^{***}$ & $-0.002^{***}$ & $-0.004^{***}$ & $1.146^{***}$ & $-0.004^{***}$ & $0.035^{***}$ & $-0.005^{***}$ & $-0.126^{***}$ & $-0.063^{***}$ & $0.199^{***}$ & $-0.004^{***}$ & $-0.001^{***}$\\
0.01 &   $-0.008^{***}$ & $-0.008^{***}$ & $-0.006^{***}$ & $-0.011^{***}$ & $1.79^{***}$ & $-0.041^{***}$ & $0.012^{***}$ & $-0.008^{***}$ & $-0.005^{***}$ & $-0.086^{***}$ & $0.205^{***}$ & $-0.008^{***}$ & $-0.008^{***}$\\
0.05 &   $-0.045^{***}$ & $-0.046^{***}$ & $0.007^{***}$ & $-0.049^{***}$ & $1.761^{***}$ & $0.04^{***}$ & $0.003^{***}$ & $-0.05^{***}$ & $-0.113^{***}$ & $-0.04^{***}$ & $0.1^{***}$ & $-0.042^{***}$ & $-0.046^{***}$\\
0.1 &   $-0.059^{***}$ & $-1.643^{***}$ & $0.146^{***}$ & $-0.06^{***}$ & $2.041^{***}$ & $-0.009^{***}$ & $-0.046^{***}$ & $-0.059^{***}$ & $-0.119^{***}$ & $0.657^{***}$ & $0.148^{***}$ & $-0.054^{***}$ & $-0.059^{***}$\\
0.25 &   $-0.058^{***}$ & $-0.064^{***}$ & $0.077^{***}$ & $-0.064^{***}$ & $2.11^{***}$ & $-0.039^{***}$ & $-0.072^{***}$ & $-0.064^{***}$ & $-0.151^{***}$ & $-0.016^{***}$ & $-0.072^{***}$ & $-0.064^{***}$ & $-0.064^{***}$\\

\hline
\hline \\[-1.8ex]
\\[-1.8ex] 
& \multicolumn{13}{c}{Panel B: Log GRS FF3 Ratio, February 2000 - January 2010} \\ \cline{2-14}
\\[-1.8ex] 
AF ($\theta_t$) & \multicolumn{1}{c}{BPZ$_F$} & \multicolumn{1}{c}{BPZ$_L$} & \multicolumn{1}{c}{BSV} & \multicolumn{1}{c}{CRW} & \multicolumn{1}{c}{DGU} & \multicolumn{1}{c}{FF3} & \multicolumn{1}{c}{FF6} & \multicolumn{1}{c}{GKX} & \multicolumn{1}{c}{HXZ} & \multicolumn{1}{c}{KNS} & \multicolumn{1}{c}{KPS} & \multicolumn{1}{c}{NN} & \multicolumn{1}{c}{RF}  \\
\\[-1.8ex] 
\hline \\[-1.8ex]

0.0001 &   $-0.061^{***}$ & $14.05^{***}$ & $-0.415^{***}$ & $-0.002^{***}$ & $0.017^{***}$ & $-0.117^{***}$ & $-0.162^{***}$ & $0.001^{***}$ & $-0.137^{***}$ & $10.628^{***}$ & $-0.009^{***}$ & $0.009^{***}$ & $-0.044^{***}$\\
0.0005 &   $0.000^{***}$ & $-0.041^{***}$ & $-0.022^{***}$ & $0.000^{***}$ & $0.627^{***}$ & $0.000^{***}$ & $0.178^{***}$ & $0.000^{***}$ & $-0.195^{***}$ & $-0.007^{***}$ & $-0.015^{***}$ & $-0.004^{***}$ & $-0.025^{***}$\\
0.001 &   $0.000^{***}$ & $0.000^{***}$ & $0.008^{***}$ & $-0.001^{***}$ & $0.942^{***}$ & $-0.014^{***}$ & $0.024^{***}$ & $-0.001^{***}$ & $-0.237^{***}$ & $0.005^{***}$ & $-0.056^{***}$ & $0.000^{***}$ & $-0.057^{***}$\\
0.005 &   $-0.002^{***}$ & $-0.062^{***}$ & $0.000^{***}$ & $-0.002^{***}$ & $0.983^{***}$ & $-0.005^{***}$ & $0.042^{***}$ & $-0.002^{***}$ & $-0.127^{***}$ & $-0.071^{***}$ & $0.175^{***}$ & $-0.002^{***}$ & $0.002^{***}$\\
0.01 &   $-0.004^{***}$ & $-0.004^{***}$ & $-0.004^{***}$ & $-0.007^{***}$ & $1.688^{***}$ & $-0.043^{***}$ & $0.015^{***}$ & $-0.004^{***}$ & $-0.015^{***}$ & $-0.088^{***}$ & $0.226^{***}$ & $-0.004^{***}$ & $-0.004^{***}$\\
0.05 &   $-0.034^{***}$ & $-0.034^{***}$ & $0.016^{***}$ & $-0.038^{***}$ & $1.302^{***}$ & $0.051^{***}$ & $0.013^{***}$ & $-0.037^{***}$ & $-0.135^{***}$ & $-0.047^{***}$ & $0.099^{***}$ & $-0.028^{***}$ & $-0.034^{***}$\\
0.1 &   $-0.055^{***}$ & $67.911^{***}$ & $0.196^{***}$ & $-0.055^{***}$ & $1.619^{***}$ & $0.000^{***}$ & $-0.038^{***}$ & $-0.055^{***}$ & $-0.145^{***}$ & $0.695^{***}$ & $0.103^{***}$ & $-0.051^{***}$ & $-0.055^{***}$\\
0.25 &   $-0.075^{***}$ & $-0.082^{***}$ & $0.07^{***}$ & $-0.083^{***}$ & $1.83^{***}$ & $-0.027^{***}$ & $-0.063^{***}$ & $-0.082^{***}$ & $-0.16^{***}$ & $-0.002^{***}$ & $-0.092^{***}$ & $-0.083^{***}$ & $-0.082^{***}$\\

\hline
\hline \\[-1.8ex]
\\[-1.8ex] 
& \multicolumn{13}{c}{Panel C: Log GRS FF6 Ratio, February 2000 - January 2010} \\ \cline{2-14}
\\[-1.8ex] 
AF ($\theta_t$) & \multicolumn{1}{c}{BPZ$_F$} & \multicolumn{1}{c}{BPZ$_L$} & \multicolumn{1}{c}{BSV} & \multicolumn{1}{c}{CRW} & \multicolumn{1}{c}{DGU} & \multicolumn{1}{c}{FF3} & \multicolumn{1}{c}{FF6} & \multicolumn{1}{c}{GKX} & \multicolumn{1}{c}{HXZ} & \multicolumn{1}{c}{KNS} & \multicolumn{1}{c}{KPS} & \multicolumn{1}{c}{NN} & \multicolumn{1}{c}{RF}  \\
\\[-1.8ex] 
\hline \\[-1.8ex]

0.0001 &   $-0.091^{***}$ & $-0.351^{***}$ & $-0.452^{***}$ & $-0.001^{***}$ & $0.086^{***}$ & $-0.024^{***}$ & $-0.139^{***}$ & $-0.002^{***}$ & $-0.068^{***}$ & $-0.023^{***}$ & $0.002^{***}$ & $0.003^{***}$ & $-0.033^{***}$\\
0.0005 &   $0.000^{***}$ & $-0.028^{***}$ & $-0.03^{***}$ & $-0.001^{***}$ & $0.699^{***}$ & $0.001^{***}$ & $0.314^{***}$ & $0.000^{***}$ & $-0.127^{***}$ & $0.007^{***}$ & $0.03^{***}$ & $0.001^{***}$ & $-0.015^{***}$\\
0.001 &   $0.001^{***}$ & $0.001^{***}$ & $0.001^{***}$ & $0.001^{***}$ & $0.985^{***}$ & $0.011^{***}$ & $0.032^{***}$ & $0.000^{***}$ & $-0.213^{***}$ & $0.018^{***}$ & $-0.078^{***}$ & $0.000^{***}$ & $-0.139^{***}$\\
0.005 &   $0.002^{***}$ & $0.074^{***}$ & $-0.031^{***}$ & $0.002^{***}$ & $0.978^{***}$ & $-0.003^{***}$ & $0.031^{***}$ & $0.002^{***}$ & $-0.03^{***}$ & $-0.049^{***}$ & $0.152^{***}$ & $0.000^{***}$ & $0.006^{***}$\\
0.01 &   $0.003^{***}$ & $0.003^{***}$ & $0.013^{***}$ & $0.000^{***}$ & $1.783^{***}$ & $-0.029^{***}$ & $0.013^{***}$ & $0.003^{***}$ & $-0.052^{***}$ & $-0.06^{***}$ & $0.11^{***}$ & $0.003^{***}$ & $0.003^{***}$\\
0.05 &   $-0.025^{***}$ & $-0.026^{***}$ & $0.048^{***}$ & $-0.029^{***}$ & $1.259^{***}$ & $0.044^{***}$ & $0.015^{***}$ & $-0.03^{***}$ & $-0.08^{***}$ & $0.01^{***}$ & $0.145^{***}$ & $-0.022^{***}$ & $-0.026^{***}$\\
0.1 &   $-0.037^{***}$ & $0.074^{***}$ & $0.25^{***}$ & $-0.036^{***}$ & $1.659^{***}$ & $0.004^{***}$ & $-0.029^{***}$ & $-0.037^{***}$ & $-0.122^{***}$ & $0.683^{***}$ & $0.112^{***}$ & $-0.036^{***}$ & $-0.036^{***}$\\
0.25 &   $0.011^{***}$ & $-0.001^{***}$ & $0.059^{***}$ & $-0.002^{***}$ & $1.644^{***}$ & $-0.022^{***}$ & $-0.051^{***}$ & $-0.002^{***}$ & $-0.168^{***}$ & $-0.034^{***}$ & $-0.023^{***}$ & $-0.004^{***}$ & $-0.001^{***}$\\

\hline
\hline \\[-1.8ex]
\\[-1.8ex] 
& \multicolumn{13}{c}{Panel D: Log GRS HXZ Ratio, February 2000 - January 2010} \\ \cline{2-14}
\\[-1.8ex] 
AF ($\theta_t$) & \multicolumn{1}{c}{BPZ$_F$} & \multicolumn{1}{c}{BPZ$_L$} & \multicolumn{1}{c}{BSV} & \multicolumn{1}{c}{CRW} & \multicolumn{1}{c}{DGU} & \multicolumn{1}{c}{FF3} & \multicolumn{1}{c}{FF6} & \multicolumn{1}{c}{GKX} & \multicolumn{1}{c}{HXZ} & \multicolumn{1}{c}{KNS} & \multicolumn{1}{c}{KPS} & \multicolumn{1}{c}{NN} & \multicolumn{1}{c}{RF}  \\
\\[-1.8ex] 
\hline \\[-1.8ex]

0.0001 &   $-0.148^{***}$ & $-0.422^{***}$ & $-0.468^{***}$ & $-0.001^{***}$ & $0.049^{***}$ & $-0.045^{***}$ & $-0.194^{***}$ & $-0.003^{***}$ & $-0.142^{***}$ & $-0.1^{***}$ & $-0.003^{***}$ & $0.005^{***}$ & $-0.023^{***}$\\
0.0005 &   $0.000^{***}$ & $-0.047^{***}$ & $-0.016^{***}$ & $0.000^{***}$ & $0.545^{***}$ & $0.001^{***}$ & $0.099^{***}$ & $0.000^{***}$ & $-0.284^{***}$ & $-0.001^{***}$ & $0.001^{***}$ & $-0.002^{***}$ & $-0.037^{***}$\\
0.001 &   $0.001^{***}$ & $0.001^{***}$ & $0.005^{***}$ & $0.001^{***}$ & $0.917^{***}$ & $0.005^{***}$ & $0.029^{***}$ & $0.000^{***}$ & $-0.357^{***}$ & $0.014^{***}$ & $-0.074^{***}$ & $0.001^{***}$ & $-0.185^{***}$\\
0.005 &   $0.003^{***}$ & $49.007^{***}$ & $-0.045^{***}$ & $0.003^{***}$ & $0.912^{***}$ & $-0.003^{***}$ & $0.028^{***}$ & $0.003^{***}$ & $-0.108^{***}$ & $-0.077^{***}$ & $0.11^{***}$ & $0.002^{***}$ & $0.008^{***}$\\
0.01 &   $0.006^{***}$ & $0.006^{***}$ & $0.014^{***}$ & $0.005^{***}$ & $1.763^{***}$ & $-0.034^{***}$ & $0.014^{***}$ & $0.007^{***}$ & $-0.043^{***}$ & $-0.069^{***}$ & $0.207^{***}$ & $0.006^{***}$ & $0.006^{***}$\\
0.05 &   $-0.01^{***}$ & $-0.012^{***}$ & $0.023^{***}$ & $-0.014^{***}$ & $1.211^{***}$ & $0.041^{***}$ & $0.022^{***}$ & $-0.015^{***}$ & $-0.127^{***}$ & $-0.001^{***}$ & $0.151^{***}$ & $-0.009^{***}$ & $-0.012^{***}$\\
0.1 &   $-0.024^{***}$ & $-0.066^{***}$ & $0.307^{***}$ & $-0.023^{***}$ & $1.453^{***}$ & $0.012^{***}$ & $-0.019^{***}$ & $-0.024^{***}$ & $-0.143^{***}$ & $0.682^{***}$ & $0.179^{***}$ & $-0.027^{***}$ & $-0.023^{***}$\\
0.25 &   $0.009^{***}$ & $-0.002^{***}$ & $0.076^{***}$ & $-0.002^{***}$ & $1.425^{***}$ & $-0.008^{***}$ & $-0.035^{***}$ & $-0.002^{***}$ & $-0.157^{***}$ & $0.01^{***}$ & $-0.023^{***}$ & $0.009^{***}$ & $-0.002^{***}$\\

\hline
\hline 
 & \multicolumn{13}{r}{\textit{Note:} $^{*}$p$<$0.1; $^{**}$p$<$0.05; $^{***}$p$<$0.01} \\

\end{tabular}
}
    \vspace{4mm}
    \caption{\textbf{Simple Experiments GRS, 2000 - 2010.} This table is similar to Table \ref{tab:grs_baseline}, except the table displays results from the decade (120 months) of February 2000 to January 2010 inclusive. }
    \label{tab:grs_2010}
\end{table}

\begin{table}[!t] \centering
    \resizebox{1\textwidth}{!}{\begin{tabular}{@{\extracolsep{5pt}}lccccccccccccc}

\\[-1.8ex]\hline
\hline \\[-1.8ex]
\\[-1.8ex] 
& \multicolumn{13}{c}{Panel A: Log CAPM GRS Ratio, February 2010 - January 2020} \\ \cline{2-14}
\\[-1.8ex] 
\\[-1.8ex] IAF ($\theta_0$) & \multicolumn{1}{c}{BPZ$_F$} & \multicolumn{1}{c}{BPZ$_L$} & \multicolumn{1}{c}{BSV} & \multicolumn{1}{c}{CRW} & \multicolumn{1}{c}{DGU} & \multicolumn{1}{c}{FF3} & \multicolumn{1}{c}{FF6} & \multicolumn{1}{c}{GKX} & \multicolumn{1}{c}{HXZ} & \multicolumn{1}{c}{KNS} & \multicolumn{1}{c}{KPS} & \multicolumn{1}{c}{NN} & \multicolumn{1}{c}{RF}  \\
\\[-1.8ex]
\hline \\[-1.8ex]

0.0001 &   $0.000^{***}$ & $0.107^{***}$ & $-0.225^{***}$ & $-0.019^{***}$ & $0.212^{***}$ & $-0.459^{***}$ & $-0.104^{***}$ & $0.009^{***}$ & $-0.034^{***}$ & $-0.007^{***}$ & $-0.009^{***}$ & $0.000^{***}$ & $-0.018^{***}$\\
0.0005 &   $0.000^{***}$ & $0.068^{***}$ & $-0.162^{***}$ & $0.001^{***}$ & $0.252^{***}$ & $-0.001^{***}$ & $-0.182^{***}$ & $-0.017^{***}$ & $0.064^{***}$ & $-0.32^{***}$ & $-0.015^{***}$ & $0.000^{***}$ & $-0.027^{***}$\\
0.001 &   $0.000^{***}$ & $0.006^{***}$ & $-0.621^{***}$ & $0.000^{***}$ & $0.109^{***}$ & $-0.453^{***}$ & $-0.243^{***}$ & $0.001^{***}$ & $-0.226^{***}$ & $-0.243^{***}$ & $0.075^{***}$ & $0.000^{***}$ & $-0.029^{***}$\\
0.005 &   $0.000^{***}$ & $0.051^{***}$ & $-0.045^{***}$ & $0.000^{***}$ & $0.099^{***}$ & $0.001^{***}$ & $-0.264^{***}$ & $-0.001^{***}$ & $0.114^{***}$ & $-0.092^{***}$ & $-0.007^{***}$ & $0.000^{***}$ & $0.000^{***}$\\
0.01 &   $0.000^{***}$ & $-0.273^{***}$ & $-0.343^{***}$ & $0.000^{***}$ & $0.454^{***}$ & $0.001^{***}$ & $-0.183^{***}$ & $-0.002^{***}$ & $0.107^{***}$ & $0.023^{***}$ & $-0.018^{***}$ & $0.000^{***}$ & $0.000^{***}$\\
0.05 &   $0.000^{***}$ & $-0.467^{***}$ & $0.03^{***}$ & $0.000^{***}$ & $1.304^{***}$ & $0.000^{***}$ & $-0.164^{***}$ & $-0.012^{***}$ & $-0.001^{***}$ & $0.011^{***}$ & $-0.002^{***}$ & $0.000^{***}$ & $0.000^{***}$\\
0.1 &   $0.000^{***}$ & $-0.316^{***}$ & $0.000^{***}$ & $0.000^{***}$ & $1.417^{***}$ & $0.000^{***}$ & $-0.381^{***}$ & $-0.001^{***}$ & $-0.001^{***}$ & $-0.013^{***}$ & $0.05^{***}$ & $0.000^{***}$ & $0.000^{***}$\\
0.25 &   $0.000^{***}$ & $-0.05^{***}$ & $0.279^{***}$ & $0.000^{***}$ & $-0.004^{***}$ & $0.000^{***}$ & $-0.033^{***}$ & $0.001^{***}$ & $0.000^{***}$ & $0.000^{***}$ & $-0.042^{***}$ & $0.000^{***}$ & $0.000^{***}$\\

\hline
\hline \\[-1.8ex]
\\[-1.8ex] 
& \multicolumn{13}{c}{Panel B: Log GRS FF3 Ratio, February 2010 - January 2020} \\ \cline{2-14}
\\[-1.8ex] 
IAF ($\theta_0$) & \multicolumn{1}{c}{BPZ$_F$} & \multicolumn{1}{c}{BPZ$_L$} & \multicolumn{1}{c}{BSV} & \multicolumn{1}{c}{CRW} & \multicolumn{1}{c}{DGU} & \multicolumn{1}{c}{FF3} & \multicolumn{1}{c}{FF6} & \multicolumn{1}{c}{GKX} & \multicolumn{1}{c}{HXZ} & \multicolumn{1}{c}{KNS} & \multicolumn{1}{c}{KPS} & \multicolumn{1}{c}{NN} & \multicolumn{1}{c}{RF}  \\
\\[-1.8ex] 
\hline \\[-1.8ex]

0.0001 &   $0.000^{***}$ & $0.054^{***}$ & $-0.244^{***}$ & $-0.019^{***}$ & $0.22^{***}$ & $-0.456^{***}$ & $-0.137^{***}$ & $0.009^{***}$ & $-0.051^{***}$ & $-0.011^{***}$ & $-0.011^{***}$ & $0.000^{***}$ & $-0.016^{***}$\\
0.0005 &   $0.000^{***}$ & $0.022^{***}$ & $-0.244^{***}$ & $0.001^{***}$ & $0.239^{***}$ & $-0.001^{***}$ & $-0.212^{***}$ & $-0.015^{***}$ & $0.089^{***}$ & $-0.345^{***}$ & $-0.011^{***}$ & $0.000^{***}$ & $-0.022^{***}$\\
0.001 &   $0.000^{***}$ & $-1.444^{***}$ & $-0.754^{***}$ & $0.000^{***}$ & $0.11^{***}$ & $-0.45^{***}$ & $-0.284^{***}$ & $0.001^{***}$ & $-0.223^{***}$ & $-0.225^{***}$ & $0.092^{***}$ & $0.000^{***}$ & $-0.021^{***}$\\
0.005 &   $0.000^{***}$ & $0.029^{***}$ & $-0.089^{***}$ & $0.000^{***}$ & $0.023^{***}$ & $0.000^{***}$ & $-0.294^{***}$ & $-0.001^{***}$ & $0.122^{***}$ & $-0.112^{***}$ & $-0.008^{***}$ & $0.000^{***}$ & $0.000^{***}$\\
0.01 &   $0.000^{***}$ & $-0.314^{***}$ & $-0.397^{***}$ & $0.000^{***}$ & $0.437^{***}$ & $0.001^{***}$ & $-0.212^{***}$ & $-0.002^{***}$ & $0.111^{***}$ & $0.006^{***}$ & $-0.018^{***}$ & $0.000^{***}$ & $0.000^{***}$\\
0.05 &   $0.000^{***}$ & $-0.527^{***}$ & $-2.041$ & $0.000^{***}$ & $1.332^{***}$ & $0.000^{***}$ & $-0.144^{***}$ & $-0.013^{***}$ & $-0.001^{***}$ & $0.01^{***}$ & $-0.004^{***}$ & $0.000^{***}$ & $0.000^{***}$\\
0.1 &   $0.000^{***}$ & $-0.348^{***}$ & $-0.11^{***}$ & $0.000^{***}$ & $1.382^{***}$ & $0.000^{***}$ & $-0.422^{***}$ & $0.000^{***}$ & $0.000^{***}$ & $-0.013^{***}$ & $0.05^{***}$ & $0.000^{***}$ & $0.000^{***}$\\
0.25 &   $0.000^{***}$ & $-0.123^{***}$ & $0.082^{***}$ & $0.000^{***}$ & $-0.005^{***}$ & $0.000^{***}$ & $-0.032^{***}$ & $0.001^{***}$ & $0.000^{***}$ & $0.000^{***}$ & $-0.042^{***}$ & $0.000^{***}$ & $0.000^{***}$\\

\hline
\hline \\[-1.8ex]
\\[-1.8ex] 
& \multicolumn{13}{c}{Panel C: Log GRS FF6 Ratio, February 2010 - January 2020} \\ \cline{2-14}
\\[-1.8ex] 
IAF ($\theta_0$) & \multicolumn{1}{c}{BPZ$_F$} & \multicolumn{1}{c}{BPZ$_L$} & \multicolumn{1}{c}{BSV} & \multicolumn{1}{c}{CRW} & \multicolumn{1}{c}{DGU} & \multicolumn{1}{c}{FF3} & \multicolumn{1}{c}{FF6} & \multicolumn{1}{c}{GKX} & \multicolumn{1}{c}{HXZ} & \multicolumn{1}{c}{KNS} & \multicolumn{1}{c}{KPS} & \multicolumn{1}{c}{NN} & \multicolumn{1}{c}{RF}  \\
\\[-1.8ex] 
\hline \\[-1.8ex]

0.0001 &   $0.001^{***}$ & $-0.009^{***}$ & $-0.266^{***}$ & $-0.02^{***}$ & $0.414^{***}$ & $-0.601^{***}$ & $-0.241^{***}$ & $0.012^{***}$ & $-0.136^{***}$ & $0.103^{***}$ & $0.001^{***}$ & $0.000^{***}$ & $-0.009^{***}$\\
0.0005 &   $0.000^{***}$ & $-0.043^{***}$ & $-0.273^{***}$ & $0.000^{***}$ & $0.3^{***}$ & $-0.001^{***}$ & $-0.269^{***}$ & $-0.011^{***}$ & $-0.132^{***}$ & $-0.338^{***}$ & $0.01^{***}$ & $0.000^{***}$ & $0.013^{***}$\\
0.001 &   $0.000^{***}$ & $-0.125^{***}$ & $-0.775^{***}$ & $0.000^{***}$ & $0.196^{***}$ & $-0.593^{***}$ & $-0.299^{***}$ & $0.001^{***}$ & $-0.185^{***}$ & $-0.242^{***}$ & $0.08^{***}$ & $0.000^{***}$ & $0.01^{***}$\\
0.005 &   $0.000^{***}$ & $0.07^{***}$ & $-0.167^{***}$ & $0.000^{***}$ & $0.411^{***}$ & $0.003^{***}$ & $-0.301^{***}$ & $0.001^{***}$ & $0.083^{***}$ & $-0.127^{***}$ & $-0.002^{***}$ & $0.000^{***}$ & $0.000^{***}$\\
0.01 &   $0.000^{***}$ & $-0.372^{***}$ & $-0.446^{***}$ & $0.000^{***}$ & $0.567^{***}$ & $0.004^{***}$ & $-0.223^{***}$ & $-0.002^{***}$ & $0.091^{***}$ & $-0.035^{***}$ & $-0.024^{***}$ & $0.000^{***}$ & $0.000^{***}$\\
0.05 &   $0.000^{***}$ & $-0.572^{***}$ & $-0.917^{***}$ & $0.000^{***}$ & $0.985^{***}$ & $-0.001^{***}$ & $-0.213^{***}$ & $-0.021^{***}$ & $-0.003^{***}$ & $0.003^{***}$ & $-0.006^{***}$ & $0.000^{***}$ & $0.000^{***}$\\
0.1 &   $0.000^{***}$ & $-0.388^{***}$ & $-0.143^{***}$ & $0.000^{***}$ & $1.379^{***}$ & $0.000^{***}$ & $-0.428^{***}$ & $-0.001^{***}$ & $-0.002^{***}$ & $-0.025^{***}$ & $0.046^{***}$ & $0.000^{***}$ & $0.000^{***}$\\
0.25 &   $0.000^{***}$ & $-0.168^{***}$ & $0.043^{***}$ & $0.000^{***}$ & $-0.011^{***}$ & $0.000^{***}$ & $-0.075^{***}$ & $0.000^{***}$ & $-0.001^{***}$ & $0.000^{***}$ & $-0.089^{***}$ & $0.000^{***}$ & $0.000^{***}$\\

\hline
\hline \\[-1.8ex]
\\[-1.8ex] 
& \multicolumn{13}{c}{Panel D: Log GRS HXZ Ratio, February 2010 - January 2020} \\ \cline{2-14}
\\[-1.8ex] 
IAF ($\theta_0$) & \multicolumn{1}{c}{BPZ$_F$} & \multicolumn{1}{c}{BPZ$_L$} & \multicolumn{1}{c}{BSV} & \multicolumn{1}{c}{CRW} & \multicolumn{1}{c}{DGU} & \multicolumn{1}{c}{FF3} & \multicolumn{1}{c}{FF6} & \multicolumn{1}{c}{GKX} & \multicolumn{1}{c}{HXZ} & \multicolumn{1}{c}{KNS} & \multicolumn{1}{c}{KPS} & \multicolumn{1}{c}{NN} & \multicolumn{1}{c}{RF}  \\
\\[-1.8ex] 
\hline \\[-1.8ex]

0.0001 &   $0.000^{***}$ & $0.068^{***}$ & $-0.253^{***}$ & $-0.021^{***}$ & $0.304^{***}$ & $-0.402^{***}$ & $-0.108^{***}$ & $0.01^{***}$ & $0.051^{***}$ & $0.028^{***}$ & $-0.01^{***}$ & $0.000^{***}$ & $-0.003^{***}$\\
0.0005 &   $0.000^{***}$ & $0.045^{***}$ & $-0.206^{***}$ & $0.001^{***}$ & $0.351^{***}$ & $0.000^{***}$ & $-0.187^{***}$ & $-0.005^{***}$ & $0.049^{***}$ & $-0.288^{***}$ & $0.022^{***}$ & $0.000^{***}$ & $-0.007^{***}$\\
0.001 &   $0.000^{***}$ & $-0.015^{***}$ & $-0.715^{***}$ & $0.000^{***}$ & $0.299^{***}$ & $-0.396^{***}$ & $-0.262^{***}$ & $0.001^{***}$ & $-0.051^{***}$ & $-0.244^{***}$ & $0.095^{***}$ & $0.000^{***}$ & $0.007^{***}$\\
0.005 &   $0.000^{***}$ & $0.018^{***}$ & $-0.083^{***}$ & $0.000^{***}$ & $0.239^{***}$ & $0.002^{***}$ & $-0.281^{***}$ & $-0.001^{***}$ & $0.064^{***}$ & $-0.089^{***}$ & $-0.004^{***}$ & $0.000^{***}$ & $0.000^{***}$\\
0.01 &   $0.000^{***}$ & $-0.294^{***}$ & $-0.368^{***}$ & $0.000^{***}$ & $0.53^{***}$ & $0.003^{***}$ & $-0.173^{***}$ & $0.000^{***}$ & $0.051^{***}$ & $-0.01^{***}$ & $-0.02^{***}$ & $0.000^{***}$ & $0.000^{***}$\\
0.05 &   $0.000^{***}$ & $-0.512^{***}$ & $-0.01^{***}$ & $0.000^{***}$ & $1.123^{***}$ & $0.000^{***}$ & $-0.172^{***}$ & $-0.003^{***}$ & $0.005^{***}$ & $0.001^{***}$ & $-0.006^{***}$ & $0.000^{***}$ & $0.000^{***}$\\
0.1 &   $0.000^{***}$ & $27.721^{***}$ & $-0.08^{***}$ & $0.000^{***}$ & $1.422^{***}$ & $0.000^{***}$ & $-0.386^{***}$ & $0.000^{***}$ & $0.001^{***}$ & $-0.013^{***}$ & $0.048^{***}$ & $0.000^{***}$ & $0.000^{***}$\\
0.25 &   $0.000^{***}$ & $-0.023^{***}$ & $-0.287^{***}$ & $0.000^{***}$ & $-0.006^{***}$ & $0.000^{***}$ & $-0.064^{***}$ & $0.000^{***}$ & $0.001^{***}$ & $0.000^{***}$ & $-0.02^{***}$ & $0.000^{***}$ & $0.000^{***}$\\

\hline
\hline 
 & \multicolumn{13}{r}{\textit{Note:} $^{*}$p$<$0.1; $^{**}$p$<$0.05; $^{***}$p$<$0.01} \\

\end{tabular}
}
    \vspace{4mm}
    \caption{\textbf{Recursive Experiments GRS.} This table is similar to Table \ref{tab:grs_baseline}, except the results represent the log ratio of the GRS statistics with the experiments where the statistical arbitrageur fraction is updated recursively. }
    \label{tab:grs_recursive}
\end{table}

\begin{table}[!t] \centering
    \resizebox{0.9\textwidth}{!}{\begin{tabular}{@{\extracolsep{5pt}}lccccccccccccc}

\\[-1.8ex]\hline
\hline \\[-1.8ex]
\\[-1.8ex] 
& \multicolumn{13}{c}{Panel A: Log CAPM GRS Ratio, February 2010 - January 2020} \\ \cline{2-14}
\\[-1.8ex] 
\\[-1.8ex] IAF ($\theta_0$) & \multicolumn{1}{c}{BPZ$_F$} & \multicolumn{1}{c}{BPZ$_L$} & \multicolumn{1}{c}{BSV} & \multicolumn{1}{c}{CRW} & \multicolumn{1}{c}{DGU} & \multicolumn{1}{c}{FF3} & \multicolumn{1}{c}{FF6} & \multicolumn{1}{c}{GKX} & \multicolumn{1}{c}{HXZ} & \multicolumn{1}{c}{KNS} & \multicolumn{1}{c}{KPS} & \multicolumn{1}{c}{NN} & \multicolumn{1}{c}{RF}  \\
\\[-1.8ex]
\hline \\[-1.8ex]

0.0001 &   $-0.1^{***}$ & $-0.25^{***}$ & $-0.281^{***}$ & $0.006^{***}$ & $-0.039^{***}$ & $0.004^{***}$ & $-0.059^{***}$ & $-0.203^{***}$ & $-0.03^{***}$ & $-0.151^{***}$ & $-0.002^{***}$ & $-0.02^{***}$ & $-0.01^{***}$\\
0.0005 &   $-0.001^{***}$ & $0.055^{***}$ & $0.338^{***}$ & $-0.028^{***}$ & $0.069^{***}$ & $-0.088^{***}$ & $0.3^{***}$ & $0.116^{***}$ & $-0.002^{***}$ & $0.136^{***}$ & $0.018^{***}$ & $-0.088^{***}$ & $0.001^{***}$\\
0.001 &   $0.000^{***}$ & $0.15^{***}$ & $-0.306^{***}$ & $-0.002^{***}$ & $0.199^{***}$ & $-0.255^{***}$ & $0.173^{***}$ & $0.01^{***}$ & $0.029^{***}$ & $1.061^{***}$ & $0.052^{***}$ & $0.001^{***}$ & $0.002^{***}$\\
0.005 &   $-0.001^{***}$ & $0.002^{***}$ & $0.002^{***}$ & $-0.003^{***}$ & $1.529^{***}$ & $0.092^{***}$ & $0.153^{***}$ & $0.004^{***}$ & $0.162^{***}$ & $-0.033^{***}$ & $-0.29^{***}$ & $0.001^{***}$ & $-0.16^{***}$\\
0.01 &   $0.004^{***}$ & $0.004^{***}$ & $0.001^{***}$ & $0.007^{***}$ & $2.008^{***}$ & $-0.224^{***}$ & $-0.026^{***}$ & $0.004^{***}$ & $-0.056^{***}$ & $0.088^{***}$ & $0.093^{***}$ & $0.003^{***}$ & $-0.203^{***}$\\
0.05 &   $0.023^{***}$ & $0.023^{***}$ & $0.429^{***}$ & $0.024^{***}$ & $2.164^{***}$ & $0.032^{***}$ & $0.003^{***}$ & $0.023^{***}$ & $-0.032^{***}$ & $0.031^{***}$ & $0.076^{***}$ & $0.023^{***}$ & $0.023^{***}$\\
0.1 &   $0.029^{***}$ & $0.331^{***}$ & $0.502^{***}$ & $0.041^{***}$ & $2.183^{***}$ & $0.011^{***}$ & $0.022^{***}$ & $0.029^{***}$ & $-0.234^{***}$ & $0.326^{***}$ & $0.15^{***}$ & $0.028^{***}$ & $0.028^{***}$\\
0.25 &   $-0.024^{***}$ & $-0.021^{***}$ & $1.341^{***}$ & $-0.023^{***}$ & $2.13^{***}$ & $-0.218^{***}$ & $0.031^{***}$ & $0.018^{***}$ & $-0.031^{***}$ & $0.823^{***}$ & $-0.294^{***}$ & $-0.018^{***}$ & $-0.021^{***}$\\

\hline
\hline \\[-1.8ex]
\\[-1.8ex] 
& \multicolumn{13}{c}{Panel B: Log GRS FF3 Ratio, February 2010 - January 2020} \\ \cline{2-14}
\\[-1.8ex] 
IAF ($\theta_0$) & \multicolumn{1}{c}{BPZ$_F$} & \multicolumn{1}{c}{BPZ$_L$} & \multicolumn{1}{c}{BSV} & \multicolumn{1}{c}{CRW} & \multicolumn{1}{c}{DGU} & \multicolumn{1}{c}{FF3} & \multicolumn{1}{c}{FF6} & \multicolumn{1}{c}{GKX} & \multicolumn{1}{c}{HXZ} & \multicolumn{1}{c}{KNS} & \multicolumn{1}{c}{KPS} & \multicolumn{1}{c}{NN} & \multicolumn{1}{c}{RF}  \\
\\[-1.8ex] 
\hline \\[-1.8ex]

0.0001 &   $-0.104^{***}$ & $-0.218^{***}$ & $-0.285^{***}$ & $0.006^{***}$ & $-0.038^{***}$ & $-0.003^{***}$ & $-0.079^{***}$ & $0.039^{***}$ & $-0.007^{***}$ & $-0.174^{***}$ & $-0.003^{***}$ & $-0.026^{***}$ & $-0.01^{***}$\\
0.0005 &   $-0.001^{***}$ & $-0.003^{***}$ & $0.295^{***}$ & $-0.028^{***}$ & $0.06^{***}$ & $-0.101^{***}$ & $0.247^{***}$ & $0.108^{***}$ & $0.026^{***}$ & $0.087^{***}$ & $0.022^{***}$ & $-0.08^{***}$ & $0.001^{***}$\\
0.001 &   $0.000^{***}$ & $-0.685^{***}$ & $-0.331^{***}$ & $-0.002^{***}$ & $0.219^{***}$ & $-0.27^{***}$ & $0.17^{***}$ & $0.007^{***}$ & $0.032^{***}$ & $1.116^{***}$ & $0.067^{***}$ & $0.001^{***}$ & $0.002^{***}$\\
0.005 &   $0.000^{***}$ & $0.003^{***}$ & $0.005^{***}$ & $-0.002^{***}$ & $1.417^{***}$ & $0.127^{***}$ & $0.13^{***}$ & $0.005^{***}$ & $0.133^{***}$ & $-0.032^{***}$ & $-0.317^{***}$ & $0.001^{***}$ & $-0.158^{***}$\\
0.01 &   $0.006^{***}$ & $0.005^{***}$ & $0.001^{***}$ & $0.008^{***}$ & $1.726^{***}$ & $-2.122$ & $12.305^{***}$ & $0.005^{***}$ & $-0.03^{***}$ & $0.093^{***}$ & $0.127^{***}$ & $0.004^{***}$ & $-0.184^{***}$\\
0.05 &   $0.029^{***}$ & $0.029^{***}$ & $0.387^{***}$ & $0.029^{***}$ & $1.786^{***}$ & $0.046^{***}$ & $0.017^{***}$ & $0.029^{***}$ & $-0.062^{***}$ & $0.038^{***}$ & $0.061^{***}$ & $0.029^{***}$ & $0.029^{***}$\\
0.1 &   $0.037^{***}$ & $-0.16^{***}$ & $0.456^{***}$ & $0.049^{***}$ & $1.532^{***}$ & $0.01^{***}$ & $0.024^{***}$ & $0.037^{***}$ & $-0.273^{***}$ & $0.303^{***}$ & $0.171^{***}$ & $0.036^{***}$ & $0.036^{***}$\\
0.25 &   $-0.011^{***}$ & $-0.007^{***}$ & $1.301^{***}$ & $-0.009^{***}$ & $1.08^{***}$ & $-0.271^{***}$ & $0.053^{***}$ & $0.03^{***}$ & $-0.006^{***}$ & $0.784^{***}$ & $-0.349^{***}$ & $-0.004^{***}$ & $-0.007^{***}$\\

\hline
\hline \\[-1.8ex]
\\[-1.8ex] 
& \multicolumn{13}{c}{Panel C: Log GRS FF6 Ratio, February 2010 - January 2020} \\ \cline{2-14}
\\[-1.8ex] 
IAF ($\theta_0$) & \multicolumn{1}{c}{BPZ$_F$} & \multicolumn{1}{c}{BPZ$_L$} & \multicolumn{1}{c}{BSV} & \multicolumn{1}{c}{CRW} & \multicolumn{1}{c}{DGU} & \multicolumn{1}{c}{FF3} & \multicolumn{1}{c}{FF6} & \multicolumn{1}{c}{GKX} & \multicolumn{1}{c}{HXZ} & \multicolumn{1}{c}{KNS} & \multicolumn{1}{c}{KPS} & \multicolumn{1}{c}{NN} & \multicolumn{1}{c}{RF}  \\
\\[-1.8ex] 
\hline \\[-1.8ex]

0.0001 &   $-0.195^{***}$ & $-0.158^{***}$ & $-0.191^{***}$ & $0.006^{***}$ & $0.000^{***}$ & $-0.001^{***}$ & $-0.026^{***}$ & $0.026^{***}$ & $-0.031^{***}$ & $-0.14^{***}$ & $-0.002^{***}$ & $-0.019^{***}$ & $-0.017^{***}$\\
0.0005 &   $0.000^{***}$ & $0.066^{***}$ & $0.227^{***}$ & $-0.029^{***}$ & $0.281^{***}$ & $-0.171^{***}$ & $0.264^{***}$ & $0.099^{***}$ & $-0.018^{***}$ & $0.038^{***}$ & $0.01^{***}$ & $-0.049^{***}$ & $-0.003^{***}$\\
0.001 &   $-0.001^{***}$ & $0.009^{***}$ & $-0.342^{***}$ & $-0.006^{***}$ & $0.314^{***}$ & $-0.255^{***}$ & $0.14^{***}$ & $-0.006^{***}$ & $-0.031^{***}$ & $1.1^{***}$ & $0.073^{***}$ & $0.001^{***}$ & $0.006^{***}$\\
0.005 &   $0.000^{***}$ & $0.003^{***}$ & $-0.011^{***}$ & $-0.004^{***}$ & $1.059^{***}$ & $-0.029^{***}$ & $0.172^{***}$ & $0.01^{***}$ & $-0.066^{***}$ & $-0.023^{***}$ & $-0.509^{***}$ & $0.003^{***}$ & $-0.232^{***}$\\
0.01 &   $0.006^{***}$ & $0.005^{***}$ & $0.007^{***}$ & $0.006^{***}$ & $1.239^{***}$ & $-0.26^{***}$ & $-0.084^{***}$ & $0.006^{***}$ & $-0.138^{***}$ & $0.069^{***}$ & $0.156^{***}$ & $0.005^{***}$ & $-0.321^{***}$\\
0.05 &   $0.032^{***}$ & $0.033^{***}$ & $0.27^{***}$ & $0.036^{***}$ & $1.161^{***}$ & $-0.016^{***}$ & $-0.015^{***}$ & $0.033^{***}$ & $-0.157^{***}$ & $0.059^{***}$ & $0.183^{***}$ & $0.036^{***}$ & $0.033^{***}$\\
0.1 &   $0.058^{***}$ & $-0.245^{***}$ & $0.426^{***}$ & $0.073^{***}$ & $0.989^{***}$ & $0.02^{***}$ & $0.043^{***}$ & $0.057^{***}$ & $-0.354^{***}$ & $0.244^{***}$ & $0.134^{***}$ & $0.058^{***}$ & $0.057^{***}$\\
0.25 &   $0.015^{***}$ & $0.018^{***}$ & $1.407^{***}$ & $0.013^{***}$ & $0.496^{***}$ & $-0.319^{***}$ & $0.027^{***}$ & $0.058^{***}$ & $-0.13^{***}$ & $0.739^{***}$ & $-0.334^{***}$ & $0.022^{***}$ & $0.019^{***}$\\

\hline
\hline \\[-1.8ex]
\\[-1.8ex] 
& \multicolumn{13}{c}{Panel D: Log GRS HXZ Ratio, February 2010 - January 2020} \\ \cline{2-14}
\\[-1.8ex] 
IAF ($\theta_0$) & \multicolumn{1}{c}{BPZ$_F$} & \multicolumn{1}{c}{BPZ$_L$} & \multicolumn{1}{c}{BSV} & \multicolumn{1}{c}{CRW} & \multicolumn{1}{c}{DGU} & \multicolumn{1}{c}{FF3} & \multicolumn{1}{c}{FF6} & \multicolumn{1}{c}{GKX} & \multicolumn{1}{c}{HXZ} & \multicolumn{1}{c}{KNS} & \multicolumn{1}{c}{KPS} & \multicolumn{1}{c}{NN} & \multicolumn{1}{c}{RF}  \\
\\[-1.8ex] 
\hline \\[-1.8ex]

0.0001 &   $-0.106^{***}$ & $-0.199^{***}$ & $-0.269^{***}$ & $0.006^{***}$ & $-0.011^{***}$ & $0.013^{***}$ & $-0.095^{***}$ & $-0.134^{***}$ & $0.077^{***}$ & $-0.156^{***}$ & $-0.001^{***}$ & $0.007^{***}$ & $-0.008^{***}$\\
0.0005 &   $-0.002^{***}$ & $0.015^{***}$ & $0.3^{***}$ & $-0.033^{***}$ & $0.085^{***}$ & $-0.059^{***}$ & $0.25^{***}$ & $0.091^{***}$ & $0.122^{***}$ & $0.101^{***}$ & $0.006^{***}$ & $-0.07^{***}$ & $0.004^{***}$\\
0.001 &   $0.000^{***}$ & $0.138^{***}$ & $-0.369^{***}$ & $-0.006^{***}$ & $0.253^{***}$ & $-0.194^{***}$ & $0.135^{***}$ & $0.011^{***}$ & $0.132^{***}$ & $0.983^{***}$ & $0.062^{***}$ & $0.001^{***}$ & $0.003^{***}$\\
0.005 &   $-0.002^{***}$ & $0.000^{***}$ & $0.014^{***}$ & $-0.003^{***}$ & $1.354^{***}$ & $0.199^{***}$ & $0.135^{***}$ & $0.004^{***}$ & $0.068^{***}$ & $-0.025^{***}$ & $-0.197^{***}$ & $-0.001^{***}$ & $-0.092^{***}$\\
0.01 &   $0.001^{***}$ & $0.000^{***}$ & $-0.004^{***}$ & $0.000^{***}$ & $1.613^{***}$ & $-0.255^{***}$ & $-0.047^{***}$ & $0.000^{***}$ & $0.091^{***}$ & $0.076^{***}$ & $0.125^{***}$ & $-0.001^{***}$ & $-0.24^{***}$\\
0.05 &   $0.003^{***}$ & $0.003^{***}$ & $0.381^{***}$ & $0.004^{***}$ & $1.612^{***}$ & $0.042^{***}$ & $-0.005^{***}$ & $0.003^{***}$ & $-0.054^{***}$ & $-0.02^{***}$ & $0.025^{***}$ & $0.003^{***}$ & $0.003^{***}$\\
0.1 &   $0.000^{***}$ & $-0.178^{***}$ & $0.481^{***}$ & $0.007^{***}$ & $1.458^{***}$ & $0.012^{***}$ & $0.009^{***}$ & $0.001^{***}$ & $-0.178^{***}$ & $0.271^{***}$ & $0.11^{***}$ & $0.000^{***}$ & $-0.001^{***}$\\
0.25 &   $-0.04^{***}$ & $-0.04^{***}$ & $1.22^{***}$ & $-0.042^{***}$ & $1.089^{***}$ & $-0.262^{***}$ & $0.014^{***}$ & $0.008^{***}$ & $0.036^{***}$ & $0.819^{***}$ & $-0.363^{***}$ & $-0.038^{***}$ & $-0.04^{***}$\\

\hline
\hline 
 & \multicolumn{13}{r}{\textit{Note:} $^{*}$p$<$0.1; $^{**}$p$<$0.05; $^{***}$p$<$0.01} \\

\end{tabular}
}
    \vspace{4mm}
    \caption{\textbf{Double Elasticity Experiments GRS.} This table is similar to Table \ref{tab:grs_recursive}, except the elasticity is doubled for the non-HSA investors in the counterfactual exercises. }
    \label{tab:grs_elasmult2}
\end{table}

\end{document}